\documentclass[preprint,aos]{imsart}

\RequirePackage{amsthm,amsmath,amsfonts,amssymb}
\RequirePackage[numbers,sort&compress]{natbib}
\RequirePackage[colorlinks,citecolor=blue,urlcolor=blue]{hyperref}
\RequirePackage{graphicx}

\startlocaldefs
\usepackage{verbatim}
\usepackage{mathrsfs}
\usepackage{mathtools}
\usepackage{caption}
\usepackage{subcaption}
\usepackage{url}
\usepackage{epstopdf}
\usepackage{pgf,tikz}
\usepackage{bbm}
\usetikzlibrary{arrows}
\usepackage{xcolor}
\usepackage{float}

\usepackage{tikz}                                   
\usetikzlibrary{shapes.geometric,calc,arrows,decorations.markings,decorations.pathmorphing}
\usepackage{pgfplots}

\usepackage{dsfont}
\usepackage{listings}
\usepackage{bbm}
\usepackage{color}
\usepackage{enumerate}
\usepackage{todonotes}
\usepackage{algorithm}
\usepackage{algpseudocode}

\usepackage{stmaryrd}
\usepackage{kantlipsum} 
\usepackage{fancyhdr} 

\renewcommand{\Delta}{\triangle}

\definecolor{darkblue}{rgb}{0,0,0.7}

\definecolor{darkgreen}{rgb}{0.01,0.75,0.24}

\def \Ee[#1]{\mathcal{E}^{\text{{#1}}}}
\def\R{\mathbf{R}}

\def\pa[#1,#2]{\frac{\partial {#1}}{\partial {#2}} }

\def\idom[#1,#2,#3]{\int_{#1}\hspace{1pt} {#2} \hspace{1pt} \text{d}{#3}}
\def\res[#1,#2]{\left.{#1}\right|_{#2}}

\def\var[#1,#2]{\langle \delta \mathcal{E}^{\text{{#1}}}({#2}),v\rangle}
\def\vars[#1,#2,#3]{\langle \delta^2\mathcal{E}^{\text{{#1}}}({#2})v,{#3}\rangle}
\def\vard[#1,#2,#3,#4]{\langle \delta\mathcal{E}^{\text{{#1}}}({#2})-\delta\mathcal{E}^{\text{{#3}}}({#4}),v\rangle}

\renewcommand{\O}{\mathcal{O}}
\def\E{\E}

\def\Ind{\mathbbm{1}}

\newcommand{\I}{\mathbb{I}}
\newcommand{\C}{\mathcal{C}}
\newcommand{\UBUBU}{\text{UBUBU}}
\newcommand{\UBUBUSG}{\text{UBUBU-SG}}

\newcommand{\UBU}{\text{UBU}}
\newcommand{\BAOAB}{\text{BAOAB}}
\newcommand{\SWR}{\mathcal{SWR}}

\renewcommand{\R}{\mathbb{R}}
\newcommand{\M}{\mathcal{M}}
\newcommand{\PP}{\mathcal{P}}
\newcommand{\HH}{\mathcal{H}}

\newcommand{\ol}{\overline}
\newcommand{\ul}{\underline}

\newcommand{\be}{\begin{equation}}
\newcommand{\en}{\end{equation}}
\newcommand{\ben}{\begin{equation*}}
\newcommand{\enn}{\end{equation*}}
\newcommand{\bea}{\begin{aligned}}
\newcommand{\ena}{\end{aligned}}

\def\ba#1\ena{\begin{align}#1\end{align}}
\def\ban#1\enan{\begin{align*}#1\end{align*}}

\usepackage{thmtools}
\usepackage{thm-restate}
\theoremstyle{plain}
\newtheorem{theorem}{Theorem}[section]

\newtheorem{lemma}[theorem]{Lemma}
\newtheorem{corollary}[theorem]{Corollary}
\newtheorem{proposition}[theorem]{Proposition}
\newtheorem{assumption}[theorem]{Assumption}
\newtheorem{example}[theorem]{Example}
\newtheorem{remark}[theorem]{Remark}
\newtheorem{definition}[theorem]{Definition}
\numberwithin{equation}{section}

\newcommand{\Var}{\mathrm{Var}}
\newcommand{\Cov}{\mathrm{Cov}}
\renewcommand{\E}{\mathbb{E}}

\endlocaldefs

\begin{document}

\begin{frontmatter}
\title{Unbiased Kinetic Langevin Monte Carlo \\ with  Inexact gradients}
\runtitle{Unbiased kinetic Langevin Monte Carlo with Inexact gradients}

\begin{aug}
\author[A]{\fnms{Neil K.}~\snm{Chada}\ead[label=e1]{neilchada123@gmail.com}},
\author[B]{\fnms{Benedict}~\snm{Leimkuhler}\ead[label=e2]{b.leimkuhler@ed.ac.uk}}, 
\author[C]{\fnms{Daniel}~\snm{Paulin}\ead[label=e3]{paulindani@gmail.com.}}\\
\and
\author[D]{\fnms{Peter A.}~\snm{Whalley}\ead[label=e4]{peter.a.whalley@gmail.com}}
\address[A]{Department of Mathematics, City  University of Hong Kong\printead[presep={,\ }]{e1}}

\address[B]{School of Mathematics, University of Edinburgh\printead[presep={,\ }]{e2}}

\address[C]{School of Physical and Mathematical Sciences, Nanyang Technological University\printead[presep={,\ }]{e3}}

\address[D]{Seminar for Statistics, ETH Zurich \printead[presep={,\ }]{e4}}
\end{aug}

\begin{abstract}

We present an unbiased method for Bayesian posterior means based on kinetic Langevin dynamics that combines advanced splitting methods with enhanced gradient approximations. Our approach avoids Metropolis correction by coupling Markov chains at different discretization levels in a multilevel Monte Carlo approach. Theoretical analysis demonstrates that our proposed estimator is unbiased,  attains finite variance, and satisfies a central limit theorem. It can achieve accuracy $\epsilon>0$ for estimating expectations of Lipschitz functions in $d$ dimensions with $\O(d^{1/4}\epsilon^{-2})$ expected gradient evaluations, without assuming warm start.   We exhibit similar bounds using both approximate and stochastic gradients, and our method's computational cost is shown to scale independently of the size of the dataset.
The proposed method is tested using a multinomial regression problem on the MNIST dataset and a Poisson regression model for soccer scores. Experiments indicate that the number of gradient evaluations per effective sample is independent of dimension, even when using inexact gradients. For product distributions, we give dimension-independent variance bounds. Our results demonstrate that in large-scale applications, the unbiased algorithm we present can be 2-3 orders of magnitude more efficient than the ``gold-standard" randomized Hamiltonian Monte Carlo.
\end{abstract}



\end{frontmatter}

\tableofcontents

\section{Introduction}
Markov chain Monte Carlo (MCMC) methods are standard computational tools for high-dimensional Bayesian inference \cite{Roberts2004}. They enable the computation of posterior means and variances and other observable averages by replacing ensemble calculations with Monte Carlo sums over discrete Markov processes.
A limitation to the broader uptake of Bayesian inference is the scaling of the computational cost of MCMC algorithms with model dimension and dataset size.
Typical MCMC methods (Metropolis Adjusted Langevin Algorithm \citep{Besag1994,roberts1998optimal}, Hamiltonian Monte Carlo \citep{duane1987hybrid, neal2011mcmc}) employ Metropolis-Hastings correction steps to ensure convergence to the desired invariant distribution. The cost of implementing such corrections scales linearly with dataset size. Even worse, in order to maintain a high acceptance rate, stepsizes must decrease as a function of the model dimension, which implies that convergence rates are also dependent on dimension  \cite{roberts2001optimal, Beskos2013, chen2023does}.   

By contrast, optimization methods typically have convergence rates that are independent of the dimension and can make use of stochastic gradients based on a subset of the data instead of the entire dataset \cite{johnson2013accelerating}. For these reasons, optimization algorithms are much more scalable than sampling methods, so practitioners often prefer machine-learning approaches.   The relative inefficiency of sampling compared to optimization also limits the uptake of uncertainty quantification techniques (typically built on a Bayesian foundation) in high-dimensional machine learning applications.

\subsection{Unbiased estimation without accept/reject steps}
This paper describes a technique for performing Bayesian inference based on unbiased unadjusted Markov chain Monte Carlo that does not rely on Metropolis-Hastings accept/reject steps. 
Our algorithm is based on a multilevel scheme \cite{MBG15} that combines several different unadjusted MCMC chains to eliminate bias efficiently.   Our approach is related to a recent paper \cite{ruzayqat2022unbiased} that introduced an unbiased unadjusted MCMC method, however we employ state-of-the-art integrators, and we extend the method with modifications for handling incomplete (or approximate) gradients, thus obtaining a procedure with improved scalability and competitiveness compared to state-of-the-art algorithms such as randomized Hamiltonian Monte Carlo (RHMC) \cite{RHMC,chen2023does}. {We also provide theoretical guarantees that are explicit in key parameters such as dimension, enabling direct comparison with state-of-the-art algorithms.}

Unbiased Monte Carlo methods have been widely studied in the recent literature; see Section 2.1 of \cite{jacob2020unbiased} for an overview. The goal of the methods of \cite{glynn2014exact,rhee2015unbiased, jacob2020unbiased, heng2019unbiased, corenflos2022unbiased} is to remove burn-in bias via couplings. \cite{kahale2022unbiased} proposed an alternative method for eliminating burn-in bias by considering a burn-in period of random length. The cited papers above all require that the stationary distribution of the Markov chain has no bias (hence, these methods typically involve Metropolization) and are not able to remove discretization bias in SDEs such as \eqref{eq:kinetic_langevin} treated using numerical methods. Middleton et al \cite{middleton2020unbiased}  extended unbiased methods to intractable likelihoods, and \cite{Douc2022Poission} created unbiased estimators of MCMC asymptotic variances. 
 
There have been several proposals for creating computationally efficient estimators for functions of SDE paths based on numerical discretization using multilevel Monte Carlo variance reduction techniques.  Our scheme relates to the method of M\"{u}ller et al \cite{muller2015improving} for approximating functions of whole paths of kinetic Langevin dynamics using integrators based on splitting. Unlike our approach, that work did not address the stationary distribution; moreover,  the burn-in bias was not eliminated, and they did not consider the incorporation of approximate or stochastic gradients. More recently, Giles et al \cite{giles2020multi} introduced a general framework for multilevel approximation of expectations with respect to the stationary distribution of overdamped Langevin dynamics and also considered stochastic gradients. However, their approach does not produce unbiased {estimates}, and overdamped Langevin dynamics generally appears less efficient at exploring distributions with high condition numbers than well-tuned kinetic Langevin dynamics \cite{GP14}, as considered here.  Until this work, multilevel approaches have not been shown to be competitive with Hamiltonian Monte Carlo methods for high-dimensional {Bayesian computation}.

We also mention that, in the area of molecular simulation, unadjusted numerical discretizations of kinetic Langevin dynamics have been employed for sampling from complex distributions for many years \cite{brunger1984stochastic, izaguirre2001langevin,leimkuhler2013rational,leimkuhler2015molecular}.
Even though such discretizations introduce bias, this is often dominated by the Monte Carlo error--even at substantially larger stepsizes than would typically be used in Metropolized calculations \cite{leimkuhler2015molecular}.  On the other hand, the magnitude of the sampling bias due to finite stepsize 
is problem-dependent and can be difficult to quantify; thus, there are situations where the ability to ameliorate
the discretization bias is crucial.  Some authors have proposed reducing the discretization bias by decreasing stepsize asymptotically \cite{welling2011bayesian, durmus2017nonasymptotic}. However, such a procedure can slow convergence or introduce heuristic schedules into the sampling apparatus.

\subsection{Proposed methodology}
We consider kinetic Langevin dynamics (also referred to as underdamped Langevin dynamics \cite{dalalyan2020sampling, cheng2018underdamped}):
\begin{equation}\label{eq:kinetic_langevin}
    \begin{split}
    dX_{t} &= V_{t}dt,\\
    dV_{t} &= -\nabla U(X_{t}) dt - \gamma V_{t}dt + \sqrt{2\gamma}dW_{t},
\end{split}
\end{equation}
where $U:\R^d\to \R$ is a potential energy function, $\{W_t\}_{t\geq 0}$ is a standard $d-$dimensional Brownian motion, and $\gamma>0$ is a friction coefficient. {One may also consider a general matrix representation for the friction parameter $\gamma>0$, which can accelerate convergence (see \cite{Chak2023}).} Under fairly weak assumptions, the unique invariant measure of the process $\{X_t,V_t\}_{t\geq 0}$ is of the form
\begin{equation}
\label{eq:inv}
\pi(dx dv) \propto \exp\left(-U(x)-\frac{ \|v\|^2}{2}\right) dx dv.
\end{equation}

This dynamics forms the basis of many sampling methods \citep{brunger1984stochastic, leimkuhler2023contractiona}, and it has a dimension-independent convergence rate for a large class of distributions \cite{lu2022explicit}. {However, in practice, one needs to discretize \eqref{eq:kinetic_langevin}, which introduces a bias in the invariant measure \cite{dalalyan2020sampling}. Typically, this bias is either ignored or corrected using Metropolization \cite{PM21}.}
In this paper, we develop a comprehensive and practical framework for unbiased estimation {which avoids Metropolis adjustment}. {We focus on} a splitting integrator called $\UBU$ \cite{sanz2021wasserstein}, which is strongly second-order accurate, where the unbiased estimator we introduce is referred to as $\UBUBU$ (Unbiased-UBU). 
 \begin{figure}
 \centering
 \includegraphics[width=0.49\linewidth]{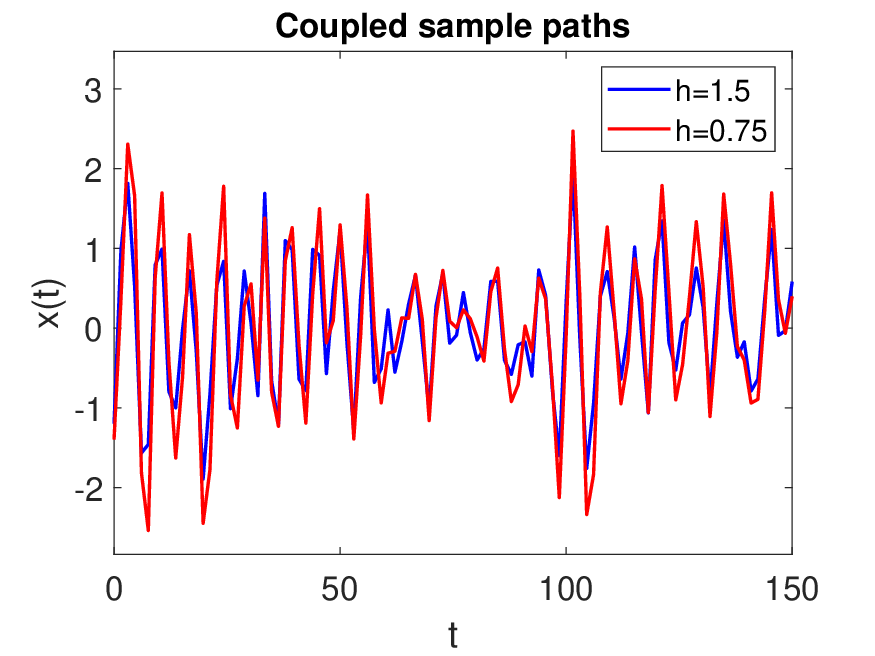}
 \includegraphics[width=0.49\linewidth]{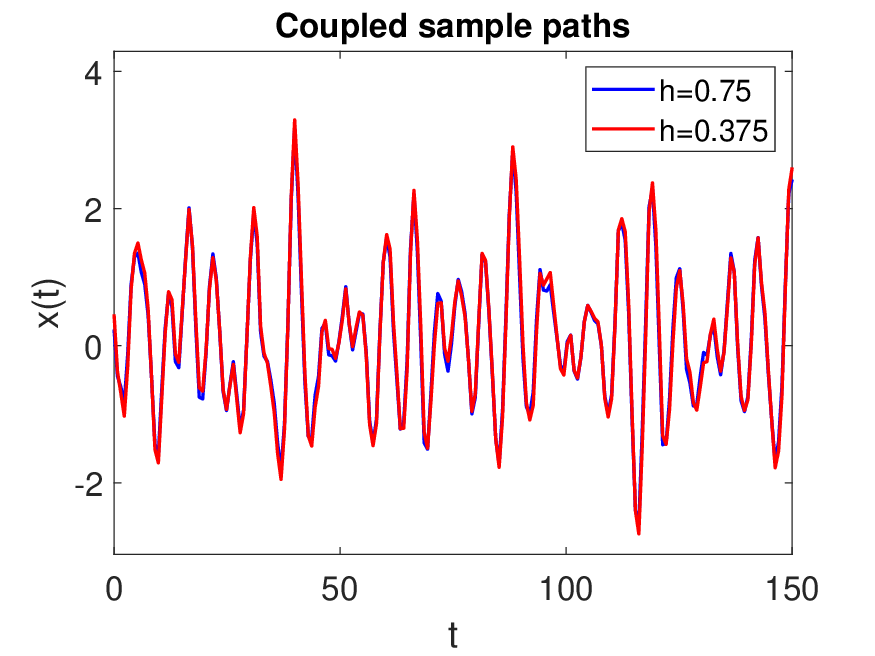}
 \caption{Coupled sample paths based on synchronous coupling from $\UBU$ (Section \ref{sec:back}) discretization scheme of kinetic Langevin diffusion for a Gaussian target at stepsizes $h=1.5,0.75$ and $h=0.75, 0.375$. $\UBU$ is strong order 2, so the typical distance between coupled paths is $\O(h^2)$.}
 \label{fig:coupled_paths}
 \end{figure}
 
In Figure \ref{fig:coupled_paths} we see that $\UBU$ discretization can be pathwise accurate even at large stepsize.  Nevertheless, there is always some residual bias, and the stationary distribution of the discretization with stepsize $h$, $\pi_h$, differs from the target distribution $\pi$.
{\cite{ruzayqat2022unbiased} proposed an unbiased estimation method which} considers a sequence of discretization levels $h_l=2^{-l}h_0$ for $l=0,1,2,\ldots$  and creates an estimator of the form
\begin{equation}
\label{eq:summ}
\hat{\pi}(f)=\hat{\pi}_{h_0}(f)+\sum_{l=0}^{\infty} \hat{\pi}_{h_{l+1},h_{l}}(f),
\end{equation}
where $f$ is some arbitrary quantity of interest, $\hat{\pi}_{h_{0}}(f)$ is an unbiased estimator of $\pi_{h_0}(f)$, and $\hat{\pi}_{h_{l+1},h_{l}}(f)$ is an unbiased estimator of  $\pi_{h_{l+1}}(f)-\pi_{h_{l}}(f)$. A sophisticated coupling construction was used for defining $\hat{\pi}_{h_{l+1},h_{l}}(f)$ based on four Markov chains using Euler–Maruyama discretization of \eqref{eq:kinetic_langevin}. Under certain weak assumptions, the estimator \eqref{eq:summ} was shown to have no bias, finite variance and finite expected computational cost. 

\subsection{{Our Contributions}}
{This paper presents UBUBU, an unbiased estimator of Bayesian posterior means for high-dimensional settings. The method combines an advanced splitting scheme, UBU, with a telescoping sum expansion \eqref{eq:summ} that is motivated by multilevel Monte Carlo. The benefits associated with our methodology are listed below.
\begin{enumerate}[(i)]
\item The burn-in bias is eliminated differently than e.g. in \cite{ruzayqat2022unbiased}, resulting in simpler couplings. Our estimator is still of the form \eqref{eq:summ}.
However, instead of estimating $\pi_{h_0}(f)$ and $\pi_{h_{l+1}}(f)-\pi_{h_{l}}(f)$, which requires eliminating the burn-in bias for both discretization levels, we let $\hat{\pi}_{h_0}(f)$ be an unbiased estimator of $\tilde{\pi}_{h_0}(f)$, and $\hat{\pi}_{h_{l+1},h_{l}}(f)$ be an unbiased estimator $\tilde{\pi}_{h_{l+1}}(f)-\tilde{\pi}_{h_{l}}(f)$. Here $\tilde{\pi}_{h_{l}}(f)$ denotes the expected value of $f$ according to the empirical distribution of a Markov chain using discretization stepsize $h_l$, thinning $2^l$, and burn-in period of length $(B_0+l\cdot B)/h_l$, for some constants $B_0, B>0$. See Figure \ref{fig:bias.elimination} for an illustration. Due to the increasing burn-in periods at smaller stepsizes, the bias of $\tilde{\pi}_{h_{l}}(f)$ shrinks to zero as $l\to \infty$. With this approach, we only need to couple two chains for creating unbiased estimators of $\tilde{\pi}_{h_{l+1}}(f)-\tilde{\pi}_{h_{l}}(f)$, and simple synchronous couplings can be used. 
\item  In our method, the number of samples per level is deterministic (except at very small stepsize), and we can use Richardson extrapolation \cite{richardson1911approximate} to further lower the variance.
\item We show unbiasedness and finite variance even when using approximate or stochastic gradients. This dramatically improves the scalability of our method to large datasets. We also prove that our unbiased estimator has computational cost that scales independently of the size of the dataset using approximate and stochastic gradients.
\item We provide a theoretical comparison between UBUBU and other well-known sampling methods, in terms of the number of gradient evaluations per effective sample size. We are able to prove that the computational complexity is state-of-the-art, with much less restrictive assumptions. A summary is provided in Table \ref{table:comp}.  
\item On a variety of applications, we illustrate in numerical experiments the considerable advantage of our unbiased estimator over state-of-the-art MCMC methods for Bayesian computation. These significant computational savings in combination with our theoretical guarantees illustrate for the first time that unbiased estimation can be a powerful alternative to Metropolis correction.
\end{enumerate}}
\begin{table}[h!]
\begin{center}
\begin{tabular}{ |c|c|c|c| } 
\hline
\textbf{Algorithm} & \textbf{Gradient Evaluations} & \textbf{Conditions} & \textbf{Reference} \\
\hline
MALA & $\O(d^{3/7})$ & $h=\O(d^{-3/7})$, \ warm start, {strongly Hessian Lipschitz} &\cite{chen2023does} \\ 
HMC  & $\O(d^{1/4})$ & $h=\O(d^{-1/4})$, \ warm start, {strongly Hessian Lipschitz}  & \cite{chen2023does} \\ 
RHMC  & $\O(d^{1/4})$ & $h=\O(d^{-1/4})$, \ warm start, Gaussian target  &  \cite{apers2022hamiltonian} \\ 
\UBUBU & $\O(d^{1/4})$ & {$h_0=\O(d^{-1/4})$}, \ {strongly Hessian Lipschitz}   &  {this work}\\ 
\hline
\end{tabular}
\end{center}
\caption{Dimension dependency of gradient evaluations per effective sample for different algorithms for $m$-strongly convex and $M$-$\nabla$Lipschitz potentials, in comparison to $\UBUBU$.}
\label{table:comp}
\end{table}

\begin{figure}
 \centering
 \includegraphics[width=\linewidth]{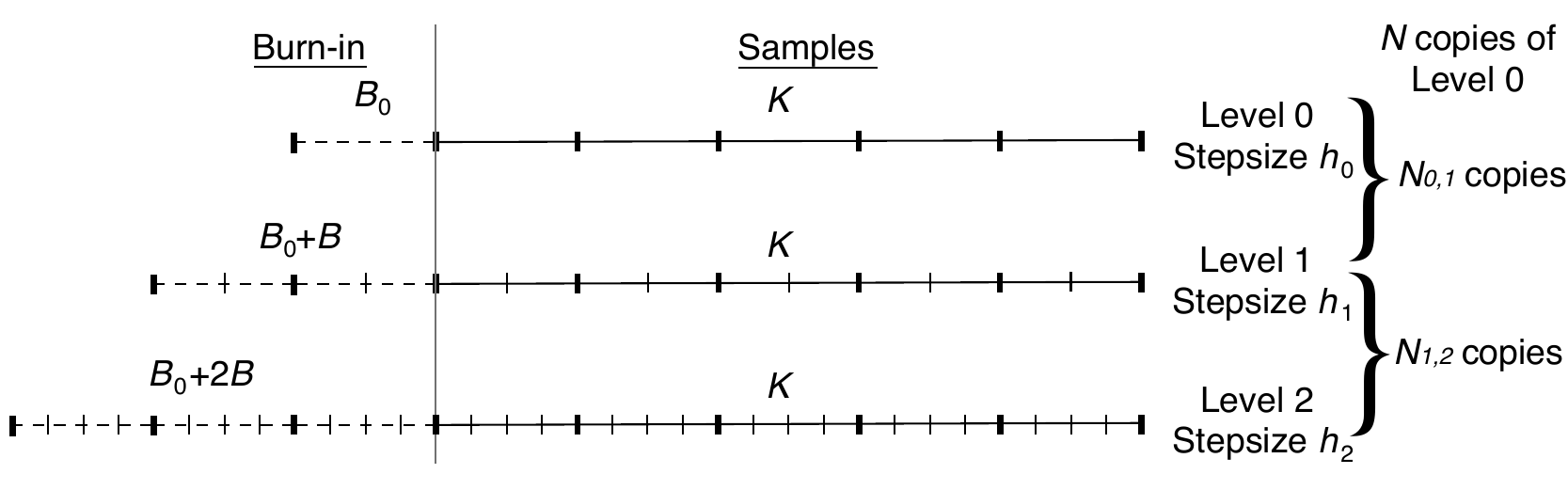}
 \caption{Elimination of bias by increasing burn-in lengths at higher discretization levels.}
 \label{fig:bias.elimination}
 \end{figure}


  \subsection{Organization}
  This article is organized as follows.
  In Section \ref{sec:back}, we provide the necessary background material related to this work, including a discussion of splitting methods for kinetic Langevin dynamics, in particular the $\UBU$ discretization, as well as the extension to stochastic gradients.
  Section \ref{sec:unbiased} is devoted to introducing  our unbiased algorithms. We first provide some simple conditions for creating unbiased estimators with finite variance based on telescopic sums, together with a central limit theorem for such estimators. We then present our method using exact gradients and discuss necessary assumptions for unbiasedness and finite variance including showing that the variance of the estimator is finite.
  {Numerical experiments are provided in Section \ref{sec:num} on  a range of high-dimensional problems, such as a Gaussian target, an MNIST classification problem and a Poisson regression model. Finally, we provide detailed proofs of all theorems and additional figures from numerical experiments in the Appendix.}

\section{Background \& preliminary material} 
\label{sec:back}

In this section, we provide the essential background material on kinetic Langevin dynamics {(see \eqref{eq:kinetic_langevin})} and a splitting-type scheme called $\UBU$. We then discuss the extension to stochastic gradients. 



\subsection{Splitting methods}\label{sec:splitting_methods}

{Discretization methods for \eqref{eq:kinetic_langevin}} with a high order of accuracy in both the weak and strong senses can be constructed by {\em splitting} \cite{SkIz2002,Me2007design,leimkuhler2013rational}, in which the SDE is broken into parts that can be either be solved analytically or which are easier to handle numerically. An accurate splitting method was introduced in \cite{AlamoSanz,BUBthesis} and was also studied in \cite{sanz2021wasserstein}. This splitting method only requires one gradient evaluation per iteration but has strong order two. The method is based on splitting the SDE \eqref{eq:kinetic_langevin} as follows 
\[
\begin{pmatrix}
dx \\
dv
\end{pmatrix} = \underbrace{\begin{pmatrix}
0 \\
-\nabla U(x)dt
\end{pmatrix}}_{\mathcal{B}} +\underbrace{\begin{pmatrix}
vdt \\
-\gamma v dt + \sqrt{2\gamma}dW_t
\end{pmatrix}}_{\mathcal{U}},
\]
which can be integrated exactly over a step of size $h$. Given $\gamma > 0$, let $\eta = \exp{\left(-\gamma h/2\right)}$, and for ease of notation, we define the following operators 
\begin{equation}\label{eq:Bdef}
\mathcal{B}(x,v,h) = (x,v - h\nabla U(x)),
\end{equation}
and
\begin{equation}\label{eq:Udef}
\begin{split}
\mathcal{U}(x,v,h/2,\xi^{(1)},\xi^{(2)}) &= \Big(x + \frac{1-\eta}{\gamma}v + \sqrt{\frac{2}{\gamma}}\left(\mathcal{Z}^{(1)}\left(h/2,\xi^{(1)}\right) - \mathcal{Z}^{(2)}\left(h/2,\xi^{(1)},\xi^{(2)}\right) \right),\\
& \eta v + \sqrt{2\gamma}\mathcal{Z}^{(2)}\left(h/2,\xi^{(1)},\xi^{(2)}\right)\Big),
\end{split}
\end{equation}
where 
\begin{equation}\label{eq:Z12def}
\begin{split}
\mathcal{Z}^{(1)}\left(h/2,\xi^{(1)}\right) &= \sqrt{\frac{h}{2}}\xi^{(1)},\\
\mathcal{Z}^{(2)}\left(h/2,\xi^{(1)},\xi^{(2)}\right) &= \sqrt{\frac{1-\eta^{2}}{2\gamma}}\Bigg(\sqrt{\frac{1-\eta}{1+\eta}\cdot \frac{4}{\gamma h}}\xi^{(1)} + \sqrt{1-\frac{1-\eta}{1+\eta}\cdot\frac{4}{\gamma h}}\xi^{(2)}\Bigg).
\end{split}
\end{equation}
The $\mathcal{B}$ operator indicated here is as given previously, whereas $\mathcal{U}$ as defined above is the exact solution in the weak sense of the remainder of the dynamics when  $\xi^{(1)}, \xi^{(2)} \sim \mathcal{N}\left(0,I_{d}\right)$ are independent random vectors. Different orders of composition of $\mathcal{B}$ and $\mathcal{U}$ can be taken to define different numerical integrators of kinetic Langevin dynamics, two such methods considered in \cite{AlamoSanz,BUBthesis} are BUB, a half step in $\mathcal{B}$, followed by a full step in $\mathcal{U}$ and a further half step in $\mathcal{B}$ and UBU, a half step in $\mathcal{U}$ followed by a full $\mathcal{B}$ step, followed by a half $\mathcal{U}$ step.

The Markov kernel for an $\UBU$ step with stepsize $h$ will be denoted by $P_{h}$, which can be described by \eqref{eq:PhUBU} as follows.
\begin{equation}\label{eq:PhUBU}
\begin{split}
\left(\xi^{(i)}_{k+1}\right)^{4}_{i = 1}, \quad &\xi^{(i)}_{k+1} \sim \mathcal{N}(0_{d},I_{d}) \text{ for all } i = 1,...,4.\\
\left(x_{k+1},v_{k+1}\right)&=\mathcal{UBU}\left(x_k,v_k,h,\xi^{(1)}_{k+1},\xi^{(2)}_{k+1},\xi^{(3)}_{k+1},\xi^{(4)}_{k+1}\right)
\\
&=\mathcal{U}\left(\mathcal{B}\left(\mathcal{U}\left(x_{k},v_{k},h/2,\xi^{(1)}_{k+1},\xi^{(2)}_{k+1}\right),h\right),h/2,\xi^{(3)}_{k+1},\xi^{(4)}_{k+1}\right).
\end{split}
\end{equation}
We have found that the strong second-order property and generally high accuracy of $\UBU$ makes it suitable for unbiased estimation, as described in Section \ref{sec:unbiased}.

The BAOAB method is an alternative splitting scheme that is known to be second-order weakly accurate and has small bias (see \cite{ BouRabeeOwhadi2010,leimkuhler2013rational,leimkuhler2013jcp,lemast2016}).
$\BAOAB$ is exact for Gaussian targets and has a robustness property for large values of the friction parameter $\gamma$ (see \cite{leimkuhler2023contractiona}), but its strong order is one. Theorem 3.3 of \cite{telatovich2017strong} claims that the stochastic velocity Verlet (SVV) method is, like $\UBU$, also strongly second-order accurate. Despite their strengths as raw sampling schemes, both BAOAB and SVV exhibited worse performance than UBU in our preliminary numerical experiments in the setting of unbiased estimation.  For this reason,  we focus on UBU in this paper. Nevertheless, it is important to note that the unbiased estimation approach of this paper is by no means limited to the UBU integrator, and its performance could be further improved by more accurate integrators developed in the future. {It can also easily be applied to many discretizations of other gradient-based stochastic processes used for sampling, for example, Hamiltonian Monte Carlo and overdamped Langevin dynamics.}

\subsection{Extension to stochastic gradients}

In this work, we also consider extending splitting methods with the use of stochastic gradients. 
We use the following definition from \cite{leimkuhler2023contractionb}.

\begin{definition} \label{def:stochastic_gradient}
A \textit{stochastic gradient approximation} of a potential $U$ is defined by a function $\mathcal{G}:\R^d \times \Omega \to \R^d$ and a probability distribution $\rho$ on a Polish space $\Omega$, such that for every $x\in \R^d$, $\mathcal{G}(x, \cdot)$ is measurable on $(\Omega,\mathcal{F})$, and for $\omega\sim \rho$,
\[\E(\mathcal{G}(x,\omega)) = \nabla U(x).\]
The function $\mathcal{G}$ and the distribution $\rho$  together define the stochastic gradient, which we denote as  $(\mathcal{G}, \rho)$.
\end{definition}

Replacing the exact gradients with such stochastic gradients in the $\mathcal{B}$ step yields
\begin{equation}\label{eq:BGdef}
\mathcal{B}_{\mathcal{G}}(x,v,h,\omega) = (x,v - h\mathcal{G}(x,\omega)),
\end{equation}
and we can use this inside $\BAOAB$ and $\UBU$ to obtain stochastic gradient variants.

\cite{leimkuhler2023contractionb} has proven convergence bounds for $\BAOAB$ with stochastic gradients in Wasserstein distance that are applicable to some widely used stochastic gradient schemes (random sampling with replacement, control variate gradient estimator).
\section{Unbiased multilevel Monte Carlo methods}
\label{sec:unbiased} 
In this section, we introduce and motivate our proposed algorithm, which we refer to as Unbiased $\UBU$ ($\UBUBU$). 
We first describe the basic unbiased Monte Carlo scheme and introduce some essential assumptions. We then give relevant results which help to motivate our estimator,  including a central limit theorem, a non-asymptotic bound on the variance with exact gradients, and other related results. Finally, we state our algorithm.

Suppose that for each $h\in (0, h_{\max}]$ (stepsize parameter), $Q_h$ is a Markov kernel on some Polish state space $\Lambda$ with stationary distribution $\mu_{h}$ such that $\mu_h$ converges to $\mu$ in distribution as $h\to 0$ (for example, these might be discretizations of a diffusion with different time stepsizes).
Assume that we are interested in computing the expectation $\mu(f)$ of a function $f$ satisfying ${\mu_h(f^2)<\infty}$ for every $h \in (0, h_{\max}]$ and $\mu(f^2)<\infty$. \cite{ruzayqat2022unbiased} suggested a multilevel estimation method  based on stepsizes 
$h_0\in (0, h_{\max}]\text{ and } h_l=h_0 \cdot 2^{-l} \text{ for }l=1,2,\ldots$,
using a telescopic sum of the form 
\[\mu(f)=\mu_{h_0}(f)+\sum_{j=1}^{\infty} (\mu_{h_j}(f)-\mu_{h_{j-1}}(f)).\]
Unbiased estimators of each term in the sum can be constructed via coupling.
A challenge with this approach is that obtaining an unbiased estimator for $\mu_{h_0}(f)$ already requires two chains to be coupled based on the approach proposed in the papers \citep{Chada2021unbiased,glynn2014exact,heng2019unbiased, jacob2020unbiased}. Estimating the expectations $\mu_{h_j}(f)-\mu_{h_{j-1}}(f)$ is even more challenging, requiring the coupling of four chains. The nature of the couplings means that it is not straightforward to use splitting methods such as $\UBU$ or $\BAOAB$ 
(as Markov kernels from different starting points need to be coupled closely in total variation distance, and this is difficult unless the distributions are Gaussian). 
To overcome such issues, we propose a different telescoping sum for estimating $\mu(f)$,
\begin{equation}\label{eq:telescopicsum2}\mu(f)=\tilde{\mu}_{h_0}(f)+\sum_{l=0}^{\infty} (\tilde{\mu}_{h_{l+1}}(f)-\tilde{\mu}_{h_{l}}(f)).\end{equation}
Here $\tilde{\mu}_{h_{l}}$ are created using some empirical averages, which will be defined in the rest of this section for exact, stochastic, and approximate gradients.

Suppose that $D_0$ is a random variable satisfying that $\E(D_0)=\tilde{\mu}_0(f)$. 
Let $\{D_0^{(r)}\}_{r=1}^{N}$ be $N$ i.i.d. copies of $D_0$, and we define
$S_0=\frac{1}{N}\sum_{r=1}^{N} D_0^{(r)}.$
Then it is clear that $\E(S_0)=\E(D_0)=\tilde{\mu}_{h_0}(f)$. Let $D_{l,l+1}$ be a random variable such that
\[\E D_{l,l+1}=\tilde{\mu}_{h_{l+1}}(f)-\tilde{\mu}_{h_{l}}(f).\]
Let $c_{0,1}\ge c_{1,2}\ge c_{2,3}\ge  \ldots$ be positive constants such that $c_{l,l+1}\to 0$ as $l\to \infty$, and let
\begin{equation}\label{eq:Nldef}
\begin{split}
L(N)&=\min\left\{l\in \mathbb{N}: c_{l,l+1} N \le 1\right\},\\
N_{l,l+1}&=\lceil c_{l,l+1} N\rceil\text{ for }l\le L(N),\\
N_{l,l+1}&\sim \text{Bernoulli}(c_{l,l+1} N) \text{ for } l>L(N).
\end{split}
\end{equation}
For each $l\ge 1$, let $\{D_{l,l+1}^{(r)}\}_{r=1}^{N_{l,l+1}}$ be $N_{l,l+1}$ i.i.d. copies of $D_{l,l+1}$, and
\begin{equation}\label{eqref:Sllp1def}S_{l,l+1}=\frac{1}{\E (N_{l,l+1})}\sum_{r=1}^{N_{l,l+1}} D_{l,l+1}^{(r)}=\left\{\begin{matrix}\frac{1}{N_{l,l+1}}\sum_{r=1}^{N_{l,l+1}} D_{l,l+1}^{(r)} \text{ for }0\le l\le L(N), \\ \frac{\Ind[N_{l,l+1}=1]}{\E(N_{l,l+1})}D_{l,l+1}^{(1)} \text{ for }l>L(N).\end{matrix}\right.
\end{equation}
It is clear from the definitions and Wald's equation that $\E S_{l,l+1}=\E D_{l,l+1}=\tilde{\mu}_{h_{l+1}}(f)-\tilde{\mu}_{h_{l}}(f)$.
Using the definition of $L(N)$ and $N_{l,l+1}$, we have $N_{L(N),L(N)+1}=1$, and hence
\begin{equation}\label{eq:SLNLNp1}
S_{L(N),L(N)+1}=D_{L(N),L(N)+1}^{(1)}.
\end{equation}

Our first estimator is defined as
\begin{equation}\label{eq:Sdef}
S = S_0 + \sum^{\infty}_{l=0}S_{l,l+1},
\end{equation}
where the terms $S_0, S_{0,1}, S_{1,2},\ldots$ are independent.

The random $D_{l,l+1}$ variable will play a key role in our approach, as it is going to link two different discretization levels with stepsizes $h_l$ and $h_{l+1}$. $\Var(S)$ depends on $\Var(D_{l,l+1})$, which is determined by how closely we couple the two discretizations. This is closely related to the strong order of the discretizations, determining how close they are to the underlying diffusion.
It is possible to improve estimator \eqref{eq:Sdef} slightly by the use of Richardson extrapolation \cite{richardson1911approximate} {(see \cite{MBG15} within the context of Multilevel Monte Carlo)}. 
The idea is that 
when $h$ is sufficiently small, for $Q_h$ defined in terms of an SDE discretization, the differences $\mu_{h}(f)-\mu(f)$ tend to follow a certain asymptotic behaviour in $h$, which can be characterized by an asymptotic expansion \cite{leimkuhler2013rational,lemast2016}. For symmetric splittings like $\BAOAB$ it is known that $\mu_{h}(f)-\mu(f)=c_{f,\mu} h^2 (1+\O(h))$ for some constant $c_{f,\mu}$ depending on $f$ and $\mu$.  The same property can be established for $\UBU$, using similar arguments.
Based on this observation taking into account that such behaviour may only be valid at small stepsizes, and using \eqref{eq:SLNLNp1}, our refined estimator is defined as
\begin{align}\label{eq:SRichardsondef}
S(c_R) &= S_0 + \sum^{L(N)-1}_{l=0}S_{l,l+1} + \frac{D_{L(N),L(N)+1}^{(1)}}{1-c_R}+\sum_{l=L(N)+1}^{\infty} \overline{S}_{l,l+1},\\
\nonumber\overline{S}_{l,l+1}&=\frac{\Ind[N_{l,l+1}=1]}{\E (N_{l,l+1})} \left[D_{l,l+1}^{(1)}-D_{L(N),L(N)+1}^{(1)}\cdot c_R^{l-L(N)}\right],
\end{align}
where $c_{R}\in [0, \phi_N^{-1/2})$ can be any number (we state the recommended choice of this in our algorithms). Our first estimator $S$ is a special case since $S(0)=S$.

The key assumptions we make on the variances are as follows:
\begin{assumption}\label{ass:var}
$f:\Lambda\to \R$ is a measurable function. $(\tilde{\mu}_{h_l})_{l\ge 0}$ is a sequence of distributions satisfying that $\tilde{\mu}_{h_l}(f)\to \mu(f)$ as $l\to \infty$. The random variable $D_0$ satisfies that $\E(D_0)={\tilde{\mu}}_{h_0}(f)$, $\Var(D_0)<\infty$, for every $l\geq 0$, the random variable $D_{l,l+1}$ satisfies that $\E(D_{l,l+1})=\tilde{\mu}_{l+1}(f)-\tilde{\mu}_{l+1}(f)$ and $\E(D_{l,l+1}^2)\le V_{D}\phi_D^{-l}$ for some finite constants $V_D>0$, $\phi_D>2$. 
\end{assumption}
\begin{assumption}\label{ass:numbsamp}
The constants $c_{l,l+1}$ controlling $N_{l,l+1}$ satisfy 
\[\underline{c}_N \phi_N^{-l}\le c_{l,l+1}\le  \ol{c}_N \phi_N^{-l},\]
for some finite constants $0<\underline{c}_N\le \ol{c}_N$, $\phi_N>2$.
\end{assumption}
\begin{assumption}\label{ass:comp}
The computational cost of generating a sample from $D_{l,l+1}$ is $\O(2^l (K+lB+B_0))$ for some finite constants $B$, $B_0$, and generating a sample from $D_0$ has a finite computational cost.
\end{assumption}
\begin{assumption}\label{ass:independence}
For $1\le l \le L(N)-1, 1\le r\le N_{l,l+1}$, the random variables $D_{l,l+1}^{(r)}$ are all independent from each other, and they are also independent from the collection of random variables $\{D_{l,l+1}^{(1)}\}_{l\ge L(N)}$. 
\end{assumption}
\begin{remark}
We do not assume that the random variables in the set $\left\{D_{l,l+1}^{(1)}\right\}_{l\ge L(N)}$ are independent. Some variance reduction may be achieved by coupling them in \eqref{eq:SRichardsondef}.
\end{remark}

\begin{restatable}{proposition}{propunb}
\label{prop:unb}
Suppose that Assumptions \ref{ass:var}, \ref{ass:numbsamp}, \ref{ass:comp} and \ref{ass:independence} hold, and that $2<\phi_N<\phi_D$. Then $S$ as defined in \eqref{eq:Sdef} is an unbiased estimator of $\mu(f)$ that has finite variance \[\Var(S)\le 
 \frac{\Var(D_0)}{N}+\frac{V_D}{N\underline{c}_N\left(1-\left(\frac{\phi_N}{\phi_D}\right)^{1/2}\right)^2},\] and finite expected computational cost.  

Similarly, for any $0\le c_R< \frac{1}{\phi_N^{1/2}}$, $S(c_R)$ as defined in \eqref{eq:SRichardsondef} is also an unbiased estimator of $\mu(f)$ with finite variance  
\[\Var(S(c_R))\le \frac{\Var(D_0)}{N}+
\frac{2V_D}{N\underline{c}_N\left(1-\left(\frac{\phi_N}{\phi_D}\right)^{1/2}\right)^2}+
\frac{1}{N^2}\cdot \frac{2V_D \ol{c}_N \phi_N^2 c_R^2}{\underline{c}_N^2(1-\phi_N c_R^2)},\]
and finite expected computational cost.
\end{restatable}
\begin{proof}
See Section C of the Appendix.
\end{proof}

We show below that a Central Limit Theorem (CLT) holds for these estimators.
\begin{restatable}{theorem}{thmCLT}
\label{thm:CLT}
Under the assumptions of Proposition \ref{prop:unb}, we have that, as $N\to \infty$,
\[\sqrt{N} (S-\mu(f)) \Rightarrow \mathcal{N}(0,\sigma^2_S) \quad \text{ and }\quad \sqrt{N} (S(c_R)-\mu(f)) \Rightarrow\mathcal{N}(0,\sigma^2_S),\]
where 
\begin{equation}\label{eq:sigma2S}\sigma^2_S:=\Var(D_0)+\sum_{l=0}^{\infty} \frac{\Var(D_{l,l+1})}{c_{l,l+1}}.\end{equation}
\end{restatable}
\begin{proof}
See Section C of the Appendix.
\end{proof}

\subsection{$\UBUBU$ with exact gradients}\label{sec:UBUBU_exact_gradient}

Now, we will specify the way $D_0$ and $D_{l,l+1}$ are defined based on $\UBU$ discretization of \eqref{eq:kinetic_langevin} with exact gradients, as defined in \eqref{eq:PhUBU}. Let $\mu_0$ be an initial distribution on $\Lambda$ that we can readily sample from, for example, a Dirac-$\delta$ at the maximum-a-posteriori (MAP) estimator. Let 
$R_0=P_{h_0}\text{ and } R_l=P_{h_l}^{2^{l}} \text{ for }l=1,2,\ldots$. 

These Markov kernels correspond to the same amount of time $h_0$ in the timescale of the limiting diffusion (and clearly, $R_l$ still has $\mu_{h_l}$ as its stationary distribution). 
Consider $B_0$ burn-in steps with kernel $R_0$ at level $0$, and $B_l=B_0+l B$ steps with kernel $R_l$ at level $l$. 
Define the approximate versions of $\mu_{h_l}$ as
\begin{equation}\label{eq:tildemu.hl}
\tilde{\mu}_{h_l}=\frac{1}{K} \sum_{i=1}^{K} \mu_{0} R_l^{B_l+i}.
\end{equation}
Estimates with respect to this can be computed by taking $B_l$ burn-in steps according to $R_l$ (equivalently $2^l B_l$ burn-in steps according to $P_{h_l}$), and then $K$ additional steps that are used for computing an empirical average. In this way, we can compute expectations with respect to $\tilde{\mu}_{h_l}$ without the use of couplings. Moreover, given that at the diffusion time scale, the burn-in time tends to infinity as $l$ grows, it is reasonable to expect that under suitable assumptions, $\tilde{\mu}_{h_l}$ converges to $\mu$ as $l\to \infty$. 

Let $D_0$ be the empirical average of a function $f$ based on $K$ samples a from Markov chain with kernel $R_0$, with burn-in $B_0$, initiated from $\mu_0$, i.e. for the Markov chain
$z_{-B_0}^{(0)}\sim \mu_0$, $z_{-B_0+1}^{(0)}\sim R_0(z_{-B_0}^{(0)},\cdot), \ldots, z_K^{(0)}\sim R_0(z_{K-1}^{(0)},\cdot)$. Let $\nu_{0}$ denote the joint distribution of $z_{-B_0}^{(0)},\ldots, z_{K}^{(0)}$, and define
\begin{equation}\label{eq:D0def}
    D_0= \frac{1}{K}\sum_{i=1}^{K} f(z_{i}^{(0)}).
\end{equation}
Let $\{D_0^{(r)}\}_{r=1}^{N}$ be $N$ i.i.d. copies of $D_0$, and  define
\begin{equation}\label{eq:S0def}
    S_0=\frac{1}{N}\sum_{r=1}^{N} D_0^{(r)}.
\end{equation}
Then it is clear that $\E(S_0)=\E(D_0)=\tilde{\mu}_{h_0}(f)$.
For $l\ge 0$, let $z_{-B_l}^{(l,l+1)},\ldots, z_{K}^{(l,l+1)}, {z'}_{-B_{l+1}}^{(l,l+1)}$,$\ldots$, ${z'}_{K}^{(l,l+1)}$ be $\Lambda$ valued random variables defined on the same probability space (i.e. coupled) such that
\begin{itemize}
\item $z_{-B_l}^{(l,l+1)},\ldots, z_{K}^{(l,l+1)}$ is a Markov chain with kernel $R_l$ initiated as $z_{-B_l}^{(l,l+1)}\sim \mu_0$, and 
\item ${z'}_{-B_{l+1}}^{(l,l+1)},\ldots, {z'}_{K}^{(l,l+1)}$ is a Markov chain with kernel $R_{l+1}$ initiated ${z'}_{-B_{l+1}}^{(l,l+1)}\sim \mu_0$.
\end{itemize}
Let
\begin{equation}\label{eq:Ddef}
    D_{l,l+1}= \frac{1}{K}\sum_{i=1}^{K} [f({z'}_{i}^{(l,l+1)})-f(z_{i}^{(l,l+1)})].
\end{equation}
From the definitions, it follows that 
\[\E D_{l,l+1}=\tilde{\mu}_{h_{l+1}}(f)-\tilde{\mu}_{h_{l}}(f),\]
hence $D_{l,l+1}$ is an unbiased estimator of the difference $\tilde{\mu}_{h_{l+1}}(f)-\tilde{\mu}_{h_{l}}(f)$.

When these Markov chains are discretizations of the same diffusion, it is natural to create synchronous couplings by using the same Brownian noise to generate the Gaussian random variables used during the periods $z_{-B_l},\ldots, z_{K}$ and $z'_{-B_l},\ldots, z'_{K}$. Such couplings can significantly reduce the variance of $D_{l,l+1}$. 
Let $\mathcal{B}$ and $\mathcal{U}$ be as in (\ref{eq:Bdef}-\ref{eq:Udef}). Further we define $\mathcal{U}^{2}$ to be
\begin{equation}\label{eq:U2def}
\mathcal{U}^{2}(x,v,h,\xi^{(1)},\xi^{(2)},\xi^{(3)},\xi^{(4)}) = \mathcal{U}\left(\mathcal{U}\left(x,v,h/2,\xi^{(1)},\xi^{(2)}\right),h/2,\xi^{(3)},\xi^{(4)}\right).
\end{equation}
As $\mathcal{U}$ is an exact solution in the weak sense to its respective component in the splitting, this is an exact solution in the weak sense which uses Brownian increments $\left(\xi^{(1)},\xi^{(2)}\right)$ in the first half step $h/2$ and $\left(\xi^{(3)},\xi^{(4)}\right)$ in the second half step $h/2$. The $\mathcal{U}^{2}$ operator is an exact solution over stepsize $h$.

A coupling can be constructed between discretization levels so that the two discretization levels share Brownian motion in the exact integration of the $\mathcal{U}$ steps. This is done by using the Brownian increments from two respective $\mathcal{U}$ solutions at the higher level and concatenating them using the $\mathcal{U}^{2}$ operator at the lower level. Next, the stochastic integrals in the two levels are coupled by sharing the same Brownian noise. The Markov kernel $P_{h,h/2}$ for the two discretization levels $h,h/2$ is defined as follows. {First let $\left(\xi^{(i)}_{k+1}\right)^{8}_{i = 1} \sim \mathcal{N}(0_{d},I_{d}) \text{ for all } i = 1,...,8$, then\\
\begin{equation}\label{eq:Phh2UBU}
\begin{split}
&\left(x'_{k+1/2},v'_{k+1/2}\right) =\mathcal{U}\left(\mathcal{B}\left(\mathcal{U}\left(x'_{k},v'_{k},h/4,\xi_{k+1}^{(1)},\xi_{k+1}^{(2)}\right),h/2\right), h/4,\xi_{k+1}^{(3)},\xi_{k+1}^{(4)}\right)\\
&\left(x'_{k+1},v'_{k+1}\right) =\mathcal{U}\left(\mathcal{B}\left(\mathcal{U}\left(x'_{k+1/2},v'_{k+1/2},h/4,\xi_{k+1}^{(5)},\xi_{k+1}^{(6)}\right),h/2\right),h/4,\xi_{k+1}^{(7)},\xi_{k+1}^{(8)}\right)\\
&\left(x_{k+1},v_{k+1}\right) =\\
&\mathcal{U}^{2}\left(\mathcal{B}\left(\mathcal{U}^{2}\left(x_{k},v_{k},h/2,\xi_{k+1}^{(1)},\xi_{k+1}^{(2)},\xi_{k+1}^{(3)},\xi_{k+1}^{(4)}\right),h\right),h/2,\xi_{k+1}^{(5)},\xi_{k+1}^{(6)},\xi_{k+1}^{(7)},\xi_{k+1}^{(8)}\right).
\end{split}
\end{equation}}
This Markov chain acts on the state space $\R^d\times \R^d \times \R^d \times \R^d$, moving from $(x_k,v_k,x_k',v_k')$ to $(x_{k+1}, v_{k+1}, x_{k+1}', v_{k+1}')$ via the steps in \eqref{eq:Phh2UBU}. When looking at the individual components, $(x_k,v_k) \to (x_{k+1}, v_{k+1})$ corresponds to one $\UBU$ step at stepsize $h$, while  $(x_k',v_k') \to (x_{k+1}', v_{k+1}')$ corresponds to two $\UBU$ steps at stepsize $h/2$. A key property here is that the stochastic integrals between two steps are synchronously coupled, which ensures that these two chains approximate the same underlying diffusion (in the strong sense). Hence, they are expected to remain close, which was observed in our numerical simulations.

We now create a coupling between levels $l$ and $l+1$, denoted by $\nu_{l,l+1}$.
\begin{algorithm}[H]
\renewcommand{\thealgorithm}{}
\floatname{algorithm}{}
\addtocounter{algorithm}{-1}
     \footnotesize 
     \begin{algorithmic}[1]
     \State For given initial distribution $\mu_0$ on $\Lambda$, we define $z_{-B_{l}}^{(l,l+1)}\sim \mu_0$ and $z'^{(l,l+1)}_{-B_{l+1}}\sim \mu_0$ as independent random variables.
     \State We let $z'^{(l,l+1)}_{-B_{l+1}},\ldots, z'^{(l,l+1)}_{-B_{l}}$ be a Markov chain evolving according to $R_{l+1}=(P_{h_{l+1}})^{2^{l+1}}$.
     \State Let $(z_{-B_{l}}^{(l,l+1)},z'^{(l,l+1)}_{-B_{l}}), (z_{-B_{l}+1}^{(l,l+1)},z'^{(l,l+1)}_{-B_{l}+1}),\ldots,  (z_{K}^{(l,l+1)},z'^{(l,l+1)}_{K})$ be a Markov chain evolving according to $R_{l, l+1}=(P_{h_l,h_{l+1}})^{2^{l}}$.
    \State Let $\nu_{l,l+1}$ denote the joint distribution of $z_{-B_l}^{(l,l+1)},\ldots, z_{K}^{(l,l+1)}, z'^{(l,l+1)}_{-B_{l+1}},\ldots, z'^{(l,l+1)}_{K}$. 
\end{algorithmic}

\caption{$\nu_{l,l+1}$ coupling}
\label{alg:nullp1}
  
\end{algorithm}

The motivation for this $\nu_{l,l+1}$ coupling is that if two coupled chains are driven by the same noise and approximate the same diffusion, they are expected to be close most of the time. Given a sufficiently long burn-in, they will likely stay close during the iterations $1,2,\ldots, K$ used for computing the differences in their empirical averages, reducing the variance of $D_{l,l+1}$. Let $c_N>0$ and $\phi_N>2$ be constants, and let 
\begin{equation}\label{eq:cllp1exact}
c_{l,l+1}=c_N\phi_N^{-l}\text{ for }l\in \mathbb{N}.
\end{equation}
Let $L(N)$ and $N_{l,l+1}$ be defined according to \eqref{eq:Nldef}, and set $l_{\max}=\max\{l: N_{l,l+1}>0\}$. Then for $l\in\{L(N),\ldots, l_{\max}\}$, we have $N_{l,l+1}\le 1$. The $D_{l,l+1}^{(1)}$ random variables at these levels will not be independent, but we define them instead based on a sequence of random variables $\{z_{-B_l}^{(l)},\ldots, z_K^{(l)}\}_{L(N)\le l\le l_{\max}+1}$, such that $\{z_{-B_l}^{(l)},\ldots, z_K^{(l)}, z_{-B_{l+1}}^{(l+1)},\ldots, z_K^{(l+1)}\}$ is distributed as $\nu_{l,l+1}$ for every $L(N)\le l\le l_{\max}$. This is possible to implement by synchronously coupling all steps to be driven by the same Brownian motion. 

In our coupled Markov chain $P_{h,h/2}$ in \eqref{eq:Phh2UBU}, we have used the double $U$ step
$\mathcal{U}^{2}(x,v,h,\xi^{(1)},\xi^{(2)},\xi^{(3)},\xi^{(4)}) = \mathcal{U}\left(\mathcal{U}\left(x,v,h/2,\xi^{(1)},\xi^{(2)}\right),h/2,\xi^{(3)},\xi^{(4)}\right)$.
Since the $U$ step is exact, this can be equivalently written as another $U$ step 
\begin{equation}\label{eq:implicit_U_eq}
\mathcal{U}^{2}(x,v,h,\xi^{(1)},\xi^{(2)},\xi^{(3)},\xi^{(4)}) = \mathcal{U}(x,v,h,{\xi^{(1)}}',{\xi^{(2)}}'),
\end{equation}
where ${\xi^{(1)}}', {\xi^{(2)}}'$ are independent standard Gaussians, that are defined {implicitly via the linear equation \eqref{eq:implicit_U_eq} by} a deterministic mapping from $\xi^{(1)},\xi^{(2)},\xi^{(3)},\xi^{(4)}$, i.e. $({\xi^{(1)}}', {\xi^{(2)}}')=\mathcal{M}_{h/2\to h}(\xi^{(1)},\xi^{(2)},\xi^{(3)},\xi^{(4)})$.
By first the Gaussian random vectors $\xi_{k+1}^{(1)},\ldots, \xi_{k+1}^{(4)}$ for the level $l_{\max}+1$, and then using the transformation $\mathcal{M}_{h_{l_{\max}+1}\to h_{l_{\max}}}, \ldots \mathcal{M}_{h_{L(N)+1}\to h_{L(N)}}$ recursively, we create a synchronous coupling of random variables $\{z_{-B_l}^{(l)},\ldots, z_K^{(l)}\}_{L(N)\le l\le l_{\max}+1}$, which we call $\nu_{L(N):l_{\max}}$.

The reason for using $\nu_{L(N):l_{\max}}$ instead of independent couplings at levels $l\ge L(N)$ is that this leads to variance reduction in the estimator \eqref{eq:SRichardsondef}. We call the overall estimator $S(c_R)$ based on formula \eqref{eq:SRichardsondef} with $D_{l,l+1}$ defined based on coupling construction $\nu_{l,l+1}$ as \emph{Unbiased $\UBU$} (or $\UBUBU$, for short). {The steps of this estimator are provided in Section B of the Appendix.}

Now, we will state our theoretical results for this algorithm. To prove unbiasedness and finite variance for our estimator $S(c_R)$, we require several assumptions, which we state below. These include assumptions on the smoothness and strong convexity of our potential, as well as restrictions on various parameters of the algorithm.

\begin{restatable}[$M$-$\nabla$ Lipschitz]{assumption}{AssumLip}
\label{assum:Lip}
$U:\R^{d} \to \R$ is twice continuously differentiable and there exists $M>0$ such that for all $x,y \in \R^d$
$$
\| \nabla U(x) - \nabla U(y)\| \leq M\|x-y\|.
$$
\end{restatable}

\begin{restatable}[$m$-strong convexity]{assumption}{AssumConvex}
\label{assum:convex}
$U:\R^{d} \to \R$ is continuously differentiable and there exists $m>0$ such that for all $x,y \in \R^d$
$$
\langle \nabla U(x) - \nabla U(y),x-y \rangle \geq m|x-y|^2.
$$
\end{restatable}
The strongly Hessian Lipschitz property relies on the following tensor norm from \cite{chen2023does}.
\begin{definition}\label{def:strongHessianLipschitznorm}
    For $A \in \mathbb{R}^{d \times d \times d}$, let
    \[
    \|A\|_{\{1,2\}\{3\}} = \sup_{x\in \mathbb{R}^{d\times d}, y\in \mathbb{R}^d}\left\{\left.\sum^{d}_{i,j,k=1}A_{ijk}x_{ij}y_{k} \right| \sum_{i,j=1}^{d}x^{2}_{ij}\leq 1, \sum^{d}_{k=1}y^{2}_{k}\leq 1\right\}.
    \]
\end{definition}
\begin{remark}
The $\|A\|_{\{1,2\}\{3\}}$ norm in Definition \ref{def:strongHessianLipschitznorm} can be equivalently written as 
\begin{equation}\label{eq:strongHessianLipschitz:alternative:def}   
\|A\|_{\{12\}\{3\}}=\left\|\sum_{i_1} A_{i_1,\cdot,\cdot}^T\cdot A_{i_1,\cdot,\cdot}\right\|^{1/2},
\end{equation}
where $A_{i_1,\cdot,\cdot}=\left(A_{i_1,i_2,i_3}\right)_{1\le i_2\le d, 1\le i_3\le d}$ is a $d\times d$ matrix, see the proof of Lemma 7 of \cite{paulin2024}. 
\end{remark}

\begin{restatable}[$M^{s}_{1}$-strongly Hessian Lipschitz]{assumption}{AssumHessLipschitz}
\label{assum:Hess_Lipschitz}
$U:\mathbb{R}^{d} \to \mathbb{R}$ is three times continuously differentiable and $M^{s}_{1}$-strongly Hessian Lipschitz if there exists $M^{s}_{1} > 0$ such that
    \[
    \|\nabla^{3}U(x)\|_{\{1,2\}\{3\}} \leq M^{s}_{1}
    \]
    for all $x \in \mathbb{R}^{d}$.
\end{restatable}
\begin{remark}
In Section I of the Appendix, we show that Bayesian multinomial regression satisfies this assumption.
\end{remark}
\begin{restatable}[$1$-Lipschitzness of $f$]{assumption}{AssumfLipschitz}
\label{assum:Lipschitz}
$f$ is a 1-Lipschitz function with respect to the Euclidean distance on $\R^{2d}$, that only depends on $x$, not $v$ (i.e. $f(x,v)=f(x,v')$ for any $x,v,v'\in \R^d$). 
\end{restatable}
\begin{restatable}[Distance of initial distribution from target]{assumption}{AssumInitialdist}
\label{assum:initialdist}
The initial distribution on $\Lambda=\R^{2d}$ satisfy that $\mathcal{W}_2(\pi,\mu_0)\le c_{\mu_0}\sqrt{\frac{d}{m}}$,  for some $c_{\mu_0}>0$.
\end{restatable}
\begin{remark}\label{remark:initdist}
It is easy to show that under Assumption \eqref{assum:convex}, for $\mu_0=\delta_{x^*}\times \mathcal{N}(0_d,I_d)$, and for $\mu_0=\mathcal{N}(x^*,(\nabla^2 U(x^*))^{-1})\times \mathcal{N}(0_d,I_d)$ (Gaussian approximation), this condition holds with $c_{\mu_0}=2$ {(see Section I of the Supplementary Material)}.
\end{remark}


\begin{restatable}{theorem}{ThmExactGradUBUBU}
\label{thm:exactgradUBUBU}
Suppose that Assumptions \ref{assum:Lip}, \ref{assum:convex}, \ref{assum:Hess_Lipschitz}, \ref{assum:Lipschitz}, \ref{assum:initialdist} hold, and in addition, 
\begin{align*}\gamma &\geq \sqrt{8M}, \quad h_0 \leq 
\frac{1}{\gamma}\cdot \frac{m}{264M}, \quad B\ge \frac{16\log(4)\gamma}{mh_0}, \quad
B_0\ge \frac{16 \gamma}{m h_0}\log\left(\frac{c_{\mu_0}+1}{ \sqrt{M}\gamma h_0^2}\right).\end{align*}
Suppose that $c_{R}\in [0,\phi_N^{-1/2})$, and $2<\phi_N<16$.
Then for any $N\ge 1$, the $\UBUBU$ estimator $S(c_R)$ has finite expected computational cost, $\E S(c_R)=\pi(f)$, and it has finite variance. 
Moreover, it satisfies a CLT as $N\to \infty$, and the asymptotic variance $\sigma^2_S$ defined in \eqref{eq:sigma2S} can be bounded as
\[\sigma^2_S\le \frac{C(m,M,M^{s}_1,\gamma,c_N,\phi_N)
}{K h_0} \left(1+\frac{1}{h_0 K}+d h_0^4\right).\]
\end{restatable}
\begin{proof}
See Section E of the Appendix.
\end{proof}

\begin{remark}
In particular, when setting $h_0=\O(d^{-1/4})$, and $K>1/h_0$, the bound simplifies to 
$\sigma^{2}_{S}\le\frac{C(\gamma, m,M,M^{s}_{1})}{N K h_0}$. This indicates that the overall number of gradient evaluations per effective sample in this setting is $\O(d^{1/4})$, which matches the best available bounds for HMC in \cite{chen2023does}, without the warm start assumption required in that paper ({which is not satisfied by the typical implementable initializations considered in Remark \ref{remark:initdist}}).

{To reach an accuracy of $\epsilon > 0$ in the RMSE of the unbiased estimator, it requires $\mathcal{O}(d^{1/4}/\epsilon^2)$ gradient evaluations (under the strongly-Hessian Lipschitz assumption). For i.i.d. samples, it would be $\mathcal{O}(1/\epsilon^2)$ and for samples from HMC under a warm start assumption, one would expect $\mathcal{O}(d^{1/4}/\epsilon^2)$ gradient evaluations according to \cite{chen2023does}. This is in the empirical average, not the invariant measure (hence why it is not polylogarithmic in $\epsilon > 0$ for HMC), and it is not to be confused with the additional polynomial dependence on $\epsilon^{-1}$, which is usually required for unadjusted algorithms \cite{dalalyan2017theoretical}. Bypassing the additional polynomial dependence on $\epsilon^{-1}$ is due to the multilevel strategy, as is typical in the multilevel Monte Carlo literature \cite{MBG15} and hence has the theoretical advantages of Metropolis-adjusted methods in terms of the dependence on the error tolerance \cite{Dwivedi2019}.} 

{If we do not make the strongly-Hessian Lipschitz assumption, but just strong convexity and $\nabla$Lipschitz, the overall number of gradient evaluations per effective sample reduces to $\mathcal{O}(d^{1/2}/\epsilon^{2})$.}

\end{remark}

The following proposition shows dimension-free bounds for product distributions. We are going to use an assumption on the initial distribution $\mu_0$.
\begin{assumption}\label{assum:product}
    Suppose that $\mu_0$ and the target distribution $\pi$ are of product form
\begin{align*}
\mu_{0}(dx,dv)&= \prod^{d}_{i=1}\mu_{0,i}(dx_i,dv_i) \quad \text{ for all }l\ge 0,\quad \pi(d x,d v) = \prod^{d}_{i=1}\tilde{\pi}_{i}(d x_i)\frac{e^{-v_i^2/2}dv_i}{\sqrt{2\pi}},
\end{align*}
for $x = (x_{1},...,x_{d}) \in \R^{d}$, $v = (v_{1},...,v_{d}) \in \R^{d}$, and that
\[\max_{1\le i\le d}\mathcal{W}_2(\pi_i,\mu_{0,i})\le c_{\mu_0}\sqrt{\frac{1}{m}},\]
for some finite constant $c_{\mu_0}$, where $\pi_{i}(d x_i,d v_i)=\tilde{\pi}_{i}(d x_i)\frac{e^{-v_i^2/2}}{\sqrt{2\pi}} dv_i$ is the joint distribution of $(x_i,v_i)$ according to the target $\pi$.
\end{assumption}

\begin{restatable}{proposition}{PropProductExactGradUBUBU}\label{prop:ProductExactGradUBUBU}
Suppose that Assumption \ref{assum:product} holds, and denote the potential $U$ as $U(x) = \sum^{d}_{i=1}U_{i}(x_{i})$. Suppose that Assumptions \ref{assum:Lip}, \ref{assum:convex}, and \ref{assum:Hess_Lipschitz} hold for each component $(U_i)_{1\le i\le d}$, and that
\begin{align*}\gamma \geq \sqrt{8M}, \quad h_0 \leq 
\frac{1}{\gamma}\cdot \frac{m}{264M}, \quad B\ge \frac{16\log(4)\gamma}{mh_0}, \quad
B_0\ge \frac{16 \gamma}{m h_0}\log\left(\frac{c_{\mu_0}+1}{ \sqrt{M}\gamma h_0^2}\right).
\end{align*}
Suppose that $f$ is of the form
\begin{equation}\label{eq:fprodwdef}
f(x,v) = g(\langle w^{(1)}, x \rangle,\ldots \langle w^{(r)}, x \rangle), 
\end{equation}
where $g:\R^r \to \R$ is 1-Lipschitz, and $w^{(1)},\ldots,w^{(r)} \in \R^{d}$. Suppose that $c_{R}\in [0, \phi_N^{-1/2})$ and $2<\phi_N<16$. Then for any $N\ge 1$, the $\UBUBU$ estimator $S(c_R)$ has finite expected computational cost, $\E S(c_R)=\pi(f)$, and it has finite variance. Moreover, it satisfies a CLT as $N\to \infty$, and the asymptotic variance can be bounded as 
\[
\sigma^2_S\le \frac{C(m,M,M^{s}_{1},\gamma,r,c_N,\phi_N)
}{K h_0} \sum_{1\le i\le r}\|w^{(i)}\|^{2}.
\]
\end{restatable}
\begin{proof}
See Section E of the Appendix.
\end{proof}
\begin{remark}
These bounds are independent of the dimension $d$. This is not surprising as the different components evolve independently according to the kinetic Langevin diffusion \eqref{eq:kinetic_langevin}, and we do not introduce any dependencies in the $\UBUBU$ algorithm. This is in contrast with Metropolized methods, where the accept/reject steps introduce dependencies in the evolution of the components. The results could be generalized to potentials which are separable into independent groups of coordinates, i.e. $U(x)=\sum_{i=1}^{s} U_i(x_{G_i})$, where $G_1,\ldots, G_{s}$ is a partition of $[d]$, and the size of each group $|G_{i}|$ is small. They could also be generalised to potentials with sparse interactions, see recent results in \cite{chen2024convergence}.
\end{remark}

\subsection{$\UBUBU$ with stochastic gradients}

In many applications, particularly in data science and machine learning, gradient computations are computationally expensive due to large datasets and the need to iterate through the entire dataset at each gradient evaluation. A common approach for reducing the cost of the gradient-based methods is to use stochastic gradient approximations based on subsampling the dataset to compute unbiased estimates  (see \cite{johnson2013accelerating,baker2019control,quiroz2018speeding, Brosse2018,Teh16,VZT16}). 

In these applications the potential $U: \R^{d} \to \R$ is typically of the form
\begin{equation}\label{eq:decom_potential}
    U(x) = U_{0}(x) + \sum^{N_D}_{i=1}U_{i}(x),
\end{equation}
where $x\in \R^{d}$, the dataset is of size $N_D \in \mathbb{N}$. $U_{0}$ can be chosen as the negative log density of the prior distribution or some other term that does not require accessing the data. In our examples, $U_{0}$ can be taken to be a quadratic function, for example, a quadratic matching the Hessian at the minimizer (which can be computed before sampling).

We remark that one of the most efficient samplers in the big data regime is the Zig-Zag sampler \cite{Bierkins2019} whose complexity is independent of the data size according to a limiting argument (although as stated in \cite{Bierkins2019}, some logarithmic factors were ignored). 
\cite{cornish2019scalable} is another recent paper that proposes a Metropolis-Hastings-type MCMC algorithm based on subsampling that only accesses $\O(1)$ or even $\O(1/\sqrt{N_{D}})$ data points per step. Although this method was shown to have state-of-the-art performance on a 10-dimensional logistic regression example, its efficiency on high-dimensional models has not yet been demonstrated.


In this section, we will develop a version of UBUBU using stochastic gradients. We are going to use random variables of the form $\omega \in [N_D]^{N_b}$, which is a random selection of $N_b$ indices to be selected uniformly on $[N_D]=\{1,\ldots,N_{D}\}$, i.i.d. with replacement \cite{baker2019control}. We denote the distribution of $\omega$ here as $\mathcal{SWR}(N_D,N_{b})$.

\begin{definition}\label{def:subsampled_sg}
The sub-sampled stochastic gradient of $U$ at $x$ with respect to $\hat{x}$ is 
\begin{equation}\label{eq:SG}
    \mathcal{G}(x,\omega|\hat{x}) = \nabla U_{0}(x) + \sum^{N_D}_{i=1}\nabla U_{i}(\hat{x}) + \frac{N_D}{N_{b}} \sum_{i \in \omega}[\nabla U_{i}(x) - \nabla U_{i}(\hat{x})],
\end{equation}
where $\omega \sim \SWR(N_D,N_{b})$.
\end{definition}
$\mathcal{G}(\cdot|\hat{x}):\R^{d} \times \Omega \to \R^{d}$ is an unbiased estimator of $\nabla U(x)$ in the sense of Definition \ref{def:stochastic_gradient}. We can use this estimator in  $\UBU$ by replacing the $\mathcal{B}$ step with
\begin{equation}\label{eq:Bstochstep}\mathcal{B}_{\mathcal{G}}(x,v,h,\omega|\hat{x}) = (x,v - h\mathcal{G}(x,\omega|\hat{x})).\end{equation}

Let $x^{*} \in \R^{d}$ be the minimizer of the potential $U$, then the selection $\hat{x}=x^{*}$ at each step corresponds to the control variate gradient estimator, see \cite{baker2019control}.
When approximating the step $\mathcal{B}$ in $\UBU$ using this control variate approach, we can only achieve strong order $1/2$.

Another possibility is to update $\hat{x}$ every $\tau = \lceil N_D/N_b\rceil$ iterations with the latest position where the gradient was evaluated (this is not $x_k$ for UBU as the gradients are evaluated after moving forward by a $\mathcal{U}$ step with stepsize $h/2$). We refer to this as the stochastic variance reduced gradient (SVRG) approach (see \citep{johnson2013accelerating,Hu21optimal}). The overall computational cost of this approach is approximately twice that of the control variate approach (due to the need for a full gradient evaluation). Since the gradient is reevaluated every $\tau$ iterations, when $h$ is small, the position $\hat{x}$ becomes closer to the positions $x$ that are considered, and the approximate dynamics provide a better approximation of the underlying diffusion \eqref{eq:kinetic_langevin}. We will show that the SVRG discretization has strong order $3/2$, which is better than the control variate estimator, hence we will use it within our unbiasing scheme. The evolution of SVRG steps can be written as follows, {let $\left(\xi^{(i)}_{k+1}\right)^{4}_{i = 1} \sim \mathcal{N}(0_{d},I_{d}) \text{ for all } i = 1,...,4$ and $\omega_{k+1}\sim \SWR(N_D,N_{b})$, then}
\begin{equation}\label{eq:PhUBUG}
\begin{split}
&\left(\ol{x}_{k},\ol{v}_{k}\right)=
\mathcal{U}\left(x_{k},v_{k},h/2,\xi^{(1)}_{k+1},\xi^{(2)}_{k+1}\right)\\
&\hat{x}_k=\ol{x}_{\lfloor k/\tau\rfloor \tau}\\
&\left(x_{k+1},v_{k+1}\right)=
\mathcal{U}\left(\mathcal{B}_{\mathcal{G}}\left(\left.\ol{x}_{k},\ol{v}_{k},h,\omega_{k+1}\right|\hat{x}_{k}\right),h/2,\xi^{(3)}_{k+1},\xi^{(4)}_{k+1}\right).\\
\end{split}
\end{equation}
Let $P_{h}^{SVRG}$ denote the time inhomogenous Markov kernel describing the evolution of $(x_k,\hat{x}_k,v_{k})$ according to the SVRG steps \eqref{eq:PhUBUG}. 

{It turns out that it can be advantageous to never compute gradients at level $0$, and use a full Gaussian approximation. This is especially relevant in scenarios where the target is very close to Gaussian. We will refer to this setting as full Gaussian approximation at level $0$.} 

{For level $0$ and initialisation purposes}, let 
\begin{equation}\label{eq:Gaussian_approximation}
    \mu_G=\mathcal{N}(x^*,(H^*)^{-1})\times \mathcal{N}(0_d,I_d)\quad \text{ with  }\quad H^*=\nabla^2 U(x^*),
\end{equation}
and define
\begin{align}
\label{eq:Hstarstep}&\mathcal{H}_{*}(x,v,h) = \left(\begin{matrix}x^*\\ 0_d\end{matrix}\right)+\exp\left(h\left(\begin{matrix}0_d & I_d\\ -H_* & 0_d\end{matrix}\right)\right)\left(\begin{matrix}x-x^*\\ v\end{matrix}\right),\\
\label{eq:Ostep}
&\mathcal{O}(x,v,h/2,\xi^{(1)},\xi^{(2)}) = \Big(x, \eta v + \sqrt{2\gamma}\mathcal{Z}^{(2)}\left(h/2,\xi^{(1)},\xi^{(2)}\right)\Big),\\
\label{eq:O2step}
&\mathcal{O}^{2}(x,v,h,\xi^{(1)},\xi^{(2)},\xi^{(3)},\xi^{(4)}) = \mathcal{O}\left(\mathcal{O}\left(x,v,h/2,\xi^{(1)},\xi^{(2)}\right),h/2,\xi^{(3)},\xi^{(4)}\right).
\end{align}
with $\mathcal{H}_{*}(x,v,h)$ corresponding the solution of the Hamiltonian dynamics on target $\mu_{G}\times \mathcal{N}(0_{d},I_d)$ initiated in $(x,v)$ at time $h$. It follows from \eqref{eq:Z12def} that $\sqrt{2\gamma}\mathcal{Z}^{(2)}\left(h/2,\xi^{(1)},\xi^{(2)}\right)\sim \mathcal{N}(0_d, (1-\eta^2)I_d)$, so this $\mathcal{O}$ steps keeps the target invariant. We are going to use the OHO scheme as part of our algorithm.
{The OHO scheme is defined in more detail in Section F of the Appendix. Our motivation for using it is that a Gaussian approximation at the top level is computationally cheaper than stochastic or full gradients. We specifically use the OHO scheme for Gaussian targets (as opposed to  UBU with a Gaussian target) as it is simple to derive an analytical formula for the iterates and there is a natural coupling to the UBU scheme (see \eqref{eq:Phh2OHOUBUG}). Further, using a Gaussian approximation provides favourable complexity results in terms of scalability with the dataset size $N_D$, which we discuss later, and this can also be seen in practice. We proceed with OHO as follows, let $\left(\xi^{(i)}_{k+1}\right)^{4}_{i = 1} \sim \mathcal{N}(0_{d},I_{d}) \text{ for all } i = 1,...,4$, then we define}
\begin{equation}\label{eq:PhOHO}
\begin{split}
&\left(\ol{x}_{k},\ol{v}_{k}\right)=
\mathcal{O}\left(x_{k},v_{k},h/2,\xi^{(1)}_{k+1},\xi^{(2)}_{k+1}\right)\\
&\left(x_{k+1},v_{k+1}\right)=
\mathcal{O}\left(\mathcal{H}_{*}\left(\ol{x}_{k},\ol{v}_{k},h\right),h/2,\xi^{(3)}_{k+1},\xi^{(4)}_{k+1}\right).
\end{split}
\end{equation}
Let $P_{h}^{OHO}$ denote the time homogeneous Markov kernel describing the evolution of $(x_k,v_{k})$ according to the OHO steps \eqref{eq:PhOHO}. 

Two chains evolving according to SVRG with step size $h$ and SVRG with step size $h/2$ can be coupled as follows. {First $\left(\xi^{(i)}_{k+1}\right)^{8}_{i = 1} \sim \mathcal{N}(0_{d},I_{d}) \text{ for all } i = 1,...,8$, $\omega'_{k+1/2}, \omega'_{k+1}\sim \SWR(N_D,N_{b}), \omega_{k+1}\sim \frac{1}{2}\delta_{\omega_{k+1/2}'}+\frac{1}{2}\delta_{\omega_{k+1}'}$, then}
\begin{align}\label{eq:Phh2UBUG}
&\left(\ol{x}_{k},\ol{v}_{k}\right) =\mathcal{U}^{2}\left(x_{k},v_{k},h/2,\xi_{k+1}^{(1)},\xi_{k+1}^{(2)},\xi_{k+1}^{(3)},\xi_{k+1}^{(4)}\right)\\
\nonumber&\hat{x}_k=\ol{x}_{\lfloor k/\tau\rfloor \tau}\\
\nonumber&\left(x_{k+1},v_{k+1}\right) =
\mathcal{U}^{2}\left(\mathcal{B_{G}}\left(\left.
\ol{x}_{k},\ol{v}_{k},h,\omega_{k+1}\right|\hat{x}_{k}\right),h/2,\xi_{k+1}^{(5)},\xi_{k+1}^{(6)},\xi_{k+1}^{(7)},\xi_{k+1}^{(8)}\right),\\
\nonumber&\left(\ol{x}'_{k},\ol{v}'_{k}\right)=
\mathcal{U}\left(x'_{k},v'_{k},h/4,\xi^{(1)}_{k+1},\xi^{(2)}_{k+1}\right)\\
\nonumber&\hat{x}'_k=\ol{x}'_{\lfloor 2k/\tau\rfloor \tau/2}\\
\nonumber&\left(x'_{k+1/2},v'_{k+1/2}\right)=
\mathcal{U}\left(\mathcal{B_{G}}\left(\left.\ol{x}'_{k},\ol{v}'_{k},h/2,\omega'_{k+1/2},v\right|\hat{x}'_{k}\right), h/4,\xi_{k+1}^{(3)},\xi_{k+1}^{(4)}\right)\\
\nonumber&\left(\ol{x}'_{k+1/2},\ol{v}'_{k+1/2}\right)=\mathcal{U}\left(x'_{k+1/2},v'_{k+1/2},h/4,\xi_{k+1}^{(5)},\xi_{k+1}^{(6)}\right)\\
\nonumber&\hat{x}'_{k+1/2}=\ol{x}'_{\lfloor (2k+1)/\tau\rfloor \tau/2}\\
\nonumber&\left(x'_{k+1},v'_{k+1}\right)= 
\mathcal{U}\left(\mathcal{B_{G}}\left(\left.\ol{x}'_{k+1/2},\ol{v}'_{k+1/2},h/2,\omega'_{k+1}\right|\hat{x}'_{k+1/2}\right),h/4,\xi_{k+1}^{(7)},\xi_{k+1}^{(8)}\right)
\end{align}
Let $P_{h,h/2}^{SVRG}$ denote the time inhomogenous Markov kernel describing the evolution of $(x_k,\hat{x}_k,v_{k},x_k',\hat{x}'_k,v_k')$ according to the SVRG steps \eqref{eq:Phh2UBUG}. 

Finally, we will also need to couple one chain with step size $h$ running OHO on the Gaussian approximation $\mu_G$, and another chain running SVRG on the target with step size $h/2$. {First $\left(\xi^{(i)}_{k+1}\right)^{8}_{i = 1} \sim \mathcal{N}(0_{d},I_{d}) \text{ for all } i = 1,...,8$, $\omega'_{k+1/2}, \omega'_{k+1}\sim \SWR(N_D,N_{b})$ then}

\begin{equation}\label{eq:Phh2OHOUBUG}
\begin{split}
&\left(\ol{x}_{k},\ol{v}_{k}\right) =\mathcal{O}^{2}\left(x_{k},v_{k},h/2,\xi_{k+1}^{(1)},\xi_{k+1}^{(2)},\xi_{k+1}^{(3)},\xi_{k+1}^{(4)}\right)\\
&\left(x_{k+1},v_{k+1}\right) =
\mathcal{O}^{2}\left(\mathcal{H}_*\left(
\ol{x}_{k},\ol{v}_{k},h\right),h/2,\xi_{k+1}^{(5)},\xi_{k+1}^{(6)},\xi_{k+1}^{(7)},\xi_{k+1}^{(8)}\right),\\
&\left(\ol{x}'_{k},\ol{v}'_{k}\right)=
\mathcal{U}\left(x'_{k},v'_{k},h/4,\xi^{(1)}_{k+1},\xi^{(2)}_{k+1}\right)\\
&\hat{x}'_k=\ol{x}'_{\lfloor 2k/\tau\rfloor \tau/2}\\
&\left(x'_{k+1/2},v'_{k+1/2}\right)=
\mathcal{U}\left(\mathcal{B_{G}}\left(\left.\ol{x}'_{k},\ol{v}'_{k},h/2,\omega'_{k+1/2},v\right|\hat{x}'_{k}\right), h/4,\xi_{k+1}^{(3)},\xi_{k+1}^{(4)}\right)\\
&\left(\ol{x}'_{k+1/2},\ol{v}'_{k+1/2}\right)=\mathcal{U}\left(x'_{k+1/2},v'_{k+1/2},h/4,\xi_{k+1}^{(5)},\xi_{k+1}^{(6)}\right)\\
&\hat{x}'_{k+1/2}=\ol{x}'_{\lfloor (2k+1)/\tau\rfloor \tau/2}\\
&\left(x'_{k+1},v'_{k+1}\right)= 
\mathcal{U}\left(\mathcal{B_{G}}\left(\left.\ol{x}'_{k+1/2},\ol{v}'_{k+1/2},h/2,\omega'_{k+1}\right|\hat{x}'_{k+1/2}\right),h/4,\xi_{k+1}^{(7)},\xi_{k+1}^{(8)}\right)\\
\end{split}
\end{equation}
Let $P_{h,h/2}^{OHO/SVRG}$ denote the time inhomogenous Markov kernel describing the evolution of $(x_k,v_{k},x_k',\hat{x}'_k,v_k')$ according to the steps \eqref{eq:Phh2OHOUBUG}. 



We now create a coupling between levels $0$ and $1$, denoted by $\nu_{0,1}^{SG}$.
\begin{algorithm}[H]
\renewcommand{\thealgorithm}{}
\floatname{algorithm}{}
\addtocounter{algorithm}{-1}
     \footnotesize 
     \begin{algorithmic}[1]
     \State Define $z^{(0,1)}_{-B_{1}}\sim \mu_G$ and let $z'^{(0,1)}_{-B_{1}}= z^{(0,1)}_{-B_{1}}$.
     \State Let $(z^{(0,1)}_{-B_{1}},z'^{(0,1)}_{-B_{1}},\hat{x}'^{(0,1)}_{-B_1}),(z^{(0,1)}_{-B_{1}+1},z'^{(0,1)}_{-B_{1}+1},\hat{x}'^{(0,1)}_{-B_1+1}),\ldots, (z^{(0,1)}_{K},z'^{(0,1)}_{K},\hat{x}'^{(0,1)}_{K})$ be a Markov chain with kernel $R_{0,1}^{OHO/SVRG}=P^{OHO/SVRG}_{h_{0},h_1}$ 
     (satisfying that $z^{(0,1)}_{k}\sim \mu_G$ for all $k$).
    \State Let $\nu_{0,1}$ denote the joint distribution of $z_{-B_0}^{(0,1)},\ldots, z_{K}^{(0,1)}, z'^{(0,1)}_{-B_{1}},\ldots, z'^{(0,1)}_{K}$. 
\end{algorithmic}
\caption{$\nu_{0,1}^{SG}$ coupling}
\label{alg:nu01SG} 
\end{algorithm}
We now create a coupling between levels $l$ and $l+1$ for $l\ge 1$, denoted by $\nu_{l,l+1}^{SG}$.
\begin{algorithm}[H]
\renewcommand{\thealgorithm}{}
\floatname{algorithm}{}
\addtocounter{algorithm}{-1}
     \footnotesize 
     \begin{algorithmic}[1]
     \State Define $z^{(l,l+1)}_{-B_{l+1}}\sim \mu_G$ and let $z'^{(l,l+1)}_{-B_{l+1}}= z^{(l,l+1)}_{-B_{l+1}}$.
     \State Let $\left(z^{(l,l+1)}_{-B_{l+1}},z'^{(l,l+1)}_{-B_{l+1}}, \hat{x}'^{(l,l+1)}_{-B_{l+1}}\right),\ldots, \left(z^{(l,l+1)}_{-B_{l}},z'^{(l,l+1)}_{-B_{l}},\hat{x}'^{(l,l+1)}_{-B_{l}}\right)$ be a Markov chain with kernel $R_{l,l+1}^{OHO/SVRG}=\left(P^{OHO/SVRG}_{h_{l},h_{l+1}}\right)^{2^{l}}$ (satisfying that $z^{(l,l+1)}_{-B_{l}}\sim \mu_G$).
     \State Let $\left(z_{-B_{l}}^{(l,l+1)},z'^{(l,l+1)}_{-B_{l}},\hat{x}'^{(l,l+1)}_{-B_{l}}\right), \left(z_{-B_{l}+1}^{(l,l+1)},z'^{(l,l+1)}_{-B_{l}+1},\hat{x}'^{(l,l+1)}_{-B_{l}+1}\right),\ldots,  \left(z_{K}^{(l,l+1)},z'^{(l,l+1)}_{K},\hat{x}'^{(l,l+1)}_{K}\right)$ be a Markov chain evolving according to $R_{l, l+1}^{SVRG}=(P_{h_l,h_{l+1}}^{SVRG})^{2^{l}}$.
    \State Let $\nu_{l,l+1}^{SG}$ denote the joint distribution of $z_{-B_l}^{(l,l+1)},\ldots, z_{K}^{(l,l+1)}, z'^{(l,l+1)}_{-B_{l+1}},\ldots, z'^{(l,l+1)}_{K}$. 
\end{algorithmic}
\caption{$\nu_{l,l+1}^{SG}$ coupling}
\label{alg:nu02SG} 
\end{algorithm}
\begin{figure}
 \centering
 \includegraphics[width=\linewidth]{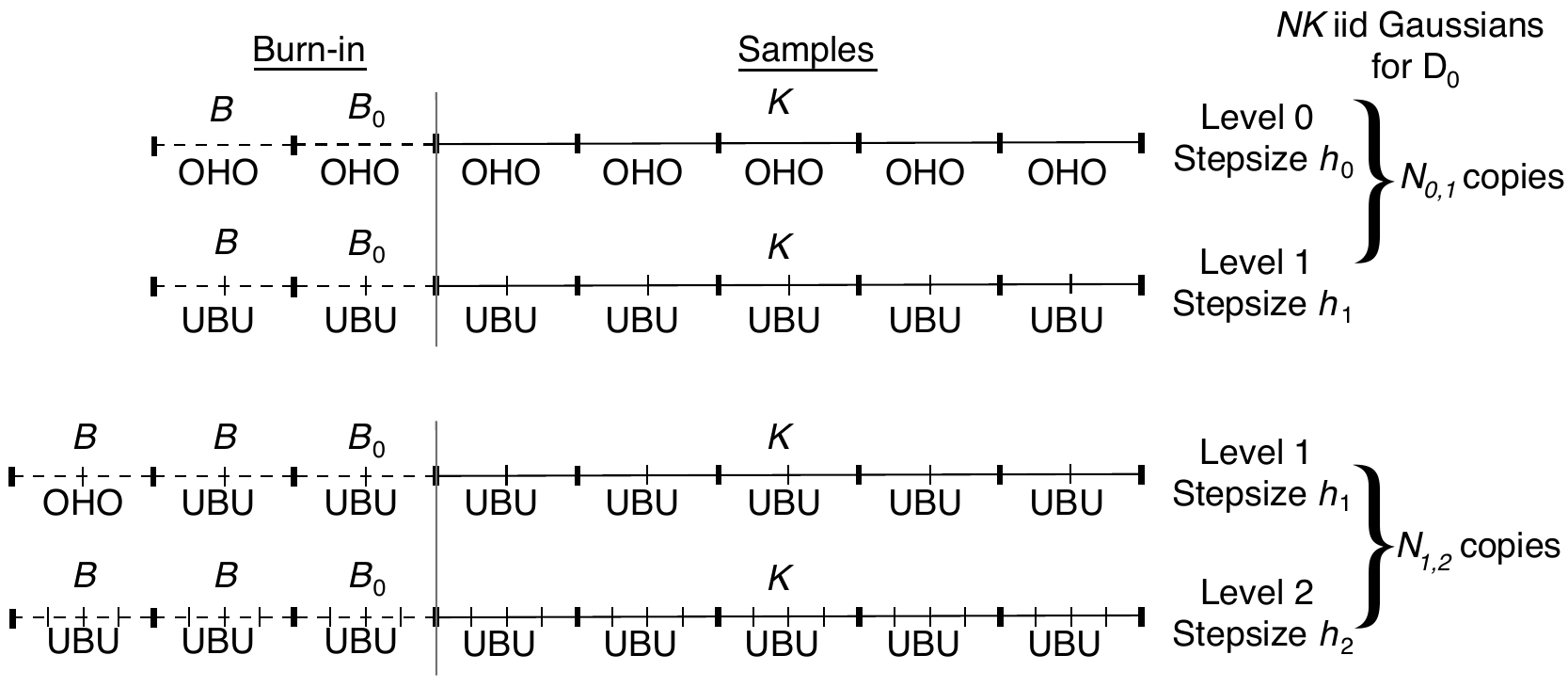}
 \caption{Coupling scheme for UBUBU-SG.}
 \label{fig:bias.elimination.sg}
 \end{figure}
Figure \ref{fig:bias.elimination.sg} illustrates our couplings between different levels using OHO/UBU discretizations. Given some constants $c_N>0$, $\phi_N>2$, we let
\begin{align}\label{eq:cllp1SG}
c_{l,l+1}&=c_N \phi_N^{-l}\,\text{ for }\,l\in \mathbb{N}.
\end{align}

Similarly to the exact gradient case, we can also define a joint coupling of levels $L(N), \ldots, l_{\max}$ in this stochastic gradient case. 
The idea is that first {generate} the Gaussian random vectors $\xi_{k+1}^{(1)},\ldots, \xi_{k+1}^{(4)}$ for the level $l_{\max}+1$, and then using the transformation $\mathcal{M}_{h_{l_{\max}+1}\to h_{l_{\max}}}, \ldots \mathcal{M}_{h_{L(N)+1}\to h_{L(N)}}$ recursively, we can generate a synchronous coupling of all the Gaussian random variables in the algorithm.
For the stochastic gradient noise terms $\omega_{k}$, we first define that at level $l_{\max}+1$, denoted by $\omega_{k}^{(l_{\max}+1)}$. From each $\omega_{2k}^{(l_{\max}+1)}$ and $\omega_{2k+1}^{(l_{\max}+1)}$, we generate $\omega_{k}^{(l_{\max})}$ to be one of them with equal probability $0.5$. We proceed recursively in the same way all the way to $\omega_{k}^{(L(N))}$.
Using the Gaussian random variables and $\omega_k^{(l)}$, we can define the synchronous coupling $\left\{z_{-B_l}^{(l)},\ldots, z_K^{(l)}\right\}_{L(N)\le l\le l_{\max}+1}$, called $\nu^{SG}_{L(N):l_{\max}}$.

Our stochastic gradient-based method (UBUBU-SG) proceeds as stated in Section B of the Appendix. 
We recommend setting the Richardson extrapolation parameter $c_{R}=\frac{1}{2\sqrt{2}}$ in this case (as SVRG has strong order 3/2). 

In order to show variance bounds for this algorithm, we make the following assumption.
\begin{restatable}[$\nabla$Lipschitz property]{assumption}{AssumLipSG}
\label{assum:LipSG}
For every $1\le i\le N_D$, $U_i:\R^{d} \to \R$ is twice differentiable and there exists a $\tilde{M}>0$ such that for all $x,y \in \R^d$,
$$
\| \nabla U_i(x) - \nabla U_i(y)\| \leq \tilde{M}\|x-y\|,
$$
for every $1\le i\le N_D$ and moreover,
$
\| \nabla U(x) - \nabla U(y)\| \leq M\|x-y\| \quad \text{for}\quad M=N_D\tilde{M}.
$
\end{restatable}

The next theorem states our bounds on the asymptotic variance for this algorithm.
\begin{restatable}{theorem}{ThmStochGradUBUBU}
\label{thm:stochgradUBUBU}
{Let us consider UBUBU with stochastic gradients}. Suppose that Assumptions \ref{assum:Lipschitz},
{and Assumptions \ref{assum:convex}, \ref{assum:LipSG} and \ref{assum:Hess_Lipschitz} hold with constants $N_{D}\Tilde{m}$, $N_{D}\Tilde{M}$ and $N_{D}\Tilde{M}^{s}_{1}$ respectively}. In addition we assume $\gamma \geq \sqrt{8M}$,
\begin{align*}  h_0 \leq \frac{C(\Tilde{\gamma},\Tilde{m},\Tilde{M},\tau/N_{D},N_{b})}{N^{3/2}_{D}}, \quad B\ge \frac{16\log(2^{3/2})\Tilde{\gamma}}{\Tilde{m}h_0N^{1/2}_{D}}, \quad B_0\ge \frac{16\Tilde{\gamma}}{\Tilde{m}h_0N^{1/2}_{D}}\log\left(\frac{1}{N^{9/4}_{D}h_0^{3/2}}\right).\end{align*}
Suppose that $c_{R}\in [0, \phi_N^{-1/2})$ and $2<\phi_N<8$. Then for any $N\ge 1$, the $\UBUBU$ estimator $S(c_R)$ has finite expected computational cost, $\E S(c_R)=\pi(f)$, and it has finite variance. 
Moreover, it satisfies a CLT as $N\to \infty$, and the asymptotic variance $\sigma^2_S$ defined in \eqref{eq:sigma2S} can be bounded as
\[\sigma^2_S\le \frac{1}{\Tilde{m}N_{D}K} + C(\Tilde{\gamma},\Tilde{m},\Tilde{M},\Tilde{M}^{s}_{1},\tau/N_{D},N_{b},\phi_{N})\frac{d^{2}}{c_{N}N^{2}_{D}}.
\]
\end{restatable}
\begin{proof}
See Section G of the Appendix.
\end{proof}
\begin{remark}
With the choice $c_N=\mathcal{O}\left(\frac{1}{N_D}\right)$ and $K=\mathcal{O}\left(1\right)$, we get a bound $\sigma_{S}^2\le \mathcal{O}\left(\frac{d^2}{\tilde{m}N_D}\right)$, which, except for the dimension dependence, is similar to the variance of a 1-Lipschitz function according to the target. Hence, obtaining an effective sample only requires evaluating a full gradient once per $\mathcal{O}(N_D)$ iteration, so there is no increase in computational cost as the dataset size $N_D$ increases. The dimension dependency $\mathcal{O}(d^2)$ in our bound is likely not sharp as we have not observed any dimension dependency in our simulations.
\end{remark}
\subsection{$\UBUBU$ with approximate gradients} 
Stochastic gradients are not the only possible approach for computing accurate approximations of the gradient. In case the potential is close to a Gaussian (which is typical in the big data regime due to the Bernstein-von-Mises theorem), the following approximation can be quite accurate.
\begin{definition}\label{def:approx_grad}
The quadratic approximate gradient of $U$ at $x$ with respect to $\hat{x}$ is defined by  
\begin{equation}\label{eq:approx_grad}
    \mathcal{Q}(x|\hat{x}) = \nabla U(\hat{x}) +\nabla^2 U(x^*) (x-\hat{x}),
\end{equation}
where $x^*$ is the minimizer of $U$.
\end{definition}

When using this approximation for the gradient, the $\mathcal{B}$ step becomes $\mathcal{B}_{\mathcal{Q}}(x,v,h|\hat{x}) = (x,v - h\mathcal{\mathcal{Q}}(x|\hat{x}))$.
The $\UBU$ iterations in this case become
\begin{equation}\label{eq:PhUBUQ}
\begin{split}
&\left(\ol{x}_{k},\ol{v}_{k}\right)=
\mathcal{U}\left(x_{k},v_{k},h/2,\xi^{(1)}_{k+1},\xi^{(2)}_{k+1}\right),\\
&\hat{x}_k=\ol{x}_{\lfloor k/\tau\rfloor \tau}\\
&\left(x_{k+1},v_{k+1}\right)=
\mathcal{U}\left(\mathcal{B_{Q}}\left(\left.\ol{x}_{k},\ol{v}_{k},h\right|\hat{x}_k\right),h/2,\xi^{(3)}_{k+1},\xi^{(4)}_{k+1}\right),
\end{split}
\end{equation}
{where $\left(\xi^{(i)}_{k+1}\right)^{4}_{i = 1} \sim \mathcal{N}(0_{d},I_{d}) \text{ for all } i = 1,...,4$.}
Let $P_{h}^{A}$ denote the time inhomogenous Markov kernel describing the evolution of $(x_k,\hat{x}_k,v_{k})$ according to the approximate gradient steps \eqref{eq:PhUBUQ}. 

The reference point $\hat{x}$ is updated after every $\tau$ iterations for some $\tau\ge 1$. We only need to evaluate the full gradient once per $\tau$ iterations, and use an approximation based on the Hessian at the minimizer otherwise. Since the Hessian $H^*=\nabla^2 U(x^*)$ only has to be computed once, this does not affect overall efficiency when the number of samples $N$ is sufficiently high. For many potentials of interest, the approximation steps in \eqref{eq:approx_grad} can be computed at a much smaller cost than the gradient of $U$. Moreover, when thinning is used (such at levels $l=1$ and higher), multiple steps according to \eqref{eq:PhUBUQ} can be combined into one using the fact that this is a linear system, further reducing the number of matrix-vector products required.

We follow a similar strategy as in the UBUBU-SG case (see Figure \ref{fig:bias.elimination.sg}). We use Gaussian samples at level $0$, and  couplings involving both OHO and UBU discretizations. 
At level $0$, we obtain i.i.d. samples from the Gaussian approximation $\mu_G=\mathcal{N}(x^*,(H^*)^{-1})\times \mathcal{N}(0_d,I_d)$. {Couplings between subsequent levels involve both OHO and UBU discretizations in the same way as UBUBU with stochastic gradients, but the stochastic gradient approximations are replaced with approximate gradient approximations described through \eqref{eq:PhUBUQ}.}



Our results for this algorithm are stated in Theorem \ref{thm:AgradUBUBUL4}.

\begin{restatable}{theorem}{ThmAGradUBUBULfour}
\label{thm:AgradUBUBUL4}
Considering UBUBU-Approx method, suppose that Assumption 
\ref{assum:Lipschitz} holds,  {and Assumptions \ref{assum:Lip}, \ref{assum:convex} and \ref{assum:Hess_Lipschitz} hold with constants $N_{D}\Tilde{M}$, $N_{D}\Tilde{m}$ and $N_{D}\Tilde{M}^{s}_{1}$ respectively}, and in addition $\gamma \geq \sqrt{8M}$,
\begin{align*} h_0 \leq \frac{C(\Tilde{\gamma},\Tilde{m},\Tilde{M},\tau/N_{D})}{N^{3/2}_{D}}, \quad B\ge \frac{16\log(2)\Tilde{\gamma}}{\Tilde{m}h_0 N^{1/2}_{D}}, \quad 
B_0\ge \frac{16\Tilde{\gamma}}{\Tilde{m}N^{1/2}_{D}h_0}\log\left(\frac{1}{ N^{3}_{D} h_0^{2}}\right).\end{align*}
Suppose that $c_{R}\in [0, \phi_N^{-1/2})$ and $2<\phi_N<4$.
Then for any $N\ge 1$, $S(c_R)$ has finite expected computational cost, $\E S(c_R)=\pi(f)$, and it has finite variance. 
Moreover, it satisfies a CLT as $N\to \infty$, and the asymptotic variance $\sigma^2_S$ defined in \eqref{eq:sigma2S} can be bounded as
\[\sigma^2_S\le \frac{1}{\Tilde{m}N_{D}K} + \frac{C(\Tilde{\gamma},\Tilde{m},\Tilde{M},\Tilde{M}^{s}_{1},\tau/N_{D},\phi_{N})d^{2}}{c_{N}N^{2}_{D}}.\]
\end{restatable}
\begin{proof}
See Section H of the Appendix.
\end{proof}

\begin{remark}
To control the asymptotic variance of Theorem \ref{thm:stochgradUBUBU} and Theorem \ref{thm:AgradUBUBUL4} for large $d$ we would need to set $h_{0} < \mathcal{O}(d^{-2})$; the dimension dependency in this bound might not be sharp, and we did not observe such limitations in our simulations. UBU iterations with AG and SVRG gradient approximations no longer form a time homogeneous Markov chain (unless the state space is extended), so it is challenging to establish $\mathcal{O}(1/K)$ scaling in the bound on $\sigma^2_S$, like in Theorem \ref{thm:exactgradUBUBU}. If we select $h_{0} \sim \mathcal{O}(1/N^{3/2}_{D})$, then for large $N_{D}$, the total computational cost of the approximate and stochastic gradient methods scales like $\mathcal{O}(N)$ due to  Proposition \ref{prop:unb}.  This is a significant improvement over UBUBU with exact gradients, which has a computational cost of $\O(N_D N)$. A comparision is provided below in Table \ref{table:comp2}.
\end{remark}
\begin{table}[h!]
\begin{center}

\begin{tabular}{ |c|c| } 
\hline
\textbf{Algorithm} & \textbf{Computational Cost} \\
\hline
\UBUBU \ (Exact gradients) & $\O(N_D N)$ \\ 
\UBUBU \ (stochastic gradients)  & $\O(N)$  \\ 
\UBUBU \ (approximate gradients) & $\O(N)$ \\ 
\hline
\end{tabular}

\end{center}
\caption{Comparison of the computational cost of the various $\UBUBU$ methods in terms of $N$ and $N_D$.}
\label{table:comp2}
\end{table}
\begin{remark}
{Although we have used Gaussian approximation at level $0$ in Theorem \ref{thm:stochgradUBUBU} and \ref{thm:AgradUBUBUL4} as this allows us to obtain better computational complexity in terms of $N_D$, one could also consider using UBU discretizations with SVRG or approximate gradients starting from level $0$. This might be advantageous when the Gaussian approximation is not accurate. One could also consider different initial distributions. It is straightforward to adapt the proofs of Theorem \ref{thm:stochgradUBUBU} to show that even in such situations, under appropriate assumptions on the burn-in times, the UBUBU-SG and UBUBU-AG methods produce unbiased estimators with finite variance. The computational complexity would have polylogarithmic dependency on $N_{D}$ in such scenarios (rather than no dependency on $N_{D}$).}
\end{remark}
\section{Numerical results}
\label{sec:num}
In this section, we provide numerical examples to demonstrate the effectiveness of our unbiased estimator $\UBUBU$ with exact, approximate and stochastic gradients. We test this on a range of problems, including (i) a Gaussian example,
(ii) a multinomial regression problem on the MNIST dataset, and (iii) a Poisson regression model for soccer scores. These computations serve to highlight the comparisons of our method with RHMC, which we view as the gold standard. 
{We briefly describe the latter in Section I of the Appendix.
}
For RHMC, we have used a partial refreshment parameter of $\alpha=0.7$, which typically performed $50\%-70\%$ better than doing full velocity refreshment ($\alpha=0$). We choose parameters $E_L$ (expected number of leapfrog steps) and $h$ (stepsize) such that the acceptance rate is in the range $0.65-0.8$ (as recommended in \cite{Beskos2013}), and that
$E_L h\approx \frac{1}{\sqrt{m}}$ ($m$ is the minimal eigenvalue of the Hessian at the mode), in line with the theoretical results for optimal convergence of the continuous time RHMC process \cite{Lu2022PDMP}. We found that the effective sample sizes (ESS) obtained in all of our experiments are in line with the continuous convergence rates of \cite{Lu2022PDMP} scaled by the stepsize $h$, so we do not think that other parameter choices can significantly improve the performance of RHMC.


Our numerical experiments with unbiased estimators are specific to the $\UBU$ splitting method, as was the analysis. We also ran some preliminary numerical experiments with an unbiased version of $\BAOAB$, but found that $\UBUBU$ was more efficient in all cases.
We estimated the ESS values based on at least 60 independent runs of each simulation. For UBUBU, the number of parallel chains $N$ was chosen in the range $N\in [64,256]$. We set $c_N=1/16$, $\phi_N=4$ for UBUBU with exact gradients, $c_N=1/16$, $\phi_N=4$ for UBUBU-SG, and $C_N=1/16$, $\phi_N=2\sqrt{2}$ for UBUBU-Approx. The friction parameter $\gamma$ was set as $\gamma=\sqrt{m}$ in all experiments, where $m$ is the minimal eigenvalue of the Hessian of the log-posterior at the mode (MAP).

The effective sample sizes values were estimated by first computing the variances of the estimators by adding together the variances of the individual terms $D_{l,l+1}^{(r)}$, which were estimated based on independent samples pooled together from all parallel runs. The variance of the last term $\ol{S}_{L(N),L(N)+1}+\ldots+\ol{S}_{l_{\max},l_{\max}+1}$ was estimated based on the values of it from independent parallel runs (one sample each). 
We also had to estimate the variances of the test functions according to the posterior distribution, this was done by computing the expectations $\E_{\pi}(f^2)$ and $\E_{\pi}(f)$ separately based on the samples from UBUBU, and using $\Var_{\pi}(f)=\E_{\pi}(f^2)-(\E_{\pi}(f))^2$. In order to estimate the errors of the effective sample size (ESS) values, we have implemented a bootstrap method, where the independent experiment's results were resampled with replacement, and generate bootstrap standard deviations \cite{diciccio1996bootstrap}.

The Python code of our simulation, based on JAX \cite{jax2018github}, is available at \texttt{\url{https://github.com/paulindani/UBUBU_JAX}}.

\subsection{Gaussian target}
Here we consider a Gaussian target in $d$ dimensions whose precision matrix has eigenvalues $$1,1+\frac{\kappa-1}{d-1},1+\frac{2(\kappa-1)}{d-1},\ldots, \kappa.$$

Theorem 4 of \cite{Lee21lower} has shown that for some Gaussian targets with condition number $\kappa$, the inverse spectral gap of HMC taking $K$ leapfrog steps per iteration is at least $\O\left(\frac{\kappa \sqrt{d}}{K\sqrt{\log(d)}}\right)$. More recently, it has been shown that randomizing the integration time can substantially improve the performance of HMC \cite{RHMC}. In continuous time, sharp convergence results have been obtained for RHMC in \cite{Lu2022PDMP}.  Moreover, for Gaussians with condition number $\kappa$, RHMC can approximate the target distribution with $\O(\sqrt{\kappa} d^{1/4})$ queries under a warm-start assumption \cite{apers2022hamiltonian}.  In our preliminary experiments, RHMC significantly outperformed HMC on high-dimensional problems, so we only consider RHMC here.
 Figure \ref{fig:Gaussian-dim-dependence-norm} shows the number of gradient evaluations per ESS for the norm test function $f(x)=\|x\|$ as a function of the dimension $d$.
As we can see UBUBU does not show any dimension dependence, while the number of gradient evaluations per ESS scales as $\O(d^{1/4})$ for RHMC. In our experiments, UBUBU is 16-18 times more efficient than RHMC for $d=10^5$. {We have additional experiments for this example in Section B of the Appendix.}


 \begin{figure}[h!]
 \centering
 \includegraphics[width=0.495\linewidth]{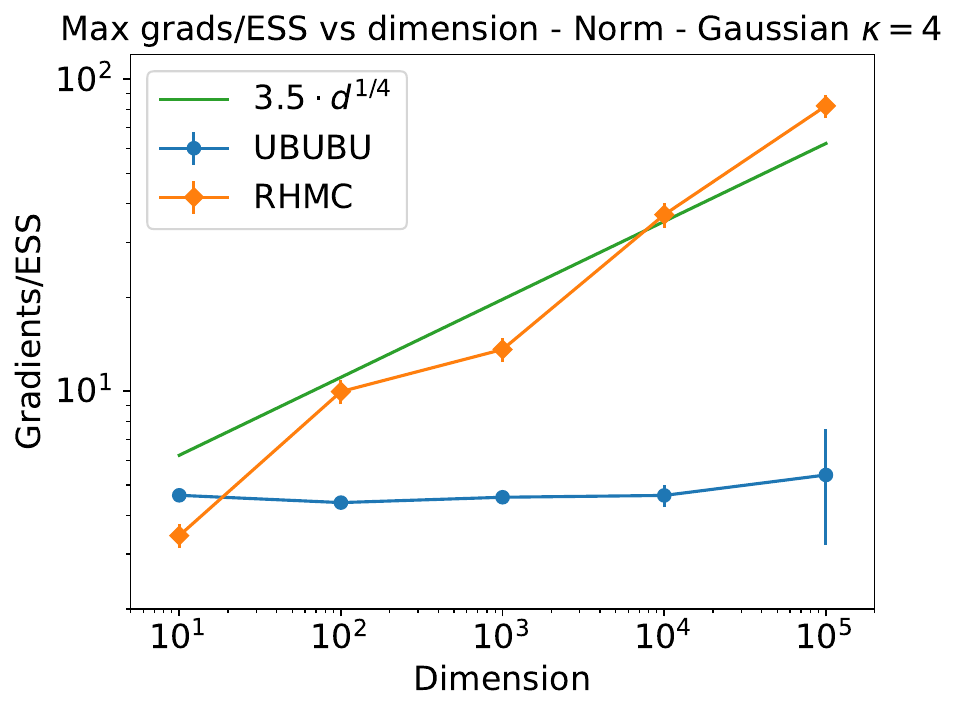} 
 \includegraphics[width=0.495\linewidth]{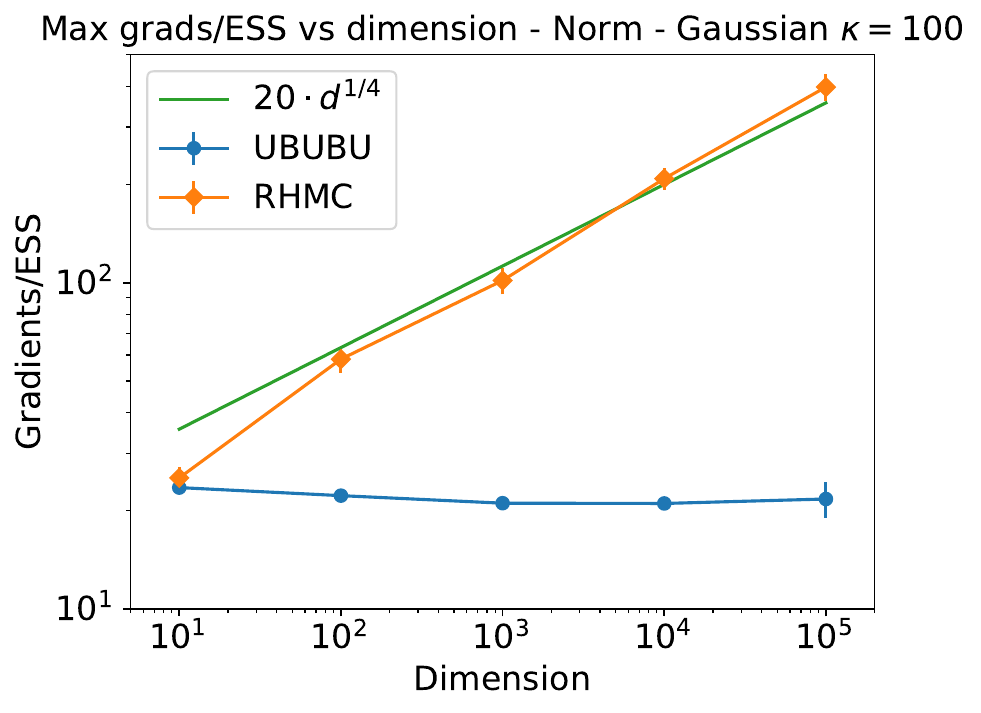} 
 \caption{Dimension dependence of gradients/ESS for test function $\|x\|$ for Gaussian targets.}
 \label{fig:Gaussian-dim-dependence-norm}
 \end{figure}


It is important to consider the dimension dependence of the variance of the original unbiased kinetic Langevin estimator based on Euler–Maruyama discretization presented in \cite{ruzayqat2022unbiased}. 
Due to the different estimator proposed there, the number of samples $N_{l,l+1}$ is random for every $l$, and the variance of the term equivalent to  $S_{l,l+1}=\frac{1}{N_{l,l+1}}\sum_{i=1}^{N_{l,l+1}}D_{l,l+1}^{(r)}$ will be proportional to $\E(D_{l,l+1}^2)$, not $\Var(D_{l,l+1})$ as in our case. For functions like the norm $f(x,v)=\|x\|$, in general, using the strong order one property of the Euler–Maruyama scheme (\cite{sanz2021wasserstein}),  $\E(D_{l,l+1})=\O(\sqrt{d} h_l)$ and $\E(D_{l,l+1}^2)=\O(d h_l^2)$. So the asymptotic variance of the final estimator is $\O(1+d h_0^2)$, and by choosing $h_0=\O(d^{-1/2})$, we expect that this will require $\O(d^{1/2})$ gradient evaluations per effective sample.

\subsection{Bayesian multinomial regression}

Our second numerical example is to consider a Bayesian multinomial regression (BMR) problem. BMR is a generalized linear regression model which estimates probabilities for $r$ different categories of dependent variable $y$ using a set of explanatory variables $x$. Here, provided $m$ classes, we let $q=(q^{1},\ldots q^{m})\in \R^{d}$ with $d=m d_o$ and $q^i\in \R^{d_o}$. The likelihood associated with the problem is given as
\begin{equation}
\label{eq:mr_like}
p(y^j|q) = \frac{\exp(\langle x^j, q^{y^j}\rangle)}{\sum_{1\le k\le m} \exp(\langle x^j, q^{k}\rangle)}.
\end{equation}
Our focus is on estimating a posterior distribution, where the posterior potential is given as
$U(q) = - \log(p_0(q)) - \sum^{N_{D}}_{k=1}\log\left(p(y^j|q) \right)$.
Here we chose $p_0$ as a Gaussian prior $p_0(q)=\frac{\exp(-\|q\|^2/(2\sigma_0^2))}{(\pi\sigma_0^2)^{d/2}}$.
In Section I of the Appendix, we show that the gradient-Lipschitz and strongly Hessian Lipschitz conditions (Assumptions \ref{assum:Lip} and \ref{assum:Hess_Lipschitz}) hold for this example.
We are interested in applying our BMR model to the MNIST dataset \cite{MNIST} about classifying handwritten digits from 0 to 9. 
The dataset contains 60,000 training data points and 10,000 test data
points where the images are of size 28 by 28 pixels. The covariate vectors $x^{j}$ are obtained by flattening the images into vectors taking values on the interval $[0,1]$, and adding a 1 in the end for the intercept term. Hence $d_0=28^2+1=785$, $m=10$, and $d=d_0m=7850$. We set the prior variance $\sigma_0^2=0.1$ (this was tested to provide good prediction performance).

For our numerical simulations, we will present two different scenarios: one without preconditioning, and one with preconditioning. By preconditioning, we mean that we obtain samples from a transformed potential $U(Ax)$ for some matrix $A$, which may have a better condition number than the original potential. It is easy to see that if $X$ follows a distribution with density proportional to $\exp(-U(x))$, then $X'=A^{-1}X$ has a density proportional to $\exp(-U(Ax))$. In addition to the coordinate test functions, we have also evaluated the efficiency of these methods for the posterior predictive probability of digits $0,1\ldots, 9$ on the test dataset consisting of 1000 images. We have only considered digits whose probability according to the model with parameters set at the MAP (maximum-a-posteriori) falls in the interval $[0.1,0.9]$ - there were 2210 such instances, these were used as our test functions. When a class has very low or very high probability, the posterior variance is very small, which makes calculating ESS values challenging.
 


{Our numerical simulations are presented in Section B in the Appendix. }
We have summarized our results in Table \ref{table:compeffMNIST}. When using preconditioning, the UBUBU approach significantly outperforms RHMC. The best performance is obtained by preconditioned UBUBU-Approx. {The last column shows the overall throughput in terms of minimal ESS/seconds amongst all components. Our JAX-based implementation ran on an RTX 5090 GPU, using \texttt{vmap} command to exploit parallel computational capacity. UBUBU-based methods significantly outperform RHMC in ESS/second. Preconditioned UBUBU-SG has a relatively low ESS/second, possibly due to the overhead of additional memory use.}


\begin{table}[H]
{
\begin{tabular}{ |c|c|c|c| } 
\hline
Algorithm & Test functions & Max grads/ESS ($\pm$ sd) & Min ESS/sec\\ 
\hline
RHMC & Coordinates & 365.94 ($\pm  26.64$) &  17.76\\
UBUBU & Coordinates & 332.06 ($\pm 2.69$)& 69.72\\
Preconditioned RHMC & Coordinates &24.39 ($\pm 1.05$) & 220.27 \\
Preconditioned UBUBU & Coordinates & 3.446 ($\pm 0.115$)& 6325.43 \\
Preconditioned UBUBU-SG & Coordinates &  0.827($\pm 0.021$)& 2567.71 \\
Preconditioned UBUBU-Approx & Coordinates &  0.192($\pm 0.0027$)&42758.04 \\
\hline
RHMC & Test set prediction & 245.80 ($\pm 19.65$)& 26.45\\
UBUBU & Test set prediction & 215.69 ($\pm 2.82$)& 107.34\\
Preconditioned RHMC & Test set prediction & 29.18 ($\pm 1.69$) &184.14  \\
Preconditioned UBUBU & Test set prediction & 3.62($\pm 0.094$) &6014.28 \\
Preconditioned UBUBU-SG & Test set prediction& 2.64 ($\pm 0.141$) &  803.716\\
Preconditioned UBUBU-Approx & Test set prediction& 0.916($\pm 0.012$) & 8961.91\\
\hline
\end{tabular}}
\caption{Computational efficiency for MNIST dataset. Standard deviations evaluated by bootstrapping.}
\label{table:compeffMNIST}
\end{table}

\subsection{Poisson regression model}
\label{subsec:BPM}

Our final example is a  Poisson regression model for predicting soccer scores taken from \cite{koopman2015foot}.

Let $g=1,\ldots, G$ be the index of games. Let $S_g^{H}$ denote the number of goals scored by the home team at game $g$, and let $S_g^{A}$ denote the number of goals scored by the away team. The independent Poisson model \cite{maher1982modelling} assumes that these scores are distributed as 
$S_g^{H}\sim \text{Poisson}(\lambda_{g}^{H}), \quad S_g^{A}\sim \text{Poisson}(\lambda_{g}^{A})$, conditionally independently given the rates $\lambda_{g}^{H}$ and  $\lambda_{g}^{A}$. In our implementation, the rates are connected to the linear predictors $\eta_{g}^{H}$ and $\eta_{g}^{A}$ using the function $\text{softplus}(x)=\log(1+\exp(x))$, 
i.e. $\lambda_{g}^{A}=\text{softplus}(\eta_{g}^{A}), \quad \lambda_{g}^{H}=\text{softplus}(\eta_{g}^{H}).$
This function is Lipschitz and also gradient Lipschitz, which is desirable given our theory. Although this is less frequently used in the literature than the log link function, it was shown to be more robust and less sensitive to outliers \citep{wiemann2021using, weiss2022softplus}. 
The linear predictors are modelled based on a random effect model with time-dependent attacking and depending strengths for each team. Let $w(g)$ denote the week of game $g$, then we set
\begin{align}\label{eq:etaAHdef}
\eta_{g}^{H}= a_{\text{home.team(g)},w(g)}+d_{\text{away.team(g)},w(g)},\quad \eta_{g}^{A}= a_{\text{away.team(g)},w(g)}+d_{\text{home.team(g)},w(g)}.
\end{align}
Let $\mathbf{a}$ be all attacking strengths of all teams over the whole period, and $\mathbf{d}$ denote all defending strengths. Then the log-likelihood $ \log(p(\mathbf{a},\mathbf{d}))$ is of the form
\begin{align*}
C(S_1^H,\ldots, S_{G}^H,S_{1}^A,\ldots, S_G^H)+\sum_{g=1}^{G}\left(-\lambda_g^H+ S_g^H\log(\lambda_g^H)-\lambda_g^A+ S_g^H\log(\lambda_g^H)\right),
\end{align*}
which can be written as a function of $\mathbf{a}$ and $\mathbf{d}$ using \eqref{eq:etaAHdef}, and the linear predictors. 
We used a Gaussian random walk prior for the attacking/defending strengths $a_{\text{team},w}$ and $d_{\text{team},w}$, together with a weak Gaussian prior on every attacking and defending strength. Let $\mathcal{T}$ denote the set of teams during the whole period considered (teams change from season to season due to relegation/promotion), then the overall log prior is of the form
\begin{align*}\log p_0(\mathbf{a},\mathbf{d})&=C(\sigma,\sigma_0)
-\sum_{\text{team}\in \mathcal{T}}\left(\sum_{w=w(1)}^{w(G)}\frac{a_{\text{team},w}^2}{2\sigma_0^2}+\sum_{w=w(1)}^{w(G)-1}  \frac{(a_{\text{team},w+1}-a_{\text{team},w})^2}{2\sigma^2}\right)\\
&-\sum_{\text{team}\in \mathcal{T}}\left(\sum_{w=w(1)}^{w(G)}\frac{d_{\text{team},w}^2}{2\sigma_0^2}+\sum_{w=w(1)}^{w(G)-1}  \frac{(d_{\text{team},w+1}-d_{\text{team},w})^2}{2\sigma^2}\right)
,\end{align*}
We set $\sigma^2=0.01$ (this means a strong correlation for about two years), and $\sigma_0^2=10$ (weakly informative prior). We considered 20 years of Premier League data (7600 games) from 19/08/2000 until 26/07/2020. Our model has $d=89526$ parameters, and the condition number of the Hessian at the mode is $\kappa\approx 4\cdot 10^3$. 

We have implemented RHMC, UBUBU and UBUBU-Approx with $\tau=10$ for this model. In the UBUBU-Approx algorithm, the target at level $0$ was chosen as the Gaussian approximation (with mean $x^*$, and precision matrix $\nabla^2 U(x^*)$), meaning that gradient evaluations were only used from level $1$ onwards. The test functions were chosen as $f(x)=x_1, \ldots, f(x)=x_d$. {Our numerical simulations are presented in Section B in the Appendix. 
However, we have summarized our results for this dataset in Table \ref{table:compeffPoisson}.
As we can see, UBUBU uses approximately 14 times fewer gradient evaluations per effective sample than RHMC, and UBUBU-Approx uses approximately 5000 times fewer gradient evaluations than RHMC. In the case of UBUBU-Approx, the vast majority of the runtime is spent on generating the Gaussian samples at level 0 (these are generated using Cholesky decomposition, and then solving sparse lower triangular matrix-based linear systems). Generating a Gaussian sample this way takes approximately 20 times longer than evaluating a gradient of the log posterior. Further speedups might be possible as the sparse solvers on GPUs become more mature.
\begin{table}[H]
{
\begin{tabular}{ |c|c|c|c| } 
\hline
Algorithm & Test functions & Max grads/ESS ($\pm sd$) & Min ESS/sec\\ 
\hline
RHMC & Coordinates & 1649.74 ($\pm 68.66$) & 57.26\\
UBUBU & Coordinates & 116.73 ($\pm 1.41$) & 1082.61\\
Approx. UBUBU & Coordinates & 0.328 ($\pm 0.0098$) & 9985.09\\
\hline
\end{tabular}}
\caption{Computational efficiency for Poisson regression model. Standard deviations evaluated by bootstrapping.}
\label{table:compeffPoisson}
\end{table}
}

\section{Conclusion}
\label{sec:conc}
In this article, we presented a new unbiased estimator which can exploit high strong-order numerical integrators for underdamped Langevin dynamics. We refer to our estimator as $\UBUBU$ which does not rely on the Metropolis acceptance/reject step. Our estimator is constructed using a telescoping sum for different discretization levels \cite{rhee2015unbiased,glynn2014exact}. We were able to gain various theoretical insights, which include showing unbiasedness and finite variance, a central limit theorem, and asymptotic and non-asymptotic bounds on the variance for three algorithms, based on exact, stochastic, and approximate gradients. We have studied the behaviour of our algorithm for product target distributions and shown that for a large class of test functions, it has dimension-independent computational complexity. For stochastic gradients, we also considered the dependency on the size of the data in the big data limit and shown that our method is very efficient in such situations. The proof of these results relies on Wasserstein contraction results for the UBU dynamics. We provided numerical experiments verifying our theory and demonstrating the performance gains over other well-known methods such as randomized HMC. {We have considered a range of models including an MNIST multinomial regression and a Poisson regression model.} 

In terms of future work, there are various directions which could be taken up. One of them is related to exploiting higher-order schemes, which were provided in \cite{Foster2023,LyonsHighOrder}. Numerical results indicate strong orders of up to 4. \cite{LyonsHighOrder} has proven strong order $3/2,5/2$ and $3$ under gradient Lipschitz, Hessian Lipschitz and third-order Lipschitz assumptions, respectively. However, the dimensional dependence obtained under each of these assumptions has not been shown to improve on the UBU scheme in \cite{sanz2021wasserstein}. Furthermore, such splitting schemes typically require more than one gradient evaluation per step, unlike our strategy.
In a different direction, one could consider integrators adapted to potentials that do not have the gradient-Lipschitz property (such as in the case of sparsity-inducing priors \cite{park2008bayesian} or log link functions).
Other potential directions are nested expectations \cite{UnbiasedInt} or the setting where one does not assume convexity \cite{Eberle2019,majka2020,chakreflection,schuh2024}.


\begin{funding}
The authors acknowledge the support of the Engineering and Physical Sciences Research Council Grant EP/S023291/1 
(MAC-MIGS Centre for Doctoral Training).
NKC is supported by an EPSRC-UKRI AI for Net Zero Grant: ``{Enabling CO2 Capture And Storage Projects Using AI}", (Grant EP/Y006143/1).
NKC is supported by a City University of Hong Kong Start-up Grant, project number: 7200809. DP is supported by a Nanyang Technological University Start-up Grant, project number: 024968-00001.

\end{funding}

\begin{appendix}

\section{Outline of results and notation}
\label{supp:appendix:outline}

The beginning of the supplementary material is devoted to providing a road-map for our results. 
{In Section \ref{supp:app:alg} we state our two main algorithms of UBUBU and SG-UBUBU. We also provide the specific form of coupling for the approximate-gradient UBUBU, while presenting our figures associated with the numerical experiments in the main file.}
In Section \ref{supp:sec:Appendix:proofs:multilevel}, which follows, we provide variance estimates of the full gradient multilevel $\UBUBU$ method. The approach we use is to bound $\Var{\left(D_{0}\right)}$ using Theorem 2 of \cite{Joulin2010} and to use the strong error estimates of \cite{sanz2021wasserstein} for $\UBU$ to estimate $\Var{\left(D_{l,l+1}\right)}$. \cite{Joulin2010} requires Ricci curvature of the $\UBU$ Markov chain and extending \cite{sanz2021wasserstein} to global strong error estimates in Section \ref{supp:sec:local_to_strong}  requires Wasserstein convergence. We provide this in Section \ref{supp:Sec:ProofConvergenceUBU} in the full gradient setting using the methods of \cite{leimkuhler2023contractiona}. We provide $L^{4}$ Lyapunov drift inequalities in the full gradient setting. We can then bound the average distance to the minimizer non-asymptotically, the key result needed to get complexity bounds in the big data setting. We also provide the proof of the central limit theorem of the estimator in Section \ref{supp:sec:Appendix:proofs:multilevel}.

In Section \ref{supp:sec:Appendix_UBUBU_variance_bnd_exact} we provide variance bounds and estimates on our estimator UBUBU with exact gradients. 
In Section \ref{supp:sec:initial_conditions} we describe the initialization and the OHO scheme for the approximate and stochastic gradient methods and some estimates of the distance between the initial measure and the target measure. We then use the techniques of \cite{Hu21optimal} to provide global strong error estimates of the SVRG method. We combine and extend the techniques of \cite{sanz2021wasserstein} and \cite{Hu21optimal} to prove new non-asymptotic stochastic gradient error bounds for the $\UBU$ integrator. 
From this we extend the results of Section \ref{supp:sec:Appendix_UBUBU_variance_bnd_exact} to providing estimates of the variance of our multilevel estimator in the SVRG stochastic gradient setting in Section \ref{supp:sec:Appendix_UBUBU_variance_bnd_svrg}. 

We further develop bounds for our new approximate gradient $\UBU$ method in Section \ref{supp:sec:Appendix_UBUBU_variance_bnd_approx_grad} using the same techniques, in the approximate gradient setting. In general, Appendices \ref{supp:sec:Appendix_UBUBU_variance_bnd_svrg}, \ref{supp:sec:Appendix_UBUBU_variance_bnd_approx_grad} follow similarly where one requires bounds on the variance of the quantity $D_0$ and $D_{l,l+1}$. However, we use an interpolation argument to improve the results in Section \ref{supp:sec:Appendix_UBUBU_variance_bnd_approx_grad} as opposed to the methods used in Section \ref{supp:sec:Appendix_UBUBU_variance_bnd_svrg}. We also use some classical results from the theory of ODEs to establish bounds between continuous diffusions to establish the variance of $D_{0,1}$ in Section \ref{supp:sec:Appendix_UBUBU_variance_bnd_svrg} and \ref{supp:sec:Appendix_UBUBU_variance_bnd_approx_grad}.
Finally, we provide some auxiliary results in Section \ref{supp:app:disc_bounds}.

\begin{itemize}
\item We let $z_{0:k}=\left(z_0,z_1,\ldots, z_k\right)$ denote a sequence of variables.
\item Let $0_d$ denote the $d$-dimensional vector of zeros. 
\item Let $I_d$ denote the $d$-dimensional identity matrix. 
\item Let $C$ denote an absolute constant (whose value may differ in each proposition or theorem).
\item Let $C(\mathrm{var}_1,\ldots,\mathrm{var}_n)$ denote a constant that is a function of variables $\mathrm{var}_1,\ldots,\mathrm{var}_n$ (this function may differ in each proposition or theorem).
\item Let $G$, $SG$ and $A$ denote an abbreviation for gradient, stochastic gradient and ``approximate gradient".
\item We let $l \in \mathbb{R}^+$ denote the level of discretization with respect to our discretized ULD, with stepsize $h_l$ defined at each level.
\item Let $D_0$ denote the empirical average of samples at level $0$. 
\item Let $D_{l,l+1}$ denote the difference of empirical averages of samples at levels $l+1$ and $l$, which are generated jointly via a synchronous coupling.
\item $N$ denotes the number of samples taken at level $0$.
\item \textcolor{black}{$N_{l,l+1}$ denotes the number of samples taken from the coupling of levels $l$ and $l+1$. }
\item $N_D$ is the size of the dataset (number of terms in potential $U(x)=U_0(x)+\sum_{i=1}^{N_D} U_i(x)$).
\item Let $z_k=(x_k,v_k)$ denotes step $k$ in a numerical discretization of kinetic Langevin dynamics with time step $h$ (specified each time this notation is used). Similarly, $Z_t$ is the solution of the continuous kinetic Langevin dynamics initialized at the invariant measure with synchronously coupled Brownian motion. $Z^k=Z_{kh}$ denotes the value of the continuous time process at the same time as $z_k$.
\item $\hat{x}_k$ denotes the point where the last gradient is evaluated for SVRG
\item $\|z\|_{L^{2}} := \left(\mathbb{E}\|z\|^{2}\right)^{1/2}$ and $\|z\|_{L^{2},a,b} := \left(\mathbb{E}\|z\|^{2}_{a,b}\right)^{1/2}$.
\end{itemize}

\newpage

\section{Algorithms and additional numerics}
\label{supp:app:alg}

\subsection{Full gradient UBUBU}
\par\noindent
\begin{algorithm}[H]
     \footnotesize
     \begin{algorithmic}[1]
    \State \textbf{Input:} \begin{itemize}
 \item Maximum stepsize $h_0$.
 \item Friction parameter $\gamma>0$.
  \item Initial distribution $\mu_0$ on $\R^d\times \R^d$ for $l\ge 0$.
  \item Potential function $U:\R^d\to \R$ of target distribution.
  \item Burn-in length parameters $B_0$ and $B$.
  \item Number of samples parameter $K$.
  \item Number of parallel chains parameters $N$, $c_N$ and $\phi_N$.
  \item Richardson extrapolation parameter  $c_{R}\in [0,\phi_N^{-1/2})$ (default value $c_R=\frac{1}{4}$).
  \item Test function $f$.
 \end{itemize}
    \State \textbf{Averages from level $0$:} 
    \For {$r=1,\ldots, N$}    
    \State Sample $z_{-B_0}^{(0,r)},\ldots, z_K^{(0,r)}$ from $\nu_0$.
    \State Compute $D_0^{(r)}$ using the samples $z_{1}^{(0,r)},\ldots, z_K^{(0,r)}$.
    \EndFor     
\State Compute $S_{0}:=\frac{1}{N}\sum_{i=1}^{N}D_0^{(r)}$.
\State \textbf{Generate number of chains}:
\State Sample $N_{l,l+1}$, let $l_{\max}=\max\{l: N_{l,l+1}>0\}$. 
\State \textbf{Averages of differences $D_{l,l+1}$ from  $l=0,\ldots, l_{\max}$:} 
\For {$l=0,\ldots, L(N)-1$}
\For {$r=1,\ldots, N_{l,l+1}$}
\State Sample
$z_{-B_l}^{(l,l+1,r)},\ldots, z_{K}^{(l,l+1,r)}, z'^{(l,l+1,r)}_{-B_{l+1}},\ldots, z'^{(l,l+1,r)}_{K}$
according to $\nu_{l,l+1}$.
\State Compute $D_{l,l+1}^{(r)}$ using   $z_{1}^{(r,l,l+1)},\ldots, z_{K}^{(r,l,l+1)}, z'^{(r,l,l+1)}_{1},\ldots, z'^{(r,l,l+1)}_{K}$.
\EndFor
\EndFor
\State Sample
$\left\{z_{-B_l}^{(l)},\ldots, z_{K}^{(l)}\right\}_{l(N)\le l\le l_{\max}+1}$
according to $\nu_{L(N):l_{\max}+1}$.
\For {$l=L(N),\ldots, l_{\max}$}
\State If $N_{l,l+1}=1$, let $\left({z'}_{1}^{(1,l,l+1)},\ldots, {z'}_{K}^{(1,l,l+1)}\right):=\left(z_{1}^{(l+1)},\ldots, z_{K}^{(l+1)}\right)$.
\State If $N_{l,l+1}=1$, let $\left(z_{1}^{(1,l,l+1)},\ldots, z_{K}^{(1,l,l+1)}\right):=\left(z_{1}^{(l)},\ldots, z_{K}^{(l)}\right)$.
\State If $N_{l,l+1}=1$, compute $D_{l,l+1}^{(1)}$. 
\State Compute $S_{l,l+1}$.
\EndFor
\State Compute $S(c_R)$.
 \State \textbf{Output:} 
 \State Unbiased estimator $S(c_R)$, 
 \State Samples  $z_{1}^{(0,r)},\ldots, z_K^{(0,r)}$ for parallel chains $1\le r\le N$,
 \State Samples $z_{1}^{(l,l+1,r)},\ldots, z_{K}^{(l,l+1,r)}, z'^{(r,l,l+1)}_{1},\ldots, z'^{(r,l,l+1)}_{K}$ for $0\le l\le l_{\max}$,  chains $1\le r\le N_{l,l+1}$.     
\end{algorithmic}

 	\caption{Unbiased-$\UBU$ ($\UBUBU$)}
 	\label{supp:alg:UBUBU}
  
 \end{algorithm}
\newpage
\subsection{Stochastic gradient UBUBU}
\par\noindent
\begin{algorithm}
     \footnotesize 
     \begin{algorithmic}[1]
    \State \textbf{Input:} \begin{itemize}
 \item Maximum stepsize $h_0$.
 \item Friction parameter $\gamma>0$.
 \item Individual potential terms $(U_i)_{0\le i\le N_D}$.
  \item Minimizer $x^*$ of Potential $U(x)$ and its Hessian $H^*=\nabla^2 U(x^*)$. 
  \item Batch size parameter $N_{b}$ (related to $\tau=\lceil N_D/N_b\rceil$).
  \item Burn-in length parameters $B_0$ and $B$.
  \item Number of samples parameter $K$.
  \item Number of parallel chains parameters $N$, $c_N$ and $\phi_N$.
  \item Richardson extrapolation parameter  $c_{R}\in [0, \phi_N^{-1/2})$ (default value $c_R=\frac{1}{2\sqrt{2}}$).
  \item Test function $f$.
 \end{itemize}
    \State \textbf{Samples from Gaussian approximation at level $0$:} 
    \State Sample $N K$ i.i.d. samples $z^{(0)}_{1},\ldots, z^{(0)}_{NK}$ from $\mu_{G}$. 
    \State Compute $S_{0}:=\frac{1}{NK}\sum_{i=1}^{NK} f(z^{(0)}_{i})$.
\State \textbf{Generate number of chains}:
\State Sample $(N_{l,l+1})_{l\ge 0}$, let $l_{\max}=\max\{l: N_{l,l+1}>0\}$. 
\State \textbf{Averages of differences $D_{l,l+1}$ from  $l=0,\ldots, l_{\max}$:} 
\For {$l=0,\ldots, L(N)-1$}
\For {$r=1,\ldots, N_{l,l+1}$}
\State Sample $z^{(l,l+1,r)}_{-B_{l}},\ldots, z_{K}^{(l,l+1,r)}, z'^{(l,l+1,r)}_{-B_{l+1}},\ldots, z'^{(l,l+1,r)}_{K}$
from $\nu_{l,l+1}^{SG}$.
\State Compute $D_{l,l+1}^{(r)}$ using   $z_{1}^{(r,l,l+1)},\ldots, z_{K}^{(r,l,l+1)}, z'^{(r,l,l+1)}_{1},\ldots, z'^{(r,l,l+1)}_{K}$.
\EndFor
\State Compute $S_{l,l+1}$.
\EndFor
\State Sample
$\left\{z_{-B_l}^{(l)},\ldots, z_{K}^{(l)}\right\}_{l(N)\le l\le l_{\max}+1}$
according to $\nu^{SG}_{L(N):l_{\max}+1}$.
\For {$l=L(N),\ldots, l_{\max}$}
\State If $N_{l,l+1}=1$, let $\left({z'}_{1}^{(1,l,l+1)},\ldots, {z'}_{K}^{(1,l,l+1)}\right):=\left(z_{1}^{(l+1)},\ldots, z_{K}^{(l+1)}\right)$.
\State If $N_{l,l+1}=1$, let $\left(z_{1}^{(1,l,l+1)},\ldots, z_{K}^{(1,l,l+1)}\right):=\left(z_{1}^{(l)},\ldots, z_{K}^{(l)}\right)$.
\State If $N_{l,l+1}=1$, compute $D_{l,l+1}^{(1)}$.
\State Compute $S_{l,l+1}$.
\EndFor
\State Compute $S(c_R)$.
 \State \textbf{Output:} 
 \State Unbiased estimator $S(c_R)$, 
 \State Samples  $z_{1}^{(0)},\ldots, z_{NK}^{(0)}$,
 \State Samples $z_{1}^{(l,l+1,r)},\ldots, z_{K}^{(l,l+1,r)}, z'^{(r,l,l+1)}_{1},\ldots, z'^{(r,l,l+1)}_{K}$ for $0\le l\le l_{\max}$,  chains $1\le r\le N_{l,l+1}$.
\end{algorithmic}

 	\caption{Unbiased-$\UBU$ with stochastic gradients ($\UBUBUSG$)}
 	\label{supp:alg:UBUBUSG}
  
 \end{algorithm}

\newpage

\subsection{Approximate gradient UBUBU}

Two chains evolving according to approximate gradients with step sizes $h$ and $h/2$ can be coupled as follows. {First $\left(\xi^{(i)}_{k+1}\right)^{8}_{i = 1} \sim \mathcal{N}(0_{d},I_{d}) \text{ for all } i = 1,...,8$, then}
\begin{equation}\label{supp:eq:Phh2UBUQ}
\begin{split}
&\left(\ol{x}_{k},\ol{v}_{k}\right)=
\mathcal{U}^{2}\left(x_{k},v_{k},h/2,\xi_{k+1}^{(1)},\xi_{k+1}^{(2)},\xi_{k+1}^{(3)},\xi_{k+1}^{(4)}\right),\\
&\hat{x}_k=\ol{x}_{\lfloor k/\tau\rfloor \tau}\\
&\left(x_{k+1},v_{k+1}\right) =\mathcal{U}^{2}\left(\mathcal{B_{Q}}\left(\left.\ol{x}_{k},\ol{v}_{k},h\right|\hat{x}_{k}\right),h/2,\xi_{k+1}^{(5)},\xi_{k+1}^{(6)},\xi_{k+1}^{(7)},\xi_{k+1}^{(8)}\right).\\
&\left(\ol{x}'_{k},\ol{v}'_{k}\right)=
\mathcal{U}\left(x'_{k},v'_{k},h/2,\xi^{(1)}_{k+1},\xi^{(2)}_{k+1}\right),\\
&\hat{x}'_k=\ol{x}'_{\lfloor 2k/\tau\rfloor \tau/2}\\
&\left(x'_{k+1/2},v'_{k+1/2}\right)=\mathcal{U}\left(\mathcal{B_{Q}}\left(\left.\ol{x}'_{k},\ol{v}'_{k},h/2,v\right|\hat{x}'_{k}\right), h/4,\xi_{k+1}^{(3)},\xi_{k+1}^{(4)}\right),\\
&\left(\ol{x}'_{k+1/2},\ol{v}'_{k+1/2}\right)=
\mathcal{U}\left(x'_{k+1/2},v'_{k+1/2},h/2,\xi^{(5)}_{k+1},\xi^{(6)}_{k+1}\right),\\
&\hat{x}'_{k+1/2}=\ol{x}'_{\lfloor (2k+1)/\tau\rfloor \tau/2}\\
&\left(x'_{k+1},v'_{k+1}\right)= \mathcal{U}\left(\mathcal{B_{Q}}\left(\left.\ol{x}'_{k+1/2},\ol{v}'_{k+1/2},h/2\right|\hat{x}'_{k+1/2}\right),h/4,\xi_{k+1}^{(7)},\xi_{k+1}^{(8)}\right),\\
\end{split}
\end{equation}
Let $P_{h,h/2}^{A}$ denote the time inhomogenous Markov kernel describing the evolution of $(x_k,\hat{x}_k,v_{k},x_k',\hat{x}'_k,v_k')$ according to the coupled approximate gradient steps \eqref{supp:eq:Phh2UBUQ}. 

As with stochastic gradients, will also need to couple one chain with step size $h$ running OHO on the Gaussian approximation $\mu_G$, and another chain based on UBU with approximate gradients on the target with step size $h/2$.
{First $\left(\xi^{(i)}_{k+1}\right)^{8}_{i = 1} \sim \mathcal{N}(0_{d},I_{d}) \text{ for all } i = 1,...,8$, then}
\begin{equation}\label{supp:eq:Phh2OHOUBUQ}
\begin{split}
&\left(\ol{x}_{k},\ol{v}_{k}\right) =\mathcal{O}^{2}\left(x_{k},v_{k},h/2,\xi_{k+1}^{(1)},\xi_{k+1}^{(2)},\xi_{k+1}^{(3)},\xi_{k+1}^{(4)}\right)\\
&\left(x_{k+1},v_{k+1}\right) =
\mathcal{O}^{2}\left(\mathcal{H}_*\left(
\ol{x}_{k},\ol{v}_{k},h\right),h/2,\xi_{k+1}^{(5)},\xi_{k+1}^{(6)},\xi_{k+1}^{(7)},\xi_{k+1}^{(8)}\right),\\
&\left(\ol{x}'_{k},\ol{v}'_{k}\right)=
\mathcal{U}\left(x'_{k},v'_{k},h/4,\xi^{(1)}_{k+1},\xi^{(2)}_{k+1}\right)\\
&\hat{x}'_k=\ol{x}'_{\lfloor 2k/\tau\rfloor \tau/2}\\
&\left(x'_{k+1/2},v'_{k+1/2}\right)=
\mathcal{U}\left(\mathcal{B_{Q}}\left(\left.\ol{x}'_{k},\ol{v}'_{k},h/2\right|\hat{x}'_{k}\right), h/4,\xi_{k+1}^{(3)},\xi_{k+1}^{(4)}\right)\\
&\left(\ol{x}'_{k+1/2},\ol{v}'_{k+1/2}\right)=\mathcal{U}\left(x'_{k+1/2},v'_{k+1/2},h/4,\xi_{k+1}^{(5)},\xi_{k+1}^{(6)}\right)\\
&\hat{x}'_{k+1/2}=\ol{x}'_{\lfloor (2k+1)/\tau\rfloor \tau/2}\\
&\left(x'_{k+1},v'_{k+1}\right)= 
\mathcal{U}\left(\mathcal{B_{Q}}(\ol{x}'_{k+1/2},\ol{v}'_{k+1/2},h/2|\hat{x}'_{k+1/2}),h/4,\xi_{k+1}^{(7)},\xi_{k+1}^{(8)}\right)\\
\end{split}
\end{equation}
Let $P_{h,h/2}^{OHO/A}$ denote the time inhomogenous Markov kernel describing the evolution of $(x_k,v_{k},x_k',\hat{x}'_k,v_k')$ according to the steps \eqref{supp:eq:Phh2OHOUBUQ}. 

We define $\nu_{0,1}^A$ as joint distribution of $z_{-B_0}^{(0,1)},\ldots, z_{K}^{(0,1)}, z'^{(0,1)}_{-B_{1}},\ldots, z'^{(0,1)}_{K}$ that is similar to the $\nu_{0,1}^{SG}$ coupling for UBUBU-SG, but using inhomogenous Markov kernels $P_{h,h/2}^{OHO/A}$ instead of $P_{h,h/2}^{OHO/SVRG}$. Similarly, we let 
$\nu_{l,l+1}^{A}$ denote the joint distribution of $z_{-B_l}^{(l,l+1)},\ldots, z_{K}^{(l,l+1)}, z'^{(l,l+1)}_{-B_{l+1}},\ldots, z'^{(l,l+1)}_{K}$, defined analogously to $\nu_{l,l+1}^{SG}$ for UBUBU-SG, but using $P_{h,h/2}^{OHO/A}$ and $P_{h,h/2}^{A}$ in place of $P_{h,h/2}^{OHO/SVRG}$ and $P_{h,h/2}^{SVRG}$.
Similarly, we define $\nu_{l(N):l_{\max}}^{A}$ as the synchronous coupling of levels $L(N),\ldots, l_{\max}+1$ analogously to $\nu_{l(N):l_{\max}}$ and $\nu_{l(N):l_{\max}}^{SG}$ - here we only need to couple the Gaussian variables driving the dynamics, which can be done in the same way as in the deterministic gradient case.

We choose $c_{l,l+1}$ as
\begin{equation}\label{supp:eq:cllp1approx}
c_{l,l+1}=c_N\phi_N^{-l}\text{ for }l\in \mathbb{N}.
\end{equation}

The UBUBU-Approx method follows similar steps as in Algorithm \ref{supp:alg:UBUBUSG}, but it uses the couplings $\nu_0^{A}$ and $\nu_{l,l+1}^{A}$ instead of $\nu_0^{SG}$, and $\nu_{l,l+1}^{SG}$. 
In terms of input, unlike in Algorithm \ref{supp:alg:UBUBUSG}, we do not use individual potential terms $U_i(x)$ and batch size $N_b$, but require gradient calculation frequency $\tau$. We recommend setting the Richardson extrapolation parameter $c_{R}=\frac{1}{2}$ in this case (as this approximate gradient scheme has strong order 1).

\subsection{Gaussian target}



Figure \ref{supp:fig:Gaussian-dim-dependence} shows the maximum number gradient evaluations per effective sample (ESS) among all components $f(x)=x_i$ for $1\le i\le d$ as a function of the dimension $d=10,10^2,\ldots, 10^5$, for condition number $\kappa\in \{4,100\}$. 

Figure \ref{supp:fig:Gaussian} presents the histograms of the number of gradient evaluations per effective sample size (ESS) amongst test functions $f(x)=x_1,\ldots, f(x)=x_d$, when comparing UBUBU with RHMC. This experiment is for a specific dimensions size of $d=10^5$ and condition numbers $\kappa \in \{4,100\}$.  As we can observe, UBUBU outperforms RHMC in terms of gradient evaluations per ESS. 

  \begin{figure}[h!]
  \centering
  \includegraphics[width=0.495\linewidth]{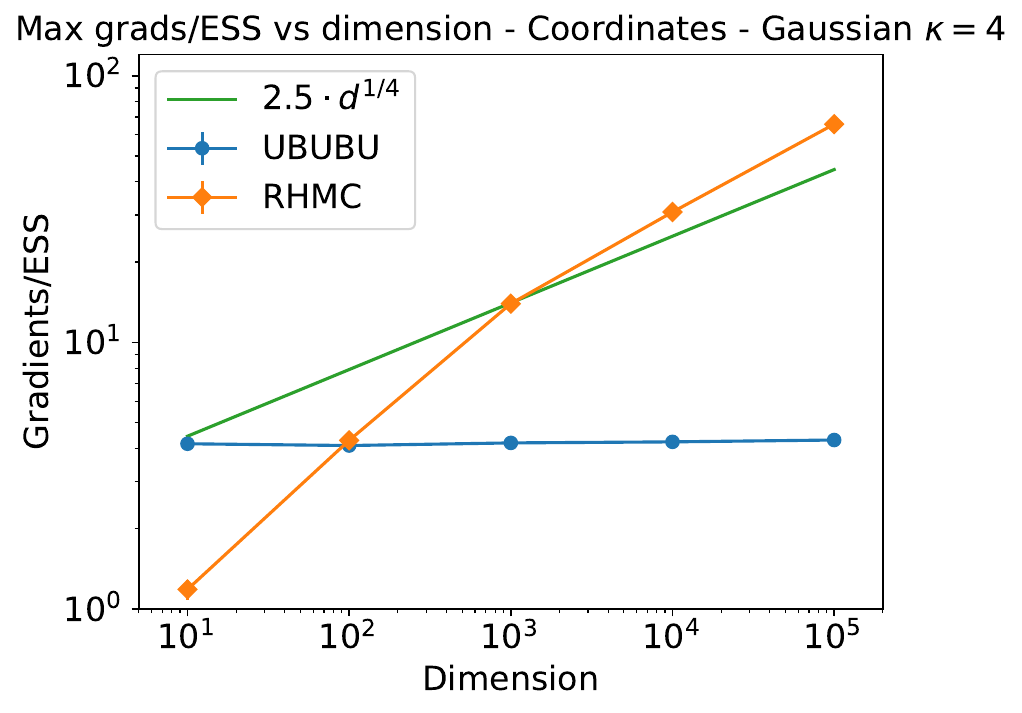} 
  \includegraphics[width=0.495\linewidth]{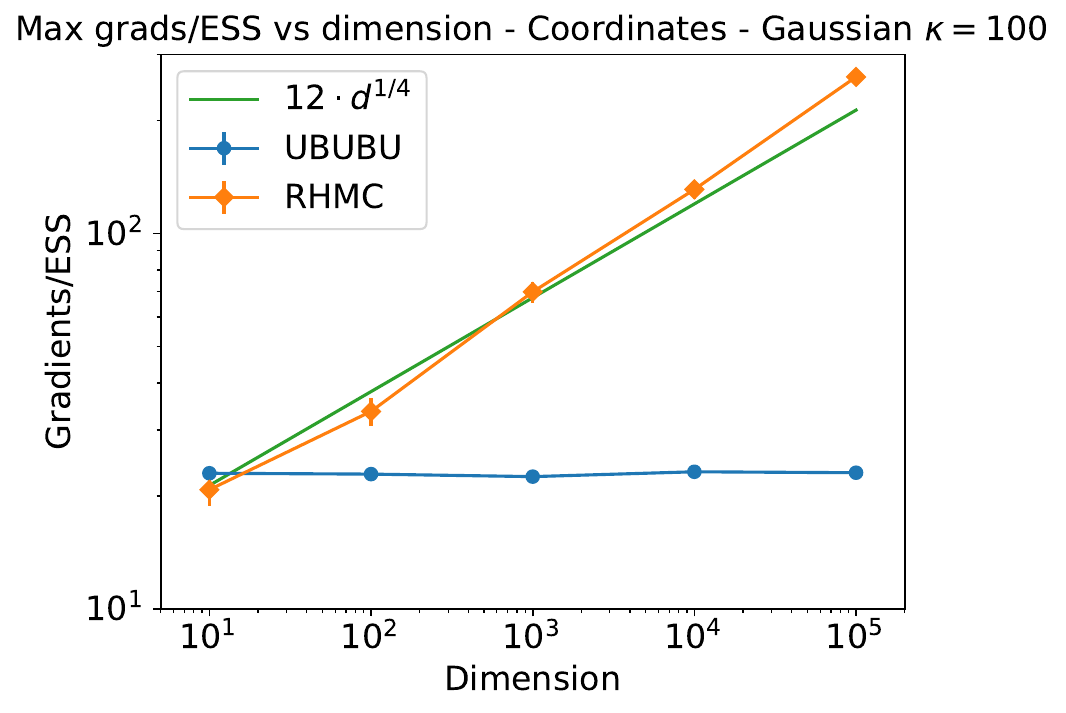} 
  \caption{Dimensional dependence of gradients/ESS over all components for Gaussian targets. Error bars represent bootstrap confidence intervals.}
  \label{supp:fig:Gaussian-dim-dependence}
  \end{figure}

  \begin{figure}[h!]
  \centering
  \includegraphics[width=0.495\linewidth]{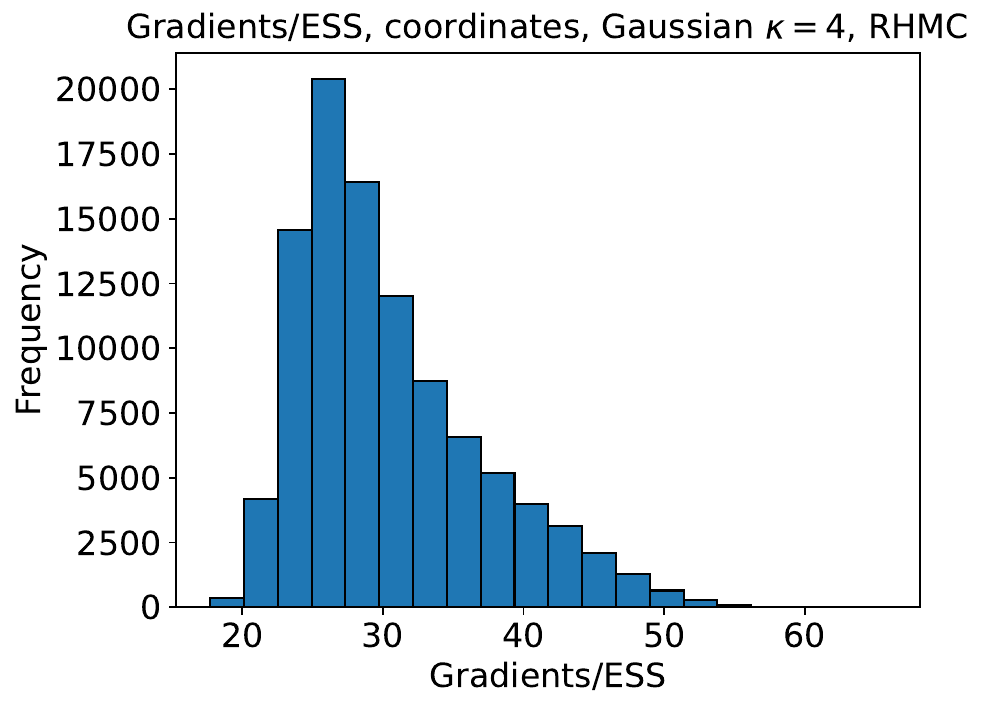} 
  \includegraphics[width=0.495\linewidth]{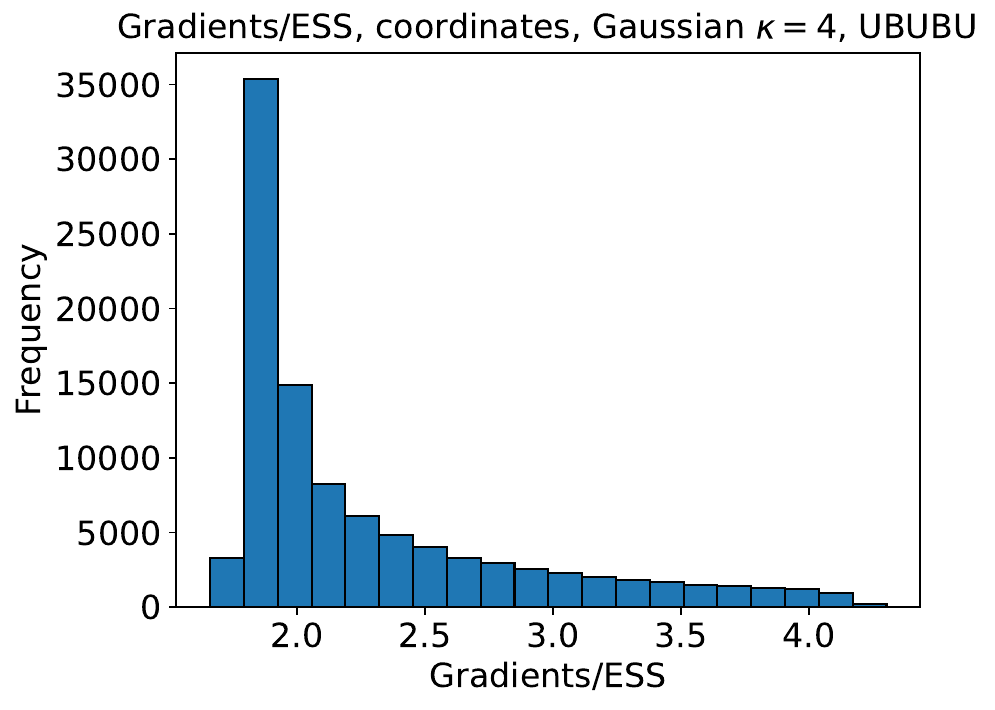}\\
  \includegraphics[width=0.495\linewidth]{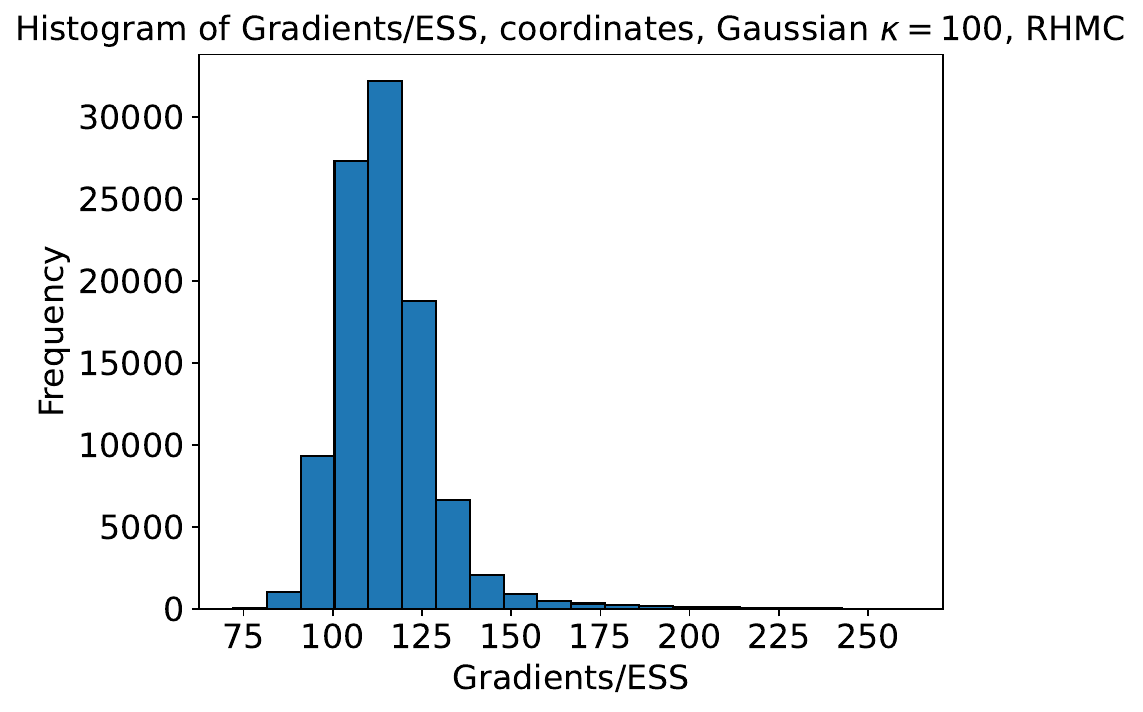} 
  \includegraphics[width=0.495\linewidth]{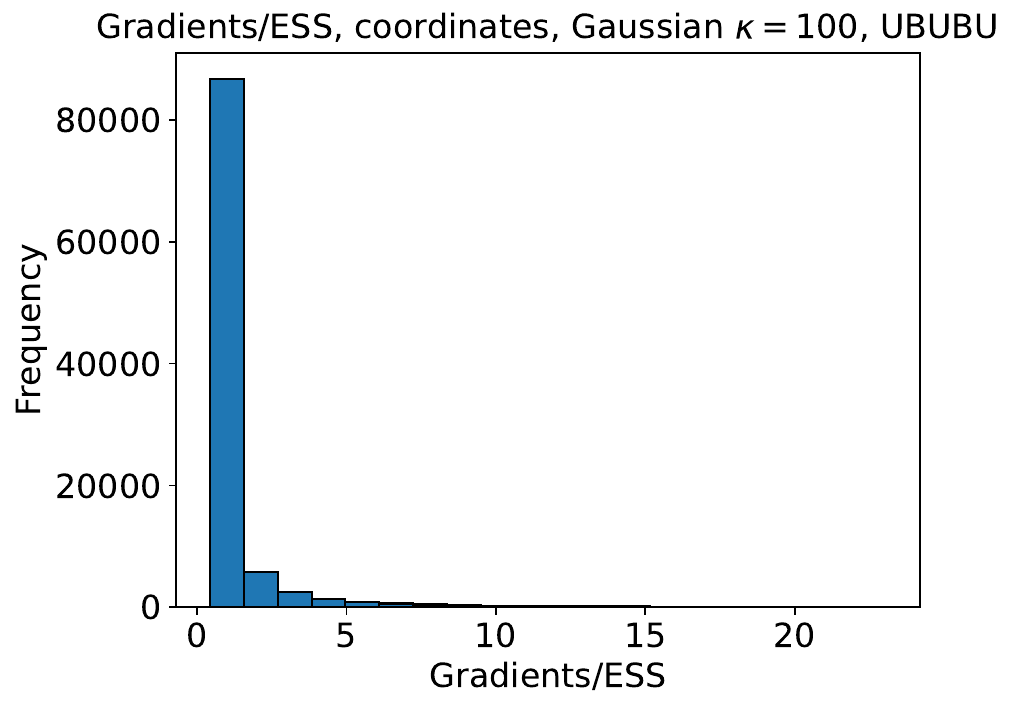}
  \caption{Gradients/ESS over all components for the 100000 dimensional Gaussian targets.}
  \label{supp:fig:Gaussian}
  \end{figure}

\subsection{Bayesian multinomial regression}

For our numerical simulations, we will present two different scenarios: one without preconditioning (Figure \ref{supp:fig:mnist_grad_per_ess1}) and one with preconditioning (Figure \ref{supp:fig:mnist_grad_per_ess2}). In both figures, we evaluated the efficiency of the methods in terms of gradient evaluations per ESS for the coordinate test functions $f(x)=x_1,\ldots, f(x)=x_d$. To compare the posterior distribution with a Gaussian approximation, we have selected a component with a relatively large third derivative. Figure \ref{supp:fig:mnist_comparison} illustrates the potential function and the Gaussian approximation with precision $\nabla^2 U(x^*)$ along the line $x^*+te_i$. Here $e_i=(0,\ldots,0,1,0,\ldots,0)$ is the unit vector of the chosen component ($i=7491$ in our implementation), and $t$ is chosen to cover up to 3 times the standard deviation difference from $x^*_i$. The distribution of this component has a significant skewness, and the density values can differ by up $40\%$ even in the bulk of the distribution. 
 \begin{figure}[H]
 \centering
 \includegraphics[width=0.47\linewidth]{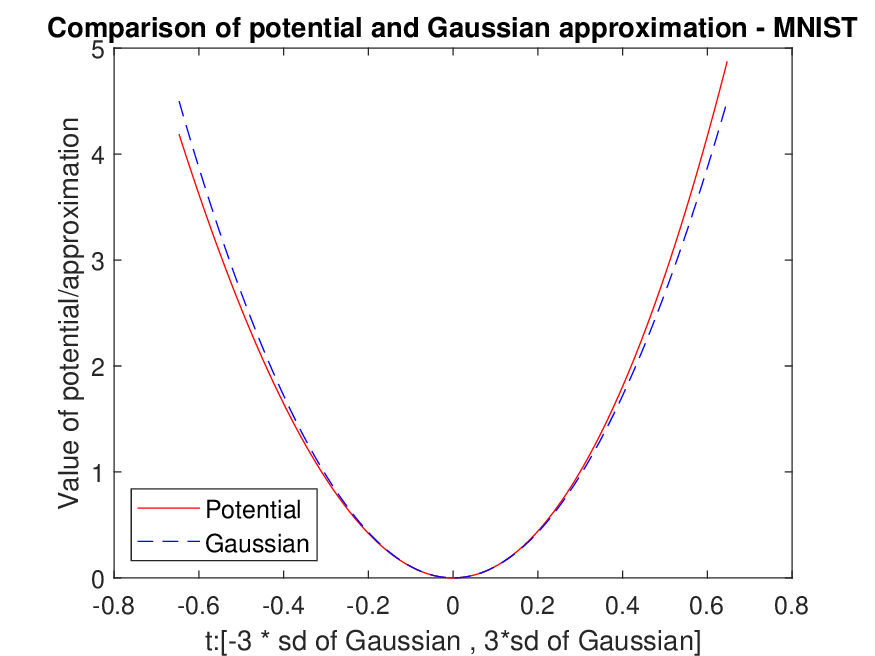} 
 \includegraphics[width=0.47\linewidth]{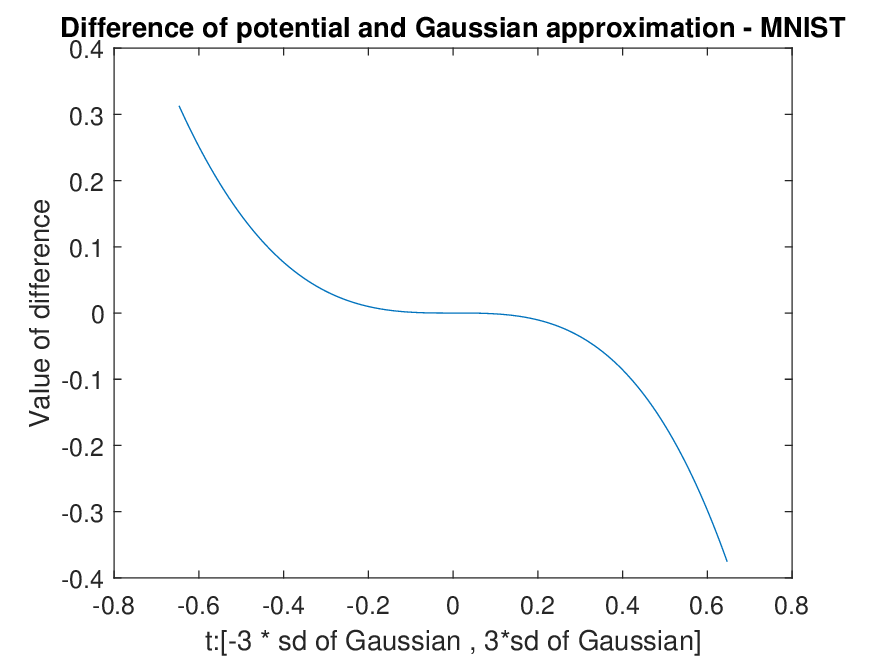}
 \caption{MNIST example. Left: Comparison between potential and quadratic approximation. Right: Difference between the potential and quadratic approximation.}
 \label{supp:fig:mnist_comparison}
 \end{figure}

In the first scenario (no preconditioning), the condition number of the Hessian at the mode $\nabla^2 U(x^*)$ is $\kappa\approx 7.2 \times 10^3$. We included simulation results with RHMC and UBUBU. As we can see, UBUBU performs similarly to RHMC in this case.

 \begin{figure}[H]
 \centering
 \includegraphics[width=0.495\linewidth]{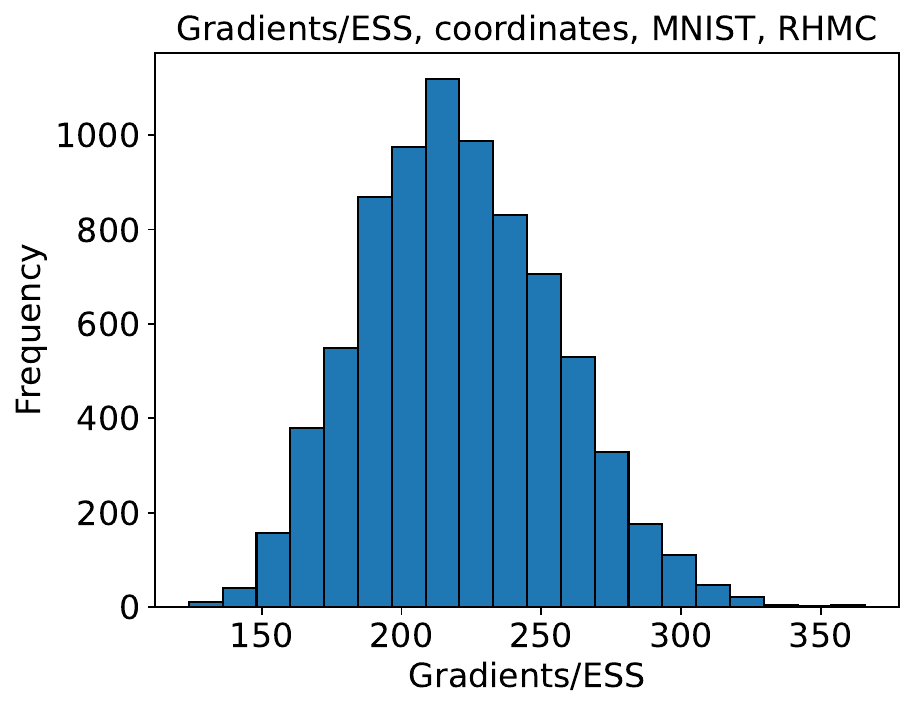} 
 \includegraphics[width=0.495\linewidth]{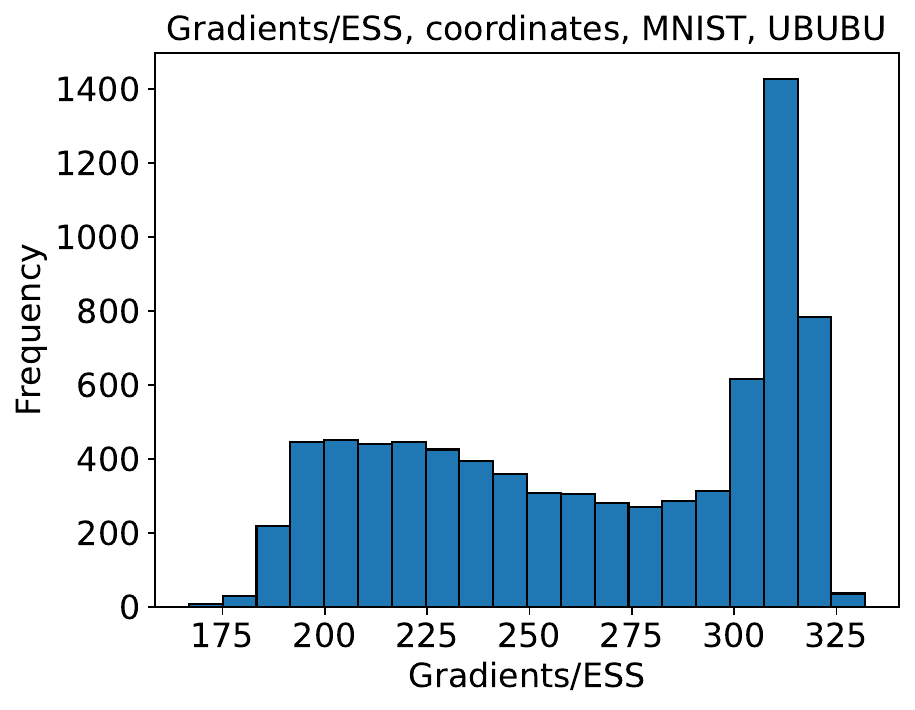}\\
 \caption{Gradients/ESS over all components for MNIST dataset without preconditioning.}
 \label{supp:fig:mnist_grad_per_ess1}
 \end{figure}

 \begin{figure}[H]
 \centering
 \includegraphics[width=0.495\linewidth]{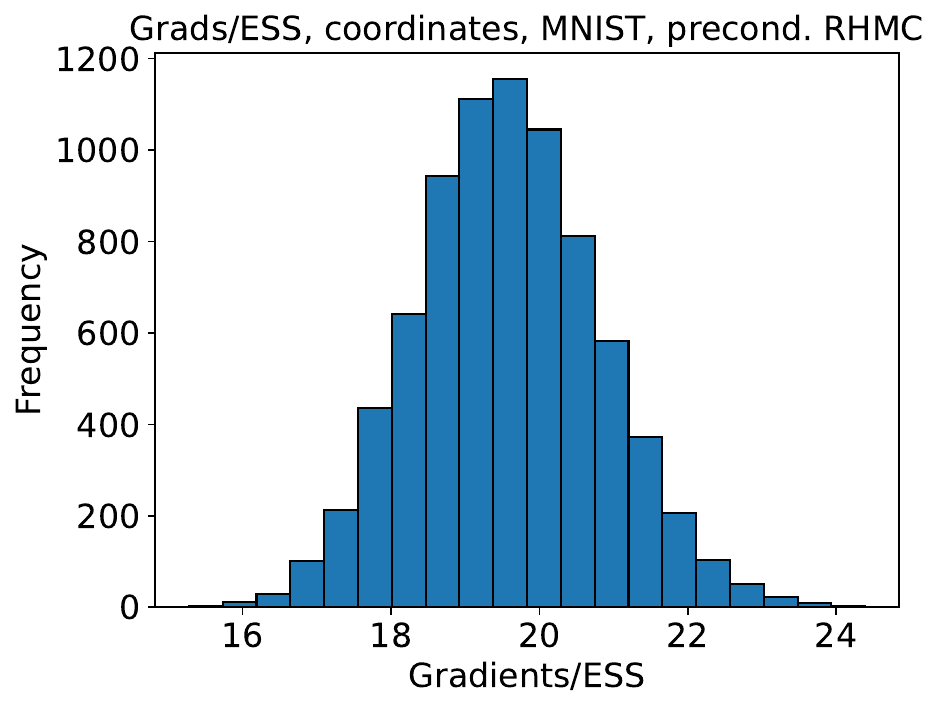} 
 \includegraphics[width=0.495\linewidth]{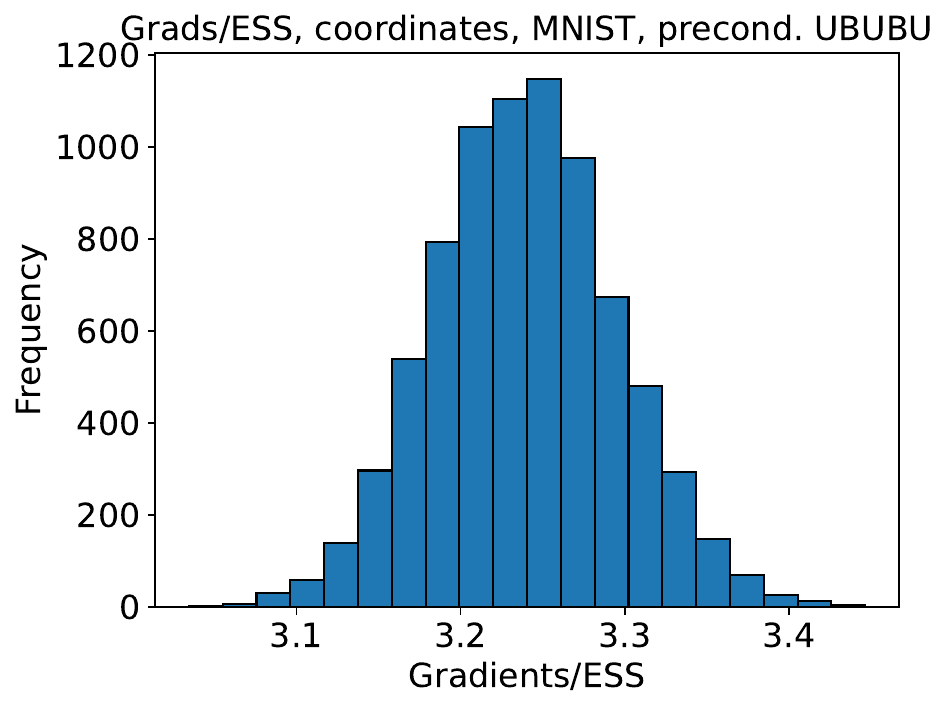}\\
\includegraphics[width=0.495\linewidth]{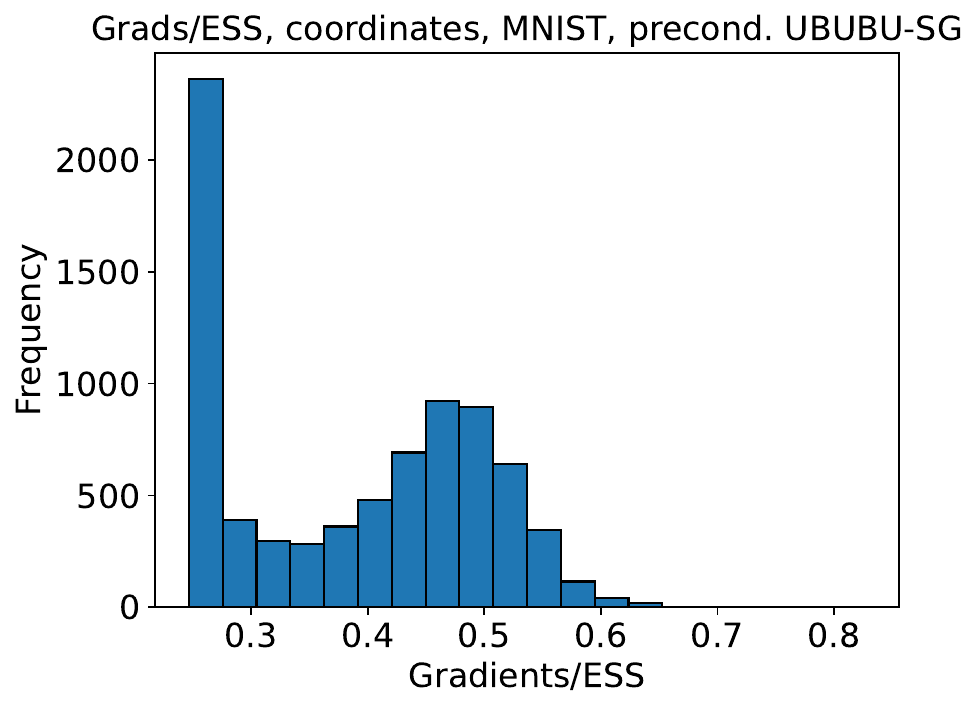} 
 \includegraphics[width=0.495\linewidth]{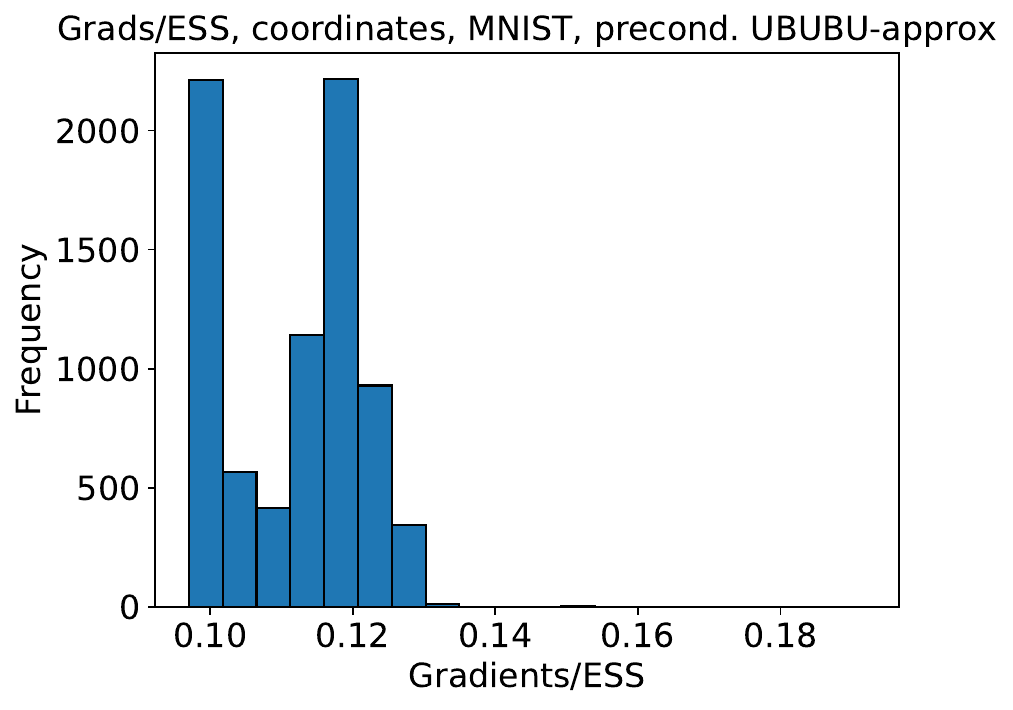}
 \caption{Gradients/ESS over all components for MNIST dataset with preconditioning.}
 \label{supp:fig:mnist_grad_per_ess2}
 \end{figure}

Figure \ref{supp:fig:mnist_predict_grad_per_ess} presents experiments comparing RHMC and UBUBU on these test functions. The experiments show that for these test functions, when compared to compared to RHMC, UBUBU-approx has 32 times improvement in gradients/ESS, and a 49 times improvement in ESS/sec.

 \begin{figure}[H]
 \centering
 \includegraphics[width=0.495\linewidth]{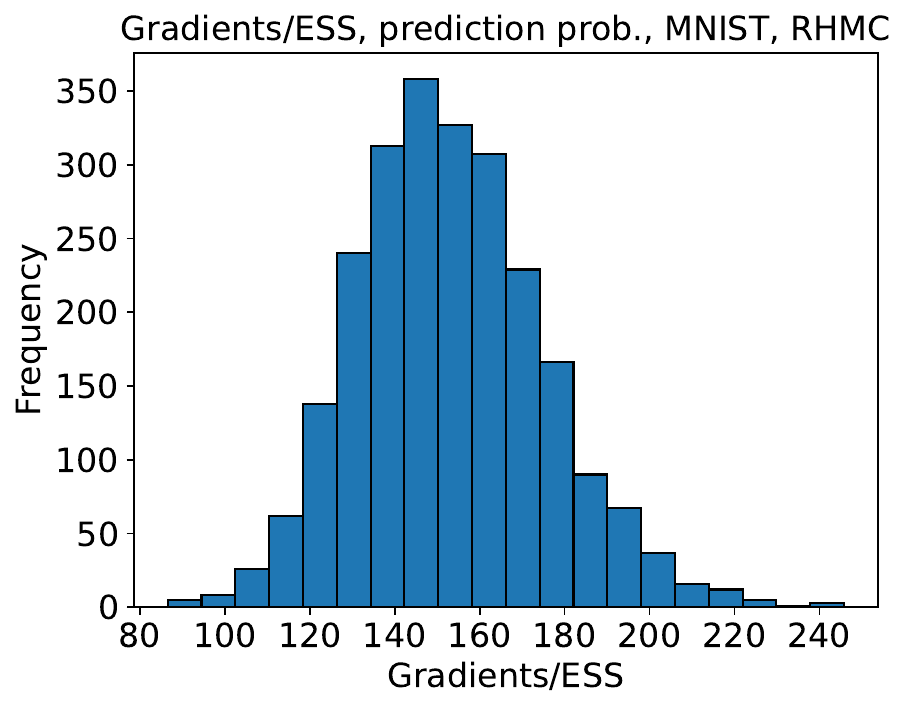} 
 \includegraphics[width=0.495\linewidth]{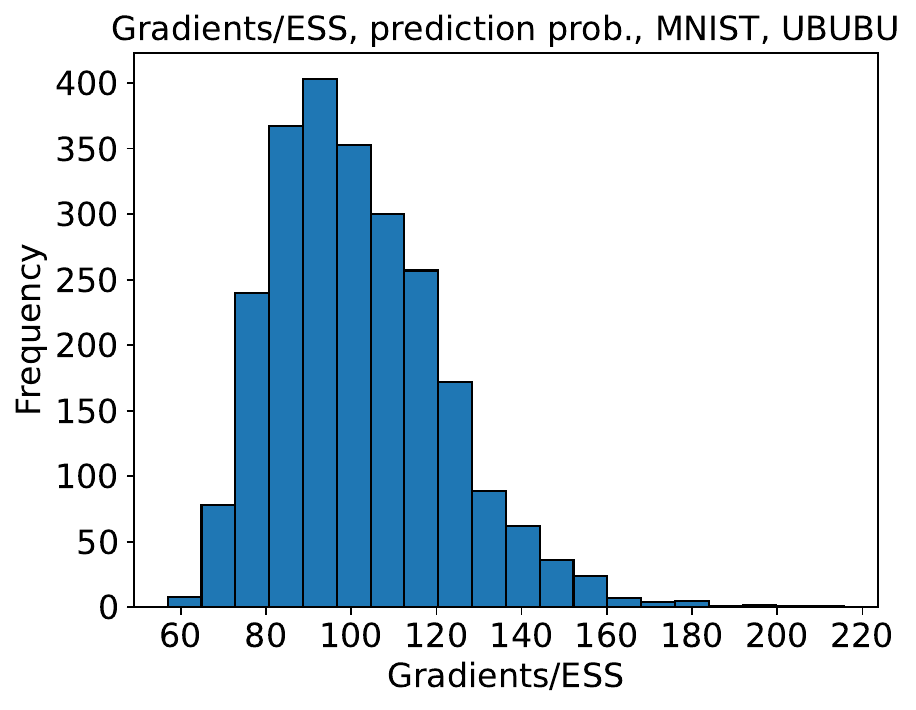}\\
 \includegraphics[width=0.495\linewidth]{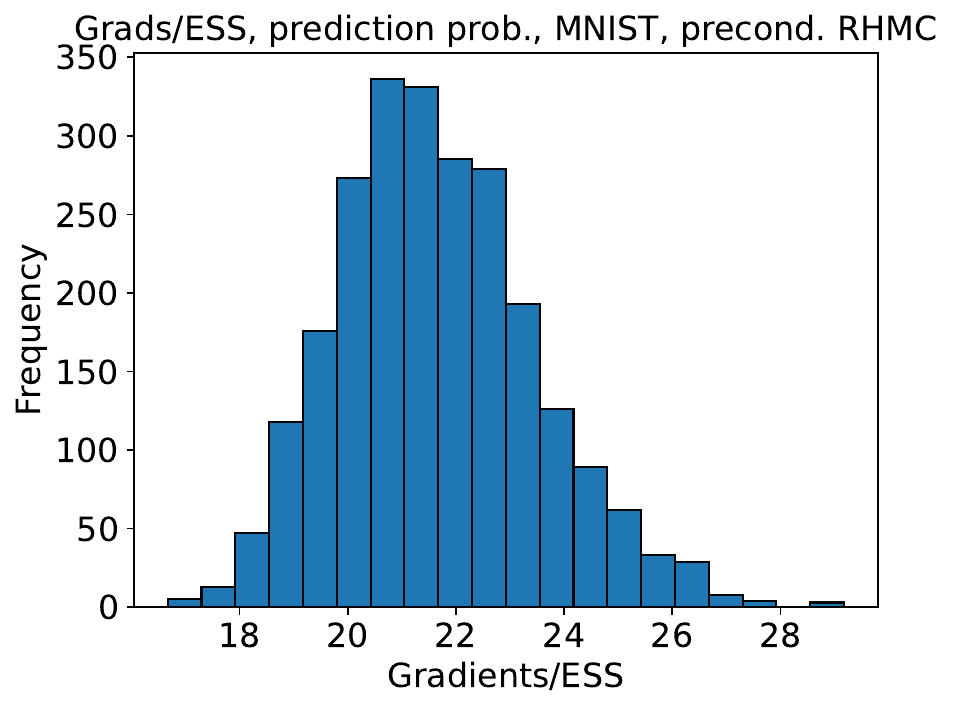} 
 \includegraphics[width=0.495\linewidth]{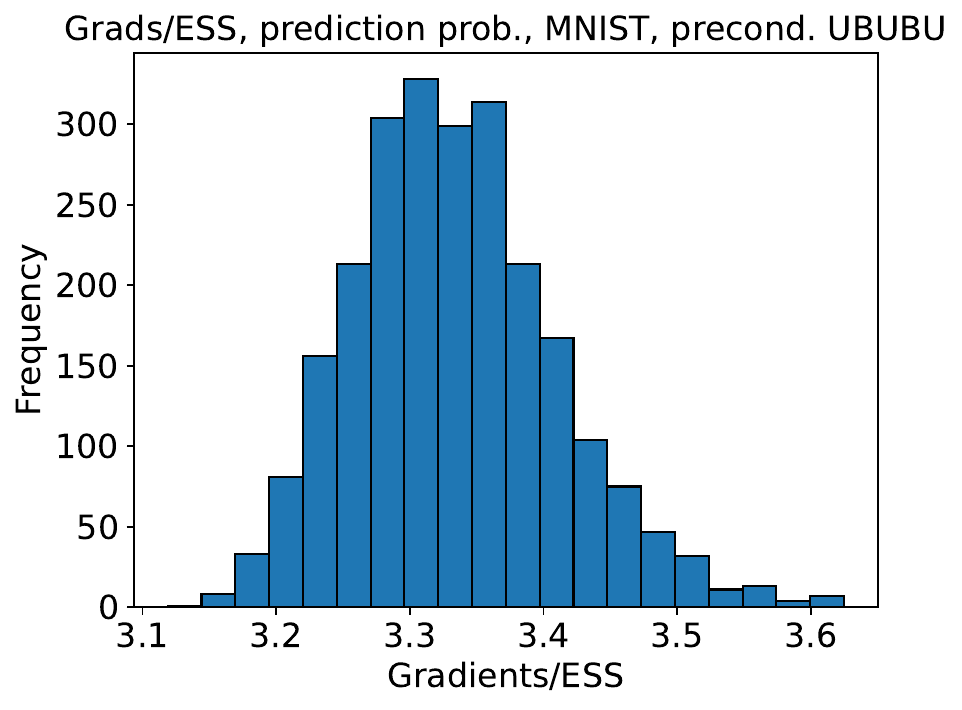}\\
 \includegraphics[width=0.495\linewidth]{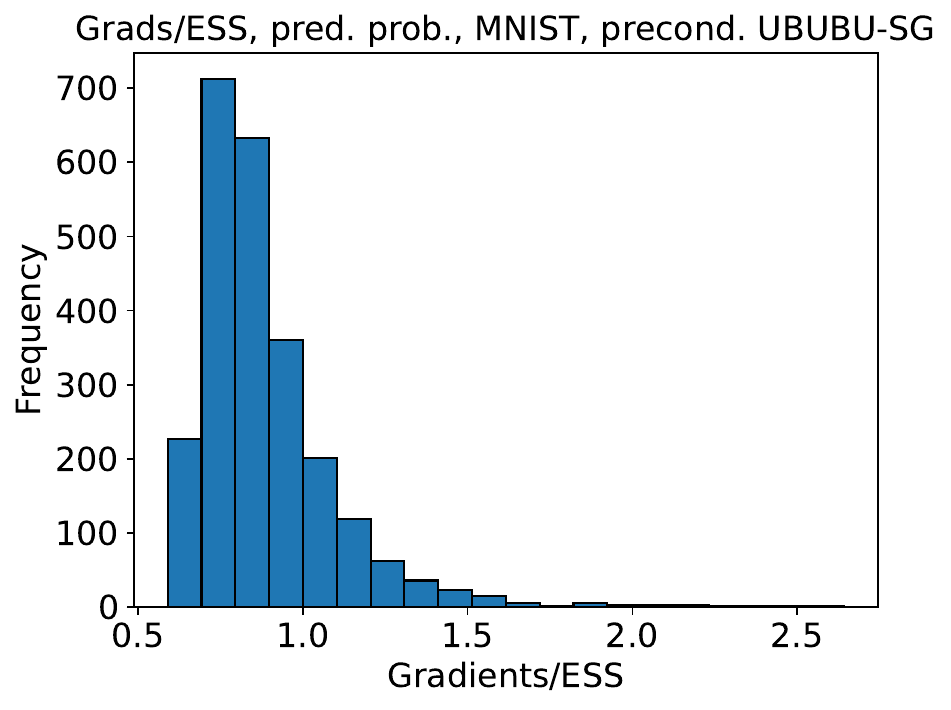} 
 \includegraphics[width=0.495\linewidth]{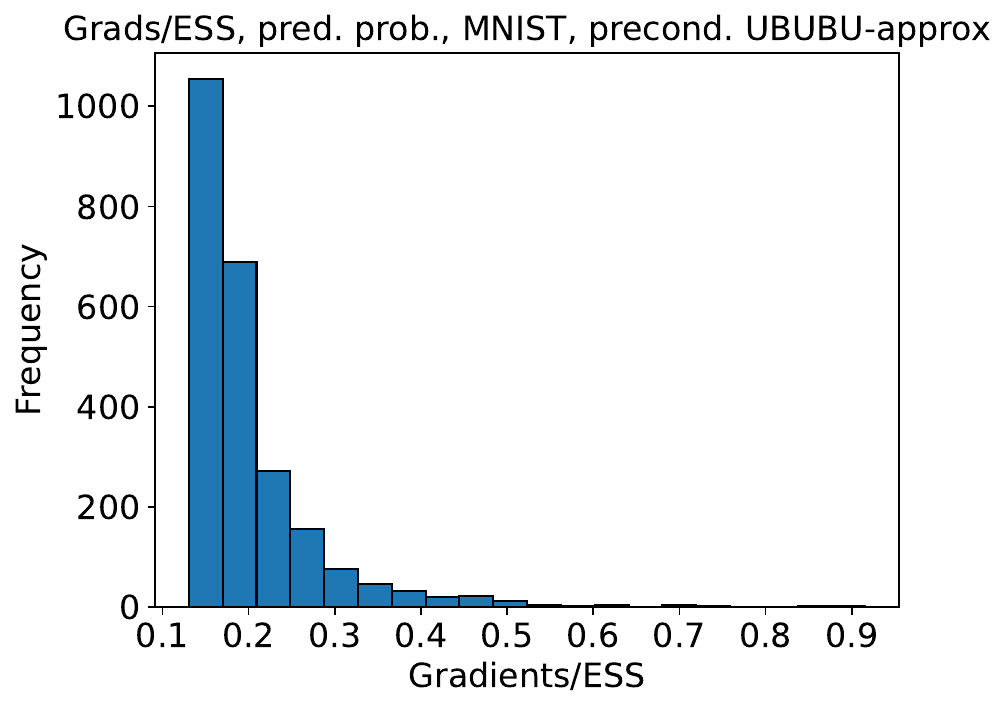}\\
 
 \caption{Gradients/ESS for 2410 probabilities on the test dataset for MNIST}
 \label{supp:fig:mnist_predict_grad_per_ess}
 \end{figure}


\subsection{Poisson regression model}
\label{supp:subsec:BPM}
Our numerical simulations are presented in Figure \ref{supp:fig:poisson_grad_per_ess}.
As we can see, UBUBU uses approximately 14 times fewer gradient evaluations per effective sample than RHMC, and UBUBU-Approx uses about 5000 times fewer gradient evaluations than RHMC. 
 \begin{figure}[h!]
 \centering
 \includegraphics[width=0.495\linewidth]{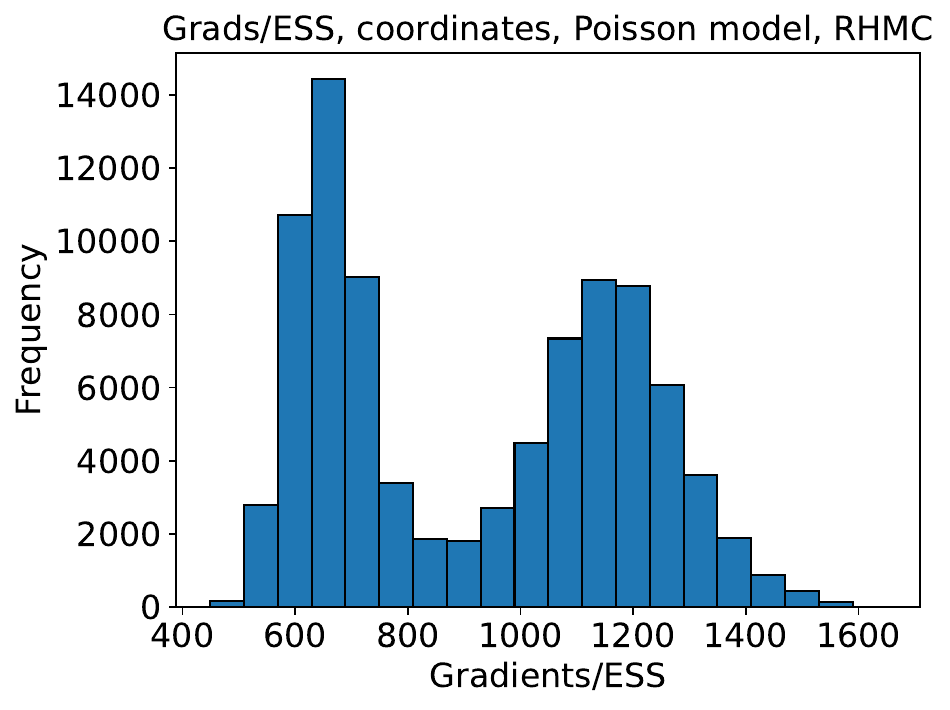}
 \includegraphics[width=0.495\linewidth]{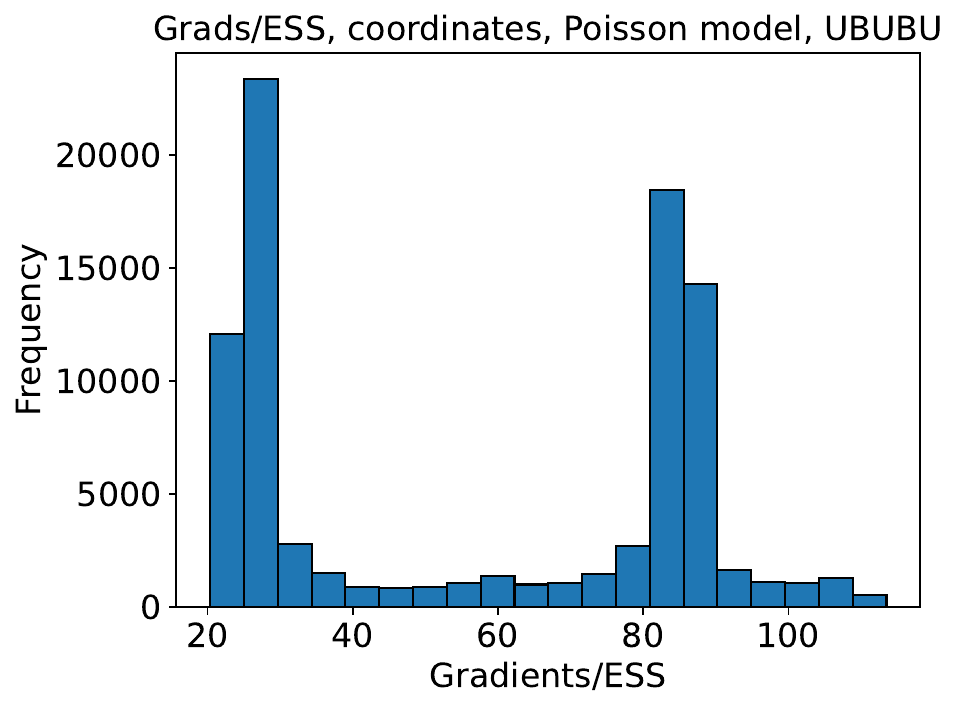}\\
 \includegraphics[width=0.495 \linewidth]{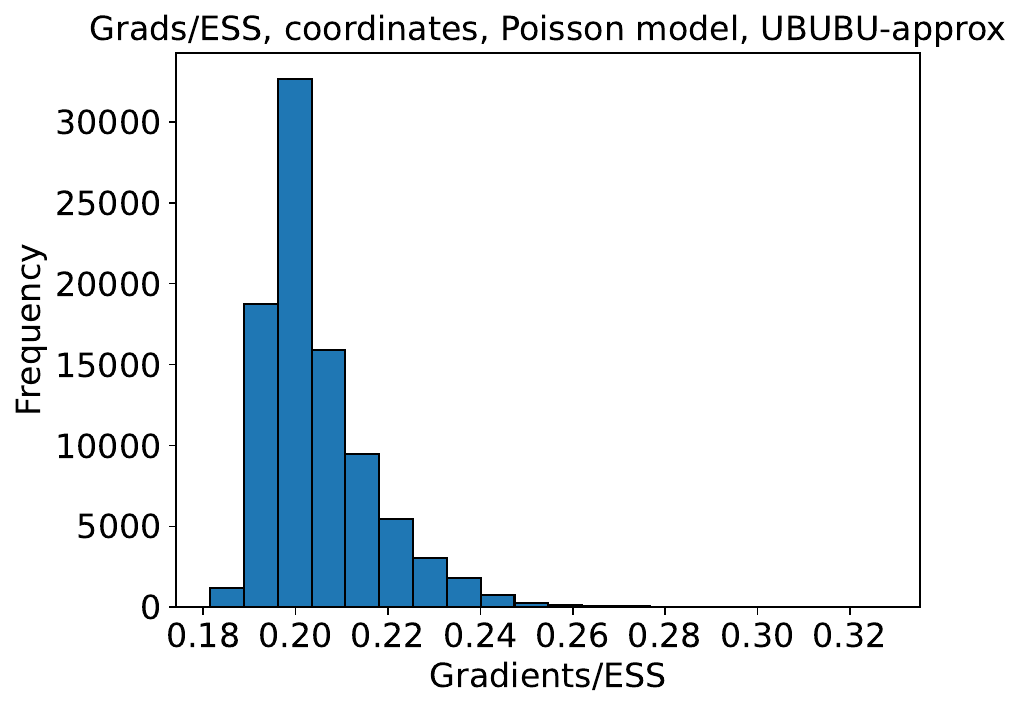}
 \caption{Gradients/ESS over all components of a Poisson regression model for soccer scores.}
 \label{supp:fig:poisson_grad_per_ess}
 \end{figure}

\section{Unbiased multilevel estimators}\label{supp:sec:Appendix:proofs:multilevel}
In this section, we provide the theoretical results for the following two estimators:
\begin{equation}\label{supp:eq:Sdef}
S = S_0 + \sum^{\infty}_{l=0}S_{l,l+1},
\end{equation}

\begin{align}\label{supp:eq:SRichardsondef}
S(c_R) &= S_0 + \sum^{L(N)-1}_{l=0}S_{l,l+1} + \frac{D_{L(N),L(N)+1}^{(1)}}{1-c_R}+\sum_{l=L(N)+1}^{\infty} \overline{S}_{l,l+1},\\
\nonumber\overline{S}_{l,l+1}&=\frac{\Ind[N_{l,l+1}=1]}{\E (N_{l,l+1})} \left[D_{l,l+1}^{(1)}-D_{L(N),L(N)+1}^{(1)}\cdot c_R^{l-L(N)}\right].
\end{align}

We restate the following key assumptions we make on the variances as follows:

\begin{assumption}\label{supp:ass:var}
$f:\Lambda\to \R$ is a measurable function. $(\tilde{\mu}_{h_l})_{l\ge 0}$ is a sequence of distributions satisfying that $\tilde{\mu}_{h_l}(f)\to \mu(f)$ as $l\to \infty$. The random variable $D_0$ satisfies that $\E(D_0)=\mu_{h_0}(f)$, $\Var(D_0)<\infty$, for every $l\geq 0$, the random variable $D_{l,l+1}$ satisfies that $\E(D_{l,l+1})=\tilde{\mu}_{l+1}(f)-\tilde{\mu}_{l+1}(f)$ and $\E(D_{l,l+1}^2)\le V_{D}\phi_D^{-l}$ for some finite constants $V_D>0$, $\phi_D>2$. 
\end{assumption}
\begin{assumption}\label{supp:ass:numbsamp}
The constants $c_{l,l+1}$ controlling $N_{l,l+1}$ satisfy 
\[\ul{c}_N \phi_N^{-l}\le c_{l,l+1}\le  \ol{c}_N \phi_N^{-l},\]
for some finite constants $0<\ul{c}_N\le \ol{c}_N$, $\phi_N>2$.
\end{assumption}
\begin{assumption}\label{supp:ass:comp}
The computational cost of generating a sample from $D_{l,l+1}$ is $\O(2^l (K+lB+B_0))$ for some finite constants $B$, $B_0$, and generating a sample from $D_0$ has a finite computational cost.
\end{assumption}
\begin{assumption}\label{supp:ass:independence}
For $1\le l \le L(N)-1, 1\le r\le N_{l,l+1}$, the random variables $D_{l,l+1}^{(r)}$ are all independent from each other, and they are also independent from the collection of random variables $\{D_{l,l+1}^{(1)}\}_{l\ge L(N)}$. 
\end{assumption}

\begin{proposition}
\label{supp:prop:unb}
Suppose that Assumptions \ref{supp:ass:var}, \ref{supp:ass:numbsamp}, \ref{supp:ass:comp} and \ref{supp:ass:independence} hold, and that $2<\phi_N<\phi_D$. Then $S$ as defined in \eqref{supp:eq:Sdef} is an unbiased estimator of $\mu(f)$ that has finite variance \[\Var(S)\le 
 \frac{\Var(D_0)}{N}+\frac{V_D}{N\ul{c}_N\left(1-\left(\frac{\phi_N}{\phi_D}\right)^{1/2}\right)^2},\] and finite expected computational cost.  

Similarly, for any $0\le c_R< \frac{1}{\phi_N^{1/2}}$, $S(c_R)$ as defined in \eqref{supp:eq:SRichardsondef} is also an unbiased estimator of $\mu(f)$ with finite variance  
\[\Var(S(c_R))\le \frac{\Var(D_0)}{N}+
\frac{2V_D}{N\ul{c}_N\left(1-\left(\frac{\phi_N}{\phi_D}\right)^{1/2}\right)^2}+
\frac{1}{N^2}\cdot \frac{2V_D \ol{c}_N \phi_N^2 c_R^2}{\ul{c}_N^2(1-\phi_N c_R^2)},\]
and finite expected computational cost.
\end{proposition}

\begin{proof}[Proof of Proposition \ref{supp:prop:unb}]
From Assumption \ref{supp:ass:comp}, and the definition of $S$, it follows that the expected computational cost of $S$ is upper bounded as follows: 
\begin{align*}&\O\left(N+\sum_{l=0}^{\infty}\E(N_{l,l+1}) 2^l (K+lB+B_0)\right)\\
&\le \O\left( N\left(1+\ol{c}_N\sum_{l=0}^{\infty} \left(\frac{2}{\phi_N}\right)^{l} (K+lB+B_0)\right)\right)<\infty.\end{align*}
From Assumptions \ref{supp:ass:var} and \ref{supp:ass:independence}, we have that
\begin{align*}&\Var(S)\le \frac{\Var(D_0)}{N}+\sum_{l=0}^{L(N)-1} \frac{\E(D_{l,l+1}^2)}{N_{l,l+1}} + \Var\left(\sum_{l=L(N)}^{\infty} \frac{D_{l,l+1}^{(1)} \Ind[N_{l,l+1}=1]}{\E(N_{l,l+1})}\right)\\
\intertext{ by the Cauchy-Schwarz inequality }
&\le \frac{\Var(D_0)}{N}+\sum_{l=0}^{L(N)-1} \frac{\E(D_{l,l+1}^2)}{N_{l,l+1}}
+ \sum_{l,l'=L(N)}^{\infty} \left(\frac{\E\left(D_{l,l+1}^2\right)}{\E(N_{l,l+1})}\right)^{1/2}\left(\frac{\E\left(D_{l',l'+1}^{2}\right)}{\E(N_{l',l'+1})}\right)^{1/2}\\
&\le \frac{\Var(D_0)}{N}+\frac{V_D}{\ul{c}_N N}\sum_{l=0}^{L(N)-1} \left(\frac{\phi_D}{\phi_N}\right)^{-l}+\frac{V_D}{\ul{c}_N N}\sum_{l,l'=L(N)}^{\infty} \left(\frac{\phi_D}{\phi_N}\right)^{-l/2-l'/2}\\
&\le \frac{\Var(D_0)}{N}+ \frac{V_D}{\ul{c}_N N}\sum_{l,l'=0}^{\infty} \left(\frac{\phi_D}{\phi_N}\right)^{-l/2-l'/2}=\frac{\Var(D_0)}{N}+\frac{V_D}{N\ul{c}_N\left(1-\left(\frac{\phi_N}{\phi_D}\right)^{1/2}\right)^2} <\infty.
\end{align*}
By Jensen's inequality, and Assumption \ref{supp:ass:var},  $\E(|S_0|+\sum_{l=0}^{\infty} |S_{l,l+1}|)<\infty$, hence by the dominated convergence theorem, 
\[\E(S)=\E(S_0)+\sum_{l=0}^{\infty} \E(S_{l,l+1})=\tilde{\mu}_{h_0}(f)+\sum_{l=0}^{\infty}\tilde{\mu}_{h_{l+1}}(f)-\tilde{\mu}_{h_{l}}(f)=\mu(f),\]
which concludes the proof for $S$.

For $S(c_R)$, the computational cost is the same as for $S$. Hence the computational cost has finite expectation. For the variance, we have
\begin{align*}\Var(S(c_R))&\le \frac{\Var(D_0)}{N}+\sum_{l=0}^{L(N)-1} \frac{\E(D_{l,l+1}^2)}{N_{l,l+1}} +\Var\left(\frac{D^{(1)}_{L(N),L(N)+1}}{1-c_R}+\sum_{l=L(N)+1}^{\infty} \overline{S}_{l,l+1}\right).
\end{align*}
The last term can be bounded as
\begin{align*}
&\Var\left(\frac{D^{(1)}_{L(N),L(N)+1}}{1-c_R}+\sum_{l=L(N)+1}^{\infty} \overline{S}_{l,l+1}\right)
\\
&=\Var\left(D^{(1)}_{L(N),L(N)+1}\left(\frac{c_R}{1-c_R}-\sum_{l=L(N)+1}^{\infty}\frac{\mathbbm{1}[N_{l,l+1}=1]}{\E(N_{l,l+1})} c_R^{l-L(N)}\right)
+\sum_{l=L(N)}^{\infty} S_{l,l+1}\right)
\\
&\le 2 \cdot \Var\left(D^{(1)}_{L(N),L(N)+1}\left(\sum_{l=L(N)+1}^{\infty}\left(\frac{\mathbbm{1}[N_{l,l+1}=1]}{\E(N_{l,l+1})}-1\right) c_R^{l-L(N)}\right)\right)
\\
&+2\cdot\Var\left(\sum_{l=L(N)}^{\infty} \frac{D_{l,l+1}^{(1)} \Ind[N_{l,l+1}=1]}{\E(N_{l,l+1})}\right).
\end{align*}
Similarly to the previous argument bounding  $\Var(S)$, we have that
\[\sum_{l=0}^{L(N)-1} \frac{\E(D_{l,l+1}^2)}{N_{l,l+1}} +2\cdot\Var\left(\sum_{l=L(N)}^{\infty} \frac{D_{l,l+1}^{(1)} \Ind[N_{l,l+1}=1]}{\E(N_{l,l+1})}\right)\le 
\frac{2V_D}{N\ul{c}_N\left(1-\left(\frac{\phi_N}{\phi_D}\right)^{1/2}\right)^2}.
\]
For the remaining term, using the independence of $N_{l,l+1}$ from the other random variables, 
\begin{align*}
&2 \cdot \Var\left(D^{(1)}_{L(N),L(N)+1}\left(\sum_{l=L(N)+1}^{\infty}\left(\frac{\mathbbm{1}[N_{l,l+1}=1]}{\E(N_{l,l+1})}-1\right) c_R^{l-L(N)}\right)\right)\\
&\le 2\cdot \E\left(D_{L(N),L(N)+1}^2\right) \left(\sum_{l=L(N)+1}^{\infty}\E\left(\left(\frac{\mathbbm{1}[N_{l,l+1}=1]}{\E(N_{l,l+1})}-1\right)^2\right) c_R^{2l-2L(N)}\right)\\
&\le \frac{2V_D}{\ul{c}_N N} \left(\frac{\phi_D}{\phi_N}\right)^{-L(N)}\cdot \sum_{l=L(N)+1}^{\infty}\left(\left(\frac{1}{\E(N_{l,l+1})}-1\right) c_R^{2l-2L(N)}\right)\\
&\le \frac{2V_D}{\ul{c}_N N} \left(\frac{\phi_D}{\phi_N}\right)^{-L(N)}\cdot \frac{\phi_N^{L(N)}}{\ul{c}_N N} \sum_{l=L(N)+1}^{\infty}
(\phi_N c_R^{2})^{l-L(N)}
\intertext{ by the definition of $L(N)$, and Assumption \ref{supp:ass:numbsamp}, $\phi_N^{-(L(N)-1)} \ol{c}_N\ge 1$, hence $\phi_N^{L(N)}\le \ol{c}_N \phi_N$,}
&\le \frac{2V_D}{\ul{c}_N N^2} \cdot \frac{\ol{c}_N \phi_N}{\ul{c}_N} \cdot \frac{\phi_N c_R^2}{1-\phi_N c_R^2}.
\end{align*}
After some rearrangement, we obtain that
\begin{align*}&\Var(S(c_R))\le \frac{\Var(D_0)}{N}+
\frac{2V_D}{N\ul{c}_N\left(1-\left(\frac{\phi_N}{\phi_D}\right)^{1/2}\right)^2}+
\frac{1}{N^2}\cdot \frac{2V_D \ol{c}_N \phi_N^2 c_R^2}{\ul{c}_N^2(1-\phi_N c_R^2)}.
\end{align*}
Finally, unbiasedness can be shown as before using the dominated convergence theorem.
\end{proof}

We show below that a central limit theorem holds for these estimators.

\begin{theorem}\label{supp:thm:CLT}
Under the assumptions of Proposition \ref{supp:prop:unb}, we have that, as $N\to \infty$,
\[\sqrt{N} (S-\mu(f)) \Rightarrow \mathcal{N}(0,\sigma^2_S) \quad \text{ and }\quad \sqrt{N} (S(c_R)-\mu(f)) \Rightarrow\mathcal{N}(0,\sigma^2_S),\]
where 
\begin{equation}\label{supp:eq:sigma2S}\sigma^2_S:=\Var(D_0)+\sum_{l=0}^{\infty} \frac{\Var(D_{l,l+1})}{c_{l,l+1}}.\end{equation}
\end{theorem}

\begin{proof}[Proof of Theorem \ref{supp:thm:CLT}]
First, we prove the result for $H:=\sqrt{N} (S-\mu(f))$.
For $\ol{l}\ge 0$, let 
\begin{align*}H^{\ol{l}}&:=\sqrt{N}(S_0-\E(S_0)) + \sum^{\ol{l}}_{l=0}(S_{l,l+1}-\E(S_{l,l+1}))\\
&=\frac{1}{\sqrt{N}}\sum_{r=1}^{N} \left(D_0^{(r)}-\E(D_0^{(r)})\right)+\sum_{l=0}^{\ol{l}} \frac{\sqrt{N}}{\E (N_{l,l+1})}\sum_{r=1}^{N_{l,l+1}} \left[D_{l,l+1}^{(r)}-\E( D_{l,l+1})\right]\\
&:=H_0+\sum_{l=0}^{\ol{l}} H_{l,l+1}.
\end{align*}
Then by using independence, and the fact that $\left(\frac{\sqrt{N}}{\E (N_{l,l+1})}\right) / \left(\frac{1}{\sqrt{N}}\right)\to \frac{1}{c_{l,l+1}}$, by the proof of the central limit theorem (see Sections 3.3-3.4 of \cite{Durrett5th}), for every $t\in \R$, $H_0$ and $H_{l,l+1}$ satisfies 
\begin{align*}
\E(e^{i t H_0})&\to e^{-t^2 \mathcal{V}_0/2}\text{ as }N\to \infty\text{ for }\mathcal{V}_0=\Var(D_{0}),\\
\E(e^{i t H_l})&\to e^{-t^2 \mathcal{V}_{l,l+1}/2}\text{ as }N\to \infty\text{ for }\mathcal{V}_{l,l+1}= \frac{\Var(D_{l,l+1})}{c_{l,l+1}}  .
\end{align*}
Using independence, we can multiply these together to obtain that for any $t\in \R$,
\[\E\left(e^{i t H^{\ol{l}}}\right)\to e^{-t^2 \left(\mathcal{V}_0+\sum_{l=0}^{\ol{l}} \mathcal{V}_{l,l+1}\right) /2}\text{ as }N\to \infty.\]
By Lemma 3.3.19 of \cite{Durrett5th}, it follows that for a random variable $X$ with $\E(X)=0$ and $\E(X^2)<\infty$, we have
\[|\E(e^{itX})-1|\le \frac{t^2\E(X^2)}{2}.\]
Suppose that $N$ is sufficiently large such that $L(N)\ge \ol{l}$, then for $X=\sqrt{N}(S-\mu(f))-H^{\ol{l}}$, we have 
\begin{align*}&\E(X^2)=\Var(X)\le N\sum_{l=\ol{l}+1}^{L(N)-1} \frac{\E(D_{l,l+1}^2)}{N_{l,l+1}}+N\cdot \Var\left(\sum_{l=L(N)}^{\infty} D_{l,l+1}^{(1)}\cdot \frac{\Ind[N_{l,l+1}=1]}{\E(N_{l,l+1})}\right)\\
&\intertext{ using the same bound as in the proof of Proposition \ref{supp:prop:unb}, and the assumption $L(N)\ge \ol{l}$,}
&\le  \frac{V_D}{\ul{c}_N}\sum_{l=\ol{l}+1}^{\infty} \left(\frac{\phi_D}{\phi_N}\right)^{-l}+ 
\frac{V_D}{\ul{c}_N}\sum_{l,l'=L(N)}^{\infty} \left(\frac{\phi_D}{\phi_N}\right)^{-l/2-l'/2}\\
&\le \left(\frac{\phi_N}{\phi_D}\right)^{\ol{l}} \frac{V_D}{\ul{c}_N} \left(\frac{1}{1-\frac{\phi_N}{\phi_D}}+\frac{1}{\left(1-\left(\frac{\phi_N}{\phi_D}\right)^{1/2}\right)^2}\right).
\end{align*}
For $N$ sufficiently large, $L(N)>\ol{l}$, so $H^{\ol{l}}$ and $X=H-H^{\ol{l}}$ are independent, thus
$\E(e^{itH})=\E(e^{itH^{\ol{l}}})\cdot \E(e^{it(H-H^{\ol{l}})})$, and
\begin{align*}&\limsup_{N\to \infty} \left|\E(e^{itH})-e^{-t^2 \left(\mathcal{V}_0+\sum_{l=0}^{\ol{l}} \mathcal{V}_{l,l+1}\right) /2}\right|\\
&\le 
e^{-t^2 \left(\mathcal{V}_0+\sum_{l=0}^{\ol{l}} \mathcal{V}_{l,l+1}\right) /2}\cdot  \frac{t^2 \E(X^2)}{2}\le  \frac{t^2}{2} \left(\frac{\phi_N}{\phi_D}\right)^{\ol{l}} \frac{V_D}{\ul{c}_N} \left(\frac{1}{1-\frac{\phi_N}{\phi_D}}+\frac{1}{\left(1-\left(\frac{\phi_N}{\phi_D}\right)^{1/2}\right)^2}\right).
\end{align*}
By letting $\ol{l}\to \infty$, it follows that $\limsup_{N\to \infty} \left|\E(e^{itH})-e^{-t^2\sigma^2_S}\right|=0$, hence the convergence follows by 
the L\'evy-Cram\'er continuity theorem (see Theorem 3.3.17 of \cite{Durrett5th}). 

The proof for $S(c_R)$ follows the same lines, except that the variances of the terms for $l\ge L(N)$ need to be controlled separately using the same bounds as in the proof of Proposition \ref{supp:prop:unb}, we omit the details.
\end{proof}

\section{Convergence results}\label{supp:Sec:ProofConvergenceUBU}
The first set of results we prove are provided below for the convergence of the UBU scheme. Proving contraction of a coupling has been a popular method for establishing convergence rates both in the continuous time setting and for the discretization for Langevin dynamics (underdamped/kinetic) and Hamiltonian Monte Carlo (see for example \cite{bou2023mixing,bou2022couplings,bou2020coupling,dalalyan2017theoretical,dalalyan2020sampling,deligiannidis2021randomized,durmus2017nonasymptotic,eberle2019couplings,monmarche2022hmc,PM21,monmarche2020almost,sanz2021wasserstein,schuh2022global} and many more).

Our approach to obtain convergence rates is based on proving contraction for a  synchronous coupling. We need an appropriate metric to attain convergence, and contraction of the $\UBU$ scheme. We introduce the Wasserstein distance in this metric.

\begin{definition}[Weighted Euclidean norm]
\label{supp:def:weightednorm}
For $z = (x,v) \in \R^{2d}$ we introduce the weighted Euclidean norm
\[
\left|\left| z \right|\right|^{2}_{a,b} = \left|\left| x \right|\right|^{2} + 2b \left\langle x,v \right\rangle + a \left|\left| v \right|\right|^{2},
\]
for $a,b > 0$ with $b^{2}<a$. 
\end{definition}

\begin{remark}
Using the assumption $b^2<a$, we can show that this is equivalent to the Euclidean norm on $\R^{2d}$. Under the condition $b^2\le a/4$, we have
\begin{equation}\label{supp:eq:normequiv}
\frac{1}{2}\min(a,1) \|z\|^2\leq \frac{1}{2}||z||^{2}_{a,0} \leq ||z||^{2}_{a,b} \leq \frac{3}{2}||z||^{2}_{a,0}\leq \frac{3}{2}\max(a,1) \|z\|^2.\end{equation}
\end{remark}

\begin{definition}[$p$-Wasserstein distance]
\label{supp:def:wass}
Let us define $\mathcal{P}_p(\R^{2d})$
to be the set of probability measures which have $p$-th moment for $p\in[1,\infty)$ (i.e. $\E (\|Z\|^p)<\infty$). Then the $p$-Wasserstein distance in norm $\|\cdot \|_{a,b}$ between two measures $\mu,\nu\in \mathcal{P}_p(\R^{2d})$ is defined as
\begin{equation}
\label{supp:eq:wass_dist}
\mathcal{W}_{p,a,b}(\nu,\mu) =\Big(\inf_{\xi \in \Gamma(\nu,\mu)} \int_{\R^{2d}} \|z_1 - z_2\|^p_{a,b}d\xi(z_1,z_2)\Big)^{1/p},
\end{equation}
where $\| \cdot \|_{a,b}$ is the norm introduced before and that $\Gamma(\nu,\mu)$  is the set of measures with respective marginals of $\nu$ and $\mu$.
\end{definition}

Before we proceed, we need to introduce the concept of Wasserstein convergence, which most of the results rely upon. 
\begin{lemma}[Wasserstein convergence]
\label{supp:lem:wass_conv}
Let $1 \leq p \leq \infty$, $\mu,\nu\in \mathcal{P}_p(\R^{2d})$, and $a,b>0$ with $b^2 <a$. Let us assume that $(z_k)_{k\ge 0} = (x_k,v_k)_{k\ge 0}$ and $(\tilde{z}_k)_{k\ge 0} = (\tilde{x}_k,\tilde{v}_k)_{k\ge 0}$ are two Markov chains with state space $\Lambda$ and kernel $P_h$ defined on the same probability space (a coupling) such that $z_0\sim \nu$,  $\tilde{z}_0\sim \mu$, and $\E(\|z_0-\tilde{z}_0\|^p)=\left[\mathcal{W}_{p,a,b}(\nu,\mu)\right]^p$. If the following contractive property holds,
\begin{equation}\label{supp:eq:contractivity_cond}
\left[\E(\| \tilde{z}_{k+1} - z_{k+1} \|^p_{a,b}|z_{0:k},\tilde{z}_{0:k})\right]^{1/p} \leq (1-c(h))\| \tilde{z}_k - z_k\|_{a,b}\quad \text{ for every }k\ge 0,
\end{equation}
then we have
$$
\mathcal{W}_{p,a,b}\left(\nu P^n_{h} ,\mu P^n_{h} \right) \leq (1-c(h))^n \mathcal{W}_{p,a,b}(\nu,\mu)\quad \text{ for every }n\ge 0.
$$
\end{lemma}
\begin{remark}
The existence of an optimal coupling satisfying that $\E(\|z_0-\tilde{z}_0\|_{a,b}^p)=\left[\mathcal{W}_{p,a,b}(\nu,\mu)\right]^p$ follows by Theorem 4.1 of \cite{villani2009optimal}.
\end{remark}
\begin{proof}
By induction, we have $\E(\| \tilde{z}_{n} - z_{n} \|^p_{a,b}|z_{0},\tilde{z}_{0})\le (1-c(h))^n \|z_{0}-\tilde{z}_{0}\|_{a,b}^p$, and the result follows by taking expectations and using Definition \ref{supp:eq:wass_dist}.
\end{proof}
Now, we present our first proposition, a convergence result of the $\UBU$ scheme with full gradients.
\begin{restatable}{proposition}{propWasserstein}
\label{supp:prop:Wasserstein}
Suppose that $U$ is $m$-strongly convex and $M$-$\nabla$Lipschitz.
Let \begin{equation}
\label{supp:eq:abeq}
a=\frac{1}{M},\quad b=\frac{1}{\gamma}, \quad c_{2}(h)=\frac{mh}{4\gamma}, \quad c(h)=\frac{mh}{8\gamma}.
\end{equation}
Let $P_{h}$ denote the transition kernel for a step of $\UBU$ with stepsize $h$. For all $\gamma \geq \sqrt{8M}$, $h < \frac{1}{2\gamma}$, $1 \leq p \leq \infty$, $\mu,\nu \in \mathcal{P}_{p}(\R^{2d})$, \eqref{supp:eq:wass_dist} holds. Hence for all  $n \in \mathbb{N}$,
\[
\mathcal{W}_{p,a,b}\left(\nu P^n_{h} ,\mu P^n_{h} \right) \leq \left(1 - c_{2}(h)\right)^{n/2} \mathcal{W}_{p,a,b}\left(\nu,\mu\right) \leq \left(1 - c(h)\right)^{n} \mathcal{W}_{p,a,b}\left(\nu,\mu\right).
\]
Further to this, $P_{h}$ has a unique invariant measure $\pi_{h}$ satisfying that $\pi_{h} \in \mathcal{P}_{p}(\R^{2d})$ for all $1 \leq p \leq \infty$.
\end{restatable}
\begin{remark}
We are going to use the same choices of $a$ and $b$ as stated in \eqref{supp:eq:abeq} everywhere in the paper.
\end{remark}
\begin{restatable}{corollary}{CorContinuousContraction}\label{supp:CorContinuousContraction}
\label{supp:cor:continuous_contraction}
Suppose that $U$ is an $m$-strongly convex $M$-$\nabla$Lipschitz potential, $\gamma \geq \sqrt{8M}$, $1 \leq p \leq 2$, $\mu,\nu \in \mathcal{P}_{p}(\R^{2d})$. Suppose that $(X_0,V_0)\sim \mu$, then  the solution of the continuous kinetic Langevin dynamics exists in the strong sense for any $t\ge 0$, and the corresponding Markov kernel $P_t^{\mathrm{cont}}$ satisfies 
\begin{equation}\label{supp:eq:continuous_contraction}
\mathcal{W}_{p,a,b}\left(\nu P_t^{\mathrm{cont}} ,\mu P_t^{\mathrm{cont}} \right) \leq \exp\left(-\frac{mt}{8\gamma}\right) \mathcal{W}_{p,a,b}\left(\nu,\mu\right) \quad \text{ for }\quad a=\frac{1}{M}, b=\frac{1}{\gamma}.
\end{equation}
\end{restatable}
\begin{remark}
    One can improve the restriction on $\gamma$ slightly by writing the potential as a perturbation of a quadratic as in \cite{schuh2022global}. Due to the restrictions on the stepsize $h$ and the friction parameter $\gamma$ in Proposition \ref{supp:prop:Wasserstein}, $c(h)=\O\left(\frac{m}{M}\right)$ for all allowed parameter choices. In general, for $\nabla$Lipschitz, strongly-convex potentials, it may be impossible to prove contraction using such a quadratic form argument and synchronous coupling for $\gamma \leq \O(\sqrt{M})$ as explained in \cite{monmarche2020almost}. In the continuous time dynamics, $\gamma=\O(\sqrt{m})$ seems to yield the fastest convergence rate, as explained in \cite{lu2022explicit}.  In Example \ref{supp:ex:Gaussian_example_UBU}, we show that for Gaussian targets, $\UBU$ has an accelerated convergence rate $c(h)=\O(\sqrt{\frac{m}{M}})$ with the choice $\gamma=\O(\sqrt{m})$ and $h=\O(1/\sqrt{M})$.
\end{remark}

\begin{proof}[Proof of Proposition \ref{supp:prop:Wasserstein}]
We follow the approach of \cite{PM21}[Corollary 20]. It is sufficient to prove contraction of a synchronous coupling of Markov chains in an appropriate norm, we will use the  $\|\cdot\|_{a,b}$ norm of Definition \ref{supp:def:weightednorm} with $a=\frac{1}{M}$, $b=\frac{1}{\gamma}$. Based on the assumptions, we have $b^{2}<a/4$.  Hence, \eqref{supp:eq:normequiv} holds.
 
We aim to show that contraction occurs in this norm for two Markov chains simulated by the same discretization $z_{n} = (x_{n},v_{n}) \in \R^{2d}$ and $\Tilde{z}_{n} = (\Tilde{x}_{n},\Tilde{v}_{n}) \in \R^{2d}$ that are synchronously coupled (i.e. share the same Gaussian random variables $\xi^{(1)},\ldots, \xi^{(4)}$), that is, 
\begin{equation}\label{supp:eq:cont_1}
  ||\tilde{z}_{k+1} - z_{k+1}||^{2}_{a,b} < \big(1 - c\left(h\right)\big)^2||\Tilde{z}_{k} - z_{k}||^{2}_{a,b}.  
\end{equation}
 Let $c_2(h)=1-(1-c(h))^2$, $z^{\Delta}_{j} = \Tilde{z}_{j} - z_{j}$ for $j \in \mathbb{N}$, then \eqref{supp:eq:cont_1} is equivalent to showing that 
\begin{equation}\label{supp:eq:contraction_matrix_form}
 \left(z^{\Delta}_{k}\right)^{T}\left(\left(1 - c_2\left(h\right)\right)\M- \PP^{T}\mathcal{M}\PP\right )z^{\Delta}_{k} > 0,    \quad \textnormal{where} \quad \mathcal{M} = \begin{pmatrix}
    I_d & b I_d \\
    b I_d & a I_d
\end{pmatrix},
\end{equation}
and $z^{\Delta}_{k+1} = \PP z^{\Delta}_{k}$ ($\PP$ depends on $z_{k}$ and $\Tilde{z}_{k}$, but we omit this in the notation).

Proving contraction for a general scheme is equivalent to showing that the matrix $\HH :=  \left(1 - c_2(h)\right)\M - \PP^{T}\M \PP \succ 0$ is positive definite. The matrix $\HH$ is symmetric and hence of the block form
\begin{equation}\label{supp:eq:contraction_matrix}
\HH = \begin{pmatrix}
    A & B \\
    B^T & C
\end{pmatrix},    
\end{equation}
where $A$, $B$, $C$ are $d\times d$ matrices, then  
\begin{equation}\HH\succ 0 \quad \Leftrightarrow \quad A \succ 0 \quad \text{ and } \quad C - BA^{-1}B \succ 0,
\end{equation}
as shown in Theorem 7.7.7 of \cite{matrixanalysis}. Further it is straightforward to show that if $A$, $B$ and $C$ commute then 
\begin{equation}\label{supp:eq:posdefcondition2}\HH\succ 0 \quad \Leftrightarrow \quad A \succ 0 \quad \text{ and } \quad AC - B^{2} \succ 0.
\end{equation}
Considering two synchronously coupled trajectories of the $\UBU$ scheme, such that they have common noise and consider the difference process $x^{\Delta}:= \left(\Tilde{x}_{j} - x_{j}\right)$, $v^{\Delta} = \left(\Tilde{v}_{j} - v_{j}\right)$ and $z^{\Delta} = \left(x^{\Delta}, v^{\Delta}\right)$, where $z^{\Delta}_{j} = \left(x^{\Delta}_{j}, v^{\Delta}_{j}\right)$ for $j = k,k +1$ for $k \in \mathbb{N}$. Let $\eta = \exp{\{-\gamma h/2\}}$, and 
\[Q = \int^{1}_{0}\nabla^{2}U\left(\Tilde{x}_{k} + t(x_{k} - \Tilde{x}_{k}\right)dt.\]
By convexity, we have $m I_d\preceq Q\preceq M I_d$.
Using the definition of the $\UBU$ scheme, we can show that 
$z^{\Delta}_{k+1} = \PP z^{\Delta}_{k}$ and $\HH :=  \left(1 - c_2(h)\right)\M - \PP^{T}\M \PP=\begin{pmatrix}
    A & B \\
    B^T & C
\end{pmatrix}$  has elements of the form
\begin{align*}
    A &= -c_2(h)I_d +Q \left(2 b h \eta+\frac{2 h \left(1-\eta\right)}{\gamma}\right)  + Q^2 \left(-a h^2 \eta^{2}-\frac{h^2 \left(1-\eta\right)^2}{\gamma^2}-\frac{2 b h^2 \eta \left(1-\eta\right)}{\gamma}\right)\\
    B &= \left(\left(1-\eta^{2}\right)\left(b - \frac{1}{\gamma}\right)-bc_2(h)\right)I_d+ Q^2 \left(-\frac{a h^2 \eta^{2} \left(1-\eta\right)}{\gamma}-\frac{2 b h^2 \eta \left(1-\eta\right)^2}{\gamma^2}-\frac{h^2
   \left(1-\eta\right)^3}{\gamma^3}\right)\\
    &+Q \left(a h \eta^{3}+\frac{h \left(\eta+1\right) \left(1-\eta\right)^2}{\gamma^2}+\frac{h \left(1-\eta\right)^2}{\gamma^2}+\frac{bh \eta^{2} \left(1-\eta\right)}{\gamma}+\frac{bh \eta \left(1-\eta\right)}{\gamma}+\frac{bh \eta \left(1-\eta^{2}\right)}{\gamma}\right)\\
    C &=  \left(a (1 - \eta^{4}) -\frac{2 b \eta^{2} \left(1-\eta^{2}\right)}{\gamma}-\frac{\left(1-\eta^{2}\right)^2}{\gamma^2}
 -ac_{2}\left(h\right)\right)I_d  \\&+ Q^2 \left(-\frac{a h^2 \eta^{2} \left(1-\eta\right)^2}{\gamma^2}-\frac{2 b h^2 \eta \left(1-\eta\right)^3}{\gamma^3}-\frac{h^2 \left(1-\eta\right)^4}{\gamma^4}\right)\\
 &+Q \left(\frac{2 a h \eta^{3} \left(1-\eta\right)}{\gamma}+\frac{2 b h \eta^{2} \left(1-\eta\right)^2}{\gamma^2}+\frac{2 b h \eta \left(\eta+1\right) \left(1-\eta\right)^2}{\gamma^2}+\frac{2 h \left(\eta+1\right) \left(1-\eta\right)^3}{\gamma^3}\right).
  \end{align*}

We will now check that $\HH\succ 0$ using \eqref{supp:eq:posdefcondition2}. By firstly considering $A$ we wish to show that all its eigenvalues are positive which can be precisely stated as
 \begin{align*}
 P_{A}(\lambda) &\geq  -c_{2}\left(h\right) + \frac{2h\lambda}{\gamma} + \left(-\frac{1}{M} - \frac{2h}{\gamma}\right)h^2 \lambda^2\\
 &\geq \frac{7h\lambda}{4\gamma} + \left(-\frac{1}{M} - \frac{1}{\gamma^2}\right)h^2 \lambda^2 > 0,
 \end{align*}
 where $\lambda$ is an eigenvalue of $Q$ ($m \leq \lambda \leq M$),  $P_{A}(\lambda)$ denotes the eigenvalue of $A$ according to the same eigenvector ($Q, A, B, C$ are all symmetric and have the same eigenvectors here). We  used our assumptions that $\gamma^{2} \geq M$, $1 - \eta \leq h\gamma/2$, and $h<\frac{1}{2\gamma}$. Hence, we have $A \succ 0$.

Now it remains to prove that $AC - B^{2} \succ 0$, now we have that $AC - B^{2}$ is a polynomial of $Q$, which we denote $P_{AC - B^{2}}(Q)$ and hence has eigenvalues dictated by the eigenvalues $\lambda$ of $Q$. Because the terms are more complicated than the previous discretizations, we choose a convenient way of expanding the expression, which can obtain positive definiteness. That is to expand the expression in terms of $a$. Therefore one can show that $P_{AC-B^{2}}(\lambda) = c_{0} + c_{1}a + c_{2}a^{2}$, where  
\begin{align*}
    &c_{1} + c_{2} a = \frac{h^2 c_2(h)\lambda^2 \eta^4}{\gamma^2}-\frac{2 h^2 c_2(h) \lambda^2
   \eta^2}{\gamma^2}-\frac{h^2 \lambda^2 \eta^4}{\gamma^2}+\frac{2 h^2 \lambda^2 \eta^2}{\gamma^2}+\frac{2 h c_2(h) \lambda \eta^4}{\gamma}-\frac{2 h \lambda \eta^4}{\gamma}\\
   &+c_2(h)
   \eta^4+\frac{h^2 c_2(h) \lambda^2}{\gamma^2}-\frac{h^2
   \lambda^2}{\gamma^2}-\frac{2 h c_2(h) \lambda}{\gamma}+\frac{2 h
   \lambda}{\gamma}+ c_2(h)^2 - c_2(h)\\
   &+  a\left(-\eta^2 h^2 \lambda^2 + \eta^2 h^{2}c_2(h)\lambda^2\right) \\
    &\geq \frac{h\lambda}{\gamma}\left(1-c_2(h)\right)\left(\frac{7}{4}(1-\eta^{4}) - \frac{4h\lambda}{\gamma} - h\gamma \right).
\end{align*}

Furthermore, we have that 
\begin{align*}
    &c_{0} = \frac{h^2 (1-c_2(h))\lambda^2 \eta^4}{\gamma^4}-\frac{2 h^2 (1-c_2(h)) \lambda^2 \eta^2}{\gamma^4}+\frac{2 h (1-c_2(h)) \lambda \eta^4}{\gamma^3}+\frac{c_2(h)(1-\eta^4)}{\gamma^2}\\
    &-\frac{c_2(h)^2}{\gamma^2}+\frac{h^2 \lambda^2(1-c_2(h))}{\gamma^4}-\frac{2 h \lambda(1-c_2(h))}{\gamma^3}\\
    &> \frac{h\lambda}{\gamma^{3}}\left(1-c_2(h)\right)\left(\frac{h\lambda}{\gamma^{3}}\left(1-\eta^{2}\right)^{2} - 2\left(1-\eta^{4}\right)\right),
\end{align*}
where now we combine this with the previous estimate 
\begin{align*}
P_{AC - B^{2}}(\lambda) &> \frac{h\left(1-c_2(h)\right)}{\gamma}\left( \frac{7}{4}(1-\eta^{4}) - \frac{4h\lambda}{\gamma} - h\gamma - \frac{2\lambda\left(1-\eta^{4}\right)}{\gamma^{2}}\right) > 0,   
\end{align*}
which is true when $\gamma \geq \sqrt{8M}$ and we have used the fact that $1-\eta^{4} \geq h\gamma$. Hence $AC - B^{2} \succ 0$ and our contraction results hold. All computations can be checked using \texttt{Mathematica}.
The first claim follows by Lemma \ref{supp:lem:wass_conv} using \eqref{supp:eq:cont_1}. The existence of a unique invariant distribution $\pi_{h}\in \mathcal{P}_{p}(\R^{2d})$ follows by the same argument as in \cite{PM21}[Corollary 20].
\end{proof}

\begin{proof}[Proof of Corollary \ref{supp:cor:continuous_contraction}]
By the triangle inequality, we have that for $a=\frac{1}{M}$, $b=\frac{1}{\gamma}$, any $n\in \mathbb{N}$ such that $n>t\cdot 2\gamma$,
\begin{align*}
&\mathcal{W}_{p,a,b}\left(\nu P_t^{\mathrm{cont}} ,\mu P_t^{\mathrm{cont}} \right) \\
&\leq 
\mathcal{W}_{p,a,b}\left(\nu P_{t/n}^{n} ,\mu P_{t/n}^n \right)
+\mathcal{W}_{p,a,b}\left(\nu P_t^{\mathrm{cont}}, \nu P_{t/n}^{n} \right)+\mathcal{W}_{p,a,b}\left(\mu P_t^{\mathrm{cont}}, \mu P_{t/n}^{n} \right).
\end{align*}
The first term can be bounded using Proposition \ref{supp:prop:Wasserstein}, and the upper bound can be shown to converge to $\exp\left(-\frac{mt}{8\gamma}\right)$ as $n\to \infty$. The second and third terms can be shown to converge to $0$ as $n\to \infty$ using the strong convergence of the $\UBU$ discretization towards the diffusion (strong order 1 under these assumptions), which was established in Section 7.7 of \cite{sanz2021wasserstein}, and the claim of the corollary now follows.
\end{proof}

\begin{example}\label{supp:ex:Gaussian_example_UBU}
Considering the anisotropic Gaussian distribution on $\R^{2}$ with a $m$-strongly convex and $M$-$\nabla$ Lipschitz potential $U: \R^{2} \mapsto \R$ given by
 \[
 U(x,y) = \frac{1}{2}mx^{2} + \frac{1}{2}My^{2}. 
\]
For the BU scheme the transition matrix for the difference chain of synchronously coupled chains is given by the matrix
 \[
P = \begin{pmatrix} I -h\left(\frac{1-\eta^{2}}{\gamma}\right)Q & \frac{1-\eta^{2}}{\gamma}I\\
-h\eta^{2} Q & \eta^{2} I 
\end{pmatrix}\text{, where }Q = \begin{pmatrix}
    m & 0 \\
    0 & M \end{pmatrix},
\]
with eigenvalues 
\[\frac{1 + \eta^{2} - h\frac{1-\eta^{2}}{\gamma}\lambda \pm \sqrt{-4\eta^{2} + \left(1 + \eta^{2} - h\frac{1-\eta^{2}}{\gamma}\lambda\right)^2}}{2},\]
for $\lambda = m,M$. For stability and contraction, we require that 

\begin{equation}\label{supp:eq:condUBU}
\lambda_{\max}:=\max_{\lambda\in \{m,M\}}\left|\frac{1 + \eta^{2} - h\frac{1-\eta^{2}}{\gamma}\lambda \pm \sqrt{-4\eta^{2} + \left(1 + \eta^{2} - h\frac{1-\eta^{2}}{\gamma}\lambda\right)^2}}{2}\right| < 1. 
\end{equation}
From this, we can compute the stepsize restrictions and the best convergence rate as, by Gelfand's formula, the asymptotic contraction rate exactly equals $1-\lambda_{\max}$. Due to the convexity of the absolute value function it is necessary that $\frac{1}{2}|1+\eta^{2}-h\frac{1-\eta^{2}}{\gamma}M|<1$, therefore $h < \sqrt{\frac{8}{M}}$, when $h < \frac{1}{2\gamma}$. In the moderate to high friction regime, the contraction rate can be written as
\[
c = \frac{1-\eta^{2} + h\frac{1-\eta^{2}}{\gamma}m - \sqrt{-4h\left(\frac{1-\eta^{2}}{\gamma}\right)m + \left(1 - \eta^{2} + h\frac{1-\eta^{2}}{\gamma}m\right)^2}}{2}
\]
which can be shown to be $\O(mh/\gamma)$ for $\gamma \geq \O(\sqrt{M})$ and $h < \O(\frac{1}{\gamma})$ for appropriate constants. In the low friction regime, we set $\gamma$ such that $-4\eta^{2} + \left(1+\eta^{2} - h\frac{1-\eta^{2}}{\gamma}m\right)^2 = 0$, noting that the solution to   this yields $\gamma$ to be $\O(\sqrt{m})$.  In this case, the eigenvalues of $P$ are
\[
\eta, \qquad \frac{1}{2}\left(1+\eta^{2} -h\frac{1-\eta^{2}}{\gamma}M \pm \sqrt{-4\eta^{2} + \left(1+\eta^{2} -h\frac{1-\eta^{2}}{\gamma}M\right)^2}\right),
\]
with modulus $\eta$ when $\left(1+\eta^{2} - h\frac{1-\eta^{2}}{\gamma}M\right)^{2} < 4\eta^{2}$. This restriction implies that $h$ is $\O(1/\sqrt{M})$.
The contraction rate is therefore given by
\[
c = 1-\eta \geq \frac{h\gamma}{4} = \O\left(\sqrt{\frac{m}{M}}\right),
\]
where $h$ is $\O(1/\sqrt{M})$. We have the corresponding contraction rate results for $\UBU$ as well due to the fact that $(\UBU)^{n} = \text{U}(\text{BU})^{n-1}\text{U}$ and $\text{U}$ is Lipschitz.
\end{example}


A key ingredient to establishing some variance bounds for the inexact gradient methods is to establish non-asymptotic bounds on the fourth moment of the distance to the minimizer. To do this we use a Lyapunov function similar to the one used for kinetic Langevin dynamics in \cite{eberle2019couplings} and inspired by \cite{mattingly2002ergodicity}. Related Lyapunov functions have also been used in \cite{durmus2021uniform} for discretized kinetic Langevin dynamics and \cite{Karoni2023} for optimizers based on Langevin dynamic methods. These bounds provide novel drift conditions in $L^{4}$ for $\UBU$ scheme and can be extended to the case of stochastic gradients.

The following lemma will be useful for the argument.
\begin{lemma}[Convexity bound]\label{supp:lem:convexity_type}
For all $x \in \R^{d}$ and for a $m-$strongly convex, $M$-$\nabla$Lipschitz potential $U:\R^{d} \to \R$ with minimizer $x^{*} \in \R^{d}$ such that $\nabla U(x^{*}) = 0$, we have 
\begin{align*}
 \left(x -x^{*}\right)\cdot \left(\nabla U(x)-\nabla U(x^{*})\right)/2 \geq \lambda \left(U(x) - U(x^{*}) + \gamma^{2}\|x-x^{*}\|^{2}/4\right)
\end{align*}
for
\begin{equation}\label{supp:eq:lambdadef}
    \lambda=\min\left(\frac{1}{4},\frac{m}{\gamma^2}\right).
\end{equation}
\end{lemma}
\begin{proof}
By convexity, it follows that 
$\left(x -x^{*}\right)\cdot \left(\nabla U(x)-\nabla U(x^{*})\right)/4\ge  (U(x) - U(x^{*}))/4$, and by $m$-strong convexity, we have $\left(x -x^{*}\right)\cdot \left(\nabla U(x)-\nabla U(x^{*})\right)/4\ge m \|x-x^{*}\|^{2}/4$. We obtain the result by adding up these two inequalities.
\end{proof}





\begin{proposition} \label{supp:prop:full_gradient_L4}
    Consider the $\UBU$ scheme with the underlying potential $U:\R^{d} \to \R$ is $M$-$\nabla$Lipschitz and $m$-strongly convex. Denote $x^{*} \in \R^{d}$ to be the minimizer of $U$ such that $\nabla U(x^{*}) = 0$ and $(x_{k},v_{k},\overline{x}_{k})_{k\in \mathbb{N}}$ to be defined by \eqref{supp:eq:disc_v}-\eqref{supp:eq:disc_x} the iterates of the full gradient $\UBU$ scheme and the points of gradient evaluation within each iteration. Further assume that $h < \min{\left( 1,\frac{1}{2\gamma}, \frac{\lambda}{8 \gamma (4+\lambda)}\right)}$ and $\gamma^{2} \geq M$, then we have 
\begin{align*}
    &\mathbb{E}\left[\|\overline{x}_{k}-x^{*}\|^{4} \mid x_{0},v_{0}\right]\leq \frac{4}{m^{2}}\Bigg[4\left(1-\frac{c_4(h)}{2}\right)^{k}\left(\gamma^{4}\|x_{0}-x^{*}\|^{4} + \|v_{0}\|^{4} +122\gamma^{2}h^{2}d^{2}\right)\\
  &+2\frac{\frac{(6h\gamma d + 160h\gamma(1+\lambda^{2}))^{2}}{4c_4(h)}  + 24h^{2}\gamma^{2}d^{2}}{c_4(h)}\Bigg],  
\end{align*}
where 
\begin{equation}\label{supp:eq:c4hdef}c_4(h) := h\lambda\gamma -8h^{2}\gamma^{2}(4 + \lambda).
\end{equation}
\end{proposition}
\begin{proof}
Using the fact that $(\mathcal{UBU})^{n} = \mathcal{U}(\mathcal{BU})^{n-1}\mathcal{BU}$ we can consider convergence of $\mathcal{BU}$,  We have that the $\mathcal{BU}$ function can be written as the update rule
\begin{align}
\label{supp:eq:BUsteps:exactx}    
    \overline{x}_{k+1} &= \overline{x}_{k} + \frac{1-\eta^{2}}{\gamma}\left(\overline{v}_{k} - h\nabla U(\overline{x}_{k})\right) + \sqrt{\frac{2}{\gamma}}\left(\mathcal{Z}^{(1)}\left(h,\xi^{(1)}_{k+1}\right) - \mathcal{Z}^{(2)}\left(h,\xi^{(1)}_{k+1},\xi^{(2)}_{k+1}\right)\right),\\
\label{supp:eq:BUsteps:exactv}
    \overline{v}_{k+1} &= \eta^{2}\left(\overline{v}_{k} - h\nabla U(\overline{x}_{k})\right) + \sqrt{2\gamma} \mathcal{Z}^{(2)}\left(h,\xi^{(1)}_{k+1},\xi^{(2)}_{k+1}\right),
\end{align}
where we used the notation $\left(\overline{x}_{k}\right)_{k \in \mathbb{N}}$ because this is the point of the gradient evaluation at each step of UBU and is the same as the $\left(\overline{x}_{k}\right)_{k \in \mathbb{N}}$ in \eqref{supp:eq:disc_y}. As a reminder,
\begin{align*}
\mathcal{Z}^{(1)}\left(h,\xi^{(1)}_{k+1}\right) &= \sqrt{h}\xi^{(1)}_{k+1}\\
\mathcal{Z}^{(2)}\left(h,\xi^{(1)}_{k+1},\xi^{(2)}_{k+1}\right) &= \sqrt{\frac{1-\eta^{4}}{2\gamma}}\left(\sqrt{\frac{1-\eta^{2}}{1+\eta^{2}}\cdot \frac{2}{\gamma h}}\xi^{(1)}_{k+1} + \sqrt{1-\frac{1-\eta^{2}}{1+\eta^{2}}\cdot\frac{2}{\gamma h}}\xi^{(2)}_{k+1}\right).
\end{align*}

We choose our Lyapunov function $\mathcal{V}:\mathbb{R}^{2d} \to \mathbb{R}$, defined for $(x,v) \in \mathbb{R}^{2d}$ by 
\begin{equation}\label{supp:eq:Vxvdef}
\mathcal{V}(x,v):=U(x) - U(x^{*}) + \frac{1}{4}\gamma^{2}\left( \|x-x^{*} + \gamma^{-1}v\|^{2} + \|\gamma^{-1}v\|^{2}-\lambda\|x-x^{*}\|^{2}\right).
\end{equation}
It is easy to check that for all $(x,v) \in \mathbb{R}^{2d}$, $\|x-x^{*} + \gamma^{-1}v\|^{2} + \|\gamma^{-1}v\|^{2}\ge \frac{1}{2}\|x-x^{*}\|^2$ and hence using $\eqref{supp:eq:lambdadef}$, 
\begin{equation}\mathcal{V}(x,v)\ge \left(\frac{m}{2}+\frac{1}{16}\gamma^2\right)\|x-x^*\|^2.\end{equation}

In order to have control over fourth moments $\E [\|\ol{x}_k-x^*\|^4]$, we start with
\begin{align*}
&\mathbb{E}\left[\mathcal{V}(\overline{x}_{k+1},\overline{v}_{k+1})^{2} \mid \overline{x}_{k},\overline{v}_{k}\right] =\\
    &\mathbb{E}\left[\left(U(\overline{x}_{k+1}) - U(x^{*}) + \frac{1}{4}\gamma^{2}\left( \|\overline{x}_{k+1}-x^{*} + \gamma^{-1}\overline{v}_{k+1}\|^{2} + \|\gamma^{-1}\overline{v}_{k+1}\|^{2}-\lambda\|\overline{x}_{k+1}-x^{*}\|^{2}\right)\right)^{2}\mid \overline{x}_{k},\overline{v}_{k}\right],
\end{align*}
and using \cite{nesterov2018lectures}[Lemma 1.2.3] we have
\begin{align*}
&U(\overline{x}_{k+1}) - U(x^{*}) \leq U(\overline{x}_{k}) - U(x^{*}) + \left[\nabla U(\overline{x}_{k}) \cdot \left(\overline{x}_{k+1}-\overline{x}_{k}\right)\right] + \frac{M}{2}\left\|\overline{x}_{k+1}-\overline{x}_{k}\right\|^{2}
\end{align*}
and
\begin{align*}
&\mathbb{E}\left[\mathcal{V}(\overline{x}_{k+1},\overline{v}_{k+1})^{2} \mid \overline{x}_{k},\overline{v}_{k}\right] \leq \mathbb{E}\Bigg[\Bigg(U(\overline{x}_{k}) - U(x^{*}) + \left[\nabla U(\overline{x}_{k}) \cdot \left(\overline{x}_{k+1}-\overline{x}_{k}\right)\right] + \frac{M}{2}\left\|\overline{x}_{k+1}-\overline{x}_{k}\right\|^{2} \\
&+ \frac{1}{4}\gamma^{2}\left( \|\overline{x}_{k+1}-x^{*} + \gamma^{-1}\overline{v}_{k+1}\|^{2} + \|\gamma^{-1}\overline{v}_{k+1}\|^{2}-\lambda\|\overline{x}_{k+1}-x^{*}\|^{2}\right)\Bigg)^{2}\mid \overline{x}_{k},\overline{v}_{k}\Bigg].
 \end{align*}
Now, we can decompose the right-hand side in the form
\[
\mathbb{E}\left(\left(r(\overline{x}_{k},\overline{v}_{k}) + \mathbf{s}(\overline{x}_{k},\overline{v}_{k}) \cdot (\xi^{(1)}_{k+1},\xi^{(2)}_{k+1}) + (\xi^{(1)}_{k+1},\xi^{(2)}_{k+1})^{T} \mathcal{T} (\xi^{(1)}_{k+1},\xi^{(2)}_{k+1})\right)^{2} \mid \overline{x}_{k},\overline{v}_{k}\right),
\]
for $r: \R^{2d} \to \R$, $\mathbf{s}:\R^{2d} \to \R^{2d}$ and $\mathcal{T} \in \R^{2d \times 2d}$. We then have
\begin{align*}
&\mathbb{E}\left[\mathcal{V}(\overline{x}_{k+1},\overline{v}_{k+1})^{2} \mid \overline{x}_{k},\overline{v}_{k}\right] \leq r^{2}(\overline{x}_{k},\overline{v}_{k}) + \mathbb{E}_{\xi^{(1)}_{k+1},\xi^{(2)}_{k+1}}\Big(\left(\mathbf{s}(\overline{x}_{k},\overline{v}_{k}) \cdot (\xi^{(1)}_{k+1},\xi^{(2)}_{k+1})\right)^{2} \\
&+ \left((\xi^{(1)}_{k+1},\xi^{(2)}_{k+1})^{T} \mathcal{T} (\xi^{(1)}_{k+1},\xi^{(2)}_{k+1})\right)^{2} + 2r(\overline{x}_{k},\overline{v}_{k})(\xi^{(1)}_{k+1},\xi^{(2)}_{k+1})^{T} \mathcal{T} (\xi^{(1)}_{k+1},\xi^{(2)}_{k+1})\Big),
\end{align*}
using the fact that $\xi^{(1)}_{k+1}$ and $\xi^{(2)}_{k+1}$ are independently distributed and have zero first and third moments. The terms $r, \mathbf{s}$ and $\mathcal{T}$ are given by 
\begin{align*}
   &r(\overline{x}_{k},\overline{v}_{k}) = \mathcal{V}(\overline{x}_{k},\overline{v}_{k}) -\frac{h\gamma}{2}\nabla U(\overline{x}_{k}) \cdot (\overline{x}_{k}-x^{*} + \gamma^{-1}\overline{v}_{k}) +\frac{1-\eta^{2}}{\gamma}\overline{v}_{k} \cdot \nabla U(\overline{x}_{k}) -\frac{\lambda(1-\eta^{2})\gamma}{2}\langle \overline{x}_{k}-x^{*},\overline{v}_{k}\rangle \\
&- \frac{1-\eta^{4}}{4}\|\overline{v}_{k}\|^{2}- \frac{h\eta^{4}}{2}\overline{v}_{k} \cdot \nabla U(\overline{x}_{k})-h\frac{1-\eta^{2}}{\gamma}\nabla U(\overline{x}_{k}) \cdot \nabla U(\overline{x}_{k}) + h^{2}\frac{(1+\eta^{4})}{4}\|\nabla U(\overline{x}_{k})\|^{2}\\
&+ \left(\frac{M}{2}-\frac{\gamma^{2}\lambda}{4}\right)\left(\frac{1-\eta^{2}}{\gamma}\right)^{2}\|\overline{v}_{k}-h\nabla U(\overline{x}_{k})\|^{2} +h\frac{\lambda(1-\eta^{2})\gamma}{2}\langle \overline{x}_{k}-x^{*},\nabla U(\overline{x}_{k})\rangle,
\\
&\mathbf{s}(\overline{x}_{k},\overline{v}_{k})\cdot \left(\xi^{(1)}_{k+1},\xi^{(2)}_{k+1} \right) = ((\sqrt{h}-a_{1})\xi^{(1)}_{k+1} - a_{2}\xi^{(2)}_{k+1})\cdot \left(\frac{M\sqrt{2\gamma}}{2\gamma}\frac{1-\eta^{2}}{\gamma} (\overline{v}_{k}-h\nabla U(\overline{x}_{k})) + \sqrt{\frac{2}{\gamma}}\nabla U(\overline{x}_{k})\right)\\
&+ \frac{\sqrt{2\gamma h}}{4}\gamma\left(\overline{x}_{k}-x^{*} + \gamma^{-1}\overline{v}_{k}-\frac{h}{\gamma}\nabla U(\overline{x}_{k})\right)\cdot \xi^{(1)}_{k+1} + \frac{\eta^{2} \sqrt{2\gamma}}{4} (\overline{v}_{k}-h\nabla U(\overline{x}_{k})) \cdot \left(a_{1}\xi^{(1)}_{k+1} + a_{2}\xi^{(2)}_{k+1}\right)\\
&-\frac{\lambda \gamma \sqrt{2\gamma}}{4}\left(\overline{x}_{k}-x^{*} + \frac{1-\eta^{2}}{\gamma}\left(\overline{v}_{k} - h \nabla U(\overline{x}_{k})\right) \cdot \left((\sqrt{h}-a_{1})\xi^{(1)}_{k+1} - a_{2}\xi^{(2)}_{k+1} \right)\right),\\
&(\xi^{(1)}_{k+1},\xi^{(2)}_{k+1})^{T}\mathcal{T}(\xi^{(1)}_{k+1},\xi^{(2)}_{k+1})= \left( \frac{M}{\gamma} - \frac{\lambda\gamma}{2}\right)\left\|\left( \sqrt{h}-a_{1}\right)\xi^{(1)}_{k+1} - a_{2}\xi^{(2)}_{k+1}\right\|^{2} \\
&+\frac{h\gamma}{2}\left\|\xi^{(1)}_{k+1}\right\|^{2} + \frac{\gamma}{2}\left\| a_{1}\xi^{(1)}_{k+1} + a_{2}\xi^{(2)}_{k+1}\right\|^{2},
\end{align*}
where we have defined $\mathcal{Z}^{(2)}(h,\xi^{(1)}_{k+1},\xi^{(2)}_{k+1}) := a_{1}\xi^{(1)}_{k+1} + a_{2}\xi^{(2)}_{k+1}$ and $\mathcal{Z}^{(1)}(h,\xi^{(1)}_{k+1}) := \sqrt{h}\xi^{(1)}_{k+1}$ and $\mathcal{Z}^{(1)}(h,\xi^{(1)}_{k+1})-\mathcal{Z}^{(2)}(h,\xi^{(1)}_{k+1},\xi^{(2)}_{k+1}) = (\sqrt{h}-a_{1})\xi^{(1)}_{k+1} - a_{2}\xi^{(2)}_{k+1}$ with $|\sqrt{h}-a_{1}| \leq 2\sqrt{h}$, $|a_{2}| \leq \sqrt{h}$ and $|a_{1}|\leq \sqrt{h}$.

We start by bounding the deterministic component $r$:
\begin{align*}
r(\overline{x}_{k},\overline{v}_{k}) &= \mathcal{V}(\overline{x}_{k},\overline{v}_{k}) -\frac{h\gamma}{2}\nabla U(\overline{x}_{k}) \cdot (\overline{x}_{k}-x^{*} + \gamma^{-1}\overline{v}_{k}) +\frac{1-\eta^{2}}{\gamma}\overline{v}_{k} \cdot \nabla U(\overline{x}_{k}) \\
&-\frac{\lambda(1-\eta^{2})\gamma}{2}\langle \overline{x}_{k}-x^{*},\overline{v}_{k}\rangle - \frac{1-\eta^{4}}{4}\|\overline{v}_{k}\|^{2}- \frac{h\eta^{4}}{2}\overline{v}_{k} \cdot \nabla U(\overline{x}_{k})+ \O(h^{2})
\end{align*}
where the higher-order terms are given by
\begin{align*}
&-h\frac{1-\eta^{2}}{\gamma}\nabla U(\overline{x}_{k}) \cdot \nabla U(\overline{x}_{k}) + \left(\frac{M}{2}-\frac{\gamma^{2}\lambda}{4}\right)\left(\frac{1-\eta^{2}}{\gamma}\right)^{2}\|\overline{v}_{k}-h\nabla U(\overline{x}_{k})\|^{2} \\&+h\frac{\lambda(1-\eta^{2})\gamma}{2}\langle \overline{x}_{k}-x^{*},\nabla U(\overline{x}_{k})\rangle
+ h^{2}\frac{(1+\eta^{4})}{4}\|\nabla U(\overline{x}_{k})\|^{2}.
\end{align*}
Using Lemma \ref{supp:lem:convexity_type} we have
\begin{align*}
&r(\overline{x}_{k},\overline{v}_{k}) \leq \mathcal{V}(\overline{x}_{k},\overline{v}_{k}) + \O(h^{2})\\
&-h\gamma\lambda\left(U(\overline{x}_{k}) - U(x^{*}) + \frac{\gamma^{2}}{4}\|\overline{x}_{k}-x^{*}\|^{2} + \frac{1-\eta^{2}}{2h}\langle \overline{x}_{k}-x^{*},\overline{v}_{k}\rangle + \frac{1-\eta^{4}}{4h\gamma\lambda}\|\overline{v}_{k}\|^{2}\right) \\
&\leq \left(1-h\lambda\gamma\right)\mathcal{V}(\overline{x}_{k},\overline{v}_{k}) + h\gamma\lambda\left(\frac{1-\eta^{4}}{4h\gamma\lambda} - \frac{1}{2\lambda}\right)\|\overline{v}_{k}\|^{2} \\&+ h\gamma\lambda\left(\frac{1-\eta^{2}}{2h} - \frac{\gamma}{2}\right)\langle \overline{x}_{k}-x^{*},\overline{v}_{k}\rangle + \O(h^{2})\\
&\leq \left(1-h\lambda\gamma\right)\mathcal{V}(\overline{x}_{k},\overline{v}_{k})  + h\gamma\lambda\left(\frac{1-\eta^{2}}{2h} - \frac{\gamma}{2}\right)\langle \overline{x}_{k}-x^{*},\overline{v}_{k}\rangle + \O(h^{2}),
\end{align*}
where we have used 
\[
\left(-\frac{h}{2} + \frac{1-\eta^{2}}{\gamma} -\frac{h\eta^{4}}{2}\right) \overline{v}_{k} \cdot \nabla U(\overline{x}_{k}) \leq \frac{h^{2}\gamma}{2}\left|\overline{v}_{k}\cdot \nabla U(\overline{x}_{k})\right|,
\]
due to the fact that for all $0< x < 1$, $0 \leq -x+2(1-e^{-x}) - xe^{-2x}\leq x^{2}$ and $0< h\gamma <1$. We group this term into higher-order terms and
use the fact that $1-\eta^{2} \geq h\gamma - \frac{(h\gamma)^{2}}{2}$ to arrive at
\[
h\gamma\lambda\left(\frac{1-\eta^{2}}{2h} - \frac{\gamma}{2}\right)\langle \overline{x}_{k}-x^{*},\overline{v}_{k}\rangle \leq h\gamma\lambda\left(\frac{\gamma}{2} - \frac{1-\eta^{2}}{2h}\right)\left|\langle \overline{x}_{k}-x^{*},\overline{v}_{k}\rangle\right|\leq \lambda \frac{h^{2}\gamma^{3}}{4}\left|\langle \overline{x}_{k}-x^{*},\overline{v}_{k}\rangle\right|.
\]
We again group this into the higher-order terms.  Assuming $h < 1$, we find that the second-order terms are bounded by 
\begin{align*}
&Mh^{2}\left(\|\overline{v}_{k}\|^{2} + h^{2}M^{2}\|\overline{x}_{k}-x^{*}\|^{2}\right) + h^{2}\frac{\gamma^{2}\lambda}{2}M\|\overline{x}_{k}-x^{*}\|^{2} + \frac{h^{2}M^{2}}{2}\|\overline{x}_{k}-x^{*}\|^{2}\\
&+ \frac{h^{2}\gamma}{2}\left(\sqrt{M}\|\overline{v}_{k}\|^{2} + M^{3/2}\|\overline{x}_{k}-x^{*}\|^{2}\right) + \lambda \frac{h^{2}\gamma^{3}}{4}\left(\gamma\|\overline{x}_{k}-x^{*}\|^{2} + \frac{1}{\gamma}\|\overline{v}_{k}\|^{2}\right) .
\end{align*}
Assuming that $\lambda \leq \frac{1}{4}$ we have, for all $x,v \in \R^{d}$,
\[
8\mathcal{V}(x,v) \geq \|v\|^{2} \qquad 16\mathcal{V}(x,v) \geq  \gamma^{2}\|x-x^{*}\|^{2}
\]
and  using $h < \frac{1}{2\sqrt{M}}$, the $\O(h^{2})$ terms are bounded by 
\begin{align*}
&h^{2}\left(\gamma^{2} + \frac{\gamma^{2}\lambda}{4} + \frac{\gamma^{2}}{2}\right)\|\overline{v}_{k}\|^{2} + \gamma^{2}h^{2}\left(\frac{M}{4}+ \frac{M\lambda}{2} + \frac{M}{2} + \frac{\gamma^{2}\lambda}{4}  + \frac{M^{2}}{2\gamma^{2}}\right)\|\overline{x}_{k}-x^{*}\|^{2}\\
&\leq 8h^{2}\gamma^{2}\left(4+\lambda\right)\mathcal{V}(\overline{x}_{k},\overline{v}_{k}).
\end{align*}
Therefore 
\[r(\overline{x}_{k},\overline{v}_{k}) \leq \left(1-h\lambda\gamma + 8h^{2}\gamma^{2}\left(4 + \lambda \right)\right)\mathcal{V}(\overline{x}_{k},\overline{v}_{k}).\]

Now let us define $c_4(h) := h\lambda\gamma -8h^{2}\gamma^{2}\left(4 + \lambda\right)$, then we have that
\[
r^{2}(\overline{x}_{k},\overline{v}_{k}) \leq (1-c_4(h))^{2}\mathcal{V}^{2}(\overline{x}_{k},\overline{v}_{k})
\]
and
\begin{align*}
&2r(\overline{x}_{k},\overline{v}_{k})\mathbb{E}_{\xi^{(1)}_{k+1},\xi^{(2)}_{k+1}}\left[(\xi^{(1)}_{k+1},\xi^{(2)}_{k+1})^{T} \mathcal{T} (\xi^{(1)}_{k+1},\xi^{(2)}_{k+1}) \right]\leq \\
& 2(1-c_4(h))\mathcal{V}(\overline{x}_{k},\overline{v}_{k})\mathbb{E}_{\xi^{(1)}_{k+1},\xi^{(2)}_{k+1}}\left[(\xi^{(1)}_{k+1},\xi^{(2)}_{k+1})^{T} \mathcal{T} (\xi^{(1)}_{k+1},\xi^{(2)}_{k+1}) \right].
\end{align*}
From the fact that $\lambda \gamma/2 \leq M/\gamma$ (due to Lemma \ref{supp:lem:convexity_type}) and $\gamma^{2} \geq 8M$ we have the estimates
\begin{align*}
&\mathbb{E}_{\xi^{(1)}_{k+1},\xi^{(2)}_{k+1}}\left[(\xi^{(1)}_{k+1},\xi^{(2)}_{k+1})^{T} \mathcal{T} (\xi^{(1)}_{k+1},\xi^{(2)}_{k+1}) \right] = \mathbb{E}_{\xi^{(1)}_{k+1},\xi^{(2)}_{k+1}}\Bigg[\left( \frac{M}{\gamma} - \frac{\lambda\gamma}{2}\right)\left\|\left( \sqrt{h}-a_{1}\right)\xi^{(1)}_{k+1} - a_{2}\xi^{(2)}_{k+1}\right\|^{2} \\
&+\frac{h\gamma}{2}\left\|\xi^{(1)}_{k+1}\right\|^{2} + \frac{\gamma}{2}\left\| a_{1}\xi^{(1)}_{k+1} + a_{2}\xi^{(2)}_{k+1}\right\|^{2}\Bigg] \leq 3h\gamma d\\
&\mathbb{E}_{\xi^{(1)}_{k+1},\xi^{(2)}_{k+1}}\left[\left((\xi^{(1)}_{k+1},\xi^{(2)}_{k+1})^{T} \mathcal{T} (\xi^{(1)}_{k+1},\xi^{(2)}_{k+1})\right)^{2} \right]\\
&\leq\mathbb{E}_{\xi^{(1)}_{k+1},\xi^{(2)}_{k+1}}\left[\left(2h\gamma \|\xi^{(1)}_{k+1}\|^{2} +  2h\gamma \|\xi^{(2)}_{k+1}\|^{2}\right)^{2}\right] \leq 24h^{2}\gamma^{2}d^{2}.
\end{align*}

Therefore the remaining term we need to bound is $\mathbb{E}_{\xi^{(1)}_{k+1},\xi^{(2)}_{k+1}}\left(\mathbf{s}(\overline{x}_{k},\overline{v}_{k}) \cdot (\xi^{(1)}_{k+1},\xi^{(2)}_{k+1})\right)^{2} = \|\mathbf{s}(\overline{x}_{k},\overline{v}_{k})\|^{2}$, where
\begin{align*}
&\mathbf{s}(\overline{x}_{k},\overline{v}_{k})\cdot \left(\xi^{(1)}_{k+1},\xi^{(2)}_{k+1} \right) \\
&= ((\sqrt{h}-a_{1})\xi^{(1)}_{k+1} - a_{2}\xi^{(2)}_{k+1})\cdot \left(\frac{M\sqrt{2\gamma}}{2\gamma}\frac{1-\eta^{2}}{\gamma} (\overline{v}_{k}-h\nabla U(\overline{x}_{k})) + \sqrt{\frac{2}{\gamma}}\nabla U(\overline{x}_{k})\right)\\
&+ \frac{\sqrt{2\gamma h}}{4}\gamma\left(\overline{x}_{k}-x^{*} + \gamma^{-1}\overline{v}_{k}-\frac{h}{\gamma}\nabla U(\overline{x}_{k})\right)\cdot \xi^{(1)}_{k+1} + \frac{\eta^{2} \sqrt{2\gamma}}{4} (\overline{v}_{k}-h\nabla U(\overline{x}_{k})) \cdot \left(a_{1}\xi^{(1)}_{k+1} + a_{2}\xi^{(2)}_{k+1}\right)\\
&-\frac{\lambda \gamma \sqrt{2\gamma}}{4}\left(\overline{x}_{k}-x^{*} + \frac{1-\eta^{2}}{\gamma}\left(\overline{v}_{k} - h \nabla U(\overline{x}_{k})\right) \cdot \left((\sqrt{h}-a_{1})\xi^{(1)}_{k+1} - a_{2}\xi^{(2)}_{k+1} \right)\right),
\end{align*}
using that $\gamma^{2}\|x-x^{*}\|^{2} \leq 16\mathcal{V}(x,v)$ and $\|v\|^{2} \leq 8\mathcal{V}(x,v)$ for all $x,v \in \R^{d}$ we have 
\[
\|\mathbf{s}(\overline{x}_{k},\overline{v}_{k})\|^{2} = \|s_{1}(\overline{x}_{k},\overline{v}_{k})\|^{2} + \|s_{2}(\overline{x}_{k},\overline{v}_{k})\|^{2},
\]
where
\begin{align*}
&\|s_{1}(\overline{x}_{k},\overline{v}_{k})\|^{2} \leq h\Bigg(\left( 2\frac{M(1-\eta^{2})}{\gamma\sqrt{2\gamma}} + \frac{\eta^{2}\sqrt{2\gamma}}{4} + 2\frac{\lambda \sqrt{2\gamma}(1-\eta^{2})}{4}\right)\|\overline{v}_{k}\| \\
&+ \left(2\frac{hM^{2}(1-\eta^{2})}{\gamma\sqrt{2\gamma}} + 2\sqrt{\frac{2}{\gamma}}M + \frac{\sqrt{2\gamma}hM}{4} + \frac{\eta^{2} \sqrt{2\gamma} hM}{4} + 2\frac{\lambda \sqrt{2\gamma} (1-\eta^{2})hM}{4}\right) \|\overline{x}_{k}-x^{*}\| \\
&+ \frac{\sqrt{2}\gamma^{3/2}}{4}\|\overline{x}_{k}-x^{*} + \gamma^{-1}\overline{v}_{k}\|\Bigg)^{2}\\
&\leq h\left((2\sqrt{\gamma} + \lambda \sqrt{\gamma})\sqrt{\mathcal{V}(\overline{x}_{k},\overline{v}_{k})} + \left(\frac{2\gamma^{3/2}}{\sqrt{8}} + \frac{\lambda \gamma^{3/2}}{32}\right)\frac{4}{\gamma}\sqrt{\mathcal{V}(\overline{x}_{k},\overline{v}_{k})} + \frac{5}{2}\sqrt{\gamma}\sqrt{\mathcal{V}(\overline{x}_{k},\overline{v}_{k})}\right)^{2}\\
&\leq 110h\gamma\left(1 + \lambda^{2}\right)\mathcal{V}(\overline{x}_{k},\overline{v}_{k})
\end{align*}
for $\gamma^{2} \geq \sqrt{8M}$ and $h<\frac{1}{2\gamma}$
and
\begin{align*}
&\|s_{2}(\overline{x}_{k},\overline{v}_{k})\|^{2} \leq h \Bigg(\left( 2\frac{M(1-\eta^{2})}{\gamma\sqrt{2\gamma}} + \frac{\eta^{2}\sqrt{2\gamma}}{4} + 2\frac{\lambda \sqrt{2\gamma}(1-\eta^{2})}{4}\right)\|\overline{v}_{k}\| \\
&+ \left(2\frac{hM^{2}(1-\eta^{2})}{\gamma\sqrt{2\gamma}} + 2\sqrt{\frac{2}{\gamma}}M +  \frac{\eta^{2} \sqrt{2\gamma} hM}{4} + 2\frac{\lambda \sqrt{2\gamma} (1-\eta^{2})hM}{4}\right) \|\overline{x}_{k}-x^{*}\|\Bigg)^{2}\\
&\leq h\left((2\sqrt{\gamma} + \lambda \sqrt{\gamma})\sqrt{\mathcal{V}(\overline{x}_{k},\overline{v}_{k})} + \left(\frac{2\gamma^{3/2}}{\sqrt{8}} + \frac{\lambda \gamma^{3/2}}{32}\right)\frac{4}{\gamma}\sqrt{\mathcal{V}(\overline{x}_{k},\overline{v}_{k})}\right)^{2}\\
&\leq 50h\gamma\left(1 + \lambda^{2}\right)\mathcal{V}(\overline{x}_{k},\overline{v}_{k}).
\end{align*}
Therefore $\mathbb{E}_{\xi^{(1)}_{k+1},\xi^{(2)}_{k+1}}\left[\left( \mathbf{s}(\overline{x}_{k},\overline{v}_{k}) \cdot (\xi^{(1)}_{k+1},\xi^{(2)}_{k+1})\right)^{2}\right] \leq 160h\gamma(1+\lambda^{2})\mathcal{V}(\overline{x}_{k},\overline{v}_{k})$. Combining estimates, we have the drift inequality
\begin{align*}
&\mathbb{E}\left[\mathcal{V}(\overline{x}_{k+1},\overline{v}_{k+1})^{2} \mid \overline{x}_{k},\overline{v}_{k}\right] \leq (1-c_4(h))^{2}\mathcal{V}^{2}(\overline{x}_{k},\overline{v}_{k}) + 6h\gamma d\left(1-c_4(h)\right)\mathcal{V}(\overline{x}_{k},\overline{v}_{k}) \\
&+  160h\gamma(1+\lambda^{2})\mathcal{V}(\overline{x}_{k},\overline{v}_{k})  + 24h^{2}\gamma^{2}d^{2}.
\end{align*}
We will now use the quadratic property that states, for  $b_{1},b_{2} > 0$,
\[
b_{2}x^{2} + \frac{b^{2}_{1}}{4b_{2}} \geq b_{1}x,
\]
for all $x \in \R$ and therefore
\[
c_4(h)\mathcal{V}^{2}(\overline{x}_{k},\overline{v}_{k}) + \frac{(6h\gamma d + 160h\gamma(1+\lambda^{2}))^{2}}{4c_4(h)} \geq  6h\gamma d\left(1-c_4(h)\right)\mathcal{V}(\overline{x}_{k},\overline{v}_{k}) +  160h\gamma(1+\lambda^{2})\mathcal{V}(\overline{x}_{k},\overline{v}_{k}).
\]
Therefore for $c_4(h) < \frac{1}{2}$ (which is satisfied when $h < \frac{1}{2\gamma}$ and $\lambda < 1$, which is satisfied as $\lambda \leq M/2\gamma^{2} \leq 1$ for $\gamma^{2} \geq \frac{M}{2}$) we have
\begin{align}\label{supp:eq:detVsquarebnd}
&\mathbb{E}\left[\mathcal{V}(\overline{x}_{k+1},\overline{v}_{k+1})^{2} \mid \overline{x}_{k},\overline{v}_{k}\right] \leq \left(1-\frac{c_4(h)}{2}\right)\mathcal{V}^{2}(\overline{x}_{k},\overline{v}_{k}) +\frac{(6h\gamma d + 160h\gamma(1+\lambda^{2}))^{2}}{4c_4(h)}  + 24h^{2}\gamma^{2}d^{2},
\end{align}
then globally, we have
\begin{align*}
&\frac{m^{2}}{4}\mathbb{E}\left[\|\overline{x}_{k}-x^{*}\|^{4}\mid y_{0},v_{0}\right] \leq \mathbb{E}\left[\mathcal{V}(\overline{x}_{k},\overline{v}_{k})^{2} \mid y_{0},v_{0}\right]  \\ &\leq \left(1-\frac{c_4(h)}{2}\right)^{k}\mathcal{V}^{2}(\overline{x}_{0},\overline{v}_{0}) +2\frac{\frac{(6h\gamma d + 160h\gamma(1+\lambda^{2}))^{2}}{4c_4(h)}  + 24h^{2}\gamma^{2}d^{2}}{c_4(h)}.
\end{align*}

Now, we have proved this for the iterates of $\mathcal{BU}$, where we wish to use the relation $(\mathcal{UBU})^{k} = \mathcal{U}(\mathcal{BU})^{k-1}\mathcal{BU}$.  In this case, we have that $\overline{x}_{k}$, the $(k+1)$-th point of approximate gradient/full gradient evaluation, is precisely the position after $\mathcal{U}(\mathcal{BU})^{k}$. It follows that
\begin{align*}
&\frac{m^{2}}{4}\mathbb{E}\left[\|\overline{x}_{k}-x^{*}\|^{4}\mid \overline{x}_{0},\overline{v}_{0}\right] \leq \left(1-\frac{c_4(h)}{2}\right)^{k}\mathcal{V}^{2}(\overline{x}_{0},\overline{v}_{0}) +2\frac{\frac{(6h\gamma d + 160h\gamma(1+\lambda^{2}))^{2}}{4c_4(h)}  + 24h^{2}\gamma^{2}d^{2}}{c_4(h)},
\end{align*}
where $(\overline{x}_{0},\overline{v}_{0}) = \mathcal{U}(x_{0},v_{0},h/2,\xi^{(1)}_{0},\xi^{(2)}_{0})$. It is easy to show that $\mathcal{V}(x,v) \leq \gamma^{2}\|x-x^{*}\|^{2} + \|v\|^{2}$ for all $(x,v) \in \mathbb{R}^{2d}$ using \cite{nesterov2018lectures}[Lemma 1.2.3] and that $\gamma^{2}\geq M$. Therefore 
\begin{align*}
&\mathbb{E}\left[\mathcal{V}^{2}(\mathcal{U}(x_{0},v_{0},h/2,\xi^{(1)}_{0},\xi^{(2)}_{0}))\mid x_{0},v_{0}\right] \leq \mathbb{E}\left[2\gamma^{4}\|\overline{x}_{0}-x^{*}\|^{4} + 2\|\overline{v}_{0}\|^{4} \mid x_{0},v_{0}\right]\\
&\leq 2\gamma^{4}\mathbb{E}\left\|x_{0}-x^{*} + \frac{1-\eta^{2}}{\gamma}v_{0} + \sqrt{\frac{2}{\gamma}}\left(\mathcal{Z}^{(1)}(h/2,\xi^{(1)}_{0}) - \mathcal{Z}^{(2)}(h/2,\xi^{(1)}_{0},\xi^{(2)}_{0})\right)\right\|^{4} \\
&+ 2\mathbb{E}\left\|\eta v_{0} + \sqrt{2\gamma}\mathcal{Z}^{(2)}(h/2,\xi^{(1)}_{0},\xi^{(2)}_{0})\right\|^{4}\\
&\leq 4\gamma^{4}\|x_{0}-x^{*}\|^{4} + 4\|v_{0}\|^{4} + \mathbb{E}\left[32\gamma^{2}\|\mathcal{Z}^{(1)}(h/2,\xi^{(1)}_{0})\|^{4} + 40\gamma^{2}\|\mathcal{Z}^{(2)}(h/2,\xi^{(1)}_{0},\xi^{(2)}_{0})\|^{4}\right]\\
&\leq 4\gamma^{4}\|x_{0}-x^{*}\|^{4} + 4\|v_{0}\|^{4} + 8\gamma^{2}h^{2}d^{2} +  480\gamma^{2}h^{2}d^{2}, 
\end{align*}
where we have used that $U(\overline{x}_{0})-U(x^{*}) \leq \frac{M}{2}\|\overline{x}_{0}-x^{*}\|^{2}$ in the first inequality and naive bounds on the fourth moments of the Gaussian increments. Hence, we arrive at the estimate
\begin{align*}
  &\mathbb{E}\left[\|\overline{x}_{k}-x^{*}\|^{4} \mid x_{0},v_{0}\right]\leq \frac{4}{m^{2}}\Bigg[4\left(1-\frac{c_4(h)}{2}\right)^{k}\left(\gamma^{4}\|x_{0}-x^{*}\|^{4} + \|v_{0}\|^{4} + 122\gamma^{2}h^{2}d^{2}\right)\\
  &+2\frac{\frac{(6h\gamma d + 160h\gamma(1+\lambda^{2}))^{2}}{4c_4(h)}  + 24h^{2}\gamma^{2}d^{2}}{c_4(h)}\Bigg],  
\end{align*}
for the $\UBU$ scheme with full gradients.

\end{proof}

\section{Variance bounds for $\UBUBU$ estimator with exact gradients}\label{supp:sec:Appendix_UBUBU_variance_bnd_exact}
\subsection{Variance bound of $D_{l,l+1}$}\label{supp:sec:local_to_strong}

To bound the variance of $D_{l,l+1}$ we use strong error estimates for the $\UBU$ integrator using the results of \cite{sanz2021wasserstein}.

In this analysis we define for random vectors $z_{1},z_{2} \in \R^{2d}$ the $L^{2}$ norm $\|z_1\|_{L^{2},a,b} = \mathbb{E}\left(\left\|z_1\right\|^{2}_{a,b}\right)^{1/2}$ and respective inner product $\langle z_1, z_2 \rangle_{L^{2},a,b} = \mathbb{E}\left(z_{1}^{T} \mathcal{M} z_{2}\right)$, where
\[
\mathcal{M} = \begin{pmatrix}
    I_d & bI_d \\
    b I_d & a I_d
\end{pmatrix}.
\]

\begin{assumption}[Local Strong Error \cite{sanz2021wasserstein}]\label{supp:ass:Local_Strong_Error_2}
Let $\phi\left(z,t,\left(W_s\right)^{t}_{s=0}\right)$ be the solution of the continuous kinetic Langevin dynamics with initial condition $z \in \R^{2d}$ up to time $t$, with Brownian motion $\left(W_s\right)^{t}_{s=0}$. Let $\psi_h\left(z,t,\left(W_s\right)^{t}_{s=0}\right)$ be the solution of a numerical discretization with initial condition $z \in \R^{2d}$ up to time $t$, with Brownian motion $\left(W_s\right)^{t}_{s=0}$ and stepsize $h$. Let $z' \sim \pi$, then we assume that
\[
\psi_h\left(z',h,\left(W_s\right)^{h}_{s=0}\right) - \phi\left(z',h,\left(W_s\right)^{h}_{s=0}\right) = \alpha_{h}\left(z',\left(W_s\right)^{h}_{s=0}\right) + \beta_{h}\left(z',\left(W_s\right)^{h}_{s=0}\right),
\]
where
\[
\left\|\alpha_{h}\left(z',\left(W_s\right)^{h}_{s=0}\right)\right\|_{L^{2},a,b} \leq C_1 h^{q+1/2},
\]
\[
\left\|\beta_{h}\left(z',\left(W_s\right)^{h}_{s=0}\right)\right\|_{L^{2},a,b} \leq C_2 h^{q+1},
\]
and
\begin{align*}
    &\left|\left\langle \psi_h\left(z',h,\left(W_s\right)^{h}_{s=0}\right) - \psi_h\left(z,h,\left(W_s\right)^{h}_{s=0}\right),\alpha_{h}\left(z',\left(W_s\right)^{h}_{s=0}\right)\right\rangle_{L^{2},a,b}\right| \\ &\leq C_{0} h\left\|z'-z\right\|_{L^{2},a,b}\left\|\alpha_{h}\left(z',\left(W_s\right)^{h}_{s=0}\right)\right\|_{L^2 ,a,b}.
\end{align*}
for some $C_{0},C_1,C_2 > 0$.
\end{assumption}

We restate Assumptions \ref{supp:assum:Lip}-\ref{supp:assum:initialdist} here for easier readability.

\begin{assumption}[$M$-$\nabla$ Lipschitz]
\label{supp:assum:Lip}
$U:\R^{d} \to \R$ is twice continuously differentiable and there exists $M>0$ such that for all $x,y \in \R^d$
$$
\| \nabla U(x) - \nabla U(y)\| \leq M\|x-y\|.
$$
\end{assumption}

\begin{assumption}[$m$-strong convexity]
\label{supp:assum:convex}
$U:\R^{d} \to \R$ is continuously differentiable and there exists $m>0$ such that for all $x,y \in \R^d$
$$
\langle \nabla U(x) - \nabla U(y),x-y \rangle \geq m|x-y|^2.
$$
\end{assumption}

\begin{assumption}[$M^{s}_{1}$-strongly Hessian Lipschitz]
\label{supp:assum:Hess_Lipschitz}
$U:\mathbb{R}^{d} \to \mathbb{R}$ is three times continuously differentiable and $M^{s}_{1}$-strongly Hessian Lipschitz if there exists $M^{s}_{1} > 0$ such that
    \[
    \|\nabla^{3}U(x)\|_{\{1,2\}\{3\}} \leq M^{s}_{1}
    \]
    for all $x \in \mathbb{R}^{d}$.
\end{assumption}

\begin{assumption}[$1$-Lipschitzness of $f$]
\label{supp:assum:Lipschitz}
$f$ is a 1-Lipschitz function with respect to the Euclidean distance on $\R^{2d}$, that only depends on $x$, not $v$ (i.e. $f(x,v)=f(x,v')$ for any $x,v,v'\in \R^d$). 
\end{assumption}

\begin{assumption}[Distance of initial distribution from target]
\label{supp:assum:initialdist}
The initial distribution on $\Lambda=\R^{2d}$ satisfy that $\mathcal{W}_2(\pi,\mu_0)\le c_{\mu_0}\sqrt{\frac{d}{m}}$,  for some $c_{\mu_0}>0$.
\end{assumption}

We make use of the following proposition, essentially due to \cite{sanz2021wasserstein}.
\begin{restatable}{proposition}{PropositionGlobalStrongError}
\label{supp:prop:GlobalStrongError}
    Suppose a numerical scheme approximating kinetic Langevin dynamics satisfies Assumption \ref{supp:ass:Local_Strong_Error_2}, with a potential which satisfies Assumptions \ref{supp:assum:Lip}-\ref{supp:assum:Hess_Lipschitz}, and $\psi_h(z,h,(W_s)_{0}^h)\sim P_h(z,\cdot)$ satisfies the Wasserstein contractivity condition \eqref{supp:eq:contractivity_cond} for $p=2$,    
    and some $a, b>0$, $b^2<a$.
    
    Let $\phi\left(z,t,\left(W_s\right)^{t}_{s=0}\right)$ be the solution of the continuous kinetic Langevin dynamics with initial condition $z \in \R^{2d}$ up to time $t$, with Brownian motion $\left(W_s\right)^{t}_{s=0}$. Let $\psi_h\left(z,t,\left(W_s\right)^{t}_{s=0}\right)$ be the solution of a numerical discretization with initial condition $z \in \R^{2d}$ up to time $t$, with Brownian motion $\left(W_s\right)^{t}_{s=0}$ and stepsize $h>0$ satisfying that 
    \begin{equation}\label{supp:eq:stepsizebnd}
        (1-c(h))^2 + C_{0}^2h^{2}<1.
    \end{equation}
    Then for any $k\ge 0$, any $z_0$ such that $\|z_0\|_{L^2,a,b}<\infty$, and $Z^0\sim \pi$, we have
    \begin{align*}
    &\left\|\psi_h\left(z_0,kh,\left(W_s\right)^{kh}_{s=0}\right)-\phi\left(Z^0,kh,\left(W_s\right)^{kh}_{s=0}\right)\right\|_{L^2,a,b} \\
    &\leq (1-R(h))^{k}\|z_0-Z^0\|_{L^2,a,b} + \sqrt{2}C_{1}\frac{h^{q+1/2}}{\sqrt{R(h)}} +\frac{2C_{2}h^{q+1}}{R(h)},
    \end{align*}
    where $R(h) = 1-\sqrt{(1-c(h))^2 + C_{0}^2 h^{2}}$.

    In particular, the discretization scheme admits a stationary distribution $\pi_h$, and its bias can be bounded as
    \begin{equation}\label{supp:eq:discretization:bias:bnd}
    \mathcal{W}_{2,a,b}(\pi_h,\pi)\leq \sqrt{2}C_{1}\frac{h^{q+1/2}}{\sqrt{R(h)}} +\frac{2C_{2}h^{q+1}}{R(h)}.
    \end{equation}
\end{restatable}

\begin{proof}
Introduce the notation
  \[
  Z^n := \phi\left(Z^0,nh, \left(W_s\right)^{nh}_{s=0}\right), \quad z_n:= \psi_h\left(z_0,nh,\left(W_s\right)^{nh}_{s=0}\right)
  \]
  for all $n \in \mathbb{N}$. Using the assumption $Z^0\sim \pi$, we also have $Z^n\sim \pi$, since the kinetic Langevin dynamics keeps $\pi$ invariant. By Assumption \ref{supp:ass:Local_Strong_Error_2}, we then have

 \begin{equation} \label{supp:eq:iteration}
      \begin{split}
      &\left\|z_k-Z^k\right\|_{L^2,a,b}=
      \left\|\psi_h\left(z_{k-1},h,\left(W_s\right)^{kh}_{(k-1)h}\right)-\phi\left(Z^{k-1},h,\left(W_s\right)^{kh}_{(k-1)h}\right)\right\|_{L^2,a,b}\\
    &=\bigg\|\psi_h\left(z_{k-1},h,\left(W_s\right)^{kh}_{(k-1)h}\right)-
    \psi_h\left(Z^{k-1},h,\left(W_s\right)^{kh}_{(k-1)h}\right)\\
    &+\psi_h\left(Z^{k-1},h,\left(W_s\right)^{kh}_{(k-1)h}\right)-
    \phi\left(Z^{k-1},h,\left(W_s\right)^{kh}_{(k-1)h}\right)\bigg\|_{L^2,a,b} \\     
    &=\bigg\|\psi_h\left(z_{k-1},h,\left(W_s\right)^{kh}_{(k-1)h}\right)-
    \psi_h\left(Z^{k-1},h,\left(W_s\right)^{kh}_{(k-1)h}\right)\\
    &+\alpha_{h} \left(Z^{k-1}, \left(W_s\right)^{kh}_{(k-1)h}\right)+ \beta_{h}\left(Z^{k-1}, \left(W_s\right)^{kh}_{(k-1)h}\right)    \bigg\|_{L^2,a,b} \\ 
     &\le \left\|\beta^{k-1}\right\|_{L^2,a,b}
     +\left\|\psi_h\left(z_{k-1},h,\left(W_s\right)^{kh}_{(k-1)h}\right)-
    \psi_h\left(Z^{k-1},h,\left(W_s\right)^{kh}_{(k-1)h}\right)+\alpha^{k-1}\right\|_{L^2,a,b},
  \end{split}
  \end{equation}
where $\alpha^{k-1}$ and $\beta^{k-1}$ are defined as
  \begin{align*}
      &\psi_{h} \left(Z^{k-1},h,\left(W_s\right)^{kh}_{(k-1)h}\right) - \phi\left(Z^{k-1},h,\left(W_s\right)^{kh}_{(k-1)h}\right) \\
      &= \alpha_{h} \left(Z^{k-1}, \left(W_s\right)^{kh}_{(k-1)h}\right)+ \beta_{h}\left(Z^{k-1}, \left(W_s\right)^{kh}_{(k-1)h}\right)\\
      &:= \alpha^{k-1} + \beta^{k-1}.      
  \end{align*}
Assumption \ref{supp:ass:Local_Strong_Error_2}, and the Wasserstein contractivity condition \eqref{supp:eq:contractivity_cond} then together imply
\begin{align*}
&\left\|\psi_h\left(z_{k-1},h,\left(W_s\right)^{kh}_{(k-1)h}\right)-
    \psi_h\left(Z^{k-1},h,\left(W_s\right)^{kh}_{(k-1)h}\right)+\alpha^{k-1}\right\|_{L^2,a,b}\\
    &=\Bigg(\left\|\alpha^{k-1}\right\|_{L^2,a,b}^2 +
    \left\|\psi_h\left(z_{k-1},h,\left(W_s\right)^{kh}_{(k-1)h}\right)-
    \psi_h\left(Z^{k-1},h,\left(W_s\right)^{kh}_{(k-1)h}\right)\right\|_{L^2,a,b}^2\\
    &+
    2 \left<\alpha^{k-1},\psi_h\left(z_{k-1},h,\left(W_s\right)^{kh}_{(k-1)h}\right)-
    \psi_h\left(Z^{k-1},h,\left(W_s\right)^{kh}_{(k-1)h}\right)\right>_{L^2,a,b}
    \Bigg)^{1/2}\\
    &\le \Bigg(\left\|\alpha^{k-1}\right\|_{L^2,a,b}^2 +
    (1-c(h))^2\left\|z_{k-1}-Z^{k-1}\right\|_{L^2,a,b}^2\\
    &+
    2 C_0 h\left\|\alpha^{k-1}\right\|_{L^2,a,b} \left\|z_{k-1}-Z^{k-1}\right\|_{L^2,a,b}
    \Bigg)^{1/2}\\
    &\le 
    \left(2\left\|\alpha^{k-1}\right\|_{L^2,a,b}^2 +
    ((1-c(h))^2+C_0^2 h^2)\left\|z_{k-1}-Z^{k-1}\right\|_{L^2,a,b}^2\right)^{1/2}
    \\
    &\le \left(2C_1^2h^{2q+1} +
    ((1-c(h))^2+C_0^2 h^2)\left\|z_{k-1}-Z^{k-1}\right\|_{L^2,a,b}^2
    \right)^{1/2}.
\end{align*}
Lemma 28 of \cite{sanz2021wasserstein} states that if a sequence of nonnegative real numbers $(a_n)_{n\ge 0}$ satisfies that 
$a_{n+1}\le \sqrt{(1-A)^2a_n^2+ B}+C$ with $A\in (0,1)$, $B\ge 0$, $C\ge 0$, then for every $n\ge 0$,
\[a_n\le (1-A)^n a_0 + \sqrt{\frac{B}{A}}+\frac{C}{A}.\]
Using this for $a_n=\left\|z_{n}-Z^{n}\right\|_{L^2,a,b}$, we have that
\begin{align*}
&\left\|z_k-Z^k\right\|_{L^2,a,b}
 \leq (1-R(h))^{k}\|z_0-Z^0\|_{L^2,a,b}+ \sqrt{2}\frac{C_{1}h^{q+1/2}}{\sqrt{R(h)}} +\frac{2C_{2}h^{q+1}}{R(h)},
\end{align*}
where $R(h)= 1-\sqrt{(1-c(h))^2 + C_{0}^2h^{2}}$, which is our first claim.

The existence of a stationary distribution $\pi_h$ follows from Lemma \ref{supp:lem:wass_conv}.
The bound on the bias follows by letting $k\to \infty$.
\end{proof}
We now are in a position to present our first result related to the variance of our unbiased scheme, which is a bound on the variance related to the global strong error or convergence. This is given below.
\begin{restatable}{proposition}{PropGlobalStrongError}
\label{supp:prop:global_strong_error}
Suppose a numerical scheme approximating kinetic Langevin dynamics satisfies the same assumptions as in Proposition \ref{supp:prop:GlobalStrongError}, and $f$ satisfies Assumption \ref{supp:assum:Lipschitz}.
If we have two chains at coarser and finer discretization levels $l$ and $l+1$ using stepsizes $h_l$ and $h_{l+1}=\frac{h_l}{2}$ satisfying \eqref{supp:eq:stepsizebnd} with synchronously coupled Brownian motions $\left(z_{k}\right)_{k \in \mathbb{N}}$ and $\left(z'_{k}\right)_{k \in \mathbb{N}}$,
such that $z_{0} \sim \pi_{0}$ and $z'_{0} \sim \pi'_{0}$, 
then we have
\begin{align*}
   &\Var\left(f(z'_{k}) - f(z_{k})\right) \leq\mathbb{E}\left[\left(f(z'_{k}) - f(z_{k})\right) ^2\right]\le \mathbb{E}\|z'_{k} - z_{k}\|^{2}_{a,b}\\
   &\leq \Bigg(\exp\left(-\frac{mkh_l}{8\gamma}\right) \left(\|z'_{0}-z_{0}\|_{L^{2},a,b}+\mathcal{W}_{2,a,b}(\pi_0,\pi)+\mathcal{W}_{2,a,b}(\pi_0',\pi)\right)
        \\
        &+(1-R(h_l))^k \mathcal{W}_{2,a,b}(\pi_0,\pi) +(1-R(h_{l+1}))^{2k} \mathcal{W}_{2,a,b}(\pi_0',\pi) \\
        &+ \sqrt{2}C_{1}\left(\frac{h^{q+1/2}_{l+1}}{\sqrt{R(h_{l+1})}} +\frac{h^{q+1/2}_{l}}{\sqrt{R(h_l)}}\right)+
        2C_{2}\left(\frac{h^{q+1}_{l+1}}{R(h_{l+1})}
        +\frac{h^{q+1}_{l}}{R(h_l)}\right)\Bigg)^2,
\end{align*}
where $R(h_i) = 1-\sqrt{(1-c(h_i))^2 + C_{0}^2 h^{2}_{i}}$ for $i = l,l+1$.
\end{restatable}
\begin{proof}[Proof of Proposition \ref{supp:prop:global_strong_error}] 
Consider the following variance bound:
  \begin{align*}
  &\Var\left(f(z'_{k}) - f(z_{k})\right) \leq\mathbb{E}\left[\left(f(z'_{k}) - f(z_{k})\right) ^2\right]\leq \mathbb{E}\|z'_{k} - z_{k}\|^{2}_{a,b}.
  \end{align*}  
  Let $\tilde{Z}_0\sim \pi$ be such that $\|\tilde{Z}_0-z_0\|_{L^2,a,b}=\mathcal{W}_{2,a,b}(\pi_0,\pi)$, and $\tilde{Z}'_0\sim \pi$ be such that $\|\tilde{Z}'_0-z_0'\|_{L^2,a,b}=\mathcal{W}_{2,a,b}(\pi_0',\pi)$ (the existence of optimal couplings was shown in Theorem 4.1 of \cite{villani2009optimal}). We use the estimate
  \begin{align*}
      &\sqrt{\mathbb{E}\|z'_{k} - z_{k}\|^2_{a,b}}=\|z'_{k} - z_{k}\|_{L^{2},a,b} \\
      &\leq \left\|z_{k}-\phi\left(\tilde{Z}_0 ,kh_{l},\left(W_s\right)^{kh_{l}}_{s=0}\right) \right\|_{L^{2},a,b}\\
      &+\left\|\phi \left(\tilde{Z}_0 ,kh_{l},\left(W_s\right)^{kh_{l}}_{s=0}\right) - \phi\left(\tilde{Z}_0'
     ,kh_{l},\left(W_s\right)^{kh_{l}}_{s=0}\right)\right\|_{L^{2},a,b}\\
      &+\left\|z'_{k} - \phi \left(\tilde{Z}_0',kh_{l},\left(W_s\right)^{kh_{l}}_{s=0}\right)\right\|_{L^{2},a,b}
      \\&=: \textnormal{(I)} + \textnormal{(II)} + \textnormal{(III)}.
  \end{align*}
  We split this into two global error terms (I) and (III) and a contraction term (II). We estimate the second term by Corollary \ref{supp:cor:continuous_contraction} as
  \begin{align*}
        \textnormal{(II)} &\leq \exp\left(-\frac{mkh_l}{8\gamma}\right) \|\tilde{Z}'_{0}-\tilde{Z}_{0}\|_{L^{2},a,b}\\
        &\le \exp\left(-\frac{mkh_l}{8\gamma}\right) \left(\|z'_{0}-z_{0}\|_{L^{2},a,b}+\mathcal{W}_{2,a,b}(\pi_0,\pi)+\mathcal{W}_{2,a,b}(\pi_0',\pi)\right).
  \end{align*}
By Proposition \ref{supp:prop:GlobalStrongError}, we have
\begin{align*}
    \textnormal{(I)} \leq (1-R(h_l))^k \mathcal{W}_{2,a,b}(\pi_0,\pi) +\sqrt{2}C_{1}\frac{2h^{q+1/2}_{l}}{\sqrt{R(h_l)}} +\frac{2C_{2}h^{q+1}_{l}}{R(h_l)}.
\end{align*}

The same argument can be applied to (III) to obtain
\begin{align*}
    \textnormal{(III)} \leq (1-R(h_{l+1}))^{2k} \mathcal{W}_{2,a,b}(\pi_0',\pi) +\sqrt{2}C_{1}\frac{2h^{q+1/2}_{l+1}}{\sqrt{R(h_{l+1})}} +\frac{2C_{2}h^{q+1}_{l+1}}{R(h_{l+1})}.
\end{align*}
Combining these we get the required result.
\end{proof}

Below are a number of useful remarks to highlight from the above theorem.

\begin{remark}
The local error, which arises from \cite{sanz2021wasserstein} is demonstrated through the bound on $\alpha_h + \beta_h$ from Assumption \ref{supp:ass:Local_Strong_Error_2}. This indicates there is an order of local strong order $q+1/2$. However, when we go to the global strong order, the order is only reduced by $1/2$ as it is order $q$. As stated in \cite{sanz2021wasserstein}, this is similar to the Euler–Maruyama scheme with local strong order $3/2$, but global strong order $1$ \cite{milstein2004stochastic}[Theorem 1.1].
\end{remark}

\begin{remark}
    Proposition \ref{supp:prop:global_strong_error} holds for $q = 2$ for the $\UBU$ scheme; \cite{sanz2021wasserstein} showed that the assumptions are true. For the $\UBU$ scheme we have for $\gamma^{2}\geq M$ and $h<\frac{1}{2\gamma}$ that $C_{2} \leq \sqrt{d}\left(\frac{7}{10}\gamma^{2} + \frac{M^{s}_1}{10\sqrt{M}}\right)$, $C_{1} = \frac{\sqrt{6dM\gamma}}{24}$ and $C_{0} \leq 4\sqrt{2M}$. These constants can be computed by following \cite{sanz2021wasserstein}[Section 7.6] where all computations are done with arbitrary $\gamma$, the constant $c$ we consider to be set to $1$ in their estimates. Constants $C_{1}$ and $C_{2}$ are estimated in the second and third step, whilst $C_{0}$ is estimated in the fourth step and fifth step. We remark that there is a missing term and a stronger assumption is needed in \cite{sanz2021wasserstein}[Section 7.6, fifth step] which has been corrected in \cite{paulin2024}. The additional term can be treated by the same argument as in the fourth step to arrive at the $C_{0}$ bound.
\end{remark}

\begin{restatable}{corollary}{CorollaryUBUGlobalStrongError}
\label{supp:corollary:global_strong_error}
Suppose that Assumptions \ref{supp:assum:Lip}, \ref{supp:assum:convex}, \ref{supp:assum:Hess_Lipschitz}, and \ref{supp:assum:initialdist} hold,  $\gamma \geq \sqrt{8M}$ and
  \begin{equation}\label{supp:eq:stepsizebndUBUBU}
        h_0\le \frac{1}{\gamma}\cdot \frac{m}{264M}.      
    \end{equation} 

Assume that the burn-in periods $B\ge \frac{16\log(4)\gamma}{mh_0}$, 
$B_0\ge \frac{16 \gamma}{m h_0}\log\left(\frac{c_{\mu_0}+1}{ \sqrt{M}\gamma h_0^2}\right)$. 
Then for every $l\ge 0$, $1\le k\le K$, the $\UBUBU$ samples satisfy
\begin{align*}
&\Var\left(f(z'^{(l,l+1)}_{k}) - f(z^{(l,l+1)}_{k})\right) \leq\mathbb{E}\left[\left(f(z'^{(l,l+1)}_{k}) - f(z^{(l,l+1)}_{k})\right) ^2 \right]\le \mathbb{E}\|z'^{(l,l+1)}_{k} - z^{(l,l+1)}_{k}\|^{2}_{a,b}\\
&\le C d \left(\left(\gamma^{{4}} + \frac{(M^{s}_1)^2}{M}\right)   \left(\frac{\gamma}{m}\right)^2 +  \frac{M\gamma^2}{m}\right)h_l^4.
\end{align*}
\end{restatable}

\begin{proof}[Proof of Corollary \ref{supp:corollary:global_strong_error}]
We have $(B_0+Bl)2^l$ burn-in steps at level $l$, and $(B_0+B(l+1))2^{l+1}$ burn-in steps at level $l+1$.
Let $\delta_*=\delta_{x^*}\times \delta_{0_d}$ be a distribution on $\Lambda$ that fixes $x=x^*$ and $v=0_d$. 
Using the assumptions, we have
\begin{align*}
&R(h_i)=1-\sqrt{(1-c(h_{i}))^2 + C_{0}^2h^{2}_{i}}
=1-\sqrt{\left(1-\frac{mh_{i}}{8\gamma}\right)^2 + C_{0}^2h^{2}_{i}}
\\
&=1-\sqrt{1-\frac{mh_{i}}{4\gamma} + \left(\left(\frac{m}{8\gamma}\right)^2+C_{0}^2\right)h^{2}_{i}}\ge 1-\sqrt{1-\frac{mh_{i}}{8\gamma} }\ge \frac{mh_{i}}{16\gamma},\\
&\mathcal{W}_{2,a,b}(\pi_0,\pi)=\mathcal{W}_{2,a,b}(\mu_0,\pi)\le c_{\mu_0}\sqrt{\frac{d}{m}},\\
&\mathcal{W}_{2,a,b}(\pi_0',\pi)\le \mathcal{W}_{2,a,b}\left(\mu_0 R_{l+1}^{B},\pi_{h_{l+1}}\right)+\mathcal{W}_{2,a,b}\left(\pi_{h_{l+1}},\pi\right),\\
&\le \mathcal{W}_{2,a,b}(\mu_0,\pi)+2\mathcal{W}_{2,a,b}\left(\pi_{h_{l+1}},\pi\right)
\le c_{\mu_0}\sqrt{\frac{d}{m}}+2\sqrt{2}C_{1}\frac{h_{l+1}^{q+1/2}}{\sqrt{R(h_{l+1})}} +\frac{4C_{2}h_{l+1}^{q+1}}{R(h_{l+1})},\\
\intertext{and}
&\|z'_{0}-z_{0}\|_{L^{2},a,b}\le \mathcal{W}_{2,a,b}(\pi_0',\delta_{*})
+\mathcal{W}_{2,a,b}(\pi_0,\delta_{*})\\
&\le \mathcal{W}_{2,a,b}(\pi_0',\pi)+\mathcal{W}_{2,a,b}(\mu_0,\pi)+2\mathcal{W}_{2,a,b}(\pi,\delta_{*})\\
&\le (3c_{\mu_0}+3)\sqrt{\frac{d}{m}}+2\sqrt{2}C_{1}\frac{h_{l+1}^{q+1/2}}{\sqrt{R(h_{l+1})}} +\frac{4C_{2}h_{l+1}^{q+1}}{R(h_{l+1})}.
\end{align*}
It is easy to check that \eqref{supp:eq:stepsizebndUBUBU} together with $C_0\le 4\sqrt{2M}$ implies that the condition \eqref{supp:eq:stepsizebnd} of Proposition \ref{supp:prop:global_strong_error} is satisfied, and we have
\begin{align*}
   &\Var\left(f(z'^{(l,l+1)}_{k}) - f(z^{(l,l+1)}_{k})\right) \leq\mathbb{E}\left[\left(f(z'^{(l,l+1)}_{k}) - f(z^{(l,l+1)}_{k})\right) ^2\right]\\
   &\leq \Bigg(\exp\left(-\frac{m(B_0+l B)h_0}{8\gamma}\right) \left(\|z'_{0}-z_{0}\|_{L^{2},a,b}+\mathcal{W}_{2,a,b}(\pi_0,\pi)+\mathcal{W}_{2,a,b}(\pi_0',\pi)\right)
        \\
        &+\exp\left(-\frac{m(B_0+l B)h_0}{16\gamma}\right)  (\mathcal{W}_{2,a,b}(\pi_0,\pi)+ \mathcal{W}_{2,a,b}(\pi_0',\pi)) \\
        &+ \sqrt{2}C_{1}\left(\frac{h^{q+1/2}_{l+1}}{\sqrt{R(h_{l+1})}} +\frac{h^{q+1/2}_{l}}{\sqrt{R(h_l)}}\right)+
        2C_{2}\left(\frac{h^{q+1}_{l+1}}{R(h_{l+1})}
        +\frac{h^{q+1}_{l}}{R(h_l)}\right)\Bigg)^2\\
   & \le \Bigg(\exp\left(-\frac{m(B_0+lB)h_0}{16\gamma}\right) ({7}c_{\mu_0}+3)\sqrt{\frac{d}{m}}+ 10\sqrt{2}C_{1}\left(\frac{h_l^{5/2}}{\sqrt{\frac{mh_{l}}{16\gamma}}}\right)+
        20 C_{2}\left(\frac{h^{3}_{l}}{\frac{mh_{l}}{16\gamma}}\right)\Bigg)^2\\
        \intertext{using the assumptions on $B_0$ and $B$}
    &\le C \left(C_1^2\frac{\gamma}{m}+C_2^2  \left(\frac{\gamma}{m}\right)^2\right) h_l^{4}\le C d \left(\left(\gamma^{{4}} + \frac{(M^{s}_1)^2}{M}\right)   \left(\frac{\gamma}{m}\right)^2 +  \frac{M\gamma^2}{m}\right)h_l^4. 
\end{align*}\end{proof}

\begin{proposition}
\label{supp:prop:boundDllp1}
Suppose that the assumptions of Proposition \ref{supp:prop:Wasserstein} hold for $h=h_l$.
Let $R_{l,l+1}=(P_{h_l,h_{l+1}})^{2^l}$ be the Markov kernel for two synchronously coupled $\UBU$ chains at discretization levels $l, l+1$. This chain is moving on state space $\Lambda^2$. Let $\ol{z}_1,\ldots, \ol{z}_K$ be a Markov chain with kernel $R_{l,l+1}$. Let $F:\Lambda^2\to \R$ be 1-Lipschitz in norm $\|\|_{a,b}$ on $\Lambda^2$, defined as $\|z_1,z_2 \|_{a,b}^2=\|z_1\|_{a,b}^2+\|z_2\|_{a,b}^2$.
Then we have 
\begin{align*}&\Var\left(\frac{\sum_{i=1}^{K}F(\ol{z}_i)}{K}\right)\le \frac{2}{K^2}\sum_{i=1}^{K}\sum_{k=0}^{K-i} \min\Bigg(\frac{\Var(F(\ol{z}_i))+\Var(F(\ol{z}_{i+k}))}{2},\\
&\sqrt{\Var(F(\ol{z}_i))\E\left[\|\ol{z}_i-\E \ol{z}_i\|_{a,b}^2\right]}\cdot \exp\left(-\frac{mh_0}{8\gamma}\cdot k\right) \Bigg),\end{align*}
\end{proposition}
\begin{proof}
We need to bound 
\[\mathrm{Cov}(F(\ol{z}_i),F(\ol{z}_{i+k}))=\E[(F(\ol{z}_i)-\E(F(\ol{z}_i)))(F(\ol{z}_{i+k})-\E(F(\ol{z}_{i+k})))].\]
Let $\tilde{z}_i$ be an independent identically distributed copy of $\ol{z}_i$. For $0\le l\le K-i-1$, and assume that conditioned on $\tilde{z}_{i:i+l}$ and $\ol{z}_{1:i+l}$, $\tilde{z}_{i+l+1}\sim P(\tilde{z}_{i+l},\cdot)$, and $(\ol{z}_{i+l+1},\tilde{z}_{i+l+1})$ are synchronously coupled, i.e. 
$\tilde{z}_{i+l+1}$ is defined based on the coupling between discretization levels using the same Gaussian variables that were used to move from $\ol{z}_{i+l}$ to $\ol{z}_{i+l+1}$. Since we have also used synchronous couplings in the proof of Proposition \ref{supp:prop:Wasserstein}, it follows from Proposition \ref{supp:prop:Wasserstein} that 
\begin{align*}
    &\E\left(\left.\left\|\tilde{z}_{i+l+1}-\ol{z}_{i+l+1}\right\|_{a,b}^2\right|\ol{z}_{i:i+l},\tilde{z}_{i:i+l}\right) \\
    &\le  \max\left(\left(1 - \frac{mh_l}{8\gamma}\right)^{2\cdot 2^l}, \left(1 - \frac{mh_{l+1}}{8\gamma}\right)^{2\cdot 2^{l+1}}\right)\left\|\tilde{z}_{i+l}-\ol{z}_{i+l}\right\|_{a,b}^2\\
    \intertext{ using that $1-x\le \exp(-x)$ for $x\ge 0$,}
    &\le \exp\left(-\frac{mh_0}{4\gamma}\right) \left\|\tilde{z}_{i+l}-\ol{z}_{i+l}\right\|^2_{a,b}.
\end{align*}
By using this bound recursively, we have
\begin{align*}&\E\left(\left.\left\|\tilde{z}_{i+k}-\ol{z}_{i+k}\right\|_{a,b}^2\right|\ol{z}_{i},\tilde{z}_{i}\right)\le \exp\left(-\frac{mh_0}{4\gamma}\cdot k\right)  \|\ol{z}_i-\tilde{z}_i\|_{a,b}^2.
\end{align*}
Since $\tilde{z}_i$ is independent of $\ol{z}_i$, and $\tilde{z}_i+1,\ldots,\tilde{z}_{i+k}$ was constructed using $\tilde{z}_i$ and Gaussians that are independent of $\ol{z}_i$ (synchronous coupling with $\ol{z}_{i+1},\ldots, \ol{z}_{i+k}$), it follows that $\tilde{z}_{i+k}$ is still independent of $\ol{z}_i$. Using this and the 1-Lipschitz property of $F$, we have
\begin{align*}&\mathrm{Cov}(F(\ol{z}_i),F(\ol{z}_{i+k}))=\E[(F(\ol{z}_i)-\E(F(\ol{z}_i)))(F(\ol{z}_{i+k})-\E(F(\ol{z}_{i+k})))
\\
&= \E[(F(\ol{z}_i)-\E(F(\ol{z}_i)))(F(\ol{z}_{i+k})-F(\tilde{z}_{i+k}))]\\
&= \E[(F(\ol{z}_i)-\E(F(\ol{z}_i)))\E(F(\ol{z}_{i+k})-F(\tilde{z}_{i+k})|\ol{z}_i,\tilde{z}_{i})]\\
&\le \sqrt{\Var(F(\ol{z}_i)) \E\left[\|\ol{z}_{i+k}-\tilde{z}_{i+k}\|_{a,b}^2\right]}\\
&\le \exp\left(-\frac{mh_0}{8\gamma}\cdot k\right)\sqrt{\Var(F(\ol{z}_i))\cdot 
\E\left[\|\ol{z}_i-\tilde{z}_i\|_{a,b}^2\right]}\\
&=\exp\left(-\frac{mh_0}{8\gamma}\cdot k\right)\sqrt{2\Var(F(\ol{z}_i))\cdot 
\E\left[\|\ol{z}_i-\E \ol{z}_i\|_{a,b}^2\right]},
\end{align*}
and the claim follows by summation.
\end{proof}

\begin{proposition}\label{supp:prop:varDllp1exact}
Under the same assumptions as in Corollary \ref{supp:corollary:global_strong_error}, the UBUBU samples satisfy that
\[\Var(D_{l,l+1})\le \frac{1}{K} C(\gamma, m,M,M^{s}_1) d h_l^4 \left(C(\gamma, m,M,M^{s}_1)-2\log(h_0)+\log(4)l+\frac{4\gamma}{mh_0}\right).\]
\end{proposition}
\begin{proof}
Note that the function $F(z_1,z_1)=f(z_1)-f(z_2)$ is 1-Lipschitz with respect to $\|(z_1,z_2)\|_{a,b}=\|z_1\|_{a,b}+\|z_2\|_{a,b}$. Let $\ol{z}_i=(z_{i}^{(l,l+1)},z'^{(l,l+1)}_i)$, then by Proposition \ref{supp:prop:boundDllp1}, we have that 
\begin{align*}&\Var(D_{l,l+1})\le \frac{1}{K^2}
\sum_{i=1}^{K}\sum_{k=0}^{K-i} \min\Bigg(\frac{\Var(F(\ol{z}_{i}))+\Var(F(\ol{z}_{i+k}))}{2},\\
&\sqrt{\Var(F(\ol{z}_i))\E\left[\|\ol{z}_i-\E \ol{z}_i\|_{a,b}^2\right]}\cdot \exp\left(-\frac{mh_0}{8\gamma}\cdot k\right) \Bigg).
\end{align*}
By a similar argument as in the proof of Corollary \ref{supp:corollary:global_strong_error}, using our assumptions on $B$ and $B_0$, we can show that
\[\left(\E\left[\|\ol{z}_i-\E \ol{z}_i\|_{a,b}^2\right]\right)^{1/2}\le 
\left(\E\left[\|\ol{z}_i-(x^*,0_d,x^*,0_d)\|_{a,b}^2\right]\right)^{1/2}\le C\sqrt{\frac{d}{m}},
\]
and by Proposition \ref{supp:prop:global_strong_error}, we have
\[\Var(F(\ol{z}_i))\le C(\gamma, m,M,M^{s}_1) d h_l^4.\]
Let \[k^*(l):=\max\left(\log\left(C\sqrt{\frac{1}{m}}\right)-\frac{1}{2}\log\left( C(\gamma, m,M,M^{s}_{1}) h_l^4 \right),0\right),\]
then for $k\ge \lceil k^*(l)\rceil$, we have 
\begin{align*}
&\sqrt{\Var(F(\ol{z}_i))\E\left[\|\ol{z}_i-\E \ol{z}_i\|_{a,b}^2\right]}\cdot \exp\left(-\frac{mh_0}{8\gamma}\cdot k\right)\\
&\le C(\gamma, m,M,M^{s}_1) d h_l^4  \exp\left(-\frac{mh_0}{8\gamma}\cdot (k-\lceil k^*(l)\rceil)\right).
\end{align*}
It is clear that $\lceil k^*(l)\rceil\le C(\gamma, m,M,M^{s}_1)-2\log(h_0)+\log(4)l$, and after
some rearrangement, we have 
\[\Var(D_{l,l+1})\le \frac{1}{K} C(\gamma, m,M,M^{s}_1) d h_l^4 \left(C(\gamma, m,M,M^{s}_1)-2\log(h_0)+\log(4)l+\frac{4\gamma}{mh_0}\right).\]
\end{proof}

\subsection{Variance bound of $D_{0}$}

\begin{proposition}\label{supp:prop:D0var}
Consider an $m$-strongly convex $M$-$\nabla$Lipschitz potential $U$ and let $P_{h}$ be the transition kernel of $\UBU$ with stepsize $h$. Suppose that $f:\Omega\to \R$ only depends on $x$ and is a 1-Lipschitz function. Suppose that $\gamma \geq \sqrt{8M}$, and $h < \frac{1}{2\gamma}$. 
Let $\mu_0$ be a distribution on $\Lambda$, and the Markov chain
$z_{-B_0}^{(0)}\sim \mu_0$, $z_{-B_0+1}^{(0)}\sim P_{h}(z_{-B_0}^{(0)},\cdot), \ldots, z_K\sim P_{h}(z_{K-1}^{(0)},\cdot)$. Then $D_0$ satisfies that
\[\Var(D_0)\le \frac{C}{c(h) K}\left(1+\frac{1}{c(h) K}\right)   \left(\frac{1}{\gamma}+\frac{\gamma}{M}\right)\frac{h}{c(h)}+\frac{(1-c(h))^{2(B_0+1)}}{2c(h)^2 K^2}\sigma_{\mu_0}^2,\]
where $$c(h)=\frac{mh}{8\gamma}, \quad \sigma_{\mu_0}^2=\int\int \|w-\tilde{w}\|^2_{a,b} d\mu_0(w)d\mu_0(\tilde{w}),$$
for some absolute constant $C$.
\end{proposition}
\begin{proof}
The bound is based on Theorem 2 of \cite{Joulin2010}. 
We need to control the following quantities for every $z\in \Lambda$:
\begin{equation}\label{supp:eq:sigmazdef}\sigma(z)^2:=\frac{1}{2}\int \int \|w-\tilde{w}\|^2_{a,b} P_{h}(z,dw) P_{h}(z,d \tilde{w}),
\end{equation}
\begin{equation}\label{supp:eq:nzdef}n(z):=\inf_{g:\Lambda\to \R,\|g\|_{a,b, Lip}\le 1} \frac{\int \int \|w-\tilde{w}\|^2_{a,b} P_{h}(z,dw) P_{h}(z,d \tilde{w})}{\int \int (g(w)-g(\tilde{w}))^2 P_{h}(z,dw) P_{h}(z,d \tilde{w})}.
\end{equation}
Here we choose $a=\frac{1}{M}$, and $b=\frac{1}{\gamma}$ as in Proposition \ref{supp:prop:Wasserstein}. To control $\sigma^2(z)$, let us define two independent identically distributed random variables $w(z)\sim P_{h}(z,\cdot )$ and $\tilde{w}(z)\sim P_{h}(z,\cdot)$. Using the definition of $\UBU$, we have
\begin{align*}
    \sigma(z)^2&=\frac{1}{2}\E(\|w(z)-\tilde{w}(z)\|_{a,b}^2)\\
    &=\frac{1}{2}\E\Bigg(\left\|\mathcal{UBU}\left(z,h,\xi^{(1)},\xi^{(2)},\xi^{(3)},\xi^{(4)}\right)-\mathcal{UBU}\left(z,h,\tilde{\xi}^{(1)},\tilde{\xi}^{(2)},\tilde{\xi}^{(3)},\tilde{\xi}^{(4)}\right)\right\|_{a,b}^2\Bigg)\\
    &\le     \E\Bigg(\bigg\|\mathcal{U}\left(\mathcal{B}\left(\mathcal{U}\left(z,h/2,\xi^{(1)},\xi^{(2)}\right),h\right),h/2,\xi^{(3)},\xi^{(4)}\right)\\
    &-\mathcal{U}\left(\mathcal{B}\left(\mathcal{U}\left(z,h/2,{\xi}^{(1)},{\xi}^{(2)}\right),h\right),h/2,{\tilde{\xi}}^{(3)},{\tilde{\xi}}^{(4)}\right)\bigg\|_{a,b}^2\Bigg)\\    &+\E\Bigg(\bigg\|\mathcal{U}\left(\mathcal{B}\left(\mathcal{U}\left(z,h/2,{\xi}^{(1)},{\xi}^{(2)}\right),h\right),h/2,{\tilde{\xi}}^{(3)},{\tilde{\xi}}^{(4)}\right)\\
    &-\mathcal{U}\left(\mathcal{B}\left(\mathcal{U}\left(z,h/2,{\tilde{\xi}}^{(1)},{\tilde{\xi}}^{(2)}\right),h\right),h/2,{\tilde{\xi}}^{(3)},{\tilde{\xi}}^{(4)}\right)\bigg\|_{a,b}^2\Bigg).
\end{align*}
Recalling the definitions of $\mathcal{U}$ and $\mathcal{B}$, we have
\[
\begin{split}
\mathcal{B}(x,v,h) &= (x,v - h\nabla U(x)),\\
\mathcal{U}(x,v,h/2,\xi^{(1)},\xi^{(2)}) &= \Big(x + \frac{1-\eta}{\gamma}v + \sqrt{\frac{2}{\gamma}}\left(\mathcal{Z}^{(1)}\left(h/2,\xi^{(1)}\right) - \mathcal{Z}^{(2)}\left(h/2,\xi^{(1)},\xi^{(2)}\right) \right),\\
& \eta v + \sqrt{2\gamma}\mathcal{Z}^{(2)}\left(h/2,\xi^{(1)},\xi^{(2)}\right)\Big),\\
\mathcal{Z}^{(1)}\left(h/2,\xi^{(1)}\right) &= \sqrt{\frac{h}{2}}\xi^{(1)},\\
\mathcal{Z}^{(2)}\left(h/2,\xi^{(1)},\xi^{(2)}\right) &= \sqrt{\frac{1-\eta^{2}}{2\gamma}}\Bigg(\sqrt{\frac{1-\eta}{1+\eta}\cdot \frac{4}{\gamma h}}\xi^{(1)} + \sqrt{1-\frac{1-\eta}{1+\eta}\cdot\frac{4}{\gamma h}}\xi^{(2)}\Bigg).
\end{split}
\]
First, 
    \begin{align*}    &\E\Bigg(\bigg\|\mathcal{U}\left(\mathcal{B}\left(\mathcal{U}\left(z,h/2,\xi^{(1)},\xi^{(2)}\right),h\right),h/2,\xi^{(3)},\xi^{(4)}\right)\\
    &-\mathcal{U}\left(\mathcal{B}\left(\mathcal{U}\left(z,h/2,{\xi}^{(1)},{\xi}^{(2)}\right),h\right),h/2,{\tilde{\xi}}^{(3)},{\tilde{\xi}}^{(4)}\right)\bigg\|_{a,b}^2\Bigg)\\
    &=\E\Bigg(\Bigg\|\Bigg(\sqrt{\frac{2}{\gamma}}\left(\mathcal{Z}^{(1)}\left(\frac{h}{2},\xi^{(1)}\right) - \mathcal{Z}^{(2)}\left(\frac{h}{2},\xi^{(1)},\xi^{(2)}\right) \right)-\sqrt{\frac{2}{\gamma}}\left(\mathcal{Z}^{(1)}\left({\tilde{\xi}}^{(1)}\right) - \mathcal{Z}^{(2)}\left(\frac{h}{2},\xi^{(1)},{\tilde{\xi}}^{(2)}\right) \right),\\ &\sqrt{2\gamma}\mathcal{Z}^{(2)}\left(\frac{h}{2},\xi^{(1)},\xi^{(2)}\right)-\sqrt{2\gamma}\mathcal{Z}^{(2)}\left(\frac{h}{2},{\tilde{\xi}}^{(1)},{\tilde{\xi}}^{(2)}\right)
    \Bigg)\Bigg\|_{a,b}^2\Bigg)\\
    \intertext{ using \eqref{supp:eq:normequiv}, and the fact that $a=\frac{1}{M}$}
    &\le \frac{6}{\gamma} \E\left(\left\|\mathcal{Z}^{(1)}\left(\frac{h}{2},\xi^{(1)}\right)- \mathcal{Z}^{(2)}\left(\frac{h}{2},\xi^{(1)},{\tilde{\xi}}^{(2)}\right)\right\|^2\right)+\frac{6\gamma}{M}\E\left(\left\|\mathcal{Z}^{(2)}\left(\frac{h}{2},\xi^{(1)},{\tilde{\xi}}^{(2)}\right)\right\|^2\right) \\
    &\le \left(\frac{3}{\gamma}+\frac{3\gamma}{M}\right)d h.
\end{align*}
Second, using the assumptions $\gamma\ge \sqrt{8M}$ and $h\le \frac{1}{\sqrt{M}}$, for any $x,v,x',v'$, 
\begin{align}
\nonumber&\bigg\|\mathcal{U}\left(x,v,h/2,{\tilde{\xi}}^{(3)},{\tilde{\xi}}^{(4)}\right)-\mathcal{U}\left(x',v',h/2,{\tilde{\xi}}^{(3)},{\tilde{\xi}}^{(4)}\right)\bigg\|_{a,b}^2\\
\nonumber&\le \frac{3}{2}\bigg\|\mathcal{U}\left(x,v,h/2,{\tilde{\xi}}^{(3)},{\tilde{\xi}}^{(4)}\right)-\mathcal{U}\left(x',v',h/2,{\tilde{\xi}}^{(3)},{\tilde{\xi}}^{(4)}\right)\bigg\|_{a,0}^2\\
\nonumber&\le \frac{3}{2}\left(\frac{1}{M}\exp(-\gamma h)\|v-v'\|^2+2\|x-x'\|^2+\frac{2 (1-\exp(-\gamma h/2))^2}{\gamma^2}\|v-v'\|^2\right)\\
\label{supp:eq:Udiff} &\le 3\|x-x'\|^2+ \frac{3}{2}\frac{1}{M}\|v-v'\|^2\le 6\|(x-x',v-v')\|_{a,b}^2, \\
\nonumber&\|\mathcal{B}(x,v,h)-\mathcal{B}(x',v',h)\|_{a,b}^2\le
\frac{3}{2}\|(x-x',v-v'+h\nabla U(x')-h \nabla U(x))\|_{0,a}^2\\
\nonumber&\le \frac{3}{2}\|x-x'\|^2 + \frac{3}{M}\|v-v'\|^2+\frac{3 h^2}{M} \| \nabla U(x')-\nabla U(x)\|^2\\
\label{supp:eq:Bdiff}&\le \left(\frac{3}{2}+3h^2 M\right)\|x-x'\|^2+\frac{3}{M}\|v-v'\|^2\le 6\|(x-x',v-v')\|_{a,b}^2
\end{align}
hence
\begin{align*}
&\E\Bigg(\bigg\|\mathcal{U}\left(\mathcal{B}\left(\mathcal{U}\left(z,h/2,{\xi}^{(1)},{\xi}^{(2)}\right),h\right),h/2,{\tilde{\xi}}^{(3)},{\tilde{\xi}}^{(4)}\right)\\
    &-\mathcal{U}\left(\mathcal{B}\left(\mathcal{U}\left(z,h/2,{\tilde{\xi}}^{(1)},{\tilde{\xi}}^{(2)}\right),h\right),h/2,{\tilde{\xi}}^{(3)},{\tilde{\xi}}^{(4)}\right)\bigg\|_{a,b}^2\Bigg)\\
&\le 36 \E\left(\left\|\mathcal{U}\left(z,h/2,{\xi}^{(1)},{\xi}^{(2)}\right)-\mathcal{U}\left(z,h/2,{\tilde{\xi}}^{(1)},{\tilde{\xi}}^{(2)}\right)\right\|_{a,b}^2\right)\\
&\le 36 
\E\Bigg(\Bigg\|\Bigg(\sqrt{\frac{2}{\gamma}}\left(\mathcal{Z}^{(1)}\left(\frac{h}{2},\xi^{(1)}\right) - \mathcal{Z}^{(2)}\left(\frac{h}{2},\xi^{(1)},\xi^{(2)}\right) \right)\\
&-\sqrt{\frac{2}{\gamma}}\left(\mathcal{Z}^{(1)}\left({\tilde{\xi}}^{(1)}\right) - \mathcal{Z}^{(2)}\left(\frac{h}{2},\xi^{(1)},{\tilde{\xi}}^{(2)}\right) \right),\\ &\sqrt{2\gamma}\mathcal{Z}^{(2)}\left(\frac{h}{2},\xi^{(1)},\xi^{(2)}\right)-\sqrt{2\gamma}\mathcal{Z}^{(2)}\left(\frac{h}{2},{\tilde{\xi}}^{(1)},{\tilde{\xi}}^{(2)}\right)
    \Bigg)\Bigg\|_{a,b}^2\Bigg)\le 36\left(\frac{3}{\gamma}+\frac{3\gamma}{M}\right)d h,
\end{align*}
using the same argument as for the previous term. Hence by summing up the above bounds, we have
\begin{equation}\label{supp:eq:sigmaz2bnd}
\sigma(z)^2\le 37\left(\frac{3}{\gamma}+\frac{3\gamma}{M}\right)d h.
\end{equation}
Now, we will lower bound $n(z)$ as defined in \eqref{supp:eq:nzdef}. By  \eqref{supp:eq:normequiv}, we have
\begin{align}
\nonumber&\E(\|w(z)-\tilde{w}(z)\|_{a,b}^2)\ge \frac{1}{2}\E(\|w(z)-\tilde{w}(z)\|_{a,0}^2)\\
&=\E\Bigg(\Bigg\|\mathcal{U}\left(\mathcal{B}\left(\mathcal{U}\left(z,h/2,\xi^{(1)},\xi^{(2)}\right),h\right),h/2,\xi^{(3)},\xi^{(4)}\right)
\\
\nonumber &-\E \mathcal{U}\left(\mathcal{B}\left(\mathcal{U}\left(z,h/2,\xi^{(1)},\xi^{(2)}\right),h\right),h/2,\xi^{(3)},\xi^{(4)}\right)\Bigg\|_{a,0}^2\Bigg)\\
\nonumber&\ge 
\E\Bigg(\Bigg\|\Bigg(\sqrt{\frac{2}{\gamma}}\left(\mathcal{Z}^{(1)}\left(\frac{h}{2},\xi^{(3)}\right) - \mathcal{Z}^{(2)}\left(\frac{h}{2},\xi^{(3)},\xi^{(4)}\right) \right),\sqrt{2\gamma}\mathcal{Z}^{(2)}\left(\frac{h}{2},\xi^{(3)},\xi^{(4)}\right)
    \Bigg)\Bigg\|_{a,0}^2\Bigg)\\
    \label{supp:eq:numeratorlbnd}&\ge \frac{\gamma}{M} d h.
\end{align}
For the denominator, we have
\begin{align*}
&\int \int (g(w)-g(\tilde{w}))^2 P_{h}(z,dw) P_{h}(z,d \tilde{w})=2\cdot \Var_{w\sim P_{h}(z,\cdot)}(g(w))\\
&=2\cdot \Var\left(g\left(\mathcal{U}\left(\mathcal{B}\left(\mathcal{U}\left(z,h/2,\xi^{(1)},\xi^{(2)}\right),h\right),h/2,\xi^{(3)},\xi^{(4)}\right)\right) \right)\\
\intertext{by the Efron-Stein inequality \cite{steele1986efron, Boucheron2013}}
&\le 2 \E(\Var_{\xi^{(1)},\xi^{(2)}}\left(g\left(\mathcal{U}\left(\mathcal{B}\left(\mathcal{U}\left(z,h/2,\xi^{(1)},\xi^{(2)}\right),h\right),h/2,\xi^{(3)},\xi^{(4)}\right)\right) \right)\\
&+2\E(\Var_{\xi^{(3)},\xi^{(4)}}\left(g\left(\mathcal{U}\left(\mathcal{B}\left(\mathcal{U}\left(z,h/2,\xi^{(1)},\xi^{(2)}\right),h\right),h/2,\xi^{(3)},\xi^{(4)}\right)\right) \right),
\end{align*}
where $\Var_{\xi^{(1)},\xi^{(2)}}(\cdot)$ means that we compute the conditional variance with respect to $\xi^{(3)},\xi^{(4)}$ (so the $\xi^{(3)},\xi^{(4)}$ are kept constant, and only the variance with respect to $\xi^{(1)},\xi^{(2)}$ is considered).
Let
\begin{align*}J_{\mathcal{U}}(h)&:=\frac{\partial \mathcal{U}(x,v,h,\xi^{(1)},\xi^{(2)})}{\partial (\xi^{(1)},\xi^{(2)})}\\
&=\left(\begin{matrix}\left(\sqrt{\frac{2h}{\gamma}}-\frac{\sqrt{2}(1-e^{-\gamma h})}{\gamma^{3/2}\sqrt{h}}\right)I_d, & -\frac{(1-e^{-\gamma h})\sqrt{2}}{\sqrt{\gamma h}} I_d\\
-\frac{1}{\gamma} \sqrt{1-e^{-2\gamma h}-\frac{2(1-e^{-\gamma h})^2}{\gamma h}} I_d, &\sqrt{1-e^{-2\gamma h}-\frac{2(1-e^{-\gamma h})^2}{\gamma h} }I_d
\end{matrix}\right),\\
\tilde{g}_h(z)&:=g\left(\mathcal{U}\left(\mathcal{B}\left(z,h\right),
h/2,\xi^{(3)},\xi^{(4)}\right)\right).
\end{align*}
Using the assumption that $g$ in 1-Lipschitz in \eqref{supp:eq:nzdef}, and the bounds (\ref{supp:eq:Udiff}-\ref{supp:eq:Bdiff}), it follows that $\tilde{g}_h$ 
 is a 6-Lipschitz function in $\|\cdot \|_{a,b}$, and \eqref{supp:eq:normequiv} implies that it is a 12-Lipschitz function in $\|\cdot\|_{a,0}$.
Since the continuously differentiable Lipschitz functions are dense amongst Lipschitz functions (see \cite{azagra2007smooth}), we can assume without loss of generality that $g$ and thus $\tilde{g}_h$ are continuously differentiable. Note that
\begin{align*}&\tilde{g}_h(z)-\tilde{g}_h(z')= \left<\nabla \tilde{g}_h(z), z-z'\right>+o(\|z-z'\|_{a,0})\\
&=\left<\left(\begin{matrix}I_d & 0\\ 0 &a^{-1/2}I_d\end{matrix} \right) \nabla \tilde{g}_h(z)  , \left(\begin{matrix}I_d & 0\\ 0 &a^{1/2}I_d\end{matrix}\right)(z-z')\right>+o(\|z-z'\|_{a,0}).
\end{align*}
Using this, it is easy to show that the $12$-Lipschitz property of $\tilde{g}_h$ in $\|\cdot\|_{a,0}$ implies that
$\left\|\nabla \tilde{g}_h(z)\right\|_{1/a,0}\le 12$ for every $z\in \Lambda$.
Hence, we obtain
\begin{align*}&\left\|\frac{\partial}{\partial (\xi^{(1)},\xi^{(2)})}\left(g\left(\mathcal{U}\left(\mathcal{B}\left(\mathcal{U}\left(z,h/2,\xi^{(1)},\xi^{(2)}\right),h\right),h/2,\xi^{(3)},\xi^{(4)}\right)\right) \right)\right\|\\
&=\left\|\frac{\partial}{\partial (\xi^{(1)},\xi^{(2)})}
 \tilde{g}_h\left( \mathcal{U}\left(z,h/2,\xi^{(1)},\xi^{(2)}\right)\right)
\right\|=\left\|J_U(h/2) \nabla\tilde{g}_h\left( \mathcal{U}\left(z,h/2,\xi^{(1)},\xi^{(2)}\right)\right)\right\|
\\
&\le 12 \sup_{w\in \Lambda: \|w\|_{1/a,0}\le 1} \left\|J_U(h/2) w\right\|=12\sup_{w\in \Lambda: \|w\|\le 1} \left\| J_U(h/2) \left(\begin{matrix}I_d & 0_d \\ 0_d & \frac{1}{\sqrt{M}} I_d\end{matrix}\right)w\right\|\\
&=12 \left\|J_U(h/2) \left(\begin{matrix}I_d & 0_d \\ 0_d & \frac{1}{\sqrt{M}} I_d\end{matrix}\right)\right\|\\
&=12\left\|\left(\begin{matrix}\left(\sqrt{\frac{h}{\gamma}}-\frac{2(1-e^{-\gamma h/2})}{\gamma^{3/2}\sqrt{h}}\right), & -\frac{2(1-e^{-\gamma h/2})}{\sqrt{M\gamma h}} \\
-\frac{1}{\gamma} \sqrt{1-e^{-\gamma h}-\frac{4(1-e^{-\gamma h/2})^2}{\gamma h}}, &\frac{1}{\sqrt{M}}\sqrt{1-e^{-\gamma h}-\frac{4(1-e^{-\gamma h/2})^2}{\gamma h}}
\end{matrix}\right)\right\|,\\
\intertext{using the fact that $-(1-e^{-x})^2\le -x^2+x^3$ for $x\ge 0$, and that $\gamma h\le 1$}
&\le 12 \left(\sqrt{\frac{h}{\sqrt{M}}}+\frac{2\gamma h}{\sqrt{M}}\right).
\end{align*}
From the Gaussian Poincar\'e inequality (see e.g. Theorem 3.20 of \cite{Boucheron2013}), and the fact that $\xi^{(1)}, \xi^{(2)}$ are standard normal, it follows that 
\begin{align*}&2\E(\Var_{\xi^{(1)},\xi^{(2)}}\left(g\left(\mathcal{U}\left(\mathcal{B}\left(\mathcal{U}\left(z,h/2,\xi^{(1)},\xi^{(2)}\right),h\right),h/2,\xi^{(3)},\xi^{(4)}\right)\right) \right)\\
&\le 2 \cdot 12^2 \left(\sqrt{\frac{h}{\sqrt{M}}}+\frac{2\gamma h}{\sqrt{M}}\right)^2\le 576\left(\frac{h}{\sqrt{M}}+4\frac{\gamma^2 h^2}{M}\right).
\end{align*}
We can bound the second term similarly, since
\begin{align*}
&\left\|\frac{\partial}{\partial (\xi^{(3)},\xi^{(4)})}\left(g\left(\mathcal{U}\left(\mathcal{B}\left(\mathcal{U}\left(z,h/2,\xi^{(1)},\xi^{(2)}\right),h\right),h/2,\xi^{(3)},\xi^{(4)}\right)\right) \right)\right\|\\
&=\left\|J_U(h/2) \nabla g\left(\mathcal{U}\left(\mathcal{B}\left(\mathcal{U}\left(z,h/2,\xi^{(1)},\xi^{(2)}\right),h\right),h/2,\xi^{(3)},\xi^{(4)}\right)\right)\right\|\\
\intertext{using the fact that $g$ is $2$-Lipschitz with respect to $\|\cdot \|_{a,0}$,}
&\le 2\left\|J_U(h/2) \left(\begin{matrix}I_d & 0_d \\ 0_d & \frac{1}{\sqrt{M}} I_d\end{matrix}\right)\right\|\le 2\sqrt{\frac{h}{\sqrt{M}}}+\frac{2\gamma h}{\sqrt{M}},
\end{align*}
and thus by the Gaussian Poincar\'e inequality,
\begin{align*}
&2\E(\Var_{\xi^{(3)},\xi^{(4)}}\left(g\left(\mathcal{U}\left(\mathcal{B}\left(\mathcal{U}\left(z,h/2,\xi^{(1)},\xi^{(2)}\right),h\right),h/2,\xi^{(3)},\xi^{(4)}\right)\right) \right)\\
&\le 8 \left(\sqrt{\frac{h}{\sqrt{M}}}+\frac{2\gamma h}{\sqrt{M}}\right)^2\le 16\left(\frac{h}{\sqrt{M}}+4\frac{\gamma^2 h^2}{M}\right).
\end{align*}
By adding these up, we obtain
\[
\int \int (g(w)-g(\tilde{w}))^2 P_{h}(z,dw) P_{h}(z,d \tilde{w})\le 592\left(\frac{h}{\sqrt{M}}+4\frac{\gamma^2 h^2}{M}\right),
\]
and hence by \eqref{supp:eq:nzdef} and \eqref{supp:eq:numeratorlbnd}, we have
\begin{equation}\label{supp:eq:nzbound}n(z)\ge \frac{\frac{\gamma}{M} d h}{
592\left(\frac{h}{\sqrt{M}}+4\frac{\gamma^2 h^2}{M}\right) 
}\ge\frac{\frac{\gamma}{M} d h}{
592\cdot 5\left(\frac{\gamma h}{M}\right) 
}\ge \frac{d}{3000}.
\end{equation}
Combining this with \eqref{supp:eq:sigmaz2bnd}, we have that
\[\sup_{z\in \Lambda} \frac{\sigma(z)^2}{n(z)}\le \left(37\left(\frac{3}{\gamma}+\frac{3\gamma}{M}\right)d h\right) \cdot \frac{3000}{d}\le 333000 \left(\frac{1}{\gamma}+\frac{\gamma}{M}\right)h,
\]
and the claim now follows by Theorem 2 of \cite{Joulin2010} and the bound on $\Var[\E(\hat{\pi}(f))|X_0)]$ on page 2427 of \cite{Joulin2010}, using the fact that $\kappa\ge 1-\sqrt{1-\frac{mh}{4\gamma}}\ge \frac{mh}{8\gamma}$ by Proposition \ref{supp:prop:Wasserstein}.
\end{proof}

\subsection{Variance of $S(c_R)$}
\label{supp:subsection:varianceScR}

\begin{theorem}
\label{supp:thm:exactgradUBUBU}
Suppose that Assumptions \ref{supp:assum:Lip}, \ref{supp:assum:convex}, \ref{supp:assum:Hess_Lipschitz}, \ref{supp:assum:Lipschitz}, \ref{supp:assum:initialdist} hold, and in addition, 
\begin{align*}\gamma &\geq \sqrt{8M}, \quad h_0 \leq 
\frac{1}{\gamma}\cdot \frac{m}{264M}, \quad B\ge \frac{16\log(4)\gamma}{mh_0}, \quad
B_0\ge \frac{16 \gamma}{m h_0}\log\left(\frac{c_{\mu_0}+1}{ \sqrt{M}\gamma h_0^2}\right).\end{align*}
Suppose that $c_{R}\in [0,\phi_N^{-1/2})$, and $2<\phi_N<16$.
Then for any $N\ge 1$, the $\UBUBU$ estimator $S(c_R)$ has finite expected computational cost, $\E S(c_R)=\pi(f)$, and it has finite variance. 
Moreover, it satisfies a CLT as $N\to \infty$, and the asymptotic variance $\sigma^2_S$ defined in \eqref{supp:eq:sigma2S} can be bounded as
\[\sigma^2_S\le \frac{C(m,M,M^{s}_1,\gamma,c_N,\phi_N)
}{K h_0} \left(1+\frac{1}{h_0 K}+d h_0^4\right).\]
\end{theorem}

\begin{proof}[Proof of Theorem \ref{supp:thm:exactgradUBUBU}]
By Corollary \ref{supp:corollary:global_strong_error}, and the fact that \[\E(D_{l,l+1}^2)\le 
\max_{1\le k\le K}\E\left[\left(f(z'^{(l,l+1)}_{k}) - f(z^{(l,l+1)}_{k})\right) ^2\right],\] it follows that under the assumptions of Corollary \ref{supp:corollary:global_strong_error}, we have 
\[\E(D_{l,l+1}^2)\le C d \left(\left(\gamma^{{4}} + \frac{(M^{s}_1)^2}{M}\right)   \left(\frac{\gamma}{m}\right)^2 +  \frac{M\gamma^2}{m}\right)h_l^4\le V_D \phi_D^{-l},\] for $V_D= C h_0^4 d \left(\left(\gamma^{{4}} + \frac{(M^{s}_1)^2}{M}\right)   \left(\frac{\gamma}{m}\right)^2 +  \frac{M\gamma^2}{m}\right)$ and $\phi_D=16$. 
From Proposition \ref{supp:prop:D0var}, and using the fact that $c(h_0)=\frac{h_0m}{8\gamma}$, and our assumptions on $B_0$, we have 
\begin{align}\nonumber&\Var(D_0)\le \frac{C}{c(h_0) K}\left(1+\frac{1}{c(h_0) K}\right)   \left(\frac{1}{\gamma}+\frac{\gamma}{M}\right)\frac{h_0}{c(h_0)}+\frac{(1-c(h_0))^{2(B_0+1)}}{2c(h_0)^2 K^2}\sigma_{\mu_0}^2\\
&\label{supp:eq:VarD0exactbnd}\le \frac{C}{ K}\cdot \frac{1}{h_0}\left(\frac{8\gamma}{m}\right)^2  \left(\frac{1}{\gamma}+\frac{\gamma}{M}\right)\left(1+\frac{8\gamma}{h_0 m}\cdot \frac{1}{K}\right).
\end{align}
The computational cost at each level satisfies the assumptions of Proposition \ref{supp:prop:unb}, so if we fix $2<\phi_N<16$, all assumptions of this proposition are satisfied. Hence $S(c_R)$ is an unbiased estimator with finite variance and computational cost.

The claim about the asymptotic variance follows by using the bounds in \eqref{supp:eq:VarD0exactbnd} and in Proposition \ref{supp:prop:varDllp1exact}, and adding up all terms according to \eqref{supp:eq:sigma2S}.
\end{proof}

\begin{proposition}
\label{supp:prop:boundDllp1independent}
Suppose that the assumptions of Proposition \ref{supp:prop:Wasserstein} hold for $h=h_l$. 
Let $R_{l,l+1}=(P_{h_l,h_{l+1}})^{2^l}$ be the Markov kernel for two synchronously coupled $\UBU$ chains at discretization levels $l, l+1$. This chain is moving on state space $\Lambda^2$. Let $\ol{z}_1,\ldots, \ol{z}_K$ be a Markov chain with kernel $R_{l,l+1}$. Let $F:\Lambda^2\to \R$ be of the form $F(z,z')=f(z)-f(z')$, where $f$ is of the form 
\eqref{supp:eq:fprodwdef}. Suppose that the target $\pi$ is a product distribution, satisfying the same conditions as in Proposition \ref{supp:prop:ProductExactGradUBUBU}. Then we have 
\begin{align*}&\Var\left(\frac{\sum_{i=1}^{K}F(\ol{z}_i)}{K}\right)\le \frac{2}{K^2}\sum_{i=1}^{K}\sum_{k=0}^{K-i} \min\Bigg(\frac{\Var(F(\ol{z}_i))+\Var(F(\ol{z}_{i+k}))}{2},\\
&\sqrt{4r \left(\sum_{s=1}^{r}\|w^{(s)}\|^2\right)}
\sqrt{\Var(F(\ol{z}_i)) \max_{1\le j\le d}\E\left[\|\ol{z}_{i,j}-\E\ol{z}_{i,j}\|_{a,b}^2\right]}\cdot \exp\left(-\frac{mh_0}{8\gamma}\cdot k\right) \Bigg).\end{align*}
\end{proposition}

\begin{proof}
We proceed similarly to the proof of Proposition \ref{supp:prop:boundDllp1}.
\[\mathrm{Cov}(F(\ol{z}_i),F(\ol{z}_{i+k}))=\E[(F(\ol{z}_i)-\E(F(\ol{z}_i)))(F(\ol{z}_{i+k})-\E(F(\ol{z}_{i+k})))].\]
Let $\tilde{z}_i$ be an independent identically distributed copy of $\ol{z}_i$, and define  $(\ol{z}_{i:i+k},\tilde{z}_{i:i+k})$ as synchronously coupled, in the same way as in the proof of Proposition \ref{supp:prop:boundDllp1}. 
It follows from applying Proposition \ref{supp:prop:Wasserstein} on each coordinate, and using independence that for every coordinate $1\le j\le d$,
\begin{align*}
 &\E\left(\left.\left\|\tilde{z}_{i+k,j}-\ol{z}_{i+k,j}\right\|_{a,b}^2\right|\ol{z}_{i,j},\tilde{z}_{i,j}\right)\le \exp\left(-\frac{mh_0}{4\gamma}\cdot k\right)  \|\ol{z}_{i,j}-\tilde{z}_{i,j}\|_{a,b}^2.
\end{align*}
With a slight abuse of notation, index $j$ here refers to both position and velocity components, hence $\tilde{z}_{i,j}=(\tilde{x}_{i,j},\tilde{v}_{i,j},\tilde{x}'_{i,j},\tilde{v}'_{i,j})\in \R^4$.
As previously, $\ol{z}_{i}$ and $\tilde{z}_{i+k}$ are independent, and
\begin{align*}&\mathrm{Cov}(F(\ol{z}_i),F(\ol{z}_{i+k}))= \E[(F(\ol{z}_i)-\E(F(\ol{z}_i)))\E(F(\ol{z}_{i+k})-F(\tilde{z}_{i+k})|\ol{z}_i,\tilde{z}_i)]\\
&\le \sqrt{\Var\left(F(\ol{z}_i)\right) \Var\left(F(\ol{z}_{i+k})-F(\tilde{z}_{i+k})\right)}
\end{align*}
By the Efron-Stein inequality \cite{steele1986efron, Boucheron2013}, and some rearrangement, we have
\begin{align*}
&\Var\left(F(\ol{z}_{i+k})-F(\tilde{z}_{i+k})\right)\le
2\sum_{j=1}^{d}\left(\sum_{s=1}^{r}|w^{(s)}_j|\right)^2 \E\left(\left\|\tilde{z}_{i+k,j}-\ol{z}_{i+k,j}\right\|_{a,b}^2\right)\\
&\le 2r \left(\sum_{s=1}^{r}\|w^{(s)}\|^2\right)
\exp\left(-\frac{mh_0}{4\gamma}\cdot k\right)
\max_{1\le j\le d}\E\left[\|\ol{z}_{i,j}-\tilde{z}_{i,j}\|_{a,b}^2\right]
\\
&=4r \left(\sum_{s=1}^{r}\|w^{(s)}\|^2\right)
\exp\left(-\frac{mh_0}{4\gamma}\cdot k\right)
\max_{1\le j\le d}\E\left[\|\ol{z}_{i,j}-\E\ol{z}_{i,j}\|_{a,b}^2\right],
\end{align*}
and the claim follows by rearrangement and summation.
\end{proof}

We are going to use an assumption on the initial distribution $\mu_0$ to show the dimension-free bounds for product distributions.
\begin{assumption}\label{supp:assum:product}
    Suppose that $\mu_0$ and the target distribution $\pi$ are of product form
\begin{align*}
\mu_{0}(dx,dv)&= \prod^{d}_{i=1}\mu_{0,i}(dx_i,dv_i) \quad \text{ for all }l\ge 0,\quad \pi(d x,d v) = \prod^{d}_{i=1}\tilde{\pi}_{i}(d x_i)\frac{e^{-v_i^2/2}dv_i}{\sqrt{2\pi}},
\end{align*}
for $x = (x_{1},...,x_{d}) \in \R^{d}$, $v = (v_{1},...,v_{d}) \in \R^{d}$, and that
\[\max_{1\le i\le d}\mathcal{W}_2(\pi_i,\mu_{0,i})\le c_{\mu_0}\sqrt{\frac{1}{m}},\]
for some finite constant $c_{\mu_0}$, where $\pi_{i}(d x_i,d v_i)=\tilde{\pi}_{i}(d x_i)\frac{e^{-v_i^2/2}}{\sqrt{2\pi}} dv_i$ is the joint distribution of $(x_i,v_i)$ according to the target $\pi$.
\end{assumption}

\begin{proposition}\label{supp:prop:ProductExactGradUBUBU}
Suppose that Assumption \ref{supp:assum:product} holds, and denote the potential $U$ as $U(x) = \sum^{d}_{i=1}U_{i}(x_{i})$. Suppose that Assumptions \ref{supp:assum:Lip}, \ref{supp:assum:convex}, and \ref{supp:assum:Hess_Lipschitz} hold for each component $(U_i)_{1\le i\le d}$, and that
\begin{align*}\gamma \geq \sqrt{8M}, \quad h_0 \leq 
\frac{1}{\gamma}\cdot \frac{m}{264M}, \quad B\ge \frac{16\log(4)\gamma}{mh_0}, \quad
B_0\ge \frac{16 \gamma}{m h_0}\log\left(\frac{c_{\mu_0}+1}{ \sqrt{M}\gamma h_0^2}\right).
\end{align*}
Suppose that $f$ is of the form
\begin{equation}\label{supp:eq:fprodwdef}
f(x,v) = g(\langle w^{(1)}, x \rangle,\ldots \langle w^{(r)}, x \rangle), 
\end{equation}
where $g:\R^r \to \R$ is 1-Lipschitz, and $w^{(1)},\ldots,w^{(r)} \in \R^{d}$. Suppose that $c_{R}\in [0, \phi_N^{-1/2})$ and $2<\phi_N<16$. Then for any $N\ge 1$, the $\UBUBU$ estimator $S(c_R)$ has finite expected computational cost, $\E S(c_R)=\pi(f)$, and it has finite variance. Moreover, it satisfies a CLT as $N\to \infty$, and the asymptotic variance can be bounded as 
\[
\sigma^2_S\le \frac{C(m,M,M^{s}_{1},\gamma,r,c_N,\phi_N)
}{K h_0} \sum_{1\le i\le r}\|w^{(i)}\|^{2}.
\]
\end{proposition}

\begin{proof}[Proof of Proposition \ref{supp:prop:ProductExactGradUBUBU}]
Unbiasedness, finite variance, and finite computational cost follow from Theorem \ref{supp:thm:exactgradUBUBU}. By \eqref{supp:eq:sigma2S}, the asymptotic variance can be expressed as
\[\sigma^2_S:=\Var(D_0)+\sum_{l=0}^{\infty} \Var(D_{l,l+1})\cdot \frac{\phi_N^l}{c}.\]

It is easy to show that 
 $f$ is $\sum_{s=1}^{r}\|w^{(s)}\|$-Lipschitz, so the variance term $\Var(D_0)$ can be bounded using Proposition \ref{supp:prop:D0var}, relying on the burn-in assumptions.

To control $\Var(D_{l,l+1})$, we first need to bound terms of the form $\Var(f(z_k')-f(z_k))$. 
Let $z_{k,j}=(x_{k,j},v_{k,j})\in \R^2$ denote components $j$ in both $x$ and $v$.
Using the Efron-Stein inequality \cite{steele1986efron, Boucheron2013}, and independence of the components, we have
\begin{align*}
&\Var(f(z'^{(l,l+1)}_k)-f(z^{(l,l+1)}_k))\le
2\sum_{j=1}^{d}\left(\sum_{s=1}^{r}|w^{(s)}_j|\right)^2 \E\left(\left\|z'^{(l,l+1)}_{k,j}-z_{k,j}^{(l,l+1)} 
  \right\|_{a,b}^2\right)\\
&\le 2r \left(\sum_{s=1}^{r}\|w^{(s)}\|^2\right) \max_{1\le j\le d}\E\left(\left\|z'^{(l,l+1)}_{k,j}-z^{(l,l+1)}_{k,j}\right\|_{a,b}^2\right).
\end{align*}
By applying Corollary \ref{supp:corollary:global_strong_error} component-wise, it follows that under our assumptions, 
\[\max_{1\le j\le d}\E\left(\left\|z'^{(l,l+1)}_{k,j}-z_{k,j}\right\|_{a,b}^2\right)\le C(m,M,\gamma, M^{s}_{1})h_l^4,\]
hence 
\[\Var(f(z'^{(l,l+1)}_k)-f(z^{(l,l+1)}_k))\le C(m,M, M^{s}_{1},\gamma,r) \sum_{1\le i\le r}\|w^{(i)}\|^{2} h_l^4.\]
Using this, and Proposition \ref{supp:prop:boundDllp1independent}, by a similar argument as in the proof of Theorem \ref{supp:thm:exactgradUBUBU}, we can show that 
\[\Var(D_{l,l+1})\le \frac{C(m,M, M^{s}_{1},\gamma, r)}{K} \left(\sum_{1\le i\le r}\|w^{(i)}\|^{2}\right) h_l^4 \left(1+\frac{4\gamma}{mh_0}+\log(4)l\right),\]
and the claim follows by summation and rearrangement.
\end{proof}

 
\section{Initialization and Gaussian approximation}\label{supp:sec:initial_conditions}

We will use the following assumptions in this section and sections \ref{supp:sec:Appendix_UBUBU_variance_bnd_svrg} and \ref{supp:sec:Appendix_UBUBU_variance_bnd_approx_grad}, which we restate here for easier readability. We will consider potentials of the form 
\begin{equation}\label{supp:eq:potential_form}
    U(x)=U_{0}(x) + \sum_{i=1}^{N_D} U_i(x),
\end{equation} 
where we aim to understand the scaling of the computational complexity when inexact gradients are used within the $\UBUBU$ framework in the large $N_D$ case. We assume that the potential has the form \eqref{supp:eq:potential_form} in this section and sections \ref{supp:sec:Appendix_UBUBU_variance_bnd_svrg} and \ref{supp:sec:Appendix_UBUBU_variance_bnd_approx_grad}, and we impose the following assumptions on the potential.

\begin{assumption}[$\nabla$Lipschitz property]\label{supp:assum:LipSG}
For every $1\le i\le N_D$, $U_i:\R^{d} \to \R$ is twice differentiable and there exists a $\tilde{M}>0$ such that for all $x,y \in \R^d$,
$$
\| \nabla U_i(x) - \nabla U_i(y)\| \leq \tilde{M}\|x-y\|,
$$
for every $1\le i\le N_D$ and moreover,
$$
\| \nabla U(x) - \nabla U(y)\| \leq M\|x-y\| \quad \text{for}\quad M=N_D\tilde{M}.
$$
\end{assumption}
If we have a potential that is not necessarily of the form \eqref{supp:eq:potential_form}, we assume the following Lipschitz assumption on the gradient.
\begin{assumption}[$\nabla$Lipschitz property]
\label{supp:assum:LipA}
There is a $\tilde{M}>0$ such that for all $x,y \in \R^d$,
$$
\| \nabla U(x) - \nabla U(y)\| \leq M\|x-y\| \quad \text{for}\quad M=N_D\tilde{M}.
$$
\end{assumption}

\begin{assumption}[$N_D\tilde{m}$-strong convexity]
\label{supp:assum:convexSG}
There exists a $\tilde{m}>0$ such that for all $x,y \in \R^d$
$$
\langle \nabla U(x) - \nabla U(y),x-y \rangle \geq m\|x-y\|^2 \quad \text{ for }\quad m=N_D\tilde{m}.
$$
\end{assumption}

\begin{assumption}[strongly Hessian Lipschitz property]
\label{supp:assum:Hess_LipschitzSG}
$U:\mathbb{R}^{d} \to \mathbb{R}$ is three times continuously differentiable and $M^{s}_{1}$-strongly Hessian Lipschitz if there exists $M^{s}_{1} > 0$ such that
    \[
    \|\nabla^{3}U(x)\|_{\{1,2\}\{3\}} \leq M^{s}_{1}\quad \text{for}\quad M^{s}_{1}=N_D \Tilde{M}^{s}_{1},
    \]
    for all $x \in \mathbb{R}^{d}$.   
\end{assumption}

For a better understanding of the scaling in terms of $N_D$, we also introduce
\begin{equation}\label{supp:eq:tgammadef}
\Tilde{\gamma}=\frac{\gamma}{\sqrt{N_D}},
\end{equation}
so that $\gamma=\sqrt{N_D}\Tilde{\gamma}$.\\

\subsection{OHO scheme}

In this section, we detail some results for the OHO scheme we use for initialization. We state results for a potential that satisfies Assumptions  \ref{supp:assum:Lip} and \ref{supp:assum:convex} that can be applied in the case of Gaussian approximation. In particular, we show strong error results using similar techniques to \cite{leimkuhler2023contractiona} and \cite{schuh2024}.

We define the solution map $\mathcal{H}$ to have update rule
\begin{equation}\label{supp:eq:H_step}
     \mathcal{H}: (x,v) \to \phi_{h}(x,v),
\end{equation}
where $\phi_{h}(x,v)$ is the solution to the ODE
\begin{align*}
    dX_{t} &= V_{t}dt, \qquad dV_{t} = -\nabla U(X_{t})dt,
\end{align*}
initialized at $(X_{0},V_{0}) := (x,v) \in \mathbb{R}^{2d}$ at time $h > 0$. We then define the OHO scheme with stepsize $h >0$ as a half step of $\mathcal{O}$ with stepsize $h/2$, followed by a full step of $\mathcal{H}$ with stepsize $h$ and a half step of $\mathcal{O}$ with stepsize $h/2$, which exactly preserves the invariant measure.

\begin{remark}
We remark that the OHO scheme is a special case of the scheme studied in \cite{monmarche2024entropic} using a hypocoercivity approach. It has also been considered as an exact splitting for discretization analysis in \cite{BouRabeeOwhadi2010,PM21,monmarche2022hmc,bou2023randomized}. In practice, this scheme is only applicable when the Hamiltonian dynamics can be solved exactly, for example for a Gaussian target.
\end{remark}

\begin{proposition}\label{supp:prop:OHO_strong_error}
Let $h<1/2\gamma$, $\gamma^{2}\geq 4M$, $k \in \mathbb{N}$ and $(X_{t},V_{t})_{t\geq 0} := (Z_{t})_{t \geq 0}$ be the solution of the continuous kinetic Langevin dynamics and $(x_{t},v_{t})_{t\geq 0} := (z_{t})_{t \geq 0}$ be the solution to the OHO scheme with stepsize $h >0$, with synchronously coupled Brownian motion and where both are initialized at $z_{0} = Z_{0} \sim \pi$. We have that
    \begin{align*}
    \|Z_{kh} - z_{kh}\|_{L^{2},a,b} &\leq \sqrt{\frac{3}{2}} e^{\frac{3}{2}hk\sqrt{M}}\left(\frac{3h\gamma \sqrt{k\gamma h d} + 5k(h\gamma)^{2}\sqrt{d}}{\sqrt{M}}\right).
\end{align*}
\end{proposition}
\begin{proof}
Considering the OHO scheme given by 
\begin{align*}
    &(x_{h},v_{h}) := \Bigg(x + h\left(\eta v + \sqrt{1-\eta^{2}}\xi_{1}\right) - \int^{h}_{0}\nabla U(x(t))(h-t)dt,\\
    &\eta\left( \eta v + \sqrt{1-\eta^{2}}\xi_{1} - h\nabla U(x) - \int^{h}_{0}\nabla^{2}U(x(t))\overline{v}(t)(h-t)dt\right) + \sqrt{1-\eta^{2}}\xi_{2}\Bigg),
\end{align*}
where the Hamiltonian dynamics $(x(t),v(t))^{h}_{t = 0}$ is initialized at $(x,\eta v + \sqrt{1-\eta^{2}}\xi_{1})$.
 The kinetic Langevin dynamics for one step can be written as 
\begin{align}
V_{h} &= \mathcal{E}(h)V_{0} - \int^{h}_{0}\mathcal{E}(h-s)\nabla U(X_{s})ds + \sqrt{2\gamma}\int^{h}_{0}\mathcal{E}(h-s)dW_s,\label{supp:eq:contv_1}\\
X_{h} &= X_0 + \mathcal{F}(h)V_{0} - \int^{h}_{0}\mathcal{F}(h-s)\nabla U(X_{s})ds + \sqrt{2\gamma}\int^{h}_{0}\mathcal{F}(h-s)dW_{s},\label{supp:eq:contx_1}
\end{align}
where
$
\mathcal{E}(h) = e^{-\gamma h}, \mathcal{F}(h) = \frac{1-e^{-\gamma h}}{\gamma}$,
and we couple the noises such that $\sqrt{1-\eta^{2}}\xi_{1} = \sqrt{2\gamma}\int^{h/2}_{0}\mathcal{E}(h/2-s)dW_{s}$ and $\sqrt{1-\eta^{2}}\xi_{2} = \sqrt{2\gamma}\int^{h}_{h/2}\mathcal{E}(h-s)dW_{s}$.
Considering the velocity component we have
\begin{align*}
    &\|V_{h} - v_{h}\|_{L^{2}} \\
    &\leq \left\| \eta^{2} \left(V_{0} - v_{0}\right) -\int^{h}_{0}\mathcal{E}(h-s)\left(\nabla U(X_{s})-\nabla U(x)\right)ds + \sqrt{2\gamma}(1-\eta) \int^{h/2}_{0}\mathcal{E}(h/2-s)dW_{s}\right\|_{L^{2}} \\
    &+ \sqrt{Md}\left|\frac{1-\eta^{2} - h\gamma \eta}{\gamma}\right| + \frac{h^{2}M\sqrt{d}}{2}\\
    &\leq \left\|\eta^{2} \left(V_{0} - v_{0}\right) -\int^{h}_{0}\mathcal{E}(h-s)\left(\nabla U(X_{0}) - \nabla U(x)\right)ds + \sqrt{2\gamma}(1-\eta) \int^{h/2}_{0}\mathcal{E}(h/2-s)dW_{s}\right\|_{L^{2}} \\
    &+ \left\|\int^{h}_{0}\mathcal{E}(h-s)\left(\nabla U(X_{s})-\nabla U(X_{0})\right)ds\right\|_{L^{2}}+ \left(\gamma \sqrt{M} + \frac{M}{2}\right)h^{2}\sqrt{d}\\
    &\leq \left\|\eta^{2} \left(V_{0} - v_{0}\right) -\int^{h}_{0}\mathcal{E}(h-s)\left(\nabla U(X_{0}) - \nabla U(x)\right)ds + \sqrt{2\gamma}(1-\eta) \int^{h/2}_{0}\mathcal{E}(h/2-s)dW_{s}\right\|_{L^{2}} \\
    &+ \left(\gamma \sqrt{M} + 3M\right)h^{2}\sqrt{d},
\end{align*}
where the final estimate is a rough estimate due to \eqref{supp:eq:contv_1}, $(X_{s},V_{s}) \sim \pi$ for all $s \in [0,h]$, the fact that $U$ is $M$-$\nabla$Lipschitz, $h < \frac{1}{2\gamma}$ and $\gamma^{2} \geq 4M$. Now considering time $kh \geq 0$  for $k \in \mathbb{N}$ and iteratively applying the argument whilst keeping Brownian components in the same $L^{2}$ norm, we have
\begin{align*}
    \|V_{kh} - v_{kh}\|_{L^{2}} &\leq \sum^{k}_{i=1}hMa_{i} + \sqrt{2\gamma}\left\|\sum^{k}_{i=1}\eta^{2(k-i)}\left(1 - \eta\right) \int^{(i-1/2)h}_{(i-1)h}\mathcal{E}((i-1/2)h-s)dW_{s}\right\|_{L^{2}} \\&+3k\gamma\sqrt{M}h^{2}\sqrt{d}\\
    &\leq \sum^{k}_{i=1}hMa_{i} + \frac{h\gamma \sqrt{2\gamma}}{2 }\left\| \sum^{k}_{i=1}\int^{(i-1/2)h}_{(i-1)h}\mathcal{E}((i-1/2)h-s)dW_{s}\right\|_{L^{2}} \\&+ k\left(h\gamma\right)^{2}\sqrt{d}+3k\gamma\sqrt{M}h^{2}\sqrt{d}\\
    &\leq \sum^{k}_{i=1}hMa_{i} + h\gamma\sqrt{k\gamma hd} + 3k\left(h\gamma\right)^{2}\sqrt{d},
\end{align*}
where $a_{i}:= \|X_{ih} - x_{ih}\|_{L^{2}}$ for $i \in \mathbb{N}$ and $b_{i}:= \|V_{ih} - v_{ih}\|_{L^{2}}$ for $i \in \mathbb{N}$. We have also used the independence of the Brownian motion over independent time intervals.

Now considering the position components we have
\begin{align*}
    \|X_{h} - x_{h}\|_{L^{2}} &\leq \Bigg\| X_{0} - x_{0} + \mathcal{F}(h) \left(V_{0} - v_{0}\right) - \int^{h}_{0}\mathcal{F}(h-s)\left(\nabla U(X_{s}) - \nabla U(x_{s})\right)ds \\
    &+ \int^{h}_{0}\sqrt{2\gamma} \mathcal{F}(h-s) dW_{s} - h\sqrt{2\gamma}\int^{h/2}_{0}\mathcal{E}(h/2-s)dW_{s}\Bigg\|_{L^{2}} + 2\gamma h^{2}\sqrt{d}\\
    &\leq \Bigg\|X_{0} - x_{0} +\mathcal{F}(h)(V_{0} - v_{0}) - \int^{h}_{0}\mathcal{F}(h-s)\left(\nabla U(X_{0}) - \nabla U(x_{0})\right)ds \\
    &+ \int^{h}_{0}\sqrt{2\gamma} \mathcal{F}(h-s) dW_{s} - h\sqrt{2\gamma}\int^{h/2}_{0}\mathcal{E}(h/2-s)dW_{s}\Bigg\|_{L^{2}} + 3\gamma h^{2}\sqrt{d}.
\end{align*}
Then as before we consider time $kh\geq 0$ for $k \in \mathbb{N}$ and we have
\begin{align*}
    &\|X_{kh} - x_{kh}\|_{L^{2}} \leq \sum^{k}_{i=1}\left(hb_{i} +h^{2}Ma_{i}\right) +3k\gamma h^{2}\sqrt{d} \\
    &+\sqrt{2\gamma}\left\|\sum^{k}_{i=1}\int^{ih}_{(i-1)h}\mathcal{F}(ih-s)dW_{s} - h\int^{(i-1/2)h}_{(i-1)h}\mathcal{E}((i-1/2)h-s)dW_{s}\right\|_{L^{2}}\\
    &\leq \sum^{k}_{i=1}\left(hb_{i} +h^{2}Ma_{i}\right)+ 3k\gamma h^{2}\sqrt{d}  + 2\sqrt{2\gamma} h\sqrt{hkd}.
\end{align*}
In $\|\cdot\|_{L^{2},a,b}$ using the preceding estimates we have
\begin{align*}
    \|Z_{kh} - z_{kh}\|_{L^{2},a,0} &\leq \frac{3}{2}\sum^{k}_{i=1}h\sqrt{M}\|Z_{ih} - z_{ih}\|_{L^{2},a,0} + \frac{3h\gamma \sqrt{k\gamma h d} + 5k(h\gamma)^{2}\sqrt{d}}{\sqrt{M}}\\
    &\leq e^{\frac{3}{2}hk\sqrt{M}}\left(\frac{3h\gamma \sqrt{k\gamma h d} + 5k(h\gamma)^{2}\sqrt{d}}{\sqrt{M}}\right).
\end{align*}
\end{proof}

\begin{theorem}\label{supp:theorem:OHO_strongerror}
Let $h<1/2\gamma$, $\gamma^{2}\geq 4M$, $l \in \mathbb{N}$ and $(X_{t},V_{t})_{t\geq 0} := (Z_{t})_{t \geq 0}$ be the solution of kinetic Langevin dynamics and $(x_{l},v_{l})_{l \in \mathbb{N}} := (z_{l})_{l \in \mathbb{N}}$ be the iterates of the solution to the $\mathcal{OHO}$ scheme with stepsize $h >0$, where both are initialized at the same point according to the invariant measure. We have that
    \begin{align*}
    \|Z_{lh} - z_{l}\|_{L^{2},a,b} &\leq 104h\gamma^{2}\sqrt{d}\left(\frac{3 \sqrt{2\gamma/\sqrt{M} } + 10\gamma/\sqrt{M}}{m}\right).
\end{align*}
\end{theorem}

\begin{proof}
    We use an approach used in \cite{leimkuhler2023contractiona} to remove the exponential constant in Proposition \ref{supp:prop:OHO_strong_error}. We define a sequence of interpolating variants $Z^{(k)}_{l}$ for every $k = 0,...,l$ as follows. We first define $Z^{(k)}_{0} = Z_{0}$, and then $(Z^{(k)}_{i})^{k}_{i=1}$ are defined by $\mathcal{OHO}$ steps followed by $(Z^{(k)}_{i})^{l}_{i=k+1}$ steps of kinetic Langevin dynamics with stepsize $h > 0$. We break the $l$ steps into blocks of size $\Tilde{l} = \lceil\frac{2}{3h\sqrt{M}} \rceil$, then we have
    \begin{align*}
        \|Z_{lh} - z_{l}\|_{L^{2},a,b} &= \|Z^{(0)}_{l} - Z^{(l)}_{l}\|_{L^{2},a,b}\\
        &\leq \sum^{\lfloor l/\Tilde{l}\rfloor - 1}_{j = 0}\left\|\left(Z^{(j\Tilde{l})}_{l} -Z^{((j+1)\Tilde{l})}_{l} \right)\right\|_{L^{2},a,b} + \left\|\left(Z^{(\lfloor l/\Tilde{l}\rfloor l)} - Z^{(l)}\right)\right\|_{L^{2},a,b},
    \end{align*}
    where we bound the terms using the fact that the first $j\Tilde{l}$ steps according to OHO keep the stationary distribution invariant and they only deviate in the following $\Tilde{l}$ steps, where we will use Proposition \ref{supp:prop:OHO_strong_error} with $l$ chosen as $\Tilde{l}$. We finally use contraction of the continuous dynamics (Corollary \ref{supp:cor:continuous_contraction}) in the remaining steps and we have
    \[
    \left\|\left(Z^{(j\Tilde{l})}_{l} -Z^{((j+1)\Tilde{l})}_{l} \right)\right\|_{L^{2},a,b} \leq 4 e^{-(l-1-(j+1)\Tilde{l})c(h)}\left(\frac{3h\gamma \sqrt{\Tilde{l}\gamma h d} + 5\Tilde{l}(h\gamma)^{2}\sqrt{d}}{\sqrt{M}}\right).
    \]
    Then summing up the terms we have that
    \begin{align*}
        \|Z_{lh} - z_{l}\|_{L^{2},a,b} &\leq 4\left(\frac{3h\gamma \sqrt{\Tilde{l}\gamma h d} + 5\Tilde{l}(h\gamma)^{2}\sqrt{d}}{\sqrt{M}}\right)\left(1 + \frac{1}{1-e^{-c(h)\Tilde{l}}}\right)\\
        &\leq 8\left(\frac{3h\gamma \sqrt{\Tilde{l}\gamma h d} + 5\Tilde{l}(h\gamma)^{2}\sqrt{d}}{\sqrt{M}}\right)\left(1 + \frac{12\gamma\sqrt{M}}{m}\right)\\
        &\leq 104h\gamma^{2}\sqrt{d}\left(\frac{3 \sqrt{2\gamma/\sqrt{M} } + 10\gamma/\sqrt{M}}{m}\right),
    \end{align*}
    as required.
\end{proof}
\subsection{Initialization and bounds}

For convex potentials, we can approximate the gradient with the Hessian at the minimizer by
\begin{equation}\label{supp:eq:hessian_grad_approx_star}
 \mathcal{Q}(x\mid x^{*}) = \nabla U(x^{*}) + H^{*}(x - x^{*}) = H^{*}(x - x^{*}),
\end{equation}
where $x^{*} \in \mathbb{R}^{d}$ is the minimizer of $U$ and $H^{*} = \nabla^{2}U(x^{*})$.
\begin{lemma}
\label{supp:lem:mean:minimizer:L4:1}
Considering the gradient approximation $\mathcal{Q}$ given by \eqref{supp:eq:hessian_grad_approx_star}, where the potential $U$ satisfies Assumption \ref{supp:assum:Hess_LipschitzSG}, and has a minimizer $x^{*} \in \mathbb{R}^{d}$ we then have the property
\[
\mathbb{E}\left\|\nabla U(x) - \mathcal{Q}(x\mid x^*) \right\|^p \leq (\tilde{M}^{s}_1)^p N_D^p    \|x-x^*\|_{L^{2p}}^{2p},
\]
for any $x \in \mathbb{R}^{d}$.
\end{lemma}
\begin{proof}
    Follows from Taylor expansion.
\end{proof}

We then define the measure $\mu_{G}=\mathcal{N}(x^*,(H^*)^{-1})\times \mathcal{N}(0_d,I_d)$ to be the Gaussian approximation of the target, which is the invariant measure of the OHO scheme and continuous kinetic Langevin dynamics with the use of the gradient approximation \eqref{supp:eq:hessian_grad_approx_star}.

\begin{proposition}\label{supp:prop:wass_initial_distance}
Let $p=2$ or $4$, then we have the following Wasserstein bound between a potential $U$ which satisfies Assumptions \ref{supp:assum:convexSG},\ref{supp:assum:Hess_LipschitzSG} and \ref{supp:assum:LipA}
\[
\mathcal{W}_{p,a,b}(\pi,\mu_{G}) \leq  \sqrt{\frac{3}{2}}\left(\frac{(2p)!}{2^{p}p!}\right)^{1/p}\frac{\Tilde{M}_{1}d}{\Tilde{m}^{2}N_{D}}.
\]
\end{proposition}
\begin{proof}
   If we let $\pi_{x}$ denote the marginal in the position of $\pi$ and $(\mu_{G})_{x}$ denote the marginal in position of $\mu_{G}$. We have from the equivalence of norms that for $p = 2,4$
    \begin{align*}
        \mathcal{W}_{p,a,b}(\pi,\mu_{G}) &\leq \sqrt{\frac{3}{2}}\mathcal{W}_{p,a,0}(\pi,\mu_{G}) \leq \sqrt{\frac{3}{2}}\mathcal{W}_{p}(\pi_{x},(\mu_{G})_{x}) \\&\leq \sqrt{\frac{3}{2}}\frac{\|\nabla U - \mathcal{Q}(\cdot\mid x^{*})\|_{L^{p}}}{m} \leq  \sqrt{\frac{3}{2}}\frac{
        M^{s}_{1}\|x-x^{*}\|^{2}_{L^{2p}}}{m}\\
        &\leq \sqrt{\frac{3}{2}}\left(\frac{(2p)!}{2^{p}p!}\right)^{1/p}\frac{M^{s}_{1}d}{m^{2}}
    \end{align*}
    where the third inequality is due to
    Proposition 22 of \cite{Vono2022Efficient}, the fourth due to Lemma \ref{supp:lem:mean:minimizer:L4:1} and the final inequality by Lemma \ref{supp:lem:moment_bounds}.
\end{proof}


\section{Variance bounds for $\UBUBU$ estimator with stochastic gradients}
\label{supp:sec:Appendix_UBUBU_variance_bnd_svrg}

For this section, we make use of the technique of the recent work of Hu et al \cite{Hu21optimal}, related
to using stochastic variance reduced gradient (SVRG). As a reminder, we use the following stochastic gradient approximation.

\begin{definition}\label{supp:def:subsampled_sg}
The sub-sampled stochastic gradient of $U$ at $x$ with respect to $\hat{x}$ is 
\begin{equation}\label{supp:eq:SG}
    \mathcal{G}(x,\omega|\hat{x}) = \nabla U_{0}(x) + \sum^{N_D}_{i=1}\nabla U_{i}(\hat{x}) + \frac{N_D}{N_{b}} \sum_{i \in \omega}[\nabla U_{i}(x) - \nabla U_{i}(\hat{x})],
\end{equation}
where $\omega \sim \SWR(N_D,N_{b})$.
\end{definition}
The approach we use is to update $\hat{x}$ every $\tau = \lceil N_D/N_b\rceil$ iterations with the latest position where the gradient was evaluated (this is not $x_k$ for UBU as the gradients are evaluated after moving forward by a $\mathcal{U}$ step with stepsize $h/2$). We refer to this as the stochastic variance reduced gradient (SVRG) approach (see \citep{johnson2013accelerating,Hu21optimal}).

A stochastic gradient version of the $\UBU$ scheme is simply constructed by replacing the $\mathcal{B}$ operator with
\[
\mathcal{B}_{\mathcal{G}}(x,v,h,\omega|\hat{x}) = (x,v- h\mathcal{G}(x,\omega|\hat{x})),
\]
where $\mathcal{G}$ is a stochastic gradient approximation of the potential $U$ as defined in the approximation given by \eqref{supp:eq:SG}.

We start with an alternative formula for the kinetic Langevin dynamics. This is used for the analysis of the $\UBU$ scheme in the full gradient setting in \cite{sanz2021wasserstein} and alternative schemes with the SVRG approximation in \eqref{supp:eq:SG}. 
The convenient way of expressing kinetic Langevin dynamics is to use It\^o's formula on the product $e^{\gamma t}V_{t}$. This results in the following set of equations for continuous kinetic Langevin dynamics with initial condition $(X_0,V_0) \in \R^{2d}$:
\begin{align}
V_{t} &= \mathcal{E}(t)V_{0} - \int^{t}_{0}\mathcal{E}(t-s)\nabla U(X_{s})ds + \sqrt{2\gamma}\int^{t}_{0}\mathcal{E}(t-s)dW_s, \label{supp:eq:cont_v}\\
X_{t} &= X_0 + \mathcal{F}(t)V_{0} - \int^{t}_{0}\mathcal{F}(t-s)\nabla U(X_{s})ds + \sqrt{2\gamma}\int^{t}_{0}\mathcal{F}(t-s)dW_{s} \label{supp:eq:cont_x},
\end{align}
where
\begin{equation}\label{supp:eq:EFdef}
\mathcal{E}(t) = e^{-\gamma t}\qquad \mathcal{F}(t) = \frac{1-e^{-\gamma t}}{\gamma}.
\end{equation}
Then the $\UBU$ scheme (as in \cite{sanz2021wasserstein}) can be expressed as 
\begin{align}
v_{k+1} &= \mathcal{E}(h)v_k - h\mathcal{E}(h/2)\nabla U(\overline{x}_{k}) + \sqrt{2\gamma}\int^{(k+1)h}_{kh}\mathcal{E}((k+1)h-s)dW_s,\label{supp:eq:disc_v}\\
\overline{x}_{k} &= x_{k} + \mathcal{F}(h/2)v_{k} + \sqrt{2\gamma}\int^{(k+1/2)h}_{kh}\mathcal{F}((k+1/2)h-s)dW_s, \label{supp:eq:disc_y}\\ 
x_{k+1} &= x_{k} + \mathcal{F}(h)v_{k} - h\mathcal{F}(h/2)\nabla U(\overline{x}_{k}) + \sqrt{2\gamma}\int^{(k+1)h}_{kh} \mathcal{F}((k+1)h-s)dW_s,\label{supp:eq:disc_x}
\end{align}
which can be more easily compared to the true dynamics via \eqref{supp:eq:cont_v} and \eqref{supp:eq:cont_x}. We will refer to $(\overline{x}_{k})_{k \in \mathbb{N}}$ as the gradient evalution points of the scheme. Similarly, stochastic gradient $\UBU$ can be expressed as \eqref{supp:eq:disc_v}-\eqref{supp:eq:disc_x} by replacing the gradients with stochastic gradient approximations,
\begin{align}
v_{k+1} &= \mathcal{E}(h)v_k - h\mathcal{E}(h/2)\mathcal{G}(\overline{x}_{k},\omega_{k+1}|\hat{x}_{k}) + \sqrt{2\gamma}\int^{(k+1)h}_{kh}\mathcal{E}((k+1)h-s)dW_s,\label{supp:eq:disc_v_SG}\\
\overline{x}_{k} &= x_{k} + \mathcal{F}(h/2)v_{k} + \sqrt{2\gamma}\int^{(k+1/2)h}_{kh}\mathcal{F}((k+1/2)h-s)dW_s, \label{supp:eq:disc_y_SG}\\
\hat{x}_k&=\overline{x}_{L(k)} \quad \text{for}\quad L(k)=\tau \lfloor k/\tau \rfloor,\label{supp:eq:disc_hatx_SG}\\
x_{k+1} &= x_{k} + \mathcal{F}(h)v_{k} - h\mathcal{F}(h/2)\mathcal{G}(\overline{x}_{k},\omega_{k+1}|\hat{x}_{k}) + \sqrt{2\gamma}\int^{(k+1)h}_{kh} \mathcal{F}((k+1)h-s)dW_s,\label{supp:eq:disc_x_SG}
\end{align}

If we are using a stochastic gradient approximation of the $\UBU$ dynamics, additional bias is introduce by the use of gradient approximations. We wish to measure the local error caused by the stochastic gradient approximation. 

\subsection{Variance bound of $D_{l,l+1}$}
Suppose now we have two $\UBU$ schemes, a UBU scheme $(z_{k})_{k \in \mathbb{N}} = (x_{k},v_{k})_{k \in \mathbb{N}}$ which uses a stochastic gradient approximation as defined in Definition \ref{supp:def:subsampled_sg} with $(\omega_{k})_{k \in \mathbb{N}}$ such that $\omega_{k} \sim \SWR(N_D,N_{b})$ for each $k \in \mathbb{N}$. Further at iteration $k$ define $z^{h}_{k} :=(x^{h}_{k},v^{h}_{k}) := \psi_{h}\left(z_{k},h,(W_{t'})^{(k+1)h}_{t'=kh}\right)$ to be a step of the full gradient UBU scheme at iteration $z_{k}$, with synchronously coupled Brownian motion. Then the local error after one step is
\begin{align*}
    \mathbb{E}\|(x_{k+1} - x^{h}_{k},v_{k+1}- v^{h}_{k})\|^{2} = h^{2}\left(\mathcal{E}(h/2) +\mathcal{F}(h/2)\right)^{2}\mathbb{E}\|\nabla U(\overline{x}_{k}) - \mathcal{G}(\overline{x}_{k},\omega_{k+1})\|^{2},
\end{align*}
and
\begin{align*}
    \mathbb{E}\|x_{k+1} - x^{h}_{k}\|^{2} \leq \frac{h^{4}}{4}\mathbb{E}\|\nabla U(\overline{x}_{k}) - \mathcal{G}(\overline{x}_{k},\omega_{k+1})\|^{2},
\end{align*}
where expectations are taken over stochastic gradient approximation and Brownian increments. The sequence $(\overline{x}_{k})_{k\in \mathbb{N}}$ are the points where the stochastic gradient approximations are evaluated defined by \eqref{supp:eq:disc_y_SG}. We now wish to bound the term $\mathbb{E}\|\nabla U(\overline{x}_{k}) - \mathcal{G}(\overline{x}_{k},\omega_{k+1})\|^{2}$,
uniformly in $k$ to control the error due to the stochastic gradient. For this, we state Lemma 1 of  \cite{Hu21optimal} with our notations, together with its proof.
\begin{lemma}\label{supp:lem:mean}
Considering iterates $(x_{k},v_{k},\overline{x}_{k})_{k \in \mathbb{N}}$ of stochastic gradient $\UBU$ with the SVRG $(\mathcal{G},\SWR(N_D,N_{b}))$ for a potential $U$ which has the form \eqref{supp:eq:potential_form}, with data size $N_{D}$ and batch size $N_b$, epoch length $\tau = \lceil N_D/N_b\rceil$, and initial condition $(x_{0},v_{0}) \in \R^{2d}$, then we have the property
\begin{align}
\nonumber&\mathbb{E} \left(\left\|\mathcal{G}(\overline{x}_{k},\omega_{k+1} \mid \overline{x}_{L(k)}) - \nabla U(\overline{x}_{k})\right\|^{2}\right) \\
\nonumber&\leq \frac{N_{D}(N_{D}-N_b)(\tau-1)^{2}}{N_{b}(N_{D}-1)}. \max_{j < k} \sum^{N_{D}}_{i=1}\mathbb{E}\left(\|\nabla U_{i}(\overline{x}_{j+1}) - \nabla U_{i}(\overline{x}_{j})\|^{2}\right).
\end{align}

\end{lemma}

\begin{corollary}
Suppose that Assumption \ref{supp:assum:LipSG} holds. For UBU with SVRG updates as defined by \eqref{supp:eq:disc_v_SG}-\eqref{supp:eq:disc_x_SG}, we have
\begin{align}
\label{supp:eq:SGcondition}&\mathbb{E}\left(\left\|\mathcal{G}(\overline{x}_{k},\omega_{k+1} \mid \overline{x}_{L(k)}) - \nabla U(\overline{x}_{k})\right\|^{2}\right) \leq 
\Theta
\max_{j < k} \mathbb{E}\|\overline{x}_{j+1} - \overline{x}_{j}\|^{2},\\
\label{supp:eq:Thetadef}&\Theta{= \frac{\tilde{M}^{2}N_D^2 (N_{D}-N_b)(\tau-1)^{2}}{N_{b}(N_{D}-1)}}.
\end{align}
\end{corollary}
\begin{proof}[Proof of Lemma \ref{supp:lem:mean}]
For the potential of the form $U(x) = U_{0}(x) + \sum^{N_{D}}_{i=1}U_{i}(x)$ and for $k \geq 1$ we define $\ol{X}_{i}= \nabla U_{i}(\overline{x}_{k}) - \nabla U_{i}(\overline{x}_{L(k)})$ and we define $Y_{i} = N_{D}\ol{X}_{i}- \sum^{N_{D}}_{j=1}\ol{X}_{j}$ for $i = 1,...,N_{D}$. Then we have that $\sum^{N_{D}}_{i=1}Y_{i} = 0$ and that
\[
\mathcal{G}(\overline{x}_{k},\omega_{k+1} \mid \overline{x}_{L(k)}) - \nabla U(\overline{x}_{k}) = \frac{1}{N_b}\sum_{i \in \omega_{k+1}}Y_{i}.
\]
Therefore our aim is to establish a bound on $\frac{1}{N_b}\sum_{i \in \omega_{k+1}}Y_{i}$. We have that
\begin{align*}
\mathbb{E}_{\omega_{k+1}}\left\|\frac{1}{N_b}\sum_{i \in \omega_{k+1}}Y_{i}\right\|^{2} &= \frac{1}{N^{2}_b}\mathbb{E}_{\omega_{k+1}}\left(\sum_{i \in \omega_{k+1}}\|Y_{i}\|^{2} + \sum_{i\neq j \in \omega_{k+1}}\langle Y_{i},Y_{j}\rangle\right)\\
&= \frac{1}{N_{b}N_{D}}\sum^{N_{D}}_{i=1}\|Y_{i}\|^{2} + \frac{b-1}{N_{b}N_{D}(N_{D}-1)}\sum_{i \neq j} \left\langle Y_{i},Y_{j}\right\rangle\\
&= \frac{N_{D}-N_b}{N_{D}-1}\frac{1}{N_{b}N_{D}}\sum^{N_{D}}_{i=1}\|Y_{i}\|^{2},
\end{align*}
where the last line is due to the fact that $\sum^{N_{D}}_{i=1}Y_{i} = 0$. Then using the fact that $\sum^{N_{D}}_{i=1}\|Y_{i}\|^{2} \leq N^{2}_{D} \sum^{N_{D}}_{i=1}\|\ol{X}_{i}\|^{2}$ and the last full gradient evaluation is at $k-\tau + 1\leq L(k) \leq k $ we have that
\begin{align*}
&
\mathbb{E}\left(\|\nabla U(\overline{x}_{k}) - \mathcal{G}(\overline{x}_{k},\omega_{k+1}\mid \overline{x}_{L(k)})\|^{2}\right) 
= \frac{N_{D}-N_b}{N_{b}N_{D}(N_{D}-1)}\sum^{N_{D}}_{i=1}\mathbb{E}\|Y_{i}\|^{2}\\
&\leq \frac{N_{D}(N_{D}-N_b)}{N_{b}(N_{D}-1)}\sum^{N_{D}}_{i=1}\mathbb{E}\|\ol{X}_{i}\|^{2}\\
&\leq \frac{N_{D}(N_{D}-N_b)}{N_{b}(N_{D}-1)}\sum^{N_{D}}_{i=1}\mathbb{E}\|\nabla U_{i}(\overline{x}_{k}) - \nabla U_{i}(\overline{x}_{L(k)})\|^{2}\\
&\leq \frac{N_{D}(N_{D}-N_b)(k-L(k))}{N_{b}(N_{D}-1)}\sum^{k-1}_{j = L(k)}\sum^{N_{D}}_{i=1}\mathbb{E}\|\nabla U_{i}(\overline{x}_{j+1}) - \nabla U_{i}(\overline{x}_{j})\|^{2}\\
&\leq \frac{N_{D}(N_{D}-N_b)(\tau-1)^{2}}{N_{b}(N_{D}-1)}\max_{j < k} \sum^{N_{D}}_{i=1}\mathbb{E}\|\nabla U_{i}(\overline{x}_{j+1}) - \nabla U_{i}(\overline{x}_{j})\|^{2},
\end{align*}
which concludes the proof.
\end{proof}

Hence it is sufficient to bound $\mathbb{E}\|\overline{x}_{k+1} - \overline{x}_{k}\|^{2}$ uniformly in $k \in \mathbb{N}$, which will be done in the following lemma.

\begin{lemma}[Displacement Lemma]\label{supp:lem:consecutive_y}
 Let a stochastic gradient $\UBU$ integrator defined by \eqref{supp:eq:disc_v_SG}-\eqref{supp:eq:disc_y_SG} with stochastic gradient $(\mathcal{G},\rho)$ satisfy 
 \[
\mathbb{E}\left(\left\|\mathcal{G}(\overline{x}_{k},\omega_{k+1}|\hat{x}_k) - \nabla U(\overline{x}_{k})\right\|^2\right) \leq \Theta\max_{j < k} \mathbb{E}\|\overline{x}_{j+1} - \overline{x}_{j}\|^{2},
\]
for some $\Theta > 0$.
If $U$ satisfies Assumptions \ref{supp:assum:LipSG}, \ref{supp:assum:convexSG}, $h < 1/2\gamma$ and $\gamma^{2} \geq 8M$, then 
\begin{align*}
\|\overline{x}_{k+1} - \overline{x}_{k}\|_{L^{2}} &\leq h^{2}\sqrt{\Theta} \max_{0\leq i < k} \|\overline{x}_{i+1}-\overline{x}_{i}\|_{L^{2}} + 7h\sqrt{M}\|z_{k-1} - Z^{k-1}\|_{L^{2},a,b} + 6h\sqrt{d},
\end{align*}
where $Z^{k}:= Z_{kh} = \phi(Z_{0},kh,(W_{t'})^{kh}_{t'=0}) \in \R^{2d}$ is the solution to continuous kinetic Langevin dynamics initialized at the invariant measure $Z_{0} \sim \pi$ at time $kh$ for $k \in \mathbb{N}$. 
\end{lemma}

\begin{proof}
Then we use the following estimate
\begin{align*}
\|\overline{x}_{k+1} - \overline{x}_{k}\|_{L^{2}} &= \|\overline{x}_{k+1} - x_{k+1} + x_{k+1} - x_{k} + x_{k}- \overline{x}_{k}\|_{L^{2}}\\
&\leq \|\mathcal{F}(h/2)(v_{k} - v_{k-1})\|_{L^{2}}  + \|x_{k+1} - x_{k}\|_{L^{2}}\\
&+ \sqrt{2\gamma}\left\|\int^{(k+3/2)h}_{(k+1)h}\mathcal{F}((k+3/2)h-s)dW_{s}\right\|_{L^{2}} \\&+ \sqrt{2\gamma}\left\|\int^{(k+1/2)h}_{kh}\mathcal{F}((k+1/2)h-s)dW_{s}\right\|_{L^{2}}\\
&=: \textnormal{(I)} + \textnormal{(II)} + \textnormal{(III)} + \textnormal{(IV)},
\end{align*}
and we bound (I), (II), (III) and (IV) separately. (III) and (IV) can be bounded above by $\sqrt{\gamma h^{3}d}$. Firstly, we will bound (II), but first we denote
\[
A_{j} = \|\overline{x}_{j+1} - \overline{x}_{j}\|^{2}_{L^{2}},
\]
for $j \in \mathbb{N}$, and $z^{h}_{k} = (x^{h}_{k},v^{h}_{k}):=\psi_{h}(z_{k},h,(W_{t'})^{(k+1)h}_{t'=kh})$ is an iterate with stepsize $h$ and initial point $(x_{k},v_k)$ of the full gradient $\UBU$ scheme and synchronously coupled Brownian motion to the stochastic gradient scheme. We then estimate
\begin{align*}
\|x_{k+1}-x_{k}\|_{L^{2}} &\leq \|x_{k+1} - x^{h}_{k}\|_{L^{2}} + \| x^{h}_{k}-x_{k}\|_{L^{2}} \\
&\leq \frac{h^{2}}{2}\sqrt{\Theta}\max_{j < k} \sqrt{ A_{j}} + \|x^{h}_{k}-x_k \|_{L^{2}},
\end{align*}
then if we define the notation $Z^{t}_{k} = (X^{t}_{k},V^{t}_{k}) := \phi(z_{k},t,(W_{t'})^{kh+t}_{t'=kh}) \in \R^{2d}$ for $k \in \mathbb{N}$ and $t \geq 0$ to be the continuous dynamics solution with initial condition $(x_k,v_k)$ at time $t$ defined by \eqref{supp:eq:cont_v} and \eqref{supp:eq:cont_x}. Then we can estimate the second term by splitting it up into discretization error and one-step displacement and bounding each of these terms separately as
\begin{align*}
\|x^{h}_{k}-x_{k}\|_{L^{2}} &\leq \|x^{h}_{k} - X^{h}_{k}\|_{L^{2}} + \|X^{h}_{k}-x_k\|_{L^{2}}.
\end{align*}
Then using \cite{sanz2021wasserstein}[Section 7.6] we have that
\begin{align*}
    &\|x^{h}_{k} - X^{h}_{k}\|_{L^{2}} \leq \\
    &\left\|\int^{h}_{0}\mathcal{F}(h/2)\left(\nabla U(X^{s}_{k}) - \nabla U(\overline{x}_{k})\right)ds + \int^{h}_{0}\left(\mathcal{F}(h-s)-\mathcal{F}(h/2)\right)\nabla U(X^{s}_{k})ds\right\|_{L^{2}}\\
    &\leq \frac{h}{2}\int^{h}_{0}\left\|\nabla U(X^{s}_{k}) - \nabla U(\overline{x}_{k})\right\|_{L^{2}}ds+h^{2}\max_{0\leq s\leq h}\|\nabla U(X^{s}_{k})\|_{L^{2}}\\
    &\leq \frac{hM}{2}\int^{h}_{0}\left\|X^{s}_{k}-\left(x_k + \mathcal{F}\left(h/2\right)v_k +\sqrt{2\gamma}\int^{(k+1/2)h}_{kh}\mathcal{F}\left((k+1/2)h-s\right)dW_s\right)\right\|_{L^{2}}ds \\
    &+ h^{2}\max_{0\leq s\leq h}\|\nabla U(X^{s}_{k})\|_{L^{2}}\\
    &\leq h^{2}\max_{0\leq s\leq h}\|\nabla U(X^{s}_{k})\|_{L^{2}} + \frac{hM}{2}\int^{h}_{0}\|X^{s}_{k}-x_k\|_{L^{2}}ds + \frac{h^{3}M}{4}\max_{0\leq s\leq h}\|V^{s}_{k}\|_{L^{2}} + \frac{h^{7/2}M\sqrt{\gamma d}}{4}.
    \end{align*}
    Now, we bound
    \begin{align*}
        \int^{h}_{0}\|X^{s}_{k}-x_k\|_{L^{2}}ds &\leq h^2\|v_k\|_{L^{2}} + h^{3}\max_{0\leq s\leq h}\|\nabla U(X^{s}_{k})\|_{L^{2}} + h^{5/2}\sqrt{2\gamma d},
    \end{align*}
    and using the fact that $h < \min\{\frac{1}{5\sqrt{M}},\frac{1}{2\gamma}\}$ we have
    \begin{align*}
    &\|x^{h}_{k} - X^{h}_{k}\|_{L^{2}} \leq \frac{3h^{2}}{2}\max_{0\leq s\leq h}\|\nabla U(X^{s}_{k})\|_{L^{2}} + h\max_{0\leq s\leq h}\|V^{s}_{k}\|_{L^{2}} + h\sqrt{d},
    \end{align*}
    and we have that 
\begin{align*}
\|X^{h}_{k} - x_{k}\|_{L^{2}} &= \left\|\mathcal{F}(h)v_{k} - \int^{h}_{0}\mathcal{F}(h)\nabla U(X^{s}_{k})ds + \sqrt{2\gamma}\int^{h}_{0}\mathcal{F}(h-s)dW_{s}\right\|_{L^{2}}\\
&\leq h\|v_{k}\|_{L^{2}} + h^{2}\max_{0 \leq s\leq h}\|\nabla U(X^{s}_{k})\|_{L^{2}} + \sqrt{2\gamma h^{3}d}.
\end{align*}
To bound the maximum terms we introduce $(Z_{t})_{t\geq 0} = (X_{t},V_{t})_{t \geq 0}$ to be the solution to continuous kinetic Langevin dynamics initialized at the invariant measure with synchronously coupled Brownian motion. We also define $Z^{k} := Z_{kh}$ for $k \in \mathbb{N}$. Then we have, in expectation, for any $0 \leq s \leq h$,
\begin{align*}
    \|V^{s}_{k}\|_{L^{2}} &\leq \|V^{s}_{k}-V_{kh + s}\|_{L^{2}} + \|V_{kh + s}\|_{L^{2}}\\
    &\leq \sqrt{2M}\|z_k-Z^{k}\|_{L^{2},a,b} + \sqrt{d},
\end{align*}
and for any $0\leq s \leq h$ we have
\begin{align*}
    \|\nabla U(X^{s}_{k})\|_{L^{2}} &\leq \|\nabla U(X^{s}_{k})-\nabla U(X_{kh + s})\|_{L^{2}} + \|\nabla U(X_{kh + s})\|_{L^{2}}\\
    &\leq M\|X^{s}_{k}-X_{kh + s}\|_{L^{2}} + \sqrt{Md}\\
    &\leq \sqrt{2}M\|z_k - Z^{k}\|_{L^{2},a,b} + \sqrt{Md},
\end{align*}
where we have used contraction of the continuous dynamics under synchronous coupling provided in Corollary \ref{supp:CorContinuousContraction} and \cite{dalalyan2017further}[Lemma 2] to bound $\|\nabla U(X_{kh + s})\|_{L^{2}}$. Therefore we have the following bound on (II)
\begin{align*}
    \textnormal{(II)} &\leq \frac{h^{2}}{2}\sqrt{\Theta}\max_{0\leq i < k}\sqrt{ A_{i}} + 4h\sqrt{M}\|z_k-Z^{k}\|_{L^{2},a,b} +  \frac{9h\sqrt{d}}{2},
\end{align*}
where $h<\frac{1}{5\sqrt{M}}$ due to the fact that $\gamma^{2} \geq 8M$ and $h < \frac{1}{2\gamma}$.

Next, we consider (I) and we can estimate
\begin{align*}
    \textnormal{(I)} &\leq \frac{h}{2}\|v_k - v_{k-1}\|_{L^{2}} \leq \frac{h}{2}\|v_k - v^{h}_{k-1}\|_{L^{2}} + \frac{h}{2}\|v^{h}_{k-1} - v_{k-1}\|_{L^{2}}\\
    &\leq \frac{h^{2}}{2}\sqrt{\Theta}\max_{0\leq i < k}\sqrt{A_{i}} + \frac{h}{2}\|v^{h}_{k-1} - v_{k-1}\|_{L^{2}},
\end{align*}
where
\[
\|v^{h}_{k-1} - v_{k-1}\|_{L^{2}} \leq \|v^{h}_{k-1} - V^{h}_{k-1}\|_{L^{2}} + \|V^{h}_{k-1} - v_{k-1}\|_{L^{2}}.
\]
Then we can bound
\begin{align*}
    &\|v^{h}_{k-1} -V^{h}_{k-1}\|_{L^{2}} \leq \\
    &\left\|\int^{h}_{0}\mathcal{E}(h/2)\left(\nabla U(X^{s}_{k-1})ds - \nabla U(\overline{x}_{k-1})\right) + \left(\mathcal{E}(h-s) - \mathcal{E}(h/2)\right)\nabla U(X^{s}_{k-1})ds \right\|_{L^{2}}\\
    &\leq M\int^{h}_{0}\left\|X^{s}_{k-1} -\overline{x}_{k-1}\right\|_{L^{2}}ds+ h\max_{0\leq s \leq h}\|\nabla U(X^{s}_{k-1}) \|_{L^{2}}\\
    &\leq \frac{3}{50}\max_{0\leq s \leq h}\|V^{s}_{k-1}\|_{L^{2}} + 2\frac{51}{50}h\max_{0\leq s \leq h}\|\nabla U(X^{s}_{k-1}) \|_{L^{2}} + \frac{2}{25}\sqrt{d},
\end{align*}
where we have used the estimate of $\int^{h}_{0}\|X^{s}_{k-1}-\overline{x}_{k-1}\|_{L^{2}}$ from the $\|x^{h}_{k-1}-X^{h}_{k-1}\|_{L^{2}}$ bound and the fact that $h < 1/5\sqrt{M}$. Using \eqref{supp:eq:cont_v} we have
\begin{align*}
    \|V^{h}_{k-1} - v_{k-1}\|_{L^{2}} &\leq \left\|(\mathcal{E}(h)-1)v_{k-1} - \int^{h}_{0}\mathcal{E}(h-s)\nabla U(X^{s}_{k-1})ds + \sqrt{2\gamma}\int^{h}_{0}\mathcal{E}(h-s)dW_{s}\right\|_{L^{2}}\\
    &\leq h\gamma\|v_{k-1}\|_{L^{2}} + h\max_{0\leq s\leq h}\|\nabla U(X^{s}_{k-1})\|_{L^{2}} +  \sqrt{2\gamma h d}\\
    &\leq 2h\sqrt{M}\left(\gamma + \sqrt{M}\right)\|z_{k-1} - Z^{k-1}\|_{L^{2},a,b} + h\sqrt{d}\left(\gamma + \sqrt{M}\right) + \sqrt{2\gamma h d},
\end{align*}
and we can combine terms to get the following bound on (I)
\begin{align*}
    \textnormal{(I)} &\leq  \frac{h^{2}}{2}\max_{0\leq i < k}\sqrt{\Theta A_{i}} + 3h\sqrt{M}\|z_{k-1} - Z^{k-1}\|_{L^{2},a,b} + \frac{5}{4}h\sqrt{d},
\end{align*}
and summing all terms we have that
\[
\|\overline{x}_{k+1}-\overline{x}_{k}\|_{L^{2}} \leq  h^{2}\max_{0\leq i < k}\sqrt{\Theta A_{i}} + 7h\sqrt{M}\|z_{k-1} - Z^{k-1}\|_{L^{2},a,b} + 6h\sqrt{d}.
\]
\end{proof}

\begin{proposition}\label{supp:prop:non-asymptotic-SG}
For a stochastic gradient $\UBU$ integrator with iterates $(z_{k})_{k \in \mathbb{N}}$, gradient evaluation points $(\overline{x}_{k})_{k \in \mathbb{N}}$, transition kernel $P_h$ and potential $U$ satisfying Assumptions \ref{supp:assum:LipSG}-\ref{supp:assum:convexSG}, where we approximate the gradient using a unbiased stochastic gradient $(\mathcal{G},\rho)$ satisfying
\[
\mathbb{E}\left(\left\|\mathcal{G}(\overline{x}_{k},\omega_{k+1}|\hat{x}_k) - \nabla U(\overline{x}_{k})\right\|^2\right) \leq \Theta\max_{j < k} \mathbb{E}\|\overline{x}_{j+1} - \overline{x}_{j}\|^{2}.
\]

Consider continuous kinetic Langevin dynamics initialized at the invariant measure $Z_{0} \sim \pi$, for $k \in \mathbb{N}$ define $Z^{k}:= Z_{kh} = \phi(Z_{0},kh,(W_{t'})^{kh}_{t'=0}) \in \R^{2d}$ with synchronously coupled Brownian motion to $(z_{k})_{k \in \mathbb{N}}$, then for all 
\[h < \min\left\{\frac{1}{2\Tilde{\gamma}N_D^{1/2}},1/2,{\frac{\tilde{m}^{1/3}N_{D}^{1/6}}{24(\Theta \Tilde{\gamma} )^{1/3}}},\frac{1}{4\Theta^{1/4}},\frac{\tilde{m}}{256\tilde{M}\Tilde{\gamma}N_D^{1/2}}\right\},\]
we have
\begin{align*}
    &\|z_{k}-Z^{k}\|_{L^{2},a,b} \leq  4(1-R_{2}(h)/2)^{k}\left(\|z_{0} - Z^{0}\|_{L^{2},a,b} + \sqrt{2}C_{1}h^{5/2}\right) \\
    &+ \frac{24\sqrt{2\gamma}h^{2}\left(C_{2}\sqrt{\gamma}+C_{1}\sqrt{m}\right)}{m}+ C(\Tilde{\gamma},\tilde{m},\tilde{M},\tilde{M}^{s}_{1}) h^{3/2} d^{1/2}N_D^{-3/4}\Theta^{1/2},
\end{align*}
where $R_{2}(h) = 1-c_{2}(h) + C^{2}_{0}h^{2}$.

Further for  all $\mu\in \mathcal{P}_{2}(\R^{2d})$, and  all $k \in \mathbb{N}$,
\begin{align*}
&\mathcal{W}_{2,a,b}(\mu P^{k}_{h},\pi)\leq  4(1-R_{2}(h)/2)^{k}\left(\mathcal{W}_{2,a,b}(\mu,\pi) + \sqrt{2}C_{1}h^{5/2}\right) \\
&+ \frac{24\sqrt{2\gamma}h^{2}\left(C_{2}\sqrt{\gamma}+C_{1}\sqrt{m}\right)}{m}+ C(\Tilde{\gamma},\tilde{m},\tilde{M},\tilde{M}^{s}_{1}) h^{3/2} d^{1/2}N_D^{-3/4}\Theta^{1/2}.
\end{align*}
\end{proposition}

\begin{proof}

Let us define the notation $Z^{t}_{k} = (X^{t}_{k},V^{t}_{k}) := \phi(z_{k},t,(W_{t'})^{kh+t}_{t'=kh}) \in \R^{2d}$ for $k \in \mathbb{N}$ and $t \geq 0$ to be the continuous dynamics solution with initial condition $(x_k,v_k)$ at time $t$ defined by \eqref{supp:eq:cont_v} and \eqref{supp:eq:cont_x}. Further define $z^{h}_{k} = (x^{h}_{k},v^{h}_{k}):=\psi_{h}(z_{k},h,(W_{t'})^{(k+1)h}_{t'=kh})$ is an iterate with stepsize $h$ and initial point $(x_{k},v_k)$ of the full gradient $\UBU$ scheme and synchronously coupled Brownian motion to the stochastic gradient scheme.

Firstly, we split up the difference in the following way
\begin{align*}
    &\|z_{k} - Z^{k}\|^{2}_{L^{2},a,b}  = \left\|\left(z_{k} - z^{h}_{k-1}\right) + \left(z^{h}_{k-1} - Z^{k}\right)\right\|^{2}_{L^{2},a,b}\\
    &= \left\|z_{k} - z^{h}_{k-1}\right\|^{2}_{L^{2},a,b} + 2\left\langle z_{k} - z^{h}_{k-1},z^{h}_{k-1} - Z^{k}\right\rangle_{L^{2},a,b} + \|z^{h}_{k-1} - Z^{k}\|^{2}_{L^{2},a,b}.
\end{align*}
Considering the inner product we have the expectation conditional on $z_{k-1}$ and $(W_{t'})^{kh}_{t'=(k-1)h}$ is zero as it is independent of the Brownian motion (due to synchronous coupling) and the stochastic gradient estimator is unbiased. Therefore
\begin{align*}
    \|z_{k} - Z^{k}\|^{2}_{L^{2},a,b}  &\leq \|z_{k} - z^{h}_{k-1}\|^{2}_{L^{2},a,b} + (\|\beta_{k-1}\|_{L^{2},a,b}\\
    &+ \|z^{h}_{k-1} - \psi(Z^{k-1},h,(W_{t'})^{kh}_{t' = (k-1)h})  + \alpha_{k-1}\|_{L^{2},a,b})^{2} \\
    &= \textnormal{(I)}' + \textnormal{(II)}'.
\end{align*}
We have that
\[
\textnormal{(I)}' \leq \frac{2h^{2}\Theta}{M}\max_{j<k}\| \overline{x}_{j+1} - \overline{x}_{j}\|^{2}_{L^{2}},
\]
and
\[
\textnormal{(II)}' \leq \left(\sqrt{(1-c_{2}(h) + C_{0}^{2}h^{2})\| z_{k-1} -Z^{k-1}\|^{2}_{L^{2},a,b} + 2C^{2}_{1}h^{5}}+C_{2}h^{3}\right)^{2}.
\]
Let $R_{2}(h) := c_{2}(h) - C_{0}^{2}h^{2}$, then assuming that $mh/8\gamma < R_{2}(h) < 1/2$ (which holds for $h<\frac{1}{2\gamma}$ and $h < \frac{m}{256\gamma M}$), using Lemma \ref{supp:lemma:recursion},
\begin{align*}
    \|z_{k} - Z^{k}\|_{L^{2},a,b}  &\leq \sqrt{2}(1-R_{2}(h)/2)^{k}\left(\|z_{0} - Z^{0}\|_{L^{2},a,b} + \sqrt{2C^{2}_{1}h^{5}}\right) + \frac{2\sqrt{2}C_{2}h^{3}}{R_{2}(h)}\\
    &+ 2\sqrt{\frac{\frac{2h^2 \Theta}{M}\max_{j<k}\|\overline{x}_{j+1}-\overline{x}_{j}\|^{2}_{L^{2}} + 2C^{2}_{1}h^{5}}{R_{2}(h)}},
\end{align*}
and now we wish to bound $\max_{j<k}\|\overline{x}_{j+1}-\overline{x}_{j}\|^{2}_{L^{2}}$. Considering Lemma \ref{supp:lem:consecutive_y}
we have that
\begin{align*}
    &\|\overline{x}_{k+1}-\overline{x}_{k}\|_{L^{2}} \leq h^{2}\sqrt{\Theta\max_{0\leq j \leq k}\|\overline{x}_{j+1} - \overline{x}_{j}\|^{2}_{L^{2}}} + 6h\sqrt{d}\\
    &+ 7h\sqrt{M}\left(\sqrt{2}(1-R_{2}(h)/2)^{k}\left(\|z_{0} - Z^{0}\|_{L^{2}} + \sqrt{2C^{2}_{1}h^{5}}\right) + \frac{16\sqrt{2}\gamma C_{2}h^{2}}{m}\right)\\
    &+ 56h\sqrt{\frac{h\gamma\Theta\max_{j<k}\|\overline{x}_{j+1} -\overline{x}_{j}\|^{2}_{L^{2}} + \gamma C^{2}_{1}h^{4}}{m}}.
\end{align*}
If we assume that $$h < \min\Big\{\frac{m^{1/3}}{24(\Theta\gamma )^{1/3}},\frac{1}{4\Theta^{1/4}},1\Big\},$$ then
\begin{align*}
    \max_{j<k}\|\overline{x}_{j+1}- \overline{x}_{j}\|_{L^{2}} &\leq
     21h\sqrt{M}\left(\sqrt{2}(1-R_{2}(h)/2)^{k}\left(\|z_{0} - Z^{0}\|_{L^{2}} + \sqrt{2C^{2}_{1}h^{5}}\right)+ \frac{16\sqrt{2}\gamma C_2 h^2}{m}\right) \\&+  18h\sqrt{d} + 168h^{3}\sqrt{M}C_{1}\sqrt{\frac{\gamma}{m}},
\end{align*}
and  
\begin{align*}
    &\|z_{k} - Z^{k}\|_{L^{2},a,b}  \leq \sqrt{2}(1-R_{2}(h)/2)^{k}\left(\|z_{0} - Z^{0}\|_{L^{2},a,b} + \sqrt{2}C_{1}h^{5/2}\right) \\&+ \frac{16\sqrt{2\gamma}h^{2}\left(C_{2}\sqrt{\gamma}+C_{1}\sqrt{m}\right)}{m} 
    + 2\sqrt{\frac{16 h\gamma \Theta\max_{j<k}\|\overline{x}_{j+1}-\overline{x}_{j}\|^{2}_{L^{2}}}{mM}}\\
    &\leq 4(1-R_{2}(h)/2)^{k}\left(\|z_{0} - Z^{0}\|_{L^{2},a,b} + \sqrt{2}C_{1}h^{5/2}\right) + \frac{24\sqrt{2\gamma}h^{2}\left(C_{2}\sqrt{\gamma}+C_{1}\sqrt{m}\right)}{m} \\
    &+ 144h^{3/2}\left(\frac{d \gamma \Theta}{mM}\right)^{1/2} + 1344h^{7/2}C_{1}\frac{\gamma\sqrt{\Theta}}{m},
\end{align*}
and the first claim follow by rewriting this bound in terms of $\tilde{m}$, $\tilde{M}$ and $\tilde{\gamma}$. 
For non-asymptotic Wasserstein results, we simply replace $Z^{k-1}$ with the continuous dynamics initialized at $\tilde{Z}_{k-1}\sim \pi$ be such that $\|\tilde{Z}_{k-1}-z_{k-1}\|_{L^2,a,b}=\mathcal{W}_{2,a,b}(\mu P^{k-1}_{h},\pi)$
as in \cite{sanz2021wasserstein}[Theorem 23]. We can then apply Lemma \ref{supp:lemma:recursion} to get the required result.
\end{proof}

\begin{remark}
    To get the non-asymptotic result to have discretization error which is of order $\O(h^{3/2})$, the gradient approximation needs to be an unbiased estimator of the gradient, without this property the discretization error reduces to order $\O(h)$.
\end{remark}

\begin{lemma}\label{supp:lem:diffusions_difference}
    Suppose we have two kinetic Langevin diffusions with synchronously coupled Brownian motion, $(Z_{t})_{t\geq 0}$ with potential $U$ satisfying Assumptions \ref{supp:assum:LipA}-\ref{supp:assum:Hess_LipschitzSG} and $(\tilde{Z}_{t})_{t\geq 0}$ with potential defined by \eqref{supp:eq:hessian_grad_approx_star}, a Gaussian approximation of $U$. We further assume that the diffusions are initialized at their invariant measures and $\gamma \geq \sqrt{8M}$. We then have that
    \begin{align*}
         \|Z_{t}-\Tilde{Z}_{t}\|_{L^{2},a,b} &\leq e^{-\frac{mt}{16\gamma}}\|Z_{0}-\Tilde{Z}_{0}\|_{L^{2},a,b} + \frac{dC(\tilde{\gamma},\Tilde{m},\tilde{M},\tilde{M}^{s}_{1})}{N_{D}}.
    \end{align*}
\end{lemma}
\begin{proof}
    We define $(X_{t},V_{t})_{t\geq 0}:=(Z_{t})_{t \geq 0}$ to be the diffusion according to \eqref{supp:eq:cont_v}-\eqref{supp:eq:cont_x} with potential $U$ and define $(\Tilde{X}_{t},\Tilde{V}_{t})_{t\geq 0}:=(\Tilde{Z}_{t})_{t\geq0}$ to be the diffusion according to \eqref{supp:eq:cont_v}-\eqref{supp:eq:cont_x} with potential $\tilde{U}$, defined by \eqref{supp:eq:hessian_grad_approx_star}, and synchronously coupled Brownian motion. By the same argument as Corollary \ref{supp:CorContinuousContraction} with the expectations rather than Wasserstein distance we have that for $h >0$,
\begin{align*}
    &\|Z_{(k+1)h} - \tilde{Z}_{(k+1)h}\|_{L^{2},a,b} 	\leq \left(1-\frac{mh}{16\gamma}\right)\|Z_{kh} - \Tilde{Z}_{kh}\|_{L^{2},a,b}\\&+ h\left\|\left(0,\nabla\Tilde{U}(\Tilde{X}_{kh})-\nabla{U}(\Tilde{X}_{kh})\right)\right\|_{L^{2},a,b} + C(\gamma,M,d)h^{2}\\
    &\leq \left(1-\frac{mh}{16\gamma}\right)\|Z_{kh} - \Tilde{Z}_{kh}\|_{L^{2},a,b}+ \frac{hM^{s}_{1}}{\sqrt{M}}\left\|\Tilde{X}_{kh}-x^{*}\right\|^{2}_{L^{4}}+ C(\gamma,M,d)h^{2},
\end{align*}
where the last inequality is due to Lemma \ref{supp:lem:mean:minimizer:L4:1} and $x^{*}$ is the minimizer of $U$ and $\tilde{U}$. Then due to Proposition \ref{supp:prop:full_gradient_L4} taking the limit as $h \to 0$,
\begin{align*}
\|\Tilde{X}_{kh}-x^{*}\|^{2}_{L^{4}} 
&\leq \frac{2}{m}\left[\sqrt{4\left(1-\frac{h\lambda\gamma}{2}\right)^{k}(\gamma^{4}\|\Tilde{X}_{0}-x^{*}\|^{4}_{L^{4}} + \|\Tilde{V}_{0}\|^{4}_{L^{4}}) + \frac{(6 d + 160 (1+\lambda^{2}))^{2}}{2\lambda^{2}}}\right]
\intertext{using Lemma \ref{supp:lem:moment_bounds}, }
&\leq \frac{2}{m}\left[\sqrt{\frac{12\gamma^{4}d^{2}}{m^{2}} + 12d^{2} + \frac{(6 d + 160 (1+\lambda^{2}))^{2}}{2\lambda^{2}}}\right],
\end{align*}
where $\lambda=\min\left(\frac{1}{4},\frac{m}{\gamma^2}\right)$ is defined as in \eqref{supp:eq:lambdadef}. We choose $k = t/h$ and define $$c_{u} :=\frac{2M^{s}_{1}}{m\sqrt{M}}\left[\sqrt{\frac{12\gamma^{4}d^{2}}{m} + 12d^{2} + \frac{(6 d + 160 (1+\lambda^{2}))^{2}}{2\lambda^{2}}}\right],$$ then we have
\begin{align*}
    &\limsup_{h\to 0}\frac{\|Z_{t+h} - \tilde{Z}_{t+h}\|_{L^{2},a,b} - \|Z_{t} - \Tilde{Z}_{t}\|_{L^{2},a,b}}{h} \leq -\frac{m}{16\gamma}\|Z_{t} - \Tilde{Z}_{t}\|_{L^{2},a,b} + c_{u}.
\end{align*}
All terms are bounded on the right-hand side due to the assumptions on the initial condition, therefore due to the Denjoy–Young–Saks theorem we have the upper Dini derivative (upper right-hand derivative) is finite. Hence considering $u: \mathbb{R} \to \mathbb{R}$ to be solution to the ODE
\begin{align*}
    &\frac{d}{dt}u(t) = -\frac{m}{16\gamma}u(t) + c_{u},
\end{align*}
with initial condition $u(0) = \|Z_{0}-\Tilde{Z}_{0}\|_{L^{2},a,b}$ which we can solve exactly. Therefore by the comparison principle for ODEs and Dini derivatives \cite{Khalil2002}[Lemma 3.4] we have
\begin{align*}
    \|Z_{t}-\Tilde{Z}_{t}\|_{L^{2},a,b} &\leq u(t) \leq e^{-\frac{mt}{16\gamma}}\|Z_{0}-\Tilde{Z}_{0}\|_{L^{2},a,b} + \frac{16\gamma}{m}c_{u}\\
    &\leq e^{-\frac{mt}{16\gamma}}\|Z_{0}-\Tilde{Z}_{0}\|_{L^{2},a,b} + \frac{dC(\tilde{\gamma},\Tilde{m},\tilde{M},\tilde{M}^{s}_{1})}{N_{D}},
\end{align*}
as required. 
\end{proof}
\begin{proposition}
\label{supp:prop:var_sgubu}
Suppose two stochastic gradient $\UBU$ chains at coarser and finer discretization levels $l$ and $l+1$, with synchronously coupled Brownian motions $\left(z_{k}\right)_{k\in \mathbb{N}}$ and $\left(z'_{k}\right)_{k \in \mathbb{N}}$ and stepsizes $h_{l}$ and $h_{l+1} = h_{l}/2$, satisfying the conditions of Proposition \ref{supp:prop:non-asymptotic-SG}, be such that $z_{0} \sim \pi_{0}$ and $z'_{0} \sim \pi'_{0}$. Then for $f$ satisfying Assumption \ref{supp:assum:Lipschitz} we have the following variance bound
\begin{align*}
    &\textnormal{Var}\left(f(z'_{k}) - f(z_{k})\right)\leq\mathbb{E}\left(f(z'_{k}) - f(z_{k})\right)^2 \leq  \\
        &\Bigg(\exp{\left(-\frac{mkh_{l}}{8\gamma}\right)}\left( \|z'_{0}-z_{0}\|_{L^{2},a,b} + \mathcal{W}_{2,a,b}(\pi_{0},\pi) + \mathcal{W}_{2,a,b}(\pi'_{0},\pi)\right)\\
    &+ 4(1-R(h_{l})/2)^{k}\left(\mathcal{W}_{2,a,b}(\pi_{0},\pi) + \sqrt{2}C_{1}h_{l}^{5/2}\right)\\
&+ 4(1-R(h_{l+1})/2)^{2k}\left(\mathcal{W}_{2,a,b}(\pi'_{0},\pi) + \sqrt{2}C_{1}h_{l+1}^{5/2}\right)\\
&+ \frac{36\sqrt{2\gamma}h^{2}_{l}\left(C_{2}\sqrt{\gamma}+C_{1}\sqrt{m}\right)}{m} +C(\Tilde{\gamma},\tilde{m},\tilde{M},\tilde{M}^{s}_{1}) d^{1/2} N_D^{-3/4}\Theta^{1/2} h_l^{3/2} \Bigg)^{2}.
\end{align*}
\end{proposition}
\begin{proof}
By following the same argument as Proposition \ref{supp:prop:global_strong_error} using Proposition \ref{supp:prop:non-asymptotic-SG} we have the desired result.
\end{proof}
\begin{restatable}{proposition}{prop:SVRG_Dl_l+1}\label{supp:prop:SVRG_Dl_l+1}
Suppose that the assumptions of Proposition \ref{supp:prop:non-asymptotic-SG} hold for the potential $U$, $h_{0} > 0$ and $\gamma > 0$ and the SVRG stochastic gradient approximation. Assume that $$h_{0} \leq \frac{C(\Tilde{\gamma},\Tilde{m},\Tilde{M},\tau/N_{D},N_{b})}{N^{3/2}_{D}}, \quad B\ge \frac{16\log(2^{3/2})\Tilde{\gamma}}{\Tilde{m}h_0 N^{1/2}_{D}}, \quad 
B_0\ge \frac{16\Tilde{\gamma}}{\Tilde{m}h_0 N^{1/2}_{D}}\log\left(\frac{1}{N^{9/4}_{D} h_0^{3/2}}\right),$$ {and the levels are initialized as described in Section \ref{supp:sec:initial_conditions}}. Then for every $l\ge 1$, $1\le k\le K$, for a test function $f$ which satisfies Assumption \ref{supp:assum:Lipschitz} the $\UBUBU$ samples satisfy
\begin{align*}
\Var\left(f(z'^{(l,l+1)}_{k}) - f(z^{(l,l+1)}_{k})\right) &\leq\mathbb{E}\left[\left(f(z'^{(l,l+1)}_{k}) - f(z^{(l,l+1)}_{k})\right) ^2 \right]\\ &\le \mathbb{E}\|z'^{(l,l+1)}_{k} - z^{(l,l+1)}_{k}\|^{2}_{a,b}\\
&\le  C(\Tilde{\gamma},\Tilde{m},\Tilde{M},\Tilde{M}^{s}_{1},\tau/N_{D},N_{b})N^{{5}/2}_{D}h_l^3 d,
\end{align*}
and further
\begin{equation}
    \Var(D_{l,l+1}) \leq C(\Tilde{\gamma},\Tilde{m},\Tilde{M},\Tilde{M}^{s}_{1},\tau/N_{D},N_{b})N^{{5}/2}_{D}h_l^3 d.
\end{equation}
\end{restatable}
\begin{proof}[Proof of Proposition \ref{supp:prop:SVRG_Dl_l+1}]
Following a similar proof as Corollary \ref{supp:corollary:global_strong_error}. However, we need to be careful with the bounding the distance between $z_{0}$ and $z'_{0}$, which is the reason for our construction of the initial conditions in Section \ref{supp:sec:initial_conditions} using the OHO scheme. In particular, we wish to have at most $\mathcal{O}(1/N_{D})$ distance in initialization. We define $(Z_{t})_{t\geq 0}$ to be defined by continuous kinetic Langevin dynamics
with the potential $U$ and $Z_{0} \sim \pi$ such that $\|Z_{0} - z'_{0}\|_{L^{2},a,b} = \mathcal{W}_{2,a,b}(\pi,\mu_{G})$ are optimally coupled. We define $(Z^{\mathcal{G}}_{t})_{t\geq 0}$ to be defined by the continuous kinetic Langevin dynamics with the potential being a Gaussian approximation of the potential such that $Z^{\mathcal{G}}_{0} = z'_{0} \sim \mu_{G}$, $(z^{\mathcal{G}}_{t})_{t\geq 0}$ to be the OHO scheme with the potential being a Gaussian approximation of the potential such that $z^{\mathcal{G}}_{0} = z'_{0} \sim \mu_{G}$. We therefore have $z^{\mathcal{G}}_{B} = z_{0}$ and
\begin{align*}
    \|z_{0} - z'_{B/h_{l+1}}\|_{L^{2},a,b} &\leq \|z^{\mathcal{G}}_{B} - Z^{\mathcal{G}}_{B} \|_{L^{2},a,b} + \| Z^{\mathcal{G}}_{B} - Z_{B} \|_{L^{2},a,b} + \|Z_{B} - z'_{B/h_{l+1}} \|_{L^{2},a,b}\\
    &\leq h\sqrt{d}C(\Tilde{\gamma},\Tilde{m},\Tilde{M}) + \frac{dC(\Tilde{\gamma},\Tilde{m},\Tilde{M},\Tilde{M}^{s}_{1})}{N_{D}}\\
    &+ 5\|Z_{0} - z'_{0}\|_{L^{2},a,b} + C(\Tilde{\gamma},\Tilde{m},\Tilde{M},\Tilde{M}^{s}_{1},\tau/N_{D},N_{b})N^{{5}/2}_{D}h_l^3 d\\
    &\leq \frac{dC(\Tilde{\gamma},\Tilde{m},\Tilde{M},\Tilde{M}^{s}_{1},\tau/N_{D},N_{b})}{N_{D}},
\end{align*}
where we have used Theorem \ref{supp:theorem:OHO_strongerror} for the first term, Lemma \ref{supp:lem:diffusions_difference} for the second and Proposition \ref{supp:prop:non-asymptotic-SG} for the third.

We also have the following rough estimates for the Wasserstein distances
\begin{align*}
    \mathcal{W}_{2,a,b}(\mu^{(l+1)}_{0},\pi) = \mathcal{W}_{2,a,b}(\pi_{0},\pi)&=\mathcal{W}_{2,a,b}(\mu_{G},\pi)\leq \frac{dC(\Tilde{\gamma},\Tilde{m},\Tilde{M},\Tilde{M}^{s}_{1})}{N_{D}},
\end{align*}
which follow from Proposition \ref{supp:prop:wass_initial_distance}, where the estimate of $\mathcal{W}_{2,a,b}(\mu^{(l+1)}_{0},\pi)$ along with Proposition \ref{supp:prop:non-asymptotic-SG} implies $\mathcal{W}_{2,a,b}(\pi'_{0},\pi) \leq \frac{dC(\Tilde{\gamma},\Tilde{m},\Tilde{M},\Tilde{M}^{s}_{1},\tau/N_{D},N_{b})}{N_{D}}$.

We have $(B_0+Bl)2^l$ burn-in steps at level $l$, and $(B_0+B(l+1))2^{l+1}$ burn-in steps at level $l+1$. 
Using the assumption that $h_{0} < \frac{\Tilde{m}}{256\Tilde{M}\Tilde{\gamma}N^{1/2}_{D}}$, we have for all $i \in \mathbb{N}$
$R(h_i)\ge \frac{mh_{i}}{8\gamma}$, and using Proposition \ref{supp:prop:non-asymptotic-SG} we have
\begin{align*}
   &\Var\left(f(z'^{(l,l+1)}_{k}) - f(z^{(l,l+1)}_{k})\right) \leq\mathbb{E}\left[\left(f(z'^{(l,l+1)}_{k}) - f(z^{(l,l+1)}_{k})\right) ^2\right]\\
   &\leq \Bigg(\exp\left(-\frac{\Tilde{m}\sqrt{N_{D}}(B_0+l B)h_0}{8\Tilde{\gamma}}\right) \left(\|z'_{B/h_{l+1}}-z_{0}\|_{L^{2},a,b}+\mathcal{W}_{2,a,b}(\pi_0,\pi)+\mathcal{W}_{2,a,b}(\pi_0',\pi)\right)
        \\
        &+4\exp\left(-\frac{\Tilde{m}\sqrt{N_{D}}(B_0+l B)h_0}{16\Tilde{\gamma}}\right)  (\mathcal{W}_{2,a,b}(\pi_0,\pi)+ \mathcal{W}_{2,a,b}(\pi_0',\pi)) \\
        &+ C(\Tilde{\gamma},\Tilde{m},\Tilde{M},\Tilde{M}^{s}_{1},\tau/N_{D},N_{b})h^{3/2}_{l}d^{1/2}N^{{5}/4}_{D}\Bigg)^2\\
   & \le \Bigg(\exp\left(-\frac{m(B_0+lB)h_0}{16\gamma}\right) \frac{dC(\Tilde{\gamma},\Tilde{m},\Tilde{M},\Tilde{M}^{s}_{1},\tau/N_{D})}{N_{D}}+ C(\Tilde{\gamma},\Tilde{m},\Tilde{M},\Tilde{M}^{s}_{1},\tau/N_{D},N_{b})h^{3/2}_{l}d^{1/2}N^{{5}/4}_{D}\Bigg)^2\\
        \intertext{using the assumptions on $B_0$ and $B$}
    &\le C(\Tilde{\gamma},\Tilde{m},\Tilde{M},\Tilde{M}^{s}_{1},\tau/N_{D},N_{b})N^{{5}/2}_{D}h_l^3 d. 
\end{align*}
We now use the simple bound
\begin{align*}
    \Var(D_{l,l+1}) &\leq \mathbb{E}(D^{2}_{l,l+1}) \leq \max_{1 \leq k \leq K}\mathbb{E}\left[\left(f(z'^{l,l+1}_{k}) - f(z^{(l,l+1)}_{k}\right)^{2}\right]\\
    &\leq C(\Tilde{\gamma},\Tilde{m},\Tilde{M},\Tilde{M}^{s}_{1},\tau/N_{D},N_{b})N^{{5}/2}_{D}h_l^3 d
\end{align*}
as required.
\end{proof}

\begin{remark}
As an alternative, one can consider a coupling with a randomized midpoint scheme, which was utilized in the work of \cite{Hu21optimal} and \cite{bou2022unadjusted} in the context of kinetic Langevin dynamics and Hamiltonian Monte Carlo. This is beyond the scope of the work, and thus we leave this as a direction to consider for future work.
\end{remark}


\begin{restatable}{proposition}{propSGKLDCVVariance}
\label{supp:prop:SGKLDCVVariance}
Suppose a full gradient Gaussian approximation OHO chain $\left(z_{k}\right)_{k \in \mathbb{N}}$ at level $0$ and a stochastic gradient $\UBU$ chain $\left(z'_{k}\right)_{k \in \mathbb{N}}$  at level $1$ using the SVRG unbiased estimator, with stepsizes $h_0$ and $h_{1}=\frac{h_0}{2}$, respectively. Further, we assume that they have synchronously coupled Brownian motions and $z_{0} \sim \pi_{0} = \mu_{G}$ and $z'_{0} \sim \pi'_{0}$. Assuming the same assumptions as Proposition \ref{supp:prop:global_strong_error} for $\left(z_{k}\right)_{k \in \mathbb{N}}$ and Proposition \ref{supp:prop:non-asymptotic-SG} for $\left(z'_{k}\right)_{k \in \mathbb{N}}$. Then for $f$ satisfying Assumption \ref{supp:assum:Lipschitz} we have the following variance bound
\begin{align*}
   &\Var\left(f(z'_{k}) - f(z_{k})\right) \leq\mathbb{E}\left[\left(f(z'_{k}) - f(z_{k})\right) ^2\right]\\
   &  \leq \Bigg(\exp{\left(-\frac{\Tilde{m}\sqrt{N_{D}}kh_{0}}{8\Tilde{\gamma}}\right)}\left(\|z'_{0}-z_{0}\|_{L^{2},a,b}+\mathcal{W}_{2,a,b}(\pi_0',\pi)\right) 
        \\
        &+ 4\left(1-R(h_{1})/2\right)^{2k}\mathcal{W}_{2,a,b}(\pi'_{0},\pi) + \frac{dC(\Tilde{\gamma},\Tilde{m},\Tilde{M},\Tilde{M}^{s}_{1})}{N_{D}} +h_{0}\sqrt{d}C(\Tilde{\gamma},\Tilde{m},\Tilde{M})\\
        &+  C(\Tilde{\gamma},\tilde{m},\tilde{M},\tilde{M}^{s}_{1}) d^{1/2} N_D^{-3/4}\Theta^{1/2} h^{3/2}_{1}\Bigg)^2,
\end{align*}
where $R(h_{1}) = 1-\sqrt{(1-c(h_{1}))^{2}+C^{2}_{0}h^{2}_{0}}$.
\end{restatable}
\begin{proof}
By following the same argument as Proposition \ref{supp:prop:global_strong_error}, but by using Proposition \ref{supp:prop:non-asymptotic-SG} and Theorem \ref{supp:theorem:OHO_strongerror} we have the desired result. However, because level zero and level one are approximating different diffusions, we can't use the contraction results for the continuous dynamics to bound $\mathbf{(II)}$, so we consider an alternative argument. For this component, we use Lemma \ref{supp:lem:diffusions_difference}. To bound $\mathbf{(I)}$ we use Theorem \ref{supp:theorem:OHO_strongerror} and to bound $\mathbf{(III)}$ we use Proposition \ref{supp:prop:non-asymptotic-SG} and we have the required result.
\end{proof}

\begin{restatable}{proposition}{prop:CV_SVRG_Dl_l+1}\label{supp:prop:CV_SVRG_Dl_l+1}
Suppose that the assumptions of Proposition \ref{supp:prop:SGKLDCVVariance} hold for the potential $U$ and $h_{0} > 0$. Assume that $$h_{0} \leq \frac{C(\Tilde{\gamma},\Tilde{m},\Tilde{M},\tau/N_{D},N_{b})}{N^{3/2}_{D}}, \quad B\ge \frac{16\log(2^{3/2})\Tilde{\gamma}}{\Tilde{m}h_0 N^{1/2}_{D}}, \quad 
B_0\ge \frac{16\Tilde{\gamma}}{\Tilde{m}N^{1/2}_{D}h_0}\log\left(\frac{1}{N^{9/4}_{D} h_0^{3/2}}\right),$$ $1<k\leq K$, the levels are initialized as described in Section \ref{supp:sec:initial_conditions} and for a function $f$ which satisfies Assumption \ref{supp:assum:Lipschitz} for stochastic gradient $\UBUBU$ we have
\begin{equation}
    \Var(D_{0,1}) \leq \frac{C(\Tilde{\gamma},\Tilde{m},\Tilde{M},\Tilde{M}^{s}_{1}) d^{2}}{N^{2}_{D}} + C(\Tilde{\gamma},\Tilde{m},\Tilde{M},\Tilde{M}^{s}_{1},\tau/N_{D},N_{b})N^{5/2}_{D}h_{0}^3 d.
\end{equation}
\end{restatable}
\begin{proof}
   Following the same proof as Proposition \ref{supp:prop:SVRG_Dl_l+1} using the results of Proposition \ref{supp:prop:SGKLDCVVariance}.
\end{proof}
\subsection{Variance of $S(c_R)$}
\label{supp:subsection:varianceScRSG}

\begin{theorem}
\label{supp:thm:stochgradUBUBU}
Considering UBUBU with stochastic gradients, suppose that Assumptions  
\ref{supp:assum:Lipschitz},  \ref{supp:assum:LipSG}, \ref{supp:assum:convexSG}, \ref{supp:assum:Hess_LipschitzSG} hold, and in addition $\gamma \geq \sqrt{8M}$,
\begin{align*}  h_0 \leq \frac{C(\Tilde{\gamma},\Tilde{m},\Tilde{M},\tau/N_{D},N_{b})}{N^{3/2}_{D}}, \quad B\ge \frac{16\log(2^{3/2})\Tilde{\gamma}}{\Tilde{m}h_0N^{1/2}_{D}}, \quad B_0\ge \frac{16\Tilde{\gamma}}{\Tilde{m}h_0N^{1/2}_{D}}\log\left(\frac{1}{N^{9/4}_{D}h_0^{3/2}}\right).\end{align*}
Suppose that $c_{R}\in [0, \phi_N^{-1/2})$ and $2<\phi_N<8$. Then for any $N\ge 1$, the $\UBUBU$ estimator $S(c_R)$ has finite expected computational cost, $\E S(c_R)=\pi(f)$, and it has finite variance. 
Moreover, it satisfies a CLT as $N\to \infty$, and the asymptotic variance $\sigma^2_S$ defined in \eqref{supp:eq:sigma2S} can be bounded as
\[\sigma^2_S\le \frac{1}{\Tilde{m}N_{D}K} + C(\Tilde{\gamma},\Tilde{m},\Tilde{M},\Tilde{M}^{s}_{1},\tau/N_{D},N_{b},\phi_{N})\frac{d^{2}}{c_{N}N^{2}_{D}}.
\]
\end{theorem}

\begin{proof}
By Propositions \ref{supp:prop:SVRG_Dl_l+1} and \ref{supp:prop:CV_SVRG_Dl_l+1}, we have for $l\geq 1$ that
\[\E(D_{l,l+1}^2)\le C(\Tilde{\gamma},\Tilde{m},\Tilde{M},\Tilde{M}^{s}_{1},\tau/N_{D},N_{b}) h^{3}_{l}dN^{5/2}_{D}\le V_{D_{1}} \phi_{D_{1}}^{-l},\] for $V_{D_{1}}=C(\Tilde{\gamma},\Tilde{m},\Tilde{M},\Tilde{M}^{s}_{1}\tau/N_{D},b) h^{3}_{0}dN^{5/2}_{D}$ and $\phi_{D_{1}} = 8$. 
For $l = 0$ we have
\[\E(D_{l,l+1}^2)\le C(\Tilde{\gamma},\Tilde{m},\Tilde{M},\Tilde{M}^{s}_{1},\tau/N_{D},N_{b}) \frac{d^{2}}{N^{2}_{D}}\le V_{D_{2}} \phi_{D_{2}}^{-l},\] for $V_{D_{2}}= C(\Tilde{\gamma},\Tilde{m},\Tilde{M},\Tilde{M}^{s}_{1},\tau/N_{D},N_{b}) \frac{d^{2}}{N^{2}_{D}}$ and $\phi_{D_{2}} = 2$.

Due to the fact that for $D_{0}$ we take $K$ i.i.d. Gaussian samples, it is easy to show using the Gaussian Poincar\'e inequality that
\begin{align}\nonumber&\Var(D_0)\le \frac{1}{\Tilde{m}N_{D}K}.
\end{align}
The computational cost at levels $l \geq 1$ satisfies the assumptions of Proposition \ref{supp:prop:unb}, so if we fix $2<\phi_{N}<8$, all assumptions of this proposition are satisfied. Hence $S(c_R)$ is an unbiased estimator with finite variance and computational cost.

For the asymptotic variance using \eqref{supp:eq:sigma2S}, and the above estimates  we have
\begin{align*}
\sigma^{2}_{S} &\leq \frac{1}{\Tilde{m}N_{D}} + C(\Tilde{\gamma},\Tilde{m},\Tilde{M},\Tilde{M}^{s}_{1},\tau/N_{D},N_{b}) \frac{d^{2}}{c_{N}N^{2}_{D}} + \sum^{\infty}_{l=1}\frac{V_{D_{1}}\phi^{-l}_{D_{1}}}{c_{N}\phi^{-l}_{N}}\\
&\leq \frac{1}{\Tilde{m}N_{D}} + C(\Tilde{\gamma},\Tilde{m},\Tilde{M},\Tilde{M}^{s}_{1},\tau/N_{D},N_{b},\phi_{N})\left(\frac{d^{2}}{c_{N}N^{2}_{D}} + \frac{h^{3}_{0}dN^{5/2}_{D}}{c_{N}}\right),\\
\intertext{if we choose $h_{0}$ to be of the order $\mathcal{O}(1/{N^{3/2}_{D}})$ then we have}
&\leq \frac{1}{\Tilde{m}N_{D}} + C(\Tilde{\gamma},\Tilde{m},\Tilde{M},\Tilde{M}^{s}_{1},\tau/N_{D},N_{b},\phi_{N})\frac{d^{2}}{c_{N}N^{2}_{D}},
\end{align*}
as required.

\end{proof}

\section{Variance bounds for $\UBUBU$ estimator with approximate gradients}
\label{supp:sec:Appendix_UBUBU_variance_bnd_approx_grad}
One can also approximate the gradient in a cheap way, which has bias, but such that the bias tends to zero with the stepsize. 
The multilevel estimator will still be an unbiased estimator from the target measure.

For convex potentials, we can approximate the gradient with the Hessian at the minimizer by
\begin{equation}\label{supp:eq:hessian_grad_approx}
 \mathcal{Q}(x\mid \hat{x}) = \nabla U(\hat{x}) + \nabla^{2}U(x^{*})(x - \hat{x}).
\end{equation}
Despite the fact that this estimator is biased, in our multilevel approach, the overall estimator will still be unbiased.

As before, the updates in $(\ol{x}_k,\ol{v}_k)_{k\ge 0}$ form a $\mathcal{B}\mathcal{U}$ step, so they can be expressed as
\begin{align}
\label{supp:eq:BUsteps:approxx}
    \overline{x}_{k+1} &= \overline{x}_{k} + \frac{1-\eta^{2}}{\gamma}\left(\overline{v}_{k} - h \mathcal{Q}(\ol{x}_k|\ol{x}_{L(k)})\right) + \sqrt{\frac{2}{\gamma}}\left(\mathcal{Z}^{(1)}\left(h,\xi^{(1)}_{k+1}\right) - \mathcal{Z}^{(2)}\left(h,\xi^{(1)}_{k+1},\xi^{(2)}_{k+1}\right)\right),\\
\label{supp:eq:BUsteps:approxv}
    \overline{v}_{k+1} &= \eta^{2}\left(\overline{v}_{k} - h\mathcal{Q}(\ol{x}_k|\ol{x}_{L(k)})\right) + \sqrt{2\gamma} \mathcal{Z}^{(2)}\left(h,\xi^{(1)}_{k+1},\xi^{(2)}_{k+1}\right),
\end{align}
where $\hat{x}_{k}=\ol{x}_{L(k)}$ and $L(k)=\tau \lfloor k/\tau\rfloor$.

It turns out that at level $0$, it can be advantageous to simply use the gradients of the Gaussian approximation, and never compute gradients of $U$. This corresponds to gradient approximation of the form
\begin{equation}
\mathcal{Q}^*(x) =\mathcal{Q}(x|x^*)= \nabla^{2}U(x^{*})(x - x^*),
\end{equation}
and so \eqref{supp:eq:BUsteps:approxx}-\eqref{supp:eq:BUsteps:approxv} holds with $\hat{x}_{k}=x^*$ for every $k\ge 0$ in this case.

\subsection{Non-asymptotic guarantees}

\begin{lemma}
\label{supp:lem:mean:approx:L4}
Considering iterates $(x_{k},v_{k},\overline{x}_{k})_{k \in \mathbb{N}}$ of approximate gradient $\UBU$, with epoch length $\tau$ and gradient approximation $\mathcal{Q}$ given by \eqref{supp:eq:hessian_grad_approx}, and initial condition $(x_{0},v_{0}) \in \R^{2d}$, we have the property
\[
\mathbb{E}\left\|\nabla U(\overline{x}_{k}) - \mathcal{Q}(\overline{x}_{k}\mid \overline{x}_{L(k)}) \right\|^2 \leq \tilde{M}_1^2 N_D^2    (\tau -1)^2\max_{j\leq k}\|\ol{x}_j-x^*\|_{L^4}^2 \cdot \max_{j < k} \|\overline{x}_{j+1} - \overline{x}_{j}\|_{L^4}^{2},
\]
and we also have
\[
\mathbb{E}\left\|\nabla U(\overline{x}_{k}) - \mathcal{Q}(\overline{x}_{k}\mid x^*) \right\|^2 \leq \tilde{M}_1^2 N_D^2    \|\ol{x}_k-x^*\|_{L^4}^4.
\]

\end{lemma}
\begin{proof}
Let the last full gradient evaluation be at iteration $L(k)$, then
\begin{align*}
&\mathbb{E}\|\nabla U(\overline{x}_{k}) - \mathcal{Q}(\overline{x}_{k}\mid \overline{x}_{L(k)})\|^{2} = \mathbb{E}\|\nabla U(\overline{x}_{k}) - \nabla U(\overline{x}_{L(k)}) - \nabla^{2}U(x^{*})(\overline{x}_{k} - \overline{x}_{L(k)})\|^2\\
&=\mathbb{E}\left\|\left(\int_{t=0}^{1}\nabla^{2}U(\ol{x}_k+t(\ol{x}_{L(k)}-\ol{x}_{k})) \mathrm{d}t - \nabla^{2}U(x^{*})\right)(\overline{x}_{k} - \overline{x}_{L(k)})\right\|^2\\
&\le \tilde{M}_1^2 N_D^2 \mathbb{E}\left(\left(\int_{t=0}^{1}\|\ol{x}_k+t(\ol{x}_{L(k)}-\ol{x}_{k})-x^*\|\right)^2 \left\|\overline{x}_{k} - \overline{x}_{L(k)}\right\|^2\right)\\
&\le \frac{\tilde{M}_1^2 N_D^2}{2}
\left(\|\ol{x}_{k}-x^*\|_{L^4}^2+\|\ol{x}_{L(k)}-x^*\|_{L^4}^2\right)
\left\|\overline{x}_{k} - \overline{x}_{L(k)}\right\|_{L^4}^2\\
&\le \tilde{M}_1^2 N_D^2    (\tau -1)^2\max_{j\leq k}\|\ol{x}_j-x^*\|_{L^4}^2 \cdot \max_{j < k} \|\overline{x}_{j+1} - \overline{x}_{j}\|_{L^4}^{2}.
\end{align*}
The second claim follows by Taylor expansion.
\end{proof}

Now, we are going to bound the terms 
$\|\ol{x}_j-x^*\|_{L^4}$ and $\|\overline{x}_{j+1} - \overline{x}_{j}\|_{L^4}$.

\begin{lemma}\label{supp:lem:approxxzdiff}
When using exact gradients, we have
\begin{align*}
    &\|\ol{x}_{k+1}-\ol{x}_k\|_{L^4}\le 2 h \sqrt{M} \|z_k-z^*\|_{L^4,a,b}+2h\sqrt{d}.
\end{align*}
With approximate gradients, we have
\begin{align*}
    &\|\ol{x}_{k+1}-\ol{x}_k\|_{L^4}\le 
2 h \sqrt{M} \left(\|\ol{z}_k-z^*\|_{L^4,a,b}+\sqrt{2}(1+M/m)\|\ol{z}_{L(k)}-z^*\|_{L^4,a,b}\right)+2h\sqrt{d}.
\end{align*}
\end{lemma}
\begin{proof}
In the case of exact gradients, we have
\begin{align*}
    \overline{x}_{k+1} &= \overline{x}_{k} + \frac{1-\eta^{2}}{\gamma}\left(\overline{v}_{k} - h\nabla U(\overline{x}_{k})\right) + \sqrt{\frac{2}{\gamma}}\left(\mathcal{Z}^{(1)}\left(h,\xi^{(1)}_{k+1}\right) - \mathcal{Z}^{(2)}\left(h,\xi^{(1)}_{k+1},\xi^{(2)}_{k+1}\right)\right),\\
    \overline{v}_{k+1} &= \eta^{2}\left(\overline{v}_{k} - h\nabla U(\overline{x}_{k})\right) + \sqrt{2\gamma} \mathcal{Z}^{(2)}\left(h,\xi^{(1)}_{k+1},\xi^{(2)}_{k+1}\right),
\end{align*}
so using $\|\nabla U(x_k)\|\le M\|x_k-x^*\|$, Lemma \ref{supp:lem:Zcovbnd}, and the fact that $\xi\sim N(0,I_d)$ satisfies that $\E(\|\xi\|^4)\le 3d^2$, we have that for $h\le 1/\sqrt{M}$, $\gamma \ge \sqrt{8M}$,
\begin{align*}
    &\|\ol{x}_{k+1}-\ol{x}_k\|_{L^4}\le h \|\ol{v}_{k}\|_{L^4} + h^2 M \|x_k-x^*\|_{L^4}+3^{1/4}2^{1/2} h \sqrt{d} \le 2 h \sqrt{M} \|z_k-z^*\|_{L^4,a,b}+2h\sqrt{d}.
\end{align*}
For approximate gradients, we have
\begin{align}
    \overline{x}_{k+1} &= \overline{x}_{k} + \frac{1-\eta^{2}}{\gamma}\left(\overline{v}_{k} - h \mathcal{Q}(\ol{x}_k|\ol{x}_{L(k)})\right) + \sqrt{\frac{2}{\gamma}}\left(\mathcal{Z}^{(1)}\left(h,\xi^{(1)}_{k+1}\right) - \mathcal{Z}^{(2)}\left(h,\xi^{(1)}_{k+1},\xi^{(2)}_{k+1}\right)\right),\\
    \overline{v}_{k+1} &= \eta^{2}\left(\overline{v}_{k} - h\mathcal{Q}(\ol{x}_k|\ol{x}_{L(k)})\right) + \sqrt{2\gamma} \mathcal{Z}^{(2)}\left(h,\xi^{(1)}_{k+1},\xi^{(2)}_{k+1}\right).
\end{align}
Let $\tilde{x}_k=x_{L(k)}-(\nabla^{2}U(x^{*}))^{-1} \nabla U(\ol{x}_{L(k)})$, and $\tilde{U}_k(x)=\frac{1}{2}(x-\tilde{x}_k)^T \nabla^{2}U(x^{*})(x-\tilde{x}_k)$. Then
 the approximate gradient step is the same as an exact gradient step with respect to the potential $\tilde{U}_k$. So
we have by the result for exact gradients that for approximate gradients, 
\begin{align*}
    &\|\ol{x}_{k+1}-\ol{x}_k\|_{L^4}\le h \|\ol{v}_{k}\|_{L^4} + h^2 M \|x_k-\tilde{x}_k\|_{L^4}+3^{1/4}2^{1/2} h d 
    \\
    &\le 2 h \sqrt{M} \left\|\ol{z}_k-\left(\begin{matrix}\tilde{x}_k\\ 0_d\end{matrix}\right)\right\|_{L^4,a,b}+2h\sqrt{d}.
\end{align*}
Here using the triangle inequality, we have
\begin{align*}
&\left\|z_k-\left(\begin{matrix}\tilde{x}_k\\ 0_d\end{matrix}\right)\right\|_{L^4,a,b}
\le \|\ol{z}_k-z^*\|_{L^4,a,b}+ \left\|\left(\begin{matrix}\tilde{x}_k\\ 0_d\end{matrix}\right)-z^*\right\|_{L^4,a,b}\\
&\le
\|\ol{z}_k-z^*\|_{L^4,a,b}+
\sqrt{2}(1+M/m)\|\ol{z}_{L(k)}-z^*\|_{L^4,a,b},
\end{align*}
hence 
\begin{align*}
    &\|\ol{x}_{k+1}-\ol{x}_k\|_{L^4}\le
2 h \sqrt{M} \left(\|\ol{z}_k-z^*\|_{L^4,a,b}+\sqrt{2}(1+M/m)\|\ol{z}_{L(k)}-z^*\|_{L^4,a,b}\right)+2h\sqrt{d}.
\end{align*}
\end{proof}

We still need to control the evolution of $\|\ol{z}_k-z^*\|_{L^4,a,b}$.
As before in \eqref{supp:eq:Vxvdef}, we define the Lyapunov function $\mathcal{V}$ as 
\[\mathcal{V}(x,v)=U(x) - U(x^{*}) + \frac{1}{4}\gamma^{2}\left( \|x-x^{*} + \gamma^{-1}v\|^{2} + \|\gamma^{-1}v\|^{2}-\lambda\|x-x^{*}\|^{2}\right).\]
The following lemma establishes some useful properties about this.
\begin{lemma}\label{supp:lemma:V12bnd}
Suppose that $\gamma \ge \sqrt{8M}$, and that Assumptions  \ref{supp:assum:Lip} and \ref{supp:assum:convex} hold for $U$. Then for any $z=(x,v)\in \Lambda$, $\mathcal{V}(x,v)\ge 0$, and 
\begin{equation}\mathcal{V}^{1/2}(x,v)\ge \frac{1}{8} (\gamma \|x-x^{*}\|+\|v\|)\ge \frac{\sqrt{M}}{8} \|z-z^*\|_{a,b}.
\end{equation}
Moreover, $\mathcal{V}^{1/2}$ is $8\gamma$-Lipschitz with respect to the $\|\cdot \|_{a,b}$ norm. 
\end{lemma}
\begin{proof}
Using the strong convexity of $U$, 
\begin{align*}\mathcal{V}(x,v)&=U(x) - U(x^{*}) + \frac{1}{4}\gamma^{2}\left( \|x-x^{*} + \gamma^{-1}v\|^{2} + \|\gamma^{-1}v\|^{2}-\lambda\|x-x^{*}\|^{2}\right)\\
&\ge \frac{m}{2} \|x-x^*\|^2+\frac{1}{4}\gamma^2 \left((1-\lambda)\|x-x^*\|^2+2\gamma^{-2}\|v\|^2+2\left<x-x^*,\gamma^{-1} v\right> \right)\\
\intertext{using that $\left|2\left<x-x^*,\gamma^{-1} v\right>\right|\le \frac{\|x-x^*\|^2}{c}+c\|\gamma^{-1} v\|^2$ with $c=8/5$, and that $0<\lambda\le \frac{1}{4}$,}
&\ge \frac{1}{4}\gamma^2\left(\frac{1}{8}\|x-x^*\|^2+\frac{2}{5}\|\gamma^{-1}v\|^2\right)\ge \frac{1}{64}\left(\gamma \|x-x^*\|+\|v\|\right)^2\ge \frac{M}{6 4}\|z-z^*\|_{a,b}^2,
\end{align*}
and our first claim follows by taking square-root.

For the second claim, note that $\nabla \mathcal{V}^{1/2}(x,v)=\frac{1}{2}\frac{\nabla\mathcal{V}}{\mathcal{V}^{1/2}(x,v)}$. Here
\begin{align*}
&\nabla_{x} \mathcal{V}(x,v)=\nabla U(x)+\frac{1}{2}\gamma^2 ((1-\lambda)(x-x^*)+\gamma^{-1}v),\\
&\|\nabla_{x} \mathcal{V}(x,v)\|\le \left(M+\frac{\gamma^2(1-\lambda)}{2}\right)\|x-x^*\|+\frac{\gamma}{2}\|v\|\le \gamma^2 \|x-x^*\|+\frac{\gamma}{2}\|v\|,\\
&\nabla_{v} \mathcal{V}(x,v)=\frac{1}{2}\gamma^2 (\gamma^{-1}((x-x^*)+\gamma^{-1}v)+\gamma^{-2} v)\\
&\|\nabla_{v} \mathcal{V}(x,v)\|\le \frac{\gamma}{2}\|x-x^*\|+\|v\|,
\end{align*}
so we have
\begin{align*}&\|\nabla_x \mathcal{V}^{1/2}(x,v)\|=\frac{\|\nabla_x \mathcal{V}(x,v)\|}{2 \mathcal{V}^{1/2}(x,v)}\le 4\gamma,\\
&\|\nabla_v \mathcal{V}^{1/2}(x,v)\|=\frac{\|\nabla_v \mathcal{V}(x,v)\|}{2 \mathcal{V}^{1/2}(x,v)}\le 4,
\end{align*}
and since $\gamma\ge \sqrt{8M}$, for any $(x,v),(x',v')\in \Lambda$, we have \begin{align*}&|\mathcal{V}^{1/2}(x,v)-\mathcal{V}^{1/2}(x',v')|=\left<\int_{t=0}^{1}\nabla \mathcal{V}^{1/2}(x+t(x'-x),v+t(x'-v))dt,(x'-x,v'-v)\right>\\
&\le 4\sqrt{2}\gamma \|(x,v)-(x',v')\|_{a,0}\le 8\gamma \|(x,v)-(x',v')\|_{a,b}.\end{align*}
\end{proof}

As previously, let $
    \lambda=\min\left(\frac{1}{4},\frac{m}{\gamma^2}\right)$, and $c_4(h) = h\lambda\gamma -8h^{2}\gamma^{2}\left(4 + \lambda\right)$. By  \eqref{supp:eq:detVsquarebnd}, for the exact gradient scheme, if $c_4(h)<\frac{1}{2}$, we have
\begin{align*}
&\mathbb{E}\left[\mathcal{V}(\overline{x}_{k+1},\overline{v}_{k+1})^{2} \mid \overline{x}_{k},\overline{v}_{k}\right] \leq \left(1-\frac{c_4(h)}{2}\right)\mathcal{V}^{2}(\overline{x}_{k},\overline{v}_{k}) +\frac{(6h\gamma d + 160h\gamma(1+\lambda^{2}))^{2}}{4c_4(h)}  + 24h^{2}\gamma^{2}d^{2}.
\end{align*}
Let $C_{\mathcal{V}}(h):=\frac{(6h\gamma d + 160h\gamma(1+\lambda^{2}))^{2}}{4c_4(h)}  + 24h^{2}\gamma^{2}d^{2}$, then
by applying this $j$ times, we have
\begin{align}\label{supp:eq:Vbndjsteps}
&\mathbb{E}\left[\mathcal{V}(\overline{x}_{k+j},\overline{v}_{k+j})^{2} \mid \overline{x}_{k},\overline{v}_{k}\right] \leq \left(1-\frac{c_4(h)}{2}\right)^j\mathcal{V}^{2}(\overline{x}_{k},\overline{v}_{k})+C_{\mathcal{V}}(h)\frac{1-(c_4(h)/2)^j}{1-c_4(h)/2}.
\end{align}

Now we are going to generalise this result to the approximate gradient scheme.

\begin{lemma}
Consider iterates $(x_{k},v_{k},\overline{x}_{k})_{k \in \mathbb{N}}$ of approximate gradient $\UBU$, with epoch length $\tau$ and gradient approximation $\mathcal{Q}$ given by \eqref{supp:eq:hessian_grad_approx}, and initial condition $(x_{0},v_{0}) \in \R^{2d}$. Suppose that $L(k)=k$ (i.e. $k$ is divisible by $\tau$), and $c_4(h)>0 $, then for any $1\le j\le \tau$, we have
\begin{align*}
&\|\mathcal{V}^{1/2}(\ol{z}_{k+j})\|_{L^4}\leq \left[\left(1-\frac{c_4(h)}{2}\right)^j\|\mathcal{V}^{1/2}(\ol{z}_{k})\|_{L^4}^4+C_{\mathcal{V}}(h)j\right]^{1/4}\\
&+ 8\gamma \left(48h^2 \sqrt{M}\cdot j^2(\|\mathcal{V}^{1/2}(\ol{z}_{k})\|_{L^4}^4+C_{\mathcal{V}}(h)j)^{1/4} + 6h^2 j^2 \sqrt{d M}\right).
\end{align*}
\end{lemma}

\begin{proof}We use an interpolation argument, inspired by the interpolation to independence coupling in \cite{chen2010stein}. For $0\le i\le j$, let $\ol{z}_{k+j}^{(i)}=\left(\ol{x}_{k+j}^{(i)},\ol{v}_{k+j}^{(i)}\right)$ be defined by performing $j-i$ $\mathcal{BU}$ steps with exact gradients starting from $(\ol{x}_{k},\ol{v}_{k})$ according to \eqref{supp:eq:BUsteps:exactx}-\eqref{supp:eq:BUsteps:exactv}, followed by $i$ steps with approximate gradients according to \eqref{supp:eq:BUsteps:approxx}-\eqref{supp:eq:BUsteps:approxv}. Then we have $\ol{z}_{k+j}=\ol{z}_{k+j}^{(j)}$, and $\ol{z}_{k+j}^{(0)}$ corresponds to taking $j$ steps with exact gradients. By the triangle inequality, we have
\[\|\ol{z}_{k+j}-\ol{z}_{k+j}^{(0)}\|_{a,b}\le \sum_{i=0}^{j-1} \|\ol{z}_{k+j}^{(i+1)}-\ol{z}_{k+j}^{(i)}\|_{a,b}.\]
Using Proposition \ref{supp:prop:Wasserstein}, we have a contraction according to $\|\cdot \|_{a,b}$ with synchronous coupling when using the approximate gradients (because these are exact gradients with respect to a Gaussian), so we have
\[\|\ol{z}_{k+j}^{(i+1)}-\ol{z}_{k+j}^{(i)}\|_{L^4,a,b}\le \| \ol{z}_{k+i+1}^{(1)}-\ol{z}_{k+i+1}^{(0)}\|_{L^4,a,b},\]
which is the one-step error of the approximate gradient scheme versus the exact gradient scheme. 
\begin{align*}
&\| \ol{z}_{k+i+1}^{(1)}-\ol{z}_{k+i+1}^{(0)}\|_{L^4,a,b}\\
&=
\left\|\left(\frac{(1-\eta^{2})h}{\gamma} (\mathcal{Q}(\ol{x}_{k+i}^{(0)}|\ol{x}_k)-\nabla U(\ol{x}_{k+i}^{(0)})),\eta^{2} h (\mathcal{Q}(\ol{x}_{k+i}^{(0)}|\ol{x}_k)-\nabla U(\ol{x}_{k+i}^{(0)})) \right)\right\|_{L^4,a,b}\\
&\le \sqrt{2}\|\ol{x}_{k+i}^{(0)}-\ol{x}_{k}\|_{L^4} \cdot M \left(h^2+\frac{h}{\sqrt{M}}\right).
\end{align*}
So, for $h<\frac{1}{\sqrt{M}}$, we have 
\begin{align*}&\|\ol{z}_{k+j}-\ol{z}_{k+j}^{(0)}\|_{L^4,a,b}
\le 3 h\sqrt{M} \sum_{i=0}^{j-1} \|\ol{x}_{k+i}^{(0)}-\ol{x}_{k}\|_{L^4} \\
&\le 3h\sqrt{M}\cdot j\cdot \sum_{0\le i\le j-1}\|\ol{x}_{k+i+1}^{(0)}-\ol{x}_{k+i}^{(0)}\|_{L^4}\\
\intertext{using Lemma \ref{supp:lem:approxxzdiff},}
&\le 6h^2 M\cdot j\cdot \sum_{0\le i\le j-1}\|\ol{z}_{k+i}^{(0)}-z^*\|_{L^4,a,b} + 6h^2 j^2 \sqrt{d M}\\
\intertext{using Lemma \ref{supp:lemma:V12bnd},}
&\le 48h^2 \sqrt{M}\cdot j\cdot \sum_{0\le i\le j-1}\|\mathcal{V}^{1/2}(\ol{z}_{k+i}^{(0)})\|_{L^4} + 6h^2 j^2 \sqrt{d M}\\
\intertext{using \eqref{supp:eq:Vbndjsteps}}
&\le 48h^2 \sqrt{M}\cdot j^2(\|\mathcal{V}^{1/2}(\ol{z}_{k})\|_{L^4}^4+C_{\mathcal{V}}(h)j)^{1/4} + 6h^2 j^2 \sqrt{d M}.
\end{align*}

We do know that
\begin{align*}
&\|\mathcal{V}^{1/2}(\ol{z}^{(0)}_{k+j})\|_{L^4}^4\leq \left(1-\frac{c_4(h)}{2}\right)^j\|\mathcal{V}^{1/2}(\ol{z}_{k})\|_{L^4}^4+C_{\mathcal{V}}(h),
\end{align*}
so by the $8\gamma$-Lipschitz property of $\mathcal{V}^{1/2}$ in $\|\cdot\|_{a,b}$ by Lemma \ref{supp:lemma:V12bnd}, we have
\begin{align*}
&\|\mathcal{V}^{1/2}(\ol{z}_{k+j})\|_{L^4}\leq \left[\left(1-\frac{c_4(h)}{2}\right)^j\|\mathcal{V}^{1/2}(\ol{z}_{k})\|_{L^4}^4+C_{\mathcal{V}}(h)j\right]^{1/4}\\
&+ 8\gamma \left(48h^2 \sqrt{M}\cdot j^2(\|\mathcal{V}^{1/2}(\ol{z}_{k})\|_{L^4}^4+C_{\mathcal{V}}(h)j)^{1/4} + 6h^2 j^2 \sqrt{d M}\right).
\end{align*}

\end{proof}

\begin{corollary}\label{supp:cor:AG_min_distance}
Consider iterates $(x_{k},v_{k}, \overline{x}_{k})_{k \in \mathbb{N}}$ of approximate gradient UBU, with epoch length $\tau$ and gradient approximation $\mathcal{Q}$ given by \eqref{supp:eq:hessian_grad_approx} approximating a potential $U$ which satisfies Assumptions \ref{supp:assum:LipA} and \ref{supp:assum:convexSG} with $z_{0} \sim \mu_{G}$. Assume that 
\[h < \min\left\{2/\tau\gamma,1,1/2\gamma,\frac{\lambda  \tau}{64(432 \sqrt{M} \tau^2+\gamma(1+\lambda)\tau)}\right\}, \quad \gamma \geq \sqrt{M},\] then 
    \[\|\overline{z}_{k}-z^{*}\|_{L^{4},a,b}\leq  \frac{C(\Tilde{\gamma},\Tilde{m},\Tilde{M})\sqrt{d}}{\sqrt{N_{D}}}.\]
\end{corollary}
\begin{proof}
If we define $b_{k}:= \|\mathcal{V}^{1/2}(\ol{z}_{\tau k})\|_{L^4}$, then for $\gamma \geq \sqrt{8M}$ and $h < \frac{8}{\tau \gamma}$, we have $c_{4}(h) \leq 2/\tau$ (here $c_4(h)$ and $\lambda$ are defined as in \eqref{supp:eq:c4hdef} and \eqref{supp:eq:lambdadef}), and so
\begin{align}
    b_{k+1} &\leq \left[\left(1-\frac{c_{4}(h)}{2}\right)^{\tau}b^{4}_{k} + C_{\mathcal{V}}(h)\tau\right]^{1/4} + 384h^{2}\gamma\sqrt{M}\tau^{2}\left(b^{4}_{k} + C_{\mathcal{V}}(h)\tau\right)^{1/4} + 48\gamma\sqrt{M}h^{2}\tau^{2}\sqrt{d} \nonumber\\
    &\leq \left[\left(1-\frac{c_{4}(h)\tau}{4}\right)b^{4}_{k}  + C_{\mathcal{V}}(h)\tau\right]^{1/4} + 384h^{2}\gamma\sqrt{M}\tau^{2}\left(b^{4}_{k} + C_{\mathcal{V}}(h)\tau\right)^{1/4} + 48\gamma\sqrt{M}h^{2}\tau^{2}\sqrt{d}.\label{supp:eq:bkp1bnd1}
\end{align}
Using this, for $b_{k} < \max\left\{\left(\frac{8C_{\mathcal{V}}(h)}{c_{4}(h)}\right)^{1/4},\sqrt{d}\right\}$ we have that
\begin{align}
    b_{k+1} &\leq (1+384h^{2}\gamma\sqrt{M}\tau^{2})\left[ 
\frac{8C_{\mathcal{V}}(h)}{c_{4}(h)} + d^{2} + C_{\mathcal{V}}(h)\tau\right]^{1/4} + 48\gamma\sqrt{M}h^{2}\tau^{2}\sqrt{d}.\label{supp:eq:bkp1bnd2}
\end{align}
For $b_{k} \ge \max\left\{\left(\frac{8C_{\mathcal{V}}(h)}{c_{4}(h)}\right)^{1/4},\sqrt{d}\right\}$, using that $(1+x)^{1/4}\le 1+\frac{x}{4}$ for $x\in [-1,\infty)$, we have 
\begin{align*}
    b_{k+1} &\leq 
    b_{k}\left[\left[\left(1-\frac{c_{4}(h)\tau}{4}\right)  + \frac{\C_{\mathcal{V}}(h)\tau}{b_k^4}\right]^{1/4} + 384h^{2}\gamma\sqrt{M}\tau^{2}\left(1 + \frac{C_{\mathcal{V}}(h)\tau}{b_k^4}\right)^{1/4} + \frac{48\gamma\sqrt{M}h^{2}\tau^{2}\sqrt{d}}{b_k}\right]\\
    &\leq \left[ 1-\frac{c_{4}(h)\tau}{32}  + 432h^{2}\gamma\sqrt{M}\tau^{2}\right]b_{k}\\
    \intertext{ using the definition $c_4(h)=h\lambda\gamma -8h^{2}\gamma^{2}\left(4 + \lambda\right)$}
    &\le \left[ 1-h \frac{\lambda \gamma \tau}{32}+h^2 (432\gamma\sqrt{M}\tau^{2}+\gamma^2(1+\lambda)\tau)\right]b_{k}
    \intertext{ using the assumption $h\le \frac{\lambda \gamma \tau}{64(432\gamma \sqrt{M} \tau^2+\gamma^2(1+\lambda)\tau)}$
    }
    &\le \left[ 1-h \frac{\lambda \gamma \tau}{64}\right]b_{k}.
\end{align*}
Therefore we have that for all $k \in \mathbb{N}$
\begin{align*}
    b_{k} &\leq \left[ 1-h \frac{\lambda \gamma \tau}{64}\right]^{k}b_{0} + (1+384h^{2}\gamma\sqrt{M}\tau^{2})\left[ 
\frac{8C_{\mathcal{V}}(h)}{c_{4}(h)} + d^{2} + C_{\mathcal{V}}(h)\tau\right]^{1/4}+ 48\gamma\sqrt{M}h^{2}\tau^{2}\sqrt{d}.
\end{align*}
Now considering $b_{k,j} := \|\mathcal{V}^{1/2}(\ol{z}_{\tau k + j})\|_{L^4}$ we have that by the same argument replacing $\tau$ by $j$ that 
\begin{align*}
b_{k,j} &\leq \left[ 1-h \frac{\lambda \gamma j}{64}\right]^{j}b_{k} + (1+384h^{2}\gamma\sqrt{M}j^{2})\left[ 
\frac{8C_{\mathcal{V}}(h)}{c_{4}(h)} + d^{2} + C_{\mathcal{V}}(h)j\right]^{1/4}+ 48\gamma\sqrt{M}h^{2}j^{2}\sqrt{d}.
\end{align*}
Therefore considering the iterates of the approximate gradient UBU scheme we have
\[
\frac{\sqrt{\Tilde{M}N_{D}}}{8}\|\overline{z}_{k}-z^{*}\|_{L^{4},a,b}\leq  C(\Tilde{\gamma},\Tilde{m},\Tilde{M})\left(\|\mathcal{V}^{1/2}(\overline{z}_{0})\|_{L^{4}} + h\tau \sqrt{N_{D}} +\sqrt{d}\right)
\]
and therefore using Lemma \ref{supp:lem:moment_bounds} for the initial distribution we have
\[\|\overline{z}_{k}-z^{*}\|_{L^{4},a,b}\leq  \frac{C(\Tilde{\gamma},\Tilde{m},\Tilde{M})\sqrt{d}}{\sqrt{N_{D}}}.\]
\end{proof}

\begin{proposition}\label{supp:theorem:non-asymptotic-AG2}
For an approximate gradient $\UBU$ integrator with iterates $(z_{k})_{k \in \mathbb{N}}$, transition kernel $P_h$ and a potential $U$ satisfying Assumptions \ref{supp:assum:Lip}-\ref{supp:assum:Hess_Lipschitz} and $z_{0} \sim \mu_{G}$, where we approximate the gradient using the gradient approximation $\mathcal{Q}$ given in \eqref{supp:eq:hessian_grad_approx}.
Consider the continuous solution to kinetic Langevin dynamics $(Z_{t})_{t \geq 0}$, and define $Z^{k}:=Z_{kh}$ for $k \in \mathbb{N}$, where $Z^{0} \sim \pi$ is initialized at the invariant measure with synchronously coupled Brownian motion to $(z_{k})_{k \in \mathbb{N}}$, then for all 
\[h < \min\left\{2/\tau\gamma,1,1/2\gamma,\frac{\lambda  \tau}{64(432 \sqrt{M} \tau^2+\gamma(1+\lambda)\tau)}\right\}, \quad \gamma \geq \sqrt{8M},\]
$k,l \in \mathbb{N}$ such that $k > l$
\begin{align*}
    \|z_{k}-Z^{k}\|_{L^{2},a,b} &\leq (1-c(h))^{k-l}\|z_{l} - Z^{l}\|_{L^{2},a,b} + \frac{h\left((\tau - 1)\sqrt{d}+ \sqrt{N_{D}}\right)C(\Tilde{\gamma},\Tilde{m},\Tilde{M},\Tilde{M}^{s}_{1})\sqrt{d}}{\sqrt{N_{D}}},
\end{align*}
and further
\begin{align*}
\mathcal{W}_{2,a,b}(\mu_{G} P^{k}_{h},\pi)&\leq (1-c(h))^{k-l}\mathcal{W}_{2,a,b}(\mu_{G}P^{l}_{h},\pi) + \frac{h\left((\tau - 1)\sqrt{d}+ \sqrt{N_{D}}\right)C(\Tilde{\gamma},\Tilde{m},\Tilde{M},\Tilde{M}^{s}_{1})\sqrt{d}}{\sqrt{N_{D}}}.
\end{align*}
\end{proposition}

\begin{proof}
Firstly, we introduce the notation $z^{h}_{k}:= (x^{h}_{k},v^{h}_{k}):= \psi_{h}\left(z_{k},h,(W_{t'})^{(k+1)h}_{t'=kh}\right)$ for all $k \in \mathbb{N}$, an iteration of the full gradient scheme with stepsize $h>0$ and initial point $z_{k}$ with synchronously coupled Brownian motion to the approximate gradient scheme. We split up the difference in the following way
\begin{align*}
    \|z_{k} &- Z^{k}\|_{L^{2},a,b} \leq \left\|z_{k} - z^{h}_{k-1}\right\|_{L^{2},a,b} + \|z^{h}_{k-1} - Z^{k}\|_{L^{2},a,b}\\
    &\leq \|z_{k} - z^{h}_{k-1}\|_{L^{2},a,b} + \|\psi(Z^{k-1},h,(W_{t'})^{kh}_{t' = (k-1)h}) - Z^{k}\|_{L^{2},a,b}\\
    &+ \|z^{h}_{k-1} - \psi(Z^{k-1},h,(W_{t'})^{kh}_{t' = (k-1)h})\|_{L^{2},a,b} \\
    &= \textnormal{(I)}' + \textnormal{(II)}' + \textnormal{(III)}'.
\end{align*}
We have by Corollary \ref{supp:cor:AG_min_distance}, Lemma \ref{supp:lem:approxxzdiff} and Lemma \ref{supp:lem:mean:approx:L4} that
\begin{align*}
\textnormal{(I)}' &\leq \frac{\sqrt{2}h}{\sqrt{M}}M^{s}_{1}(\tau-1)\max_{j\leq k-1}\|\overline{x}_{j}-x^{*}\|_{L^{4}}\cdot \max_{j<k-1}\|\overline{x}_{j+1}-\overline{x}_{j}\|_{L^{4}}\\
&\leq (\tau - 1)h^{2}C(\Tilde{\gamma},\Tilde{m},\Tilde{M},\Tilde{M}^{s}_{1})d, \\
\intertext{by the discretization results in Section \ref{supp:app:disc_bounds}}
\textnormal{(II)}'&\leq \Tilde{C}h^{2} \leq \frac{3}{7}\sqrt{d}\left(\sqrt{M} + \gamma \right)h^2
\end{align*}
and
\[
\textnormal{(III)}' \leq (1-c(h))\| z_{k-1} -Z^{k-1}\|_{L^{2},a,b},
\]
where the inequality for $\textnormal{(II)}'$ is shown in Section \ref{supp:app:disc_bounds}.
Therefore going from local to global we have that 
\begin{align*}
    \|z_{k} - Z^{k}\|_{L^{2},a,b}  &\leq (1-c(h))^{k-l}\|z_{l} - Z^{l}\|_{L^{2},a,b} + \frac{h^{2}\left((\tau - 1)d+ \sqrt{N_{D}d}\right)C(\Tilde{\gamma},\Tilde{m},\Tilde{M},\Tilde{M}^{s}_{1})}{c(h)}\\
    &= (1-c(h))^{k-l}\|z_{l} - Z^{l}\|_{L^{2},a,b} + \frac{h\left((\tau - 1)d+ \sqrt{N_{D}d}\right)C(\Tilde{\gamma},\Tilde{m},\Tilde{M},\Tilde{M}^{s}_{1})}{\sqrt{N_{D}}}.
\end{align*}
For non-asymptotic Wasserstein results, we simply replace $Z^{k-1}$ with the continuous dynamics initialized at $\tilde{Z}_{k-1}\sim \pi$ be such that $\|\tilde{Z}_{k-1}-z_{k-1}\|_{L^2,a,b}=\mathcal{W}_{2,a,b}(\mu P^{k-1}_{h},\pi)$ as in \cite{sanz2021wasserstein}[Theorem 23]. We can then apply Lemma \ref{supp:lemma:recursion} to get the required result.
\end{proof}

\subsection{Variance bound of $D_{l,l+1}$}

\begin{proposition}
\label{supp:prop:var_agubu_2}
Suppose two approximate gradient $\UBU$ chains at coarser and finer discretization levels $l$ and $l+1$, with synchronously coupled Brownian motions $\left(z_{k}\right)_{k\in \mathbb{N}}$ and $\left(z'_{k}\right)_{k \in \mathbb{N}}$ and stepsizes $h_{l}$ and $h_{l+1} = h_{l}/2$, satisfying the conditions of Proposition \ref{supp:theorem:non-asymptotic-AG2}, be such that $z_{0} \sim \pi_{0} {= \mu_{G}}$ and $z'_{0} \sim \pi'_{0} = {\mu_{G}(P^{A}_{h_{l+1}})^{B/h_{l+1}}}$. Then for $f$ satisfying Assumption \ref{supp:assum:Lipschitz} we have the following variance bound
\begin{align*}
    &\textnormal{Var}\left(f(z'_{k}) - f(z_{k})\right)\leq\mathbb{E}\left(f(z'_{k}) - f(z_{k})\right)^2 \leq  \\
        &\Bigg(\exp{\left(-\frac{\Tilde{m}\sqrt{N_{D}}kh_{l}}{8\Tilde{\gamma}}\right)}\left( \|z'_{0}-z_{0}\|_{L^{2},a,b} + \mathcal{W}_{2,a,b}(\pi_{0},\pi) + \mathcal{W}_{2,a,b}(\pi'_{0},\pi)\right)\\
    &+ (1-c(h_{l}))^{k}\mathcal{W}_{2,a,b}(\pi_{0},\pi) + (1-c(h_{l+1}))^{2k}\mathcal{W}_{2,a,b}(\pi'_{0},\pi) \\
    &+ \frac{h_{l}\left((\tau - 1)d+ \sqrt{N_{D}d}\right)C(\Tilde{\gamma},\Tilde{m},\Tilde{M},\Tilde{M}^{s}_{1})}{\sqrt{N_{D}}}  \Bigg)^{2}.
\end{align*}
\end{proposition}
\begin{proof}
By following the same argument as Proposition \ref{supp:prop:global_strong_error} using Proposition \ref{supp:theorem:non-asymptotic-AG2} we have the desired result.
\end{proof}

\begin{corollary}\label{supp:cor:var_agubu_2}
    Suppose that the assumptions of Proposition \ref{supp:prop:var_agubu_2} hold for the potential $U$ and $h_{0} > 0$.  Assume that $$h_{0}\leq \frac{C(\Tilde{\gamma},\Tilde{m},\Tilde{M},\tau/N_{D})}{N^{3/2}_{D}},\quad B\ge \frac{16\log(2)\Tilde{\gamma}}{\Tilde{m}h_0 N^{1/2}_{D}}, \quad 
B_0\ge \frac{16\Tilde{\gamma}}{\Tilde{m}N^{1/2}_{D}h_0}\log\left(\frac{1}{N^{3}_{D} h_0^{2}}\right),$$ and the levels are initialized as described in Section \ref{supp:sec:initial_conditions}, Let $l \geq 1$, $1 \leq k \leq K$, and a test function $f$ satisfy Assumption \ref{supp:assum:Lipschitz} then for approximate gradient UBUBU with $\tau = N_{D}$ we have
    \begin{equation}
        \textnormal{Var}(D_{l,l+1}) \leq C(\Tilde{\gamma},\Tilde{m},\Tilde{M},\Tilde{M}^{s}_{1},\tau/N_{D})d^{2}h^{2}_{l}N_{D}.
    \end{equation}
\end{corollary}
\begin{proof}
Following the proof of Proposition \ref{supp:prop:SVRG_Dl_l+1}, but using the results of Proposition \ref{supp:theorem:non-asymptotic-AG2} you get the required result. 
\end{proof}

\begin{proposition}
\label{supp:prop:var_agubu_1}
Suppose a OHO chain at level $0$ using a full gradient Gaussian approximation and a approximate gradient $\UBU$ chain at level $1$, with synchronously coupled Brownian motions $\left(z_{k}\right)_{k\in \mathbb{N}}$ and $\left(z'_{k}\right)_{k \in \mathbb{N}}$ and stepsizes $h_{0}$ and $h_{1} = h_{0}/2$, satisfying the conditions of Proposition \ref{supp:theorem:non-asymptotic-AG2}, be such that $z_{0} \sim \pi_{0}{=\mu_{G}}$ and $z'_{0} \sim \pi'_{0}{=\mu_{G}(P^A_{h_{1}})^{B/h_{1}}}$. Then for $f$ satisfying Assumption \ref{supp:assum:Lipschitz} we have the following variance bound
\begin{align*}
   &\Var\left(f(z'_{k}) - f(z_{k})\right) \leq\mathbb{E}\left[\left(f(z'_{k}) - f(z_{k})\right) ^2\right]\le \mathbb{E}\|z'_{k} - z_{k}\|^{2}_{a,b}\\
   &\leq \Bigg(\exp\left(-\frac{mkh_0}{16\gamma}\right) \left(\|z'_{0}-z_{0}\|_{L^{2},a,b}+\mathcal{W}_{2,a,b}(\pi_0',\pi)\right)
        \\
        &+ (1-c(h_{1}))^{k}\mathcal{W}_{2,a,b}(\pi'_{0},\pi) + \frac{dC(\tilde{\gamma},\Tilde{m},\tilde{M},\tilde{M}^{s}_{1})}{N_{D}}+ C(\Tilde{\gamma},\Tilde{m},\Tilde{M})h_{0}\sqrt{d}\\
        &+ \frac{h_{1}((\tau-1)d + \sqrt{N_{D}d})C(\Tilde{\gamma},\Tilde{m},\Tilde{M},\Tilde{M}^{s}_{1}))}{\sqrt{N_{D}}}\Bigg)^2.
\end{align*}
\end{proposition}
\begin{proof}

We aim to consider the same argument as Proposition \ref{supp:prop:SGKLDCVVariance} using the results from Proposition \ref{supp:theorem:non-asymptotic-AG2}, Lemma \ref{supp:lem:diffusions_difference} and Proposition \ref{supp:theorem:OHO_strongerror}.
\end{proof}
\begin{corollary}\label{supp:cor:var_agubu_1}
    Suppose that the assumptions of Proposition \ref{supp:prop:var_agubu_1} hold for the potential $U$ and $h_{0} > 0$.  Assume that  $$h_{0} \leq \frac{C(\tilde{\gamma},\Tilde{m},\Tilde{M},\tau/N_{D})}{N^{3/2}_{D}},\qquad B\ge \frac{16\log(2)\Tilde{\gamma}}{\Tilde{m}h_0 N^{1/2}_{D}}, \quad 
B_0\ge \frac{16\Tilde{\gamma}}{\Tilde{m}N^{1/2}_{D}h_0}\log\left(\frac{1}{ N^{3}_{D} h_0^{2}}\right),$$ and the levels are initialized as described in Section \ref{supp:sec:initial_conditions}. Let $1 \leq k \leq K$, and a test function $f$ satisfy Assumption \ref{supp:assum:Lipschitz} then for approximate gradient UBUBU with $\tau = N_{D}$ we have
    \begin{equation}
        \textnormal{Var}(D_{0,1}) \leq \frac{C(\Tilde{\gamma},\Tilde{m},\Tilde{M},\Tilde{M}^{s}_{1},\tau/N_{D})d^{2}}{N^{2}_{D}}.
    \end{equation}
\end{corollary}
\begin{proof}
Following the same argument as Corollary \ref{supp:cor:var_agubu_2}, but using Proposition \ref{supp:prop:var_agubu_1}.
\end{proof}

\subsection{Variance bound of $S(c_{R})$}\label{supp:subsection:varianceScRapprox}

\begin{theorem}
\label{supp:thm:AgradUBUBUL4}
Considering UBUBU-Approx method, suppose that Assumptions 
\ref{supp:assum:Lipschitz}, \ref{supp:assum:convexSG}, \ref{supp:assum:Hess_LipschitzSG}, \ref{supp:assum:LipA} hold, and in addition $\gamma \geq \sqrt{8M}$,
\begin{align*} h_0 \leq \frac{C(\Tilde{\gamma},\Tilde{m},\Tilde{M},\tau/N_{D})}{N^{3/2}_{D}}, \quad B\ge \frac{16\log(2)\Tilde{\gamma}}{\Tilde{m}h_0 N^{1/2}_{D}}, \quad 
B_0\ge \frac{16\Tilde{\gamma}}{\Tilde{m}N^{1/2}_{D}h_0}\log\left(\frac{1}{ N^{3}_{D} h_0^{2}}\right).\end{align*}
Suppose that $c_{R}\in [0, \phi_N^{-1/2})$ and $2<\phi_N<4$.
Then for any $N\ge 1$, $S(c_R)$ has finite expected computational cost, $\E S(c_R)=\pi(f)$, and it has finite variance. 
Moreover, it satisfies a CLT as $N\to \infty$, and the asymptotic variance $\sigma^2_S$ defined in \eqref{supp:eq:sigma2S} can be bounded as
\[\sigma^2_S\le \frac{1}{\Tilde{m}N_{D}K} + \frac{C(\Tilde{\gamma},\Tilde{m},\Tilde{M},\Tilde{M}^{s}_{1},\tau/N_{D},\phi_{N})d^{2}}{c_{N}N^{2}_{D}}.\]
\end{theorem}

\begin{proof}
Following the same argument as Theorem \ref{supp:thm:stochgradUBUBU} using Corollaries \ref{supp:cor:var_agubu_2} and \ref{supp:cor:var_agubu_1}.

\end{proof}


\section{Auxiliary results \& RHMC algorithm}
\label{supp:app:disc_bounds}

\begin{lemma}\label{supp:lemma:recursion}
    If we have a sequence of non-negative numbers $(r_{k})_{k \in \mathbb{N}}$ such that for constants $A \in (0,1/2)$, $B,C,D \in \R_{\geq 0}$ such that
    \[
    r^{2}_{k+1} \leq \left(\left(\left(1-A\right)r^{2}_{k} +B\right)^{1/2}+C\right)^{2} + D
    \]
    then 
    \[
    r_{k} \leq \sqrt{2}\left(1-\frac{A}{2}\right)^{k}(r_{0} + \sqrt{B}) + \frac{2\sqrt{2}C}{A} + 2\sqrt{\frac{D + B}{A}}.
    \]
\end{lemma}
\begin{proof}
    If we define $\Tilde{r}_{k} := \sqrt{(1-A)r^{2}_{k} + B}$, then we have that
    \begin{align*}
        \Tilde{r}^{2}_{k+1} &\leq (1-A)\left(\Tilde{r}_{k} + C\right)^{2} + (1-A)D + B\\
        &\leq ((1-A/2)\Tilde{r}_{k} + C)^{2} + D + B.
    \end{align*}
    Then using \cite{DalKar2019}[Lemma 7] we have that
    \begin{align*}
        \Tilde{r}_{k} \leq (1-A/2)^{k}\Tilde{r}_{0} + \frac{2C}{A} + \sqrt{\frac{2(D + B)}{A}},
    \end{align*}
    then 
    \begin{align*}
        r_{k}\sqrt{1-A} &\leq \Tilde{r}_{k} \leq (1-A/2)^{k}\left(r_{0} + \sqrt{B}\right) + \frac{2C}{A} + \sqrt{\frac{2(D + B)}{A}},
    \end{align*}
    and, using the fact that $A \leq 1/2$, we obtain the required result.
\end{proof}

\begin{lemma} \label{supp:lem:moment_bounds}
If a potential $U: \mathbb{R}^{d} \to \mathbb{R}$ is such that $\nabla^{2}U \succ mI$, and $\nabla U(x^{*}) = 0$ then for $x \sim \pi \propto e^{-U(x)}$ we have
\[
\mathbb{E}\left[\left\|x - x^{*}\right\|^{4}\right]^{1/4} \leq 3^{1/4}\sqrt{\frac{d}{m}}\quad \textnormal{and} \quad \mathbb{E}\left[\left\|x - x^{*}\right\|^{8}\right]^{1/8} \leq 105^{1/8}\sqrt{\frac{d}{m}}.
\]
\end{lemma}
\begin{proof}
By using integration by parts and the convexity of $U$ we have that
\begin{align*}
&\int_{x \in \mathbb{R}^{d}}\|x - x^{*}\|^{4}e^{-U(x)}dx \leq \int_{x\in \mathbb{R}^{d}} \sum^{d}_{i=1}\sum^{d}_{j=1}\left(x_{i} - x^{*}_{i}\right)^{2}\left(x_{j}-x^{*}_{j}\right)^{2}e^{-U(x)}dx\\
&\leq d\sum^{d}_{i=1}\int_{x\in \mathbb{R}^{d}} \left(x_{i} - x^{*}_{i}\right)^{4}e^{-U(x)}dx\\
&\leq \frac{d}{m}\sum^{d}_{i=1}\int_{x_{-i}\in \mathbb{R}^{d-1}}\int_{x_{i} \in \mathbb{R}} \left(x_{i} - x^{*}_{i}\right)^{3}\partial_{i}U(x)e^{-U(x)}dx_{i}dx_{-i}\\
&= \frac{3d}{m}\sum^{d}_{i=1}\int_{x_{-i}\in \mathbb{R}^{d-1}}\int_{x_{i} \in \mathbb{R}} \left(x_{i} - x^{*}_{i}\right)^{2}e^{-U(x)}dx_{i}dx_{-i}\\
&\leq \frac{3d}{m^{2}}\sum^{d}_{i=1}\int_{x_{-i}\in \mathbb{R}^{d-1}}\int_{x_{i} \in \mathbb{R}} \left(x_{i} - x^{*}_{i}\right)\partial_{i}U(x)e^{-U(x)}dx_{i}dx_{-i}\\
&= \frac{3d}{m^{2}}\sum^{d}_{i=1}\int_{x\in \mathbb{R}^{d}} e^{-U(x)}dx,
\end{align*}
and similarly, we have
\begin{align*}
&\int_{x \in \mathbb{R}^{d}}\|x - x^{*}\|^{8}e^{-U(x)}dx \\
&\leq \int_{x\in \mathbb{R}^{d}} \sum^{d}_{i,j,k,l=1}\left(x_{i} - x^{*}_{i}\right)^{2}\left(x_{j}-x^{*}_{j}\right)^{2}\left(x_{k}-x^{*}_{k}\right)^{2}\left(x_{l}-x^{*}_{l}\right)^{2}e^{-U(x)}dx\\
&\leq \int_{x\in \mathbb{R}^{d}} d^{3}\sum^{d}_{i=1}\left(x_{i} - x^{*}_{i}\right)^{8}e^{-U(x)}dx\\
&\leq \frac{105d^{3}}{m^{4}}\sum^{d}_{i=1}\int_{x\in \mathbb{R}^{d}} e^{-U(x)}dx,
\end{align*}
as required.

\end{proof}

 \begin{proposition}[Local error bounds for $\UBU$]
 Suppose we have a potential $U$ which satisfies Assumptions \ref{supp:assum:Lip} and \ref{supp:assum:convex}. Let $\phi\left(\xi,h,(W_{t'})^{h}_{t'=0}\right)$ be the solution to the continuous kinetic Langevin dynamics at time $h>0$ with initial condition $\xi \sim \pi$, using Brownian motion $(W_{t'})^{t}_{t'=0}$.  Let $\psi_{h}\left(\xi,h,(W_{t'})^{t}_{t'=0}\right)$ to be the solution of the numerical discretization $\UBU$ with stepsize $h>0$, the same initial condition and Brownian motion. Then we have the following local error bound
 \[
 \|\phi(\xi,h,(W_{t'})^{h}_{t' = 0}) - \psi_{h}(\xi,h,(W_{t'})^{h}_{t' = 0})\|_{L^{2},a,b} \leq \frac{3}{7}\sqrt{d}\left(\sqrt{M} + \gamma \right)h^{2},
 \]
 for $h < \min\left\{\frac{1}{5\sqrt{M}},\frac{1}{2\gamma}\right\}$.
 \end{proposition}
 \begin{proof}
 Using the method of \cite{sanz2021wasserstein} we wish to bound the local error of the $\UBU$ scheme, when initialized at the target measure of the continuous dynamics. When considering \eqref{supp:eq:cont_v} and \eqref{supp:eq:disc_v} we have that for $\xi \sim \pi$
 \[
 \phi(\xi,h,\left(W_{t'}\right)^{h}_{t'=0}) - \psi_{h}(\xi,h,\left(W_{t'}\right)^{h}_{t'=0}) = (\Delta_x,\Delta_v),
 \]
 \[
 \Delta_x = -\int^{h}_{0}\mathcal{F}(h-s)\nabla U(x(s))ds + h\mathcal{F}(h/2)\nabla U(y).
 \]
 and
 \[
    \Delta_v = -\int^{h}_{0}\mathcal{E}(h-s)\nabla U(x(s))ds + h\mathcal{E}(h/2)\nabla U(y).
 \]
 Next, we use the fundamental theorem of calculus
 \begin{align*}
 \mathcal{E}(h-s)\nabla U(x(s)) &= \mathcal{E}(h/2)\nabla U(x(h/2)) \\&+ \int^{s}_{h/2}\left( \mathcal{E}(h-s')\nabla^{2}U(x(s'))v(s') + \gamma \mathcal{E}(h-s')\nabla U(x(s'))\right)ds'.
 \end{align*}
 Then 
 \[
 \Delta_v = - h\mathcal{E}(h/2)\left(\nabla U(x(h/2)) - \nabla U(y)\right) + \Tilde{I}_1
+ \Tilde{I}_2, \]
where
\[
\Tilde{I}_{1} = -\int^{h}
_{0}\int^{s}_{h/2}\mathcal{E}(h-s')\nabla^{2}U(x(s'))v(s')ds'ds,\]
and
\[
\Tilde{I}_{2} = -\int^{h}_{0}\int^{s}_{h/2} \gamma \mathcal{E}(h-s')\nabla U(x(s'))ds'ds.
\]
Hence
\[
\|h\mathcal{E}(h/2)\left(\nabla U(x(h/2)) - \nabla U(y)\right)\|_{L^{2}} \leq \frac{h^{3}M^{3/2}\sqrt{d}}{\sqrt{48}}
\]
from \cite{sanz2021wasserstein}[Eq. 36]. Now, we estimate $\Tilde{I}_1$ as
\begin{align*}
    \mathbb{E}\left(\|\Tilde{I}_{1}\|^{2}\right) &\leq \mathbb{E}\left[ \left(\int^{h}_{0}\left|\int^{s}_{h/2}\mathcal{E}(h-s')^2 ds'\right|ds\right)\times \left(\int^{h}_{0} \left| \int^{s}_{h/2}\|\nabla^{2}U(x(s'))v(s')\|^2ds'\right|ds\right)\right]\\
    &\leq \frac{\mathcal{F}(h)^{2}}{4} \times \frac{h^{2}M^{2}d}{4} \leq \frac{h^{4}M^{2}d}{16},
\end{align*}
and we estimate $\Tilde{I}_2$ as
\begin{align*}
     \mathbb{E}\left(\|\Tilde{I}_{2}\|^{2}\right) &\leq \gamma^2 \mathbb{E}\left[ \left(\int^{h}_{0}\left|\int^{s}_{h/2}\mathcal{E}(h-s')^2 ds'\right|ds\right)\times \left(\int^{h}_{0}\left| \int^{s}_{h/2}\left\|\nabla U(x(s'))\right\|^2 ds'\right|ds \right)\right]\\
    &\leq \gamma^2\frac{\mathcal{F}(h)^{2}}{4} \times \frac{h^{2}Md}{4} \leq \frac{h^{4}M \gamma^2 d}{16},
\end{align*}
then 
\[
\|\Delta_v\|_{L^{2}} \leq \frac{h^{3}M^{3/2}\sqrt{d}}{\sqrt{48}} + \frac{h^{2}M\sqrt{d}}{4} + \frac{h^{2}\gamma \sqrt{Md}}{4}.
\]
Using \cite{sanz2021wasserstein}[Eq 42 Estimate] we get the bound
\[
\|\Delta_x\|_{L^{2}} \leq \frac{h^{3}}{24}\left(\sqrt{3}hM^{3/2} + \left(\frac{\sqrt{42}}{2} + 1\right)M + \gamma M^{1/2}\right)\sqrt{d}.
\]
In the modified Euclidean norm we have
\begin{align*}
   \left\|(\Delta_x,\Delta_v)\right\|_{L^{2},a,b} &\leq \sqrt{\frac{3}{2}}\left(\left\|\Delta_x\right\|_{L^{2}} + \frac{1}{\sqrt{M}}\left\|\Delta_v\right\|_{L^{2}}\right) \\
   &\leq \sqrt{\frac{3d}{2}}h^2\left(\frac{h}{24}\left(\sqrt{3}hM^{3/2} + \frac{9}{2}M + \gamma M^{1/2}\right) + \frac{\sqrt{M}}{4} + \frac{\gamma}{4}\right),
\end{align*}
and under the assumption that $h<\min\{\frac{1}{5\sqrt{M}},\frac{1}{2\gamma}\}$ we see that
\begin{align*}
   \left\|(\Delta_x,\Delta_v)\right\|_{L^{2},a,b} \leq \frac{3}{7}\sqrt{d}\left(\sqrt{M} + \gamma \right)h^2.
\end{align*}
\end{proof}

The following lemma will be bound the variances of $\mathcal{Z}^{(1)}$, $\mathcal{Z}^{(2)}$ and $\mathcal{Z}^{(1)}-\mathcal{Z}^{(2)}$. 

\begin{lemma}\label{supp:lem:Zcovbnd}
For $\mathcal{Z}^{(1)}$ and $\mathcal{Z}^{(2)}$ as defined
\[
\begin{split}
\mathcal{Z}^{(1)}\left(h,\xi^{(1)}\right) &= \sqrt{h}\xi^{(1)},\\
\mathcal{Z}^{(2)}\left(h,\xi^{(1)},\xi^{(2)}\right) &= \sqrt{\frac{1-\eta^{4}}{2\gamma}}\Bigg(\sqrt{ \frac{1-\eta^{2}}{1+\eta^{2}}\cdot \frac{2}{\gamma h}}\xi^{(1)} + \sqrt{1-\frac{1-\eta^{2}}{1+\eta^{2}}\cdot\frac{2}{\gamma h}}\xi^{(2)}\Bigg),
\end{split}
\]
we have
\begin{align*}
\Cov\left(\mathcal{Z}^{(1)}\left(h,\xi^{(1)}\right)\right)&=h I_d,\\
\Cov\left(\mathcal{Z}^{(2)}\left(h,\xi^{(1)},\xi^{(2)}\right)\right)&\preceq h I_d,\\
\Cov\left(\mathcal{Z}^{(1)}\left(h,\xi^{(1)}\right)-\mathcal{Z}^{(2)}\left(h,\xi^{(1)},\xi^{(2)}\right)\right)&\preceq \frac{\gamma h^2}{4} I_d.
\end{align*}
\end{lemma}
\begin{proof}
From the definitions of $\mathcal{Z}^{(1)}$ and $\mathcal{Z}^{(2)}$ it is clear that $\Cov\left(\mathcal{Z}^{(1)}\left(h,\xi^{(1)}\right)\right)=h I_d$ and $\Cov\left(\mathcal{Z}^{(2)}\left(h,\xi^{(1)},\xi^{(2)}\right)\right)=\frac{1-\eta^4}{2\gamma} I_d\preceq h I_d$. For the last claim, we have
\begin{align*}
&\Cov\left(\mathcal{Z}^{(2)}\left(h,\xi^{(1)},\xi^{(2)}\right)-\mathcal{Z}^{(1)}\left(h,\xi^{(1)}\right)\right)
\\&=\left(\sqrt{\frac{1-\eta^{4}}{2\gamma}}\sqrt{\frac{1-\eta^{2}}{1+\eta^{2}}\cdot \frac{2}{\gamma h}}-\sqrt{h}\right)^2
+ \frac{1-\eta^4}{2\gamma}\left(1-\frac{1-\eta^{2}}{1+\eta^{2}}\cdot\frac{2}{\gamma h}\right)\\
&=\frac{1-\eta^{4}}{2\gamma}+h-2\frac{(1-\eta^{2})}{\gamma}=\frac{1-e^{-2\gamma h}-4(1-e^{-\gamma h})+2\gamma h}{2\gamma}\le \frac{(\gamma h)^2}{2} \cdot \frac{1}{2\gamma}\le \frac{\gamma h^2}{4}.
\end{align*}
\end{proof}

\begin{lemma}\label{supp:lem:A123bnd}
Let $A=\sum_{l=1}^{n}A^{(l)}$, with $A^{(l)}\in \R^{d\times d}$ for every $1\le l\le n$. Then we have
\begin{align*}
\|A\|_{\{12\}\{3\}}=\left\|\sum_{i_1,l,m} (A^{(l)}_{i_1,\cdot,\cdot})^T\cdot A^{(m)}_{i_1,\cdot,\cdot}\right\|^{1/2}.
\end{align*}

\end{lemma}
\begin{proof}
This follows by expanding the formula $\|A\|_{\{12\}\{3\}}=\left\|\sum_{i_1} A_{i_1,\cdot,\cdot}^T\cdot A_{i_1,\cdot,\cdot}\right\|^{1/2}$ shown in Lemma 7 of \cite{paulin2024}.
\end{proof}

The following lemma shows some bounds for the gradient-Lipschitz constant $M$ and strongly Hessian Lipschitz constant $M_1^{s}$ for the Bayesian multinomial regression example.
\begin{lemma}\label{supp:lem:BMR}
Consider the Bayesian multinomial regression likelihood of the form,
\begin{equation}
p(y^j|q) = \frac{\exp(\langle x^j, q^{y^j}\rangle)}{\sum_{1\le k\le m} \exp(\langle x^j, q^{k}\rangle)},
\end{equation}
where the posterior potential is given as
\begin{equation}
U(q) = - \log(p_0(q)) - \sum^{N_{D}}_{k=1}\log\left(p(y^j|q) \right),
\end{equation}
with $p_0(q)=\frac{\exp(-\|q\|^2/(2\sigma_0^2))}{(\pi\sigma_0^2)^{d/2}}$. This satisfies the following bounds,
\begin{align*}
&\sup_{q\in \R^d}\|\nabla^2 U(q)\|\le \sigma_0^{-2}+\left\|\sum_{l=1}^{N_D} (x^l)(x^l)^T\right\|,\\
&\sup_{q\in \R^d}\|\nabla^3 U(q)\|_{\{12\}\{3\}}\le 6\left\|\sum_{l=1}^{N_D}\left[(x^l)(x^l)^T  \left(\sum_{m=1}^{N_D}\langle x^{l},x^m\rangle^2\right)\right]\right\|^{1/2}.
\end{align*}
\end{lemma}
\begin{remark}
    If $N_D\to \infty$, and $(x^l)_{1\le l\le N_D}$ are i.i.d. samples from a continuous $d$-dimensional distribution that is non-degenerate with $\E(\|x^l\|^6)=\mathcal{O}(1)$, then we would expect $\|\nabla^2 U(q)\|\propto \frac{N_D}{d}$, and $\|\nabla^3 U(q)\|_{\{12\}\{3\}} \propto \frac{N_D}{d}$.
\end{remark}
\begin{proof}
For $1\le i\le m$, let $E^i=\left(\begin{matrix}&0_{d_{o}}\\ &\vdots \\ &I_{d_{o}} \\ &\vdots \\ &0_{d_{o}}\end{matrix}\right)$ be an $d\times d_{o}$ block matrix with an identity matrix at block $i$. Let 
\[S(x,q)=\sum_{1\le l\le m} \exp(\langle x, q^l\rangle).\]
Let $x\otimes y\otimes z\in \R^{d\times d\times d}$ denote the tensor product of 3 vectors, i.e. $(x\otimes y\otimes z)_{ijk}=x_i y_j z_k$. Then we can express the likelihood term $\log\left(p(y|q) \right)$ and its derivatives as follows.
\begin{align*}
&\log\left(p(y|q) \right)=\langle x, q^{y}\rangle - \log(S(x,q))\\
&\nabla_{q}\log\left(p(y|q) \right)=\sum_{i=1}^{m} (E^i x) \left(\I[y=i] - \exp(\langle x, q^i\rangle)(S(x,q))^{-1}\right)\\
&\nabla^2_{q}\log\left(p(y|q) \right)=\sum_{i,j=1}^{m} (E^i x)(E^j x)^T \exp(\langle x, q^i\rangle+\langle x, q^j\rangle) (S(x,q))^{-2} \\
&-\sum_{i=1}^{m} (E^i x)(E^i x)^T \exp(\langle x, q^i\rangle) (S(x,q))^{-1}\\
&\nabla^3_{q}\log\left(p(y|q) \right)=-\sum_{i=1}^{m} (S(x,q))^{-1} \exp(\langle x, q^i\rangle) (E^i x)\otimes (E^i x) \otimes (E^i x)\\
&+\sum_{i,j=1}^{m} (S(x,q))^{-2} \exp(\langle x, q^i\rangle+\langle x, q^j\rangle)\cdot \\
&\cdot \left((E^i x)\otimes (E^i x) \otimes (E^j x)+(E^i x)\otimes (E^j x) \otimes (E^i x) +(E^i x)\otimes (E^j x) \otimes (E^j x)\right)\\
&-2 (S(x,q))^{-3}\sum_{i,j,k=1}^{m} (E^i x)\otimes (E^j x) \otimes (E^k x)\exp(\langle x, q^i\rangle+\langle x, q^j\rangle+\langle x, q^k\rangle).
\end{align*}
The first claim of the lemma bounding $\|\nabla^2 U(q)\|$ follows from the fact that 
\[0_d \preceq -\nabla^2_{q}\log\left(p(y|q) \right)\preceq \sum_{i}^{m} (E^i x)(E^i x)^T \exp(\langle x, q^i\rangle) (S(x,q))^{-1}\preceq \sum_{i}^{m} (E^i x)(E^i x)^T,\]
here $\preceq$ denotes the semidefinite order.

For the second claim, note that 
\begin{align*}
&\left\|-\sum_{l=1}^{N_D}\sum_{i=1}^{m} (S(x^l,q))^{-1} \exp(\langle x^l, q^i\rangle) (E^i x^l)\otimes (E^i x^l) \otimes (E^i x^l)\right\|_{\{12\}\{3\}}\\
\intertext{using Lemma \ref{supp:lem:A123bnd}}
&\le \left\|\sum_{l,m=1}^{N_D} \langle x^{l},x^m\rangle \left((x^l)(x^m)^T\right)^2 \right\|^{1/2}= \left\|\sum_{l,m=1}^{N_D} \langle x^{l},x^m\rangle^2 (x^l)(x^m)^T \right\|^{1/2}\\
&\le \left\|\frac{1}{2}\sum_{l,m=1}^{N_D} \langle x^{l},x^m\rangle^2 [(x^l)(x^l)^T+(x^m)(x^m)^T] \right\|^{1/2}\\
&=\left\|\sum_{l=1}^{N_D}\left[(x^l)(x^l)^T  \left(\sum_{m=1}^{N_D}\langle x^{l},x^m\rangle^2\right)\right]\right\|^{1/2}.
\end{align*}
The other terms in the sum can be bounded similarly as 
\begin{align*}
&\Bigg\|\sum_{i,j=1}^{m} (S(x,q))^{-2} \exp(\langle x, q^i\rangle+\langle x, q^j\rangle)\cdot \\
&\cdot \left((E^i x)\otimes (E^i x) \otimes (E^j x)+(E^i x)\otimes (E^j x) \otimes (E^i x) +(E^i x)\otimes (E^j x) \otimes (E^j x)\right)\Bigg\|_{\{12\}\{3\}}\\
&\le 3\left\|\sum_{l=1}^{N_D}\left[(x^l)(x^l)^T  \left(\sum_{m=1}^{N_D}\langle x^{l},x^m\rangle^2\right)\right]\right\|^{1/2},\\
&\left\|-2 (S(x,q))^{-3}\sum_{i,j,k=1}^{m} (E^i x)\otimes (E^j x) \otimes (E^k x)\exp(\langle x, q^i\rangle+\langle x, q^j\rangle+\langle x, q^k\rangle)
\right\|_{\{12\}\{3\}}\\
&\le 2\left\|\sum_{l=1}^{N_D}\left[(x^l)(x^l)^T  \left(\sum_{m=1}^{N_D}\langle x^{l},x^m\rangle^2\right)\right]\right\|^{1/2},
\end{align*}
and the claim follows by the triangle inequality.
\end{proof}

\begin{algorithm}[h]
     \footnotesize
     \begin{algorithmic}[1]
    \State \textbf{Input:} \begin{itemize}
 \item stepsize $h$.
  \item Initial distribution $\mu_0$ on $\R^d\times \R^d$.
  \item Potential function $U:\R^d\to \R$ of target distribution.
  \item Number of samples parameter $K$.
  \item Expected number of leapfrog steps parameter $E_{L}\ge 1$.
  \item Partial refreshment parameter $\alpha$.
 \end{itemize}
    \State \textbf{Initialise $(x_0,v_0)\sim \mu_0$.} 
    \For {$i=1,\ldots, K$} 
    \State Sample $L\sim \mathrm{Geom}(1/E_L)$.
    \State \textbf{Perform $L$ leapfrog steps.}
    \State Set $(\tilde{x}_0,\tilde{v}_0):=(x_i,v_i)$.
    \For{$j=0,\ldots, L-1$}
        \State             $\tilde{v}_{j+1/2}:=\tilde{v}_j-\frac{h}{2}\nabla U(\tilde{x}_j)$
      \State $\tilde{x}_{j+1}:=\tilde{x}_j+h\tilde{v}_{j+1/2}$
      \State
      $\tilde{v}_{j+1}:=\tilde{v}_{j+1/2}-\frac{h}{2}\nabla U(\tilde{x}_{j+1})$
     \EndFor        
    \State Let $(x_{i}',v_{i}')=(\tilde{x}_L,\tilde{v}_L)$
    \State \textbf{Compute Hamiltonian.}
    \State $H(x_i,v_i)=U(x_i)+\frac{1}{2}\|v_i\|^2$, $H(x_i',v_i')=U(x_i')+\frac{1}{2}\|v_i'\|^2$.
    \State \textbf{Perform Metropolis-Hastings accept/reject step (with flipping the velocity upon rejection).}
    \State With probability $\min\left[1,\exp(H(x_i,v_i)-H(x_i',v_i'))\right]$, set $(x_{i+1},v_{i+1})=(x_i',v_i')$ (accept proposal).
    \State Otherwise, set $(x_{i+1},v_{i+1})=(x_i,-v_i)$ (reject proposal).
    \State \textbf{Partial velocity refreshment.}
    \State Sample $Z\sim \mathcal{N}(0_{d}, I_d)$ and update $v_{i+1}\to \alpha v_{i+1}+(1-\alpha^2)^{1/2} Z$.
    \EndFor     

 \State \textbf{Output:} 
 \State Samples  $(x_1,v_1),\ldots, (x_K,v_K)$.     
\end{algorithmic}
 	\caption{Randomized Hamiltonian Monte Carlo with Partial Refreshment (RHMC)}
 	\label{supp:alg:RHMC}
\end{algorithm}
\end{appendix}

\bibliographystyle{imsart-number} 
\bibliography{references}




\end{document}